\newif\ifsubmission
\submissiontrue

\newif\ifconference
\conferencefalse

\ifconference
    \documentclass{llncs}
\else
\documentclass[11pt]{article} 
\fi

\usepackage[utf8]{inputenc}
\ifconference\else\usepackage[pdftex,left=1in,top=1in,bottom=1in,right=1in]{geometry}\fi
\usepackage{amsmath, amssymb, dsfont, mathrsfs}
\ifconference \else \usepackage{amsthm} \fi
\usepackage{braket}
\usepackage{graphicx}
\usepackage{caption}
\usepackage{bm}
\usepackage{xcolor}
\usepackage[framemethod=tikz]{mdframed} %
\usepackage[colorlinks=true,citecolor=blue,linkcolor=blue]{hyperref}
\usepackage[nameinlink, noabbrev, capitalize]{cleveref}
\usepackage{physics}
\usepackage{verbatim} 
\usepackage{enumitem}
\usepackage{mathtools}
\hypersetup{
    colorlinks,
    linkcolor={red!50!black},
    citecolor={blue!50!black},
    urlcolor={blue!80!black}
}

\ifconference
\else
\usepackage{tocloft}

\fi

\definecolor{darkgreen}{rgb}{0,0.5,0}
\definecolor{darkblue}{rgb}{0,0,0.6}
\newcommand{\addrefs}{\bibliographystyle{alpha}
\bibliography{000-INCLUDE/bibfiles/abbrev0,000-INCLUDE/bibfiles/crypto,000-INCLUDE/bibfiles/additional}}

\DeclareMathAlphabet{\mathpzc}{OT1}{pzc}{m}{it}

\newcommand{\lrpkegame}[1]{\mathsf{LR-PKE-GAME}_{#1}}

\ifsubmission
\newcommand{\todo}[1]{}
\newcommand{\anote}[1]{}
\newcommand{\takashi}[1]{}
\newcommand{\vipul}[1]{}
\newcommand{\ryo}[1]{}
\newcommand{\fuyuki}[1]{}
\else
\newcommand{\todo}[1]{{\color{blue} \footnotesize(Alper: #1)}}
\newcommand{\anote}[1]{{\color{blue} \footnotesize(Alper: #1)}}
\newcommand{\takashi}[1]{{\color{purple} \footnotesize(Takashi: #1)}}
\newcommand{\vipul}[1]{{\color{red} \footnotesize(Vipul: #1)}}
\newcommand{\ryo}[1]{{\color{darkgreen} \footnotesize(Ryo: #1)}}
\newcommand{\fuyuki}[1]{{\color{chocolate} \footnotesize(Fuyuki: #1)}}
\fi

\ifconference
\else
\newtheorem{theorem}{Theorem}
 
 \newtheorem{claim}{Claim}
 \newtheorem{lemma}{Lemma}
 
 \newtheorem{corollary}{Corollary}
 
 \newtheorem{definition}{Definition}
 
 \newtheorem{remark}{Remark}
 
\fi

\newcommand{\C}{\mathds{C}}
\newcommand{\N}{\mathds{N}}
\newcommand{\BN}{\mathbb{N}}

\newcommand{\Z}{\mathbb{Z}}
\newcommand{\R}{\mathbb{R}}
\newcommand{\E}{\mathds{E}}
\newcommand{\cA}{\mathcal{A}}
\newcommand{\cC}{\mathcal{C}}

\newcommand{\cD}{\mathcal{D}}
\newcommand{\cF}{\mathcal{F}}

\newcommand{\cI}{\mathcal{I}}
\newcommand{\cP}{\mathcal{P}}

\newcommand{\cU}{\mathcal{U}}

\newcommand{\cX}{\mathcal{X}}
\newcommand{\cY}{\mathcal{Y}}

\newcommand{\pk}{\mathsf{pk}}
\newcommand{\sk}{\mathsf{sk}}

\newcommand{\pke}{\mathsf{PKE}}

\newcommand{\dss}{\mathsf{SIG}}
\newcommand{\cdss}{\mathsf{CSIG}}
\newcommand{\prf}{\mathsf{PRF}}
\newcommand{\hyb}{\mathsf{Hyb}}

\newcommand{\samp}{\leftarrow}
\newcommand{\reg}{\qreg{R}}
\newcommand{\regi}[1]{\reg_{\mathsf{#1}}}

\newcommand{\io}{i\mathcal{O}}
\newcommand{\minibank}{\mathsf{MiniBank}}
\newcommand{\ofe}{\mathsf{1FE}}

\newcommand{\npke}{\mathsf{NPKE}}

\newcommand{\biti}[2]{\mathsf{Bit}_{#1}(#2)}
\newcommand{\prefii}[2]{\mathsf{Prefix}_{#1}(#2)}
\newcommand{\csrd}{C_{\mathsf{SRD}}}
\newcommand{\ocsrd}{C_{\mathsf{OSRD}}}
\newcommand{\crcons}{\mathfrak{S}}

\newcommand{\cR}{\mathcal{R}}
\newcommand{\cM}{\mathcal{M}}
\newcommand{\cE}{\mathcal{E}}
\newcommand{\poly}{\mathsf{poly}}

\newcommand{\bit}{\{0,1\}}

\newcommand{\Enc}{\mathsf{Enc}}
\newcommand{\Dec}{\mathsf{Dec}}

\newcommand{\Ver}{\mathsf{Verify}}

\newcommand{\zo}{\{0, 1\}}

\newcommand{\negl}{\mathsf{negl}}

\newcommand{\cN}{\mathcal{N}}

\newcommand{\enc}{\mathsf{Enc}}
\newcommand{\dec}{\mathsf{Dec}}

\newcommand{\adve}{\mathcal{A}}

\newcommand{\Sign}{\mathsf{Sign}}

\newcommand{\fe}{\mathsf{FE}}
\newcommand{\bank}{\mathsf{Bank}}

\newcommand{\cfgame}[1]{\mathsf{PKQM-CF}_{#1}}

\newcommand{\bzgame}[1]{\mathsf{BZ-GAME}_{#1}}

\newcommand{\advb}[3]{\mathsf{Adv}_{#1}^{\mathsf{#2} \mbox{-} \mathsf{#3}}}
\newcommand{\advc}[4]{\mathsf{Adv}_{#1}^{\mathsf{#2} \mbox{-} \mathsf{#3} \mbox{-} \mathsf{#4}}}
\newcommand{\advd}[5]{\mathsf{Adv}_{#1}^{\mathsf{#2} \mbox{-} \mathsf{#3} \mbox{-} \mathsf{#4} \mbox{-} \mathsf{#5}}}
\newcommand{\expb}[3]{\mathsf{Exp}_{#1}^{ \mathsf{#2} \mbox{-} \mathsf{#3}}}
\newcommand{\expc}[4]{\mathsf{Exp}_{#1}^{ \mathsf{#2} \mbox{-} \mathsf{#3} \mbox{-} \mathsf{#4}}}
\newcommand{\expd}[5]{\mathsf{Exp}_{#1}^{\mathsf{#2} \mbox{-} \mathsf{#3} \mbox{-} \mathsf{#4} \mbox{-} \mathsf{#5}}}
\newcommand*{\algo}[1]{\ensuremath{\mathsf{#1}}}
\newcommand{\seteq}{\coloneqq}
\newcommand{\setbk}[1]{\{#1\}}
\newcommand{\mat}[1]{\boldsymbol{#1}}
\newcommand{\chosen}{\leftarrow}
\newcommand{\lrun}{\leftarrow}
\newcommand{\ra}{\rightarrow}
\newcommand{\concat}{\|}
\newcommand{\tensor}{\otimes}
\newcommand{\cB}{\mathcal{B}}
\newcommand{\qB}{\cB}
\newcommand{\qA}{\cA}
\newcommand{\qC}{\cC}
\newcommand{\qD}{\mathsf{q}\cD}

\newcommand{\sfmode}{\mathsf{mode}}
\newcommand{\cind}{\approx}
\newcommand{\projimp}{\algo{ProjImp}}

\newcommand*{\qreg}[1]{{\color{gray}{\mathsf{#1}}}}

\newcommand{\coin}{\mathsf{coin}}
\newcommand{\secp}{\lambda}
\newcommand{\ct}{\mathsf{ct}}
\newcommand*{\qstate}[1]{\mathpzc{#1}}
\newcommand{\qsk}{\qstate{sk}}
\newcommand{\msk}{\mathsf{msk}}
\newcommand*{\msg}{\mathsf{m}}
\newcommand{\ek}{\mathsf{ek}}
\newcommand{\dk}{\mathsf{dk}}

\newcommand{\One}{\mathsf{One}}
\newcommand{\one}{\mathsf{One}}

\newcommand{\event}[1]{\mathbf{#1}}

\newcommand{\FE}{\mathsf{FE}}
\newcommand{\PKFE}{\mathsf{PKFE}}
\newcommand{\RE}{\algo{RE}}
\newcommand{\Setup}{\algo{Setup}}
\newcommand{\KG}{\algo{KeyGen}}
\newcommand{\Gen}{\algo{Gen}}

\newcommand{\sffe}{\mathsf{fe}}
\newcommand{\fsk}{\mathsf{fsk}}

\newcommand{\qKG}{\algo{KG}}
\newcommand{\qDec}{\algo{Dec}}
\newcommand{\PRF}{\algo{PRF}}
\newcommand{\prfF}{\mathsf{F}}
\newcommand{\prfgen}{\PRF.\Gen}
\newcommand{\prfkey}{\mathsf{K}}
\newcommand{\Puncture}{\algo{Puncture}}
\newcommand{\PuncPRF}{\algo{PPRF}}

\newcommand\SDE{\algo{SDE}}

\newcommand{\sde}{\mathsf{sde}}

\newcommand\UE{\algo{UE}}

\newcommand{\ue}{\mathsf{ue}}

\newcommand{\TI}{\algo{TI}}
\newcommand{\ATI}{\algo{ATI}}
\newcommand{\API}{\algo{API}}

\newcommand{\Convert}{\algo{Convert}}
\newcommand{\shiftdis}[1]{\Delta_{\mathsf{Shift}}^{#1}}

\newcommand{\genstate}{\mathsf{GenState}}
\newcommand{\prs}{\mathsf{PRS}}

\newcommand{\qstateq}{\rho}
\newcommand{\qalgo}{\mathsf}

\newenvironment{boxfig}[2]{\begin{figure}[#1]\fbox{\begin{minipage}{0.97\linewidth}
                        \vspace{0.2em}
                        \makebox[0.025\linewidth]{}
                        \begin{minipage}{0.95\linewidth}
            {{
                        #2 }}
                        \end{minipage}
                        \vspace{0.2em}
                        \end{minipage}}
                        }
                        {\end{figure}}

\newcommand{\pprotocol}[4]{
\begin{boxfig}{h!}{\footnotesize 
\centering{\textbf{#1}}
    #4
\vspace{0.2em} } \caption{\label{#3} #2}
\end{boxfig}
}

\newcommand{\protocol}[4]{
\pprotocol{#1}{#2}{#3}{#4} }

\DeclareMathOperator*{\Exp}{\mathbb{E}}

\let\originalleft\left
\let\originalright\right
\renewcommand{\left}{\mathopen{}\mathclose\bgroup\originalleft}
\renewcommand{\right}{\aftergroup\egroup\originalright}

\begin{document}
\allowdisplaybreaks	

\title{Multi-Copy Security in Quantum Cryptography and More}
\date{}

\ifconference

\author{Anonymous Submission}

\else
\author{Alper \c{C}akan\thanks{Carnegie Mellon University.  Part of the work done while the author was visiting NTT Social Informatics Laboratories, Tokyo, Japan. \texttt{acakan@andrew.cmu.edu}.} \and Vipul Goyal\thanks{NTT Research \& Carnegie Mellon University.  \texttt{vipul@vipulgoyal.org}}  \and Fuyuki Kitagawa\thanks{NTT Social Informatics Laboratories, Tokyo, Japan. \texttt{fuyuki.kitagawa@ntt.com}.}  \and Ryo Nishimaki \thanks{NTT Social Informatics Laboratories, Tokyo, Japan. \texttt{ryo.nishimaki@ntt.com}.} \and Takashi Yamakawa\thanks{NTT Social Informatics Laboratories, Tokyo, Japan. \texttt{takashi.yamakawa@ntt.com}.}}
\fi

\maketitle

\begin{abstract}
Unclonable cryptography leverages the quantum no-cloning principle to achieve strong security guarantees that are impossible to achieve in a classical world. Most existing works in this area only consider the basic single-copy security, and there been only a few works that achieve security in the more realistic model of \emph{collusion-resistance} (where adversary receives multiple keys), which is the gold standard in cryptography. Further, the few existing works that do consider collusion-resistance have convoluted solutions that are highly tailored to their own applications, with little hope to generalize, and they often re-invent the tools from both single-key quantum cryptography as well as collusion-resistant classical cryptography. Moreover, the question of \emph{multi-copy security}, where the adversary receives multiple copies of the same quantum state (rather than merely getting multiple independently sampled states) is almost completely open.

In this work, we develop a large toolset of generic compilers and technical lemmata for dealing with collusion-resistance and multi-copy security in quantum cryptography. Using our toolset, we obtain a large number of new feasibility results, with proofs that are significantly \emph{simpler} than the existing proofs in literature.

As part of our toolset, we introduce a generic compiler that upgrades single-key secure quantum protection (copy-protection/LOCC leakage-resilience/secure leasing) schemes for decryption keys to collusion-resistant secure schemes. Then, we also introduce a generic compiler that upgrades collusion-resistant primitives to achieve multi-copy security, assuming only one-way functions.

As part of our technical lemmata, we show the collusion-resistant analogue of the \emph{one-way-to-hiding (O2H) lemma}, a quantum-state analogue of the \emph{small-range-distributions lemma}, a \emph{quantum pigeonhole lemma} for entangled adversaries and the first deterministic signature scheme with quantum-query security. We also show that independent-challenge security implies identical-challenge security in collusion-resistant copy-protection search games, and thus we obtain the first schemes that satisfy the latter security notion.

Finally, using our toolset,  we obtain a large number of new feasibility results. For example, we obtain the first multi-copy secure constructions of public-key quantum money (termed \emph{quantum coins}), single-decryptor encryption (SDE), unclonable encryption, and more. We obtain the first collusion-resistant secure key-leasing scheme with a fully classical vendor. We also obtain the first LOCC leakage-resilient PKE scheme with multi-copy security, thus making progress towards achieving \emph{quantum key-fire} in the plain model.

\vfill\par\textbf{Keywords:} Quantum cryptography, copy-protection
\end{abstract}

\newpage
\tableofcontents
\newpage

\section{Introduction}\label{sec:intro}
Starting with the seminal work of Wiesner (\cite{Wie83}), cryptographers have realized that quantum information can help us achieve security guarantees that are simply impossible in a classical world. While initially these guarantees were mainly related to hardness of simply making more copies of a quantum state (e.g. quantum money), a long sequence of works have shown various more advanced applications where the quantum state also encodes some functionality; such as unclonable ciphertexts (\cite{TQC:BroLor20}), copy-protection of software/cryptographic keys (\cite{Aar09}), unbounded/LOCC leakage-resilience (\cite{TCC:CGLR24}), and secure leasing of keys (\cite{AP21}).

\paragraph{Single-Key vs Collusion-Resistance} While there have been a large amount of works achieving significant results related to these various uniquely-quantum security notions (\cite{C:CLLZ21,cryptoeprint:2025/1197,cryptoeprint:2025/1880,cryptoeprint:2024/1564,TCC:CGLR24}), unfortunately, except for a handful number of works, most research so far has focused on the simpler \emph{single-key setting}  where we assume the adversary obtains a single quantum object. For example, in the case of copy-protection, most works consider an adversary who receives a single copy of a copy-protected software as a quantum state, and tries to create two copies of it. Unfortunately, this security model is highly unrealistic: In real-life, huge numbers of users can own the same program (think Microsoft Excel), and it is unlikely that they will not collaborate to pirate the software. In fact, in most existing schemes, even just two users collaborating can easily convert the protected software into an untraceable classical string (which is easily distributable). Similarly, in the setting of key leasing, we can consider users who subscribe to a content provider (such as a pay-to-view TV channel, like HBO) and we want to achieve a security notion where after they return their key at the end of their membership, they should not be able to decrypt future content. Again, the notion of single-key security is simply unrealistic. Similarly, in unclonable encryption most existing works consider only a single ciphertext, and in unbounded/LOCC leakage resilience most works consider a single copy of the secret quantum keys being stored - which are far from the real-life scenarios.

\paragraph{State of Affairs} To summarize, single-key security is simply unrealistic, and indeed collusion-resistant security models have been the gold standard in classical cryptography. However, there have been only a limited amount of works (\cite{TCC:CakGoy24,cryptoeprint:2025/1842,kitagawa2025pke,PRV24,ITC:AnaMutPor25}) dealing with collusion-resistance in quantum cryptography; thus, in most settings, collusion-resistance has not been achieved. Moreover, some of these works have results in very restricted models, and all of these existing works have convoluted approaches that are highly tailored to their setting. On top of the hardness of achieving advanced quantum cryptography notions even in the single-key setting, these works in addition deal with the challenges of collusion-resistance. Thus, they usually re-invent the techniques from both the single-key quantum cryptography as well as techniques from classical collusion-resistance; and they have little generalizability to other open questions in collusion-resistant quantum cryptography. 

\paragraph{Modular Approaches} In light of the above, we ask the following question.
\begin{quote}
    \begin{center}\textit{
    Can we achieve collusion-resistance in quantum cryptography in a generic way from single-key quantum schemes and collusion-resistant classical schemes?    
    }\end{center}
\end{quote}
Note that this would allow us to obtain collusion-resistant quantum constructions by tapping into the huge amount of literature dealing with collusion-resistance in classical cryptography, as well as the literature on single-key security in quantum cryptography; thus stopping us from reinventing the wheel. It would also allow us to inherit any future developments to the underlying quantum and classical schemes (e.g. constructions from better assumptions).

\paragraph{Feasibilities} Aside from a more modular approach, we also note that, as discussed above, many questions in collusion-resistant quantum cryptography have not been solved. We mention some of these now. For example, consider the following problem, which we will refer to as the \emph{broadcast problem}. Suppose a broadcaster/content-provider would like to securely sell subscriptions to their broadcast, and they want to ensure that any user who has terminated their membership no longer has access to the future content. A classical subscriber can keep a copy of their decryption key, so returning or deleting the key cannot by itself prevent access to future broadcasts encrypted under the same public key. However, quantumly, the broadcaster can securely lease the decryption key for the broadcast to subscribers, and through quantum mechanics, will be assured that security is achieved once a user destroys/returns their quantum key. Indeed, the broadcast problem has been used as one of the main motivating examples for secure leasing since the latter's introduction. However, surprisingly, there is no satisfying solution so far to this natural problem: Existing key leasing schemes either need the broadcaster to send the keys through a quantum channel to the users (\cite{kitagawa2025pke}), or provide security only under the unrealistic assumption that not even two adversaries collude (\cite{EC:CGJL25}).

Another question is unclonable encryption, where existing schemes either lose their security as soon as multiple ciphertexts are available to the adversary (\cite{TQC:BroLor20,C:AKLLZ22}), or are in extremely restricted oracle models where the adversary can only make a single query to the oracle (\cite{PRV24}). Yet another question is unbounded/LOCC leakage-resilient keys: Existing schemes (\cite{TCC:CGLR24}) lose their security when the adversary obtains leakage from two keys, whereas in real-life the same key is usually stored in multiple servers. Note that this question is also related to the recently introduced notion of \emph{quantum key-fire} (\cite{ccakan2025public}), which considers unbounded/LOCC leakage-resilient keys that are clonable\footnote{Collusion-resistance/multi-copy security is a prerequisite for quantum key-fire since adversary can create multiple copies himself even when he is given only a single key.}.
\begin{quote}
    \begin{center}\textit{
Can we solve the broadcasting problem (classical-lessor collusion-resistant key leasing)?\\Can we achieve collusion-resistant unclonable encryption?\\ Can we achieve collusion-resistant LOCC leakage-resilience and make progress towards quantum-key fire?
    }\end{center}
\end{quote}

\paragraph{Multi-Copy Security} Another related question is the question of \emph{multi-copy security}, where we insist that the multiple keys/ciphertexts and so on that are given to the adversary are not samples from the same mixed state, but rather they are exactly the same pure state. Multi-copy security\footnote{Note that this is a stronger notion than collusion-resistance: Any multi-copy scheme is automatically collusion-resistant, but not vice versa.} has interesting applications in, for example, secure leasing\footnote{If the keys are copies of the same pure state, a user with multiple keys can sublease them, and do verification upon return by comparing to his own keys through SWAP test} and quantum money\footnote{Identical banknotes provide anonymity against \emph{semi-honest} adversaries that try to track the banknotes without modifying them, and it also allows us to construct upgradable quantum coins (\cref{sec:upgrade})}, as well as being interesting from a theoretical perspective and having connections to physics, such as to \emph{black-holes} (\cite{PRV24}). Unfortunately, the situation for multi-copy security is even more bleak than collusion-resistance: Despite being considered for a long time (at least since 2012\footnote{For example, quantum coins are mentioned as an open question in \cite{Aaronson2012QuantumMoneyACM}}), even simple questions like multi-copy secure public-key quantum money (also called \emph{quantum coins}) have been open.
\begin{quote}
    \begin{center}\textit{
Can we achieve multi-copy security in quantum cryptography?
    }\end{center}
\end{quote}

\subsection{Our Results}
In this paper, we answer all of these questions affirmatively: We develop a large toolset of compilers and technical lemmata for dealing with collusion-resistance and multi-copy security in quantum cryptography. Using our toolset, we obtain a large number of new feasibility results with generic constructions, with proofs that are significantly \emph{simpler} than the convoluted proofs of the existing results in literature (which also have \emph{weaker} statements)! 

Our new feasibility results include
\begin{itemize}
\item The first public-key quantum coin scheme, assuming indistinguishability obfuscation (iO) and one-way functions (OWF)
    \item The first collusion-resistant and multi-copy secure unclonable encryption (UE) schemes, with both search (in the plain model, assuming iO) and indistinguishability (in QROM, unconditional) security
    \item The first collusion-resistant classical-lessor secure key leasing scheme (assuming LWE in the bounded-collusion case, and LWE and functional encryption in the general case), solving the \emph{broadcast problem} discussed above,
    \item The first collusion-resistant classical-vendor copy-protection schemes for decryption keys (assuming subexponential iO and subexponential LWE),
    \item The first multi-copy secure bounded quantum and LOCC leakage-resilient public-key encryption (PKE) schemes (assuming iO and one-way functions), marking the first progress towards quantum key-fire (\cite{ccakan2025public}) in the plain model\footnote{Quantum key-fire implies multi-copy secure LOCC leakage-resilient PKE}.
\end{itemize}

Below, we give a more detailed summary of our feasibility results and of our toolset. We believe our toolset will be helpful for future work dealing with collusion-resistance/multi-copy security.

\subsubsection{Generic Compilers}
\paragraph{Quantum-Protection Collusion-Resistance Compiler (\cref{sec:fe-compiler})} We give a \emph{generic}\footnote{Technically, we use the circuit of the classical encryption algorithm of the single-key scheme. However, its key generation and decryption are used in a black-box way.} compiler that uses classical functional encryption (FE) to achieve collusion-resistance for quantum protection constructions for decryption keys (copy-protection, LOCC leakage-resilience, secure leasing) from single-key schemes. Our compiler is \emph{quantum-key-preserving} which means the quantum part of the key in the resulting scheme is the same as the original scheme, thus it preserves various properties of the single-key scheme such as classical-leasing/vendor.

Since bounded-collusion-resistant FE can be obtained from any PKE, this also allows us to achieve bounded-collusion-resistance for quantum protection. 

Note that since our compilers use the single-key scheme's decryption in a black-box way, our schemes will inherit any improvements (e.g. efficiency, assumptions) to single-key schemes.

\paragraph{Purification Compiler (\cref{sec:purification_compiler})} We show how to compile in a black-box way any \emph{quantum-state game} that involves multiple samples from a mixed state into a pure-state (i.e. multi-copy) game while preserving security, assuming existence of one-way functions. This means we can securely compile collusion-resistant schemes into multi-copy schemes. We require the state-generation algorithm to have classically determined outputs: its classical input and randomness fully determine its pure output state.

\subsubsection{New Feasibilities}
Using our compiler and our technical tools (discussed below), we obtain a large set of new feasibility results.

\paragraph{Quantum Money} (\cref{sec:coincomp,sec:upgrade})
\begin{itemize}
    \item \cref{cor:qcoin} (informal): We obtain the first public-key quantum coin scheme, assuming  a single-banknote public-key quantum money scheme (which can be instantiated with iO and one-way functions). We actually show security of two different purifications: First, in \cref{sec:money_PR_compiler}, we use our compiler from \cref{sec:purification_compiler} based on pseudorandom states, and then in \cref{sec:money_eqsup_compiler} we also show the security of the conjectured folklore construction (where the purified banknote is equal-superposition of individual banknotes).
    \item In \cref{sec:upgrade}, using quantum coins, we show the existence of a notion called \emph{upgradable coins} where the security initially holds relative to a weak assumption such as public-key encryption (but verification is \emph{almost-public} \cite{bs21}) and later when a strong cryptographic primitive (such as subspace hiding obfuscation) becomes available, bank only needs to publish a classical value to switch to full public-key verification (and banknotes do not need to be changed). This can be considered an intermediate solution to the problem of public-key quantum money without obfuscation.
\end{itemize}

\paragraph{Unclonable Encryption}(\cref{Sec:UE})
\begin{itemize}
    \item \cref{thm:UE_IND_QROM,thm:ue_ind} (informal):  We obtain the first multi-copy indistinguishability secure UE in QROM (unconditionally). Previous indistinguishability secure results were also in QROM but they were only in the single-ciphertext setting (\cite{C:AKLLZ22}).
    \item \cref{thm:ue,thm:bounded_ue} (informal): We obtain the first multi-copy search secure UE in the plain model (assuming LWE for bounded, and one-way functions and iO for unbounded number of copies). Previous results were only in the single-ciphertext setting (\cite{TQC:BroLor20}), or significantly restricted structured oracle models where the adversary receives a bounded number of copies and can only make a single query to the oracle (\cite{PRV24}).
\end{itemize}

\paragraph{Copy Protection} (\cref{sec:CR-transformation})
\begin{itemize}
    \item \cref{cor:cr_SDE_from_IO} (informal): We obtain the first \emph{multi-copy} secure public-key single-decryptor encryption (SDE)\footnote{Encryption scheme with copy-protected decryption keys} scheme, from polynomially secure iO and one-way functions. Previously, even collusion-resistant constructions results required subexponentially secure iO and subexponential LWE (\cite{TCC:CakGoy24}), and there was no multi-copy secure construction.
    \item \cref{cor:classical-vendor-sde} (informal): We obtain the first collusion-resistant copy-protection scheme with a classical-vendor\footnote{This means that the copy-protected key can be generated through classical interaction with the owner (who is fully classical) of the master key. This is also called \emph{semi-quantum copy-protection}}, from subexponential iO and subexponential LWE. Previous classical-vendor schemes had the same assumptions but were secure only in the single-key setting (\cite{chevalier2023semi}).
    \item \cref{cor:cr_SKSDE_from_LWE} (informal): We obtain the first bounded-multi-copy secret-key\footnote{We do not allow access to an encryption oracle.} SDE scheme assuming LWE. Previous secret-key schemes but were secure only in the single-key setting (\cite{TCC:KitNis23}).
    \item \cref{thm:ind-to-iden-sde}  (informal):  We show that any SDE scheme that satisfies \emph{strong} search security (\cref{def:cr_strong_anti_piracy}) also satisfies identical-challenge search security in the collusion-resistant/multi-copy setting. This gives us identical-challenge search security in all the cases above. Previous works had identical-challenge security only for the single-key setting (\cite{C:AnaKalLiu23}).
\end{itemize}

\paragraph{Secure Leasing (\cref{sec:skl})} 
\begin{itemize}
\item \cref{cor:skl-cl} (informal): We obtain the first collusion-resistant PKE with classical-lessor secure key leasing, assuming functional encryption and hardness of LWE. Moreover, our scheme uses only $2$-rounds and thus it is round-optimal. This solves the \emph{broadcast problem} mentioned in \cref{sec:intro}. Previously, no collusion-resistant classical-lessor scheme was known, and even achieving only classical-deletion required iO (\cite{kitagawa2025pke}).

\item \cref{cor:skl-cdel} (informal): We obtain the first \emph{multi-copy} secure PKE with secure key leasing with classical deletion, assuming any classical FE. Previously only collusion-resistance was known (assuming iO) with a non-black-box approach. 

\item \cref{cor:skl-bdd} (informal):  We obtain the first bounded-multi-copy secure PKE with secure key leasing with classical deletion, assuming merely a classical PKE scheme. Instead assuming LWE, we obtain the first bounded-collusion-resistant PKE with classical-lessor key leasing (with only $2$-rounds). Previously, constructions from PKE (in quantum lessor setting) or LWE (in classical lessor setting) were only known for single-key security \cite{cryptoeprint:2024/1564,EC:CGJL25}.
\end{itemize}

\paragraph{LOCC Leakage-Resilience (\cref{sec:lr})} 
\begin{itemize}
    \item \cref{cor:leakpure} (informal): We obtain the first schemes with security against LOCC and bounded-quantum leakage in the collusion-resistant/multi-copy setting, assuming polynomially secure iO and one-way functions. Previously only single-key secure schemes were known, with the same assumptions as ours (\cite{TCC:CakGoy24,qleak}).
\end{itemize}

We also develop several technical tools for collusion-resistance and multi-copy security, which we discuss in \cref{sec:tech-overview}.

\subsection{Organization}
The remainder of the paper is organized as follows. In \cref{sec:tech-overview}, we give an overview of our techniques. In \Cref{sec:preliminaries}, we provide preliminaries by reviewing some relevant definitions and known results. 

\paragraph{Tools}
In \Cref{sec:purification_compiler}, we present our purification compiler, which upgrades collusion-resistance into multi-copy security.

In \cref{sec:fe-compiler}, we give a collusion-resistance compiler for quantum protection notions (copy-protection, secure leasing, LOCC leakage-resilience), based on classical FE.

In \Cref{sec:detsig}, we give a quantum-query-secure deterministic digital signature, which is used in the security proof of the folklore construction of quantum coins (\cref{sec:money_eqsup_compiler}). In \cref{sec:readonce}, we give a quantum-state analogue of the small-range distributions technique, which is also used in the same proof. 

In \cref{sec:lemma_collusion_resistance}, we give our simultaneous distinguisher finding theorem.
In \cref{sec:o2h}, we give our collusion-resistant one-way to hiding theorem. In \cref{sec:relating}, we give our result relating identical- and independent-challenge security in the collusion-resistant setting.

\paragraph{Quantum Money} In \Cref{sec:coincomp}, we present two constructions of quantum coins. The first (\Cref{sec:money_PR_compiler}) is obtained by directly applying the compiler developed in \Cref{sec:purification_compiler}, while the second (\Cref{sec:money_eqsup_compiler}) is the folklore ``equal-superposition'' construction whose security we establish using the tools developed in \Cref{sec:readonce,sec:detsig}.
In \Cref{sec:upgrade}, we introduce upgradable quantum coins and present its construction.

\paragraph{Quantum Protection Schemes} In \Cref{sec:CR-transformation}, we show single-decryptor encryption schemes. More specifically, we prove the security of our collusion-resistant quantum protection compiler (from \cref{sec:fe-compiler}) in the setting of copy-protection, thus showing generically how to convert a single-key SDE into collusion-resistant SDE, which is then compiled into multi-copy secure SDE by the compiler in \Cref{sec:purification_compiler}. 

In \cref{sec:lr} and \cref{sec:skl}, we prove the security of our collusion-resistant quantum protection compiler in the settings of LOCC leakage-resilience and secure leasing.

\paragraph{Unclonable Encryption} In \Cref{Sec:UE}, we demonstrate a conversion from collusion-resistant SDE to multi-challenge secure UE, which is then compiled into multi-copy secure UE by the compiler in \Cref{sec:purification_compiler}. We also show how to construct collusion-resistant/multi-copy UE in the QROM unconditionally.

\section*{Acknowledgments}
Generative AI tools were used for language editing. 
A\c{C} was supported by the following grants of VG: NSF award 1916939, DARPA SIEVE program, a gift from Ripple, a DoE NETL award, a JP Morgan Faculty Fellowship, a PNC center for financial services innovation award, and a Cylab seed funding award. Part of the work done while A\c{C} was visiting  NTT Social Informatics Laboratories, Tokyo, Japan.

\newpage
\section{Technical Overview}\label{sec:tech-overview}
In this section, we give an overview of our work. We refer the reader to the relevant sections for formal definitions and technical details.

\subsection{Definitions and Models}
We first informally recall the various security notions we consider in this work.

\paragraph{Collusion-Resistance and Multi-Copy Security} In collusion-resistant security, we require the security condition to hold even when the adversary receives multiple keys or cryptographic objects. For example, in the case of unclonable encryption, it will receive multiple ciphertexts. In the case of single-decryptor encryption, it will be multiple decryption keys. We discuss precise definitions below. If the total number of keys the adversary receives, $q(\lambda)$, is bounded beforehand (meaning that the scheme will depend on/grow with $q(\lambda)$), we call it \emph{bounded collusion-resistance}. When we say collusion-resistance without a modifier, we mean full collusion-resistance, meaning $q(\lambda)$ is not bounded and can be any polynomial. In the context of unclonable encryption, we sometimes refer to collusion-resistance as multi-challenge security.

Multi-copy security is a strengthening of collusion-resistant security: We require that the multiple keys given to the adversary are not simply independently sampled keys from the same distribution (i.e. $\rho^{\otimes q(\lambda)}$), but rather they are actually the exact same copies of some pure state (i.e. $\ket{\psi}^{\otimes q(\lambda)}$).

\paragraph{Quantum Protection and Classical Vendors:} A quantum protection scheme (\cref{sec:qprotdefn}) is the syntax for schemes where encode a decryption key into a quantum state (\emph{protected key}) to satisfy some security notion, such as copy-protection or secure leasing. When we say that a scheme is classical-vendor or classical-lessor, we mean that the protected key can be sent to the receiver through a classical channel, and the key sender is fully classical.

\paragraph{Copy-Protection:} We will consider copy-protection for decryption keys, also called single-decryptor encryption (SDE). In an SDE scheme, the adversary receives multiple ($q(\lambda)$) quantum decryption keys, and outputs (possibly entangled) $q(\lambda)+1$ adversaries. Then, these adversaries are tested. We consider search and CPA-style security, where the adversaries receive a uniformly random message $m \samp \mathcal{M}$ (search), or an encryption of $m_b$ for random bit $b$ (CPA-style) respectively. The security requirement is that the probability of all adversaries succeeding is $1/|\mathcal{M}|+\negl(\lambda)$ in case of search and $1/2+\negl(\lambda)$ in case of CPA-style security. We also consider identical (\cref{def:cr_iden-chal_search_anti_piracy}) versus independent challenge security, where the adversaries receive the same challenge (i.e. the same ciphertext) or independently sampled challenges. We also consider \emph{strong} security (\cref{def:cr_strong_anti_piracy,def:cr_strong_search_anti_piracy}) where the adversaries are tested through \emph{threshold implementations} (\cref{def:threshold_implementation}), which is easier to work with in proofs.

\paragraph{Secure Leasing} We will consider secure leasing schemes for decryption keys (\cref{sec:skl-def}). In the security game, the adversary receives multiple quantum decryption keys, and at the end, it produces deletion certificates for all of them. If all certificates are valid, then the adversary receives a challenge ciphertext ($\mathsf{Enc}(\pk, m_b)$), and we require that it cannot guess it with probability better than $1/2+\negl(\lambda)$. We also allow adversary to query the deletion-certificate verification algorithm as an oracle at any points. In a classical-lessor scheme, both sending the key and proof of deletion are classical.

\paragraph{LOCC Leakage-Resilience} We will consider LOCC leakage-resilient (\cref{def:coll-res-lr}) and bounded quantum leakage-resilient schemes (\cref{def:coll-res-lr-q}). In LOCC leakage-resilient security game, the adversary $\adve_1$ receives multiple quantum decryption keys, and communicates with adversary $\adve_2$ by exchanging classical message of any length over any number of rounds. After the leakage ends, $\adve_2$ receives a challenge ciphertext ($\mathsf{Enc}(\pk, m_b)$), and we require that it cannot guess it with probability better than $1/2+\negl(\lambda)$. 

\paragraph{Bounded Quantum Leakage-Resilience} In bounded quantum leakage-resilient security game (\cref{def:coll-res-lr-q}), we consider the multiple keys being given to separate adversaries $\adve_1^1, \dots, \adve_1^q$, and $\adve_2$ interacts with each of them, exchanging arbitrary classical messages (like in the LOCC leakage-resilience game), but also additionally exchanging $\ell$ qubits at most with each\footnote{Note that if a single $\adve_1$ were to receive multiple keys, it could simply send one of the keys in full as the qubit leakage if the allowed leakage rate is non-trivial. That's why different adversaries receive the keys.}. After the leakage-phase, we test the adversary $\adve_2$ as above.

We note that our collusion-resistant bounded quantum-leakage resilience result (\cref{cor:leakpure}) also uncovers an interesting phenomenon. Suppose the size of the keys of the scheme is $s(\lambda)$ qubits, and let $c$ denote the constant leakage rate/fraction. Even though the adversary receives quantum leakage from any number of copies of the same pure state $\ket{\psi}$, it still cannot break security. For example, by leaking $c\cdot s(\lambda)$ qubits from each copy of $\ket{\psi}$, the adversary can receive a total of $q(\lambda)\cdot c \cdot s(\lambda)$ qubits of leakage, where $q(\lambda)$ is an arbitrary polynomial. So, in total, it can even receive $(s(\lambda))^{1000}$ qubits of information about $\ket{\psi}$ (which is $s(\lambda)$ qubits), but essentially learn no information! This is in stark contrast to the classical case where leaking on the same string will eventually help us recover it. Thus, this is a striking example of entanglement and non-recombinability of quantum states.

\paragraph{Unclonable Encryption} In unclonable encryption (\cref{sec:uedef}), we encode a message $m$ into a quantum ciphertext. In the security game, the adversary receives multiple ($q(\lambda)$) ciphertexts,  encrypting messages $m^b_1, \dots, m^b_q$ where $b$ is a random bit.  Then, it outputs $q(\lambda)+1$ possibly entangled adversaries. Then, all these adversaries receive the same value $sk$, the decryption key of the scheme, and output their guess for $b$. We require that the probability of all of them succeeding is at most $1/2+\negl(\lambda).$ We also consider search security where the adversary instead receives $q$ ciphertexts encrypting the same random message $m$. We require that all $q+1$ adversaries recover $m$ with probability at most $1/|\mathcal{M}|+\negl(\lambda)$, which is negligible when $1/|\mathcal{M}|$ is negligible.

\subsection{Collusion-Resistance Compiler (\cref{sec:fe-compiler})}
\newcommand{\comppke}{\mathsf{CompiledPKE}}
In this section, we discuss our collusion-resistance compiler based on classical functional encryption (FE). Our compiler uses a single-key-secure quantum protection scheme $\one$, and a classical public-key FE scheme $\fe$ to construct a new scheme $\comppke$ that satisfies the collusion-resistant version of the quantum protection game (i.e. from single-key copy-protection  to collusion-resistant copy-protection). 

Our compiler is black-box\footnote{As mentioned in \cref{sec:intro}, we do use the encryption circuit of the single-key scheme, but key-generation and decryption are used in a black-box way.} in both schemes $\one$ and $\fe$, and further, it is \emph{quantum-key} preserving, meaning that the quantum part of the scheme is the same as the keys of $\one$. This in particular means that if $\one$ is a classical-vendor scheme, then so is $\comppke$. We note that while there have been some collusion-resistant quantum protection constructions in the past (e.g. copy-protection construction of \cite{TCC:CakGoy24}), they were all highly non-black-box, and thus they cannot be adapted to newer notions such as classical-vendor schemes, whereas our scheme immediately inherits any improvements to single-key schemes.

\paragraph{Construction} Now we describe our scheme $\comppke$. During setup, we simply sample an instance of $\fe$, and its master key and public key become the master and public keys of $\comppke$. To generate new quantum key, we first sample a fresh instance of $\one$ as  $(\one.\pk,\one.\qsk)\lrun \One.\Setup(1^\secp)$. Then, we create a functional key $fsk$ (under $\fe$) for the circuit $C_{\one.\pk, tag}(m || K)$ that (up to some caveats - see \cref{{fig:func_reenc}}) simply outputs\footnote{If it looks unusual that the randomness if generating using $\prf(K, tag)$ instead of $\prf(K, m)$, recall that each ciphertext will come with a fresh PRF key $K$ embedded.} $\one.\mathsf{Enc}(\one.\pk, m; \prf(K, tag))$. Finally, $fsk$ together with $\one.\qsk$ is output as the quantum key of $\comppke$. To encrypt $m$ with $\comppke$, we simply sample a fresh PRF key $K$ and output $\fe.\mathsf{Enc}(pk, m || K)$. Decryption is straightforward: We first decrypt with $\fe.\mathsf{Dec}$ using $fsk$, which gives us a ciphertext of $\one$, and we decrypt this with $\one.\sk$. The correctness of the scheme is straightforward.

\paragraph{Security} For security, as discussed above, we show that it satisfies collusion-resistant $X$ security if $\fe$ is a secure classical functional encryption scheme and $\one$ is a secure single-key scheme with $X$ security, where $X \in \{\text{copy-protection,}$ $\text{secure-leasing,}$ $\text{LOCC leakage-resilience,}$ $\text{bounded-quantum leakage resilience} \}$. The copy-protection proof proceeds differently than others since we need to deal with multiple entangled adversaries, whereas in the last three notions there is a single adversary that is tested (even though it receives multiple keys).

To prove the security for the last three notions, we introduce a somewhat technical lemma (\cref{thm:fe-coll_res}) that is reused in all three proofs. The lemma says that if there is an adversary that receives multiple ($q(\lambda)$) keys from $\comppke.\KG$ by interacting with a challenger and then succeeds in guessing the challenge bit $b$ given the challenge ciphertext $ct \samp \comppke.\mathsf{Enc}(pk, m_b)$ with probability $1/2 + \epsilon$, then there exists another adversary $\adve'$ that succeeds with probability $1/2 + \epsilon/(16\cdot q(\lambda))$ when $ct$ is sampled according to the distribution $\mathcal{D}_{j^*}$ for some $j^* \in [q(\lambda)]$ (see the theorem for the distribution). Crucially, the new adversary $\adve'$ is the same as the old adversary $\adve$, except for the last part of the game where it receives the challenge ciphertext. This means that when going from $\adve$ to $\adve'$, we do not break the interface/rules of the quantum protection security game (e.g., our leakage adversaries do not suddenly start communicating over unbounded quantum channels). Again, crucially, the ciphertext distributions $\mathcal{D}_j$ are so that their ciphertexts can be simulated using a sample $\one.\ct \samp \one.\mathsf{Enc}(\one.\pk_j\footnote{Here, $\pk_j$ means the public-key of the $\one$ instance sampled during the $j$-th key generation call to $\comppke$.}, m_b)$, without knowing the message $m_b$ or the randomness for $\one.\ct$.  This means that when we are reducing security to the single-key game for $\one$, we will simply receive the challenge ciphertext $\one.\ct$ from our challenger, and simulate a sample from $\mathcal{D}_{j^*}$, for which we know our adversary succeeds (thanks to \cref{thm:fe-coll_res}).

\paragraph{Secure Leasing and Leakage-Resilience} With \cref{thm:fe-coll_res}, our proofs for collusion-resistant secure leasing and leakage-resilience follows in a straightforward way: If there is an adversary $\adve$ that succeeds in the collusion-resistant security game, then by our theorem, there is an adversary that succeeds for the distribution $\mathcal{D}_{j^*}$. However, this gives us a successful admissible\footnote{Since only the last part of the adversary changes in \cref{thm:fe-coll_res}.} adversary for $\one$: Simulate the instance of $\fe$ and all instances of $\one$ except for $j^*$, then when the challenge $ct$ is received, simulate $\mathcal{D}_{j^*}$ and execute $\adve'$ to obtain a successful guess.

\paragraph{Copy-Protection} As mentioned before, while there are some similarities, we need a different proof for the copy-protection security of the compiler. Like above, we would like to construct an adversary that violates the single-key copy-protection security of some instance of $\one$. For this, unlike above, we need to construct two adversaries that simultaneously succeed for the same distribution $\mathcal{D}_{j^*}$ (associated with $\one.pk_{j^*}$). Once we do this, as above, we can simulate the FE instance and other instance of $\one$ to obtain an adversary for single-key copy-protection security of $\one$ instance associated with $\one.pk_{j^*}$. An argument similar to \cref{thm:fe-coll_res} shows that for each individual adversary $A_i$ (for $i \in [q+1]$), there will be some $j_i^* \in [q]$ such that $\adve_{i}$ succeeds for $\mathcal{D}_{j^*_i}$. Then, by pigeonhole, there must be two $i_1 \neq i_2$ such that $j_{i_1}^* = j_{i_2}^* = j^*$. However, this argument is unfortunately not sufficient: While $\adve_{i_1}$ and $\adve_{i_2}$ will be both individually successful for $\mathcal{D}_{j^*}$, due to possible entanglement between them, it is not clear if they will be successful at the same time. In particular, $\adve_{i_1}$ succeeding might cause $\adve_{i_2}$ to immediately collapse to a bad adversary\footnote{For example, consider an entanglement between them the where they essentially share a single decryption key $\ket{\psi}$  as $\frac{1}{\sqrt{2}}\ket{\psi}_{\adve_{i_1}}\otimes\ket{Garbage}_{\adve_{i_2}} + \frac{1}{\sqrt{2}}\ket{Garbage}_{\adve_{i_1}}\otimes\ket{\psi}_{\adve_{i_2}} $}. Thus, to resolve this, we use our new \emph{quantum pigeonhole lemma} (\cref{thm:collusion_resistant}) that says that indeed there will be some $j^*$ such that two adversaries will simultaneously succeed for $\mathcal{D}_{j^*}$. To make the quantum pigeonhole lemma work, we need work to work with threshold implementations (\cref{def:threshold_implementation}) instead of directly testing adversaries, which also means we need to tweak our techniques from \cref{thm:fe-coll_res} a little in this case. Finally, as discussed above, we end up with simultaneous successful adversaries for the same instance of $\one$, which completes the reduction. The full proof is technical and is given in \cref{sec:compiler_SDE}.

\paragraph{Bounded-Collusion Resistance} We also present in \cref{sec:compiler_SDE_bounded} a different collusion-resistance compiler for SDE that works only in the bounded-collusion setting, but it does not make any additional assumptions (other than underlying single-key scheme's security). We refer the reader to \cref{sec:compiler_SDE_bounded} for the construction (which is significantly different than the compiler above) and the proof (which shows some similarities to the proof of the other compiler).

\subsection{Purification Compiler}
In the last section, we discussed our collusion-resistance compiler, which allows us to obtain collusion-resistance from single-key schemes. In this section, we discuss our purification compiler, which allows us to obtain multi-copy secure schemes from collusion-resistant schemes.

We provide two purification compilers. The first one (\cref{sec:purification_compiler}) is based on pseudorandom states (PRS) and only assumes existence of one-way functions, and allows us to obtain multi-copy security from collusion-resistant security in a large number of settings. In particular, it works for \emph{state-query games} (\cref{def:state-query-game}), which captures copy-protection, secure leasing, LOCC/bounded-quantum leakage-resilience, unclonable encryption and quantum money. The second compiler (\cref{sec:money_eqsup_compiler}) only works for quantum money, and assumes deterministic quantum-query-secure digital signatures (\cref{sec:detsig}). While the first compiler is more general, we use the second compiler to prove the security of the folklore quantum coin scheme (whose security was only conjectured before) where the coin states are simply equal superpositions of all possible banknotes of a quantum money mini-scheme. In fact, it can even show the security of arbitrary superpositions that place negligible weight on each possible banknote, whereas our first compiler shows security for a particular purification.

Now we discuss our first compiler. Our compiler shows that for any state-query game, its compiled version which uses pure states is also secure. Our compiler simply uses a pseudorandom state generator (PRS), and whenever some mixed state (generated by $\mathsf{GenState}$ with some randomness) is used, it instead outputs its purified version which is obtained by creating a PRS state, applying a PRF to it, and supplying the randomness $\mathsf{GenState}$ of using this state (by running $\mathsf{GenState}$ coherently). In the proof, using the relation between pseudorandom/Haar random states and \emph{type states}, we show that the purified state can be simulated efficiently using samples from the mixed state. The full proof is technical: We refer the reader to \cref{sec:purification_compiler} for details.

Now we discuss our second compiler. In this compiler, we simply output a equal superposition of exponentially many banknote states from independent (generated using a PRF) single-copy quantum money schemes, with their serial numbers signed in superposition using our deterministic\footnote{Note that if the signature scheme's security relied on independent randomness for each quantum query (which is the case for all prior scheme before our work), we would be outputting mixed states.} signature scheme that is quantum-query secure (\cref{sec:detsig}). To prove security, we first rely on the quantum-query security of our signature scheme to argue that there will be a serial number collision in a pair of the forged banknotes. Then we rely on our quantum-state read-once small-range distributions lemma (\cref{sec:readonce}) to simulate the adversary's view using a small number of banknotes such that only a single copy of one of the banknotes is used. Finally we place a challenge banknote from the single-banknote scheme in place of the banknote whose single copy is used to get a reduction to the security of the single-banknote scheme. We refer the reader to \cref{sec:money_eqsup_compiler} for the construction and the proof.

\subsection{Unclonable Encryption}
We now discuss our unclonable encryption schemes. We give two different constructions, one is in the plain model (assumes iO and one-way functions), and satisfies search security; the other one is in QROM (it is unconditionally secure) and satisfies the more desirable indistinguishability security. Note that achieving (even non-collusion-resistant) indistinguishability security in the plain model is one of the biggest open questions in quantum cryptography, thus, our result can be considered optimal (until a breakthrough is achieved).

We first describe our plain model construction (\cref{sec:SDE_to_UE}). In this construction, we simply rely on a identical-challenge search-secure SDE scheme (\cref{def:cr_iden-chal_search_anti_piracy}). SDE scheme's unclonable keys are used in our ciphertexts, along with $m \oplus s$, where $s$ is a one-time pad that is part of UE secret key. Again as part of the UE secret key, we also use SDE ciphertext for $s$ along with the master key of SDE. We show how to reduce security to collusion-resistant identical-challenge search-security of SDE, which is also achieved for the first time by our work (\cref{thm:ind-to-iden-sde,thm:ind-to-iden}).

We now describe our QROM construction (\cref{sec:ueqrom}). In our construction, the secret key is a random string $s$. To encrypt $m$, we sample random $r$ and a random BB84 state description $(x, \theta)$, and output $r, G(s, r) \oplus \theta, H(x) \oplus m$ along with the BB84 state $H^\theta\ket{x}$. 

Note that this construction uses BB84 states, which are easy to prepare and measure (the scheme only needs Hadamard gates to encrypt or decrypt). Observe that the hash function/the random oracle is only evaluated classically by the construction (but the adversary of course has quantum access). In the proof, given a succesful adversary, we rely on our new collusion-resistant one-way-to-hiding lemma to extract information about one of the BB84 states and violate its search security (\cref{lem:BB84_aug}). The full proof is technical and is given in \cref{sec:ue-ind-proof}.

In both constructions, multi-copy security is achieved by our generic purification compiler (\cref{sec:purification_compiler}). Also note that assuming PKE, our constructions can be made public-key through a simple transformation.

\subsection{Technical Tools}
In this section, we describe the interpretations of our technical tools. Their proofs are technical thus we refer the reader to the relevant sections for the proofs. 

The quantum pigeonhole lemma supports our reduction from collusion-resistant SDE to single-key security (\cref{sec:compiler_SDE}). The collusion-resistant O2H lemma is used in our QROM unclonable-encryption proof (\cref{sec:ue-ind-proof}). Our challenge-comparison lemma provides the identical-challenge SDE security needed for the plain-model UE construction (\cref{sec:SDE_to_UE}). Finally, the small-range lemma and deterministic signatures are used in the equal-superposition quantum coin construction (\cref{sec:money_eqsup_compiler}).

\paragraph{Quantum Pigeonholes for Entangled Adversaries: Simultaneous Distinguisher Finding (\cref{sec:lemma_collusion_resistance})} Let $\mathcal{D}_0, \dots, \mathcal{D}_q$ be some challenge distributions in some game. We show that if $q$ possibly entangled adversaries all succeed on distribution $\mathcal{D}_0$, and any possible efficient adversary fails on  distribution $\mathcal{D}_q$, then we can obtain a pair of adversaries that can simultaneously distinguish between the same distribution pair $\mathcal{D}_{j-1}, \mathcal{D}_{j}$. Note that this is non-trivial due to entanglement: While we can easily argue by a hybrid argument and the classical pigeonhole principle that two adversaries will \emph{individually} distinguish between the same pair $\mathcal{D}_{j-1}, \mathcal{D}_{j}$, it is not clear if they can simultaneously do so since one adversary's success might collapse the other one. Also note that distinguishing a pair of distributions versus answering a challenge from a distribution are different measurements, thus, all adversaries being simultaneously successful for answering challenges from $\mathcal{D}_0$ does not immediately imply anything.

This lemma helps us significantly in copy-protection proofs for getting a reduction a single-key secure state, as explained in the relevant section.

\paragraph{Collusion-Resistant One-way To Hiding (\cref{sec:o2h})} We show a collusion-resistant version of the one-way-to-hiding lemma (\cite{C:AmbHamUnr19}) which only applies to a single adversary. Our lemma says that if multiple ($q(\lambda) + 1$) possibly entangled adversaries can simultaneously distinguish between an oracle and its reprogrammed version (which differs from the original oracle on disjoint regions $S_1, \dots, S_q$), then we can \emph{simultaneously} extract points $x_1, x_2$ from two adversaries such that the oracles differ on $x_1, x_2$, and in particular, $x_1, x_2$ are in the same region $S_i$. Note that this is quite non-trivial since extracting from one adversary can cause the other adversaries to collapse to states who are not good oracle distinguishers. This lemma allows us to obtain reductions to single-copy games, since each region $S_i$ will correspond to some information about one of the single-copy unclonable states.

We note that our new lemma also simplifies proofs (even in the case of two adversaries). For example, previous proofs of single-ciphertext unclonable encryption (\cite{C:AKLLZ22}) relied on explicit use of tools like Jordan's lemma, whereas with our new collusion-resistant one-way to hiding lemma, our unclonable encryption proof is much simpler and intuitive.

\paragraph{Collusion-Resistant Identical and Independent Challenge Security (\cref{sec:relating})} This lemma shows that strong search security (which is defined using threshold implementations) implies identical challenge security even in the collusion-resistant setting.  Previously, a similar result was only known in the case of two adversaries. Note that unlike the case of two parties (which works also with independent challenge security rather strong security), our collusion-resistant proof relies crucially on spectral properties of threshold implementations.

Since it is known that strong search security also implies independent challenge security, our result in some sense means that strong search security is the right security notion for copy-protection.

\paragraph{Quantum-State Read-Once Small-Range Distributions (\cref{sec:readonce})} In this lemma, we show how to simulate an quantum-state-oracle adversary's queries using a small number of samples from the oracle's output distribution, and in particular using a state whose support consists only of single copies of these samples. This allows us to obtain reductions to single-copy secure games from games that use superpositions/purifications of quantum states.

\paragraph{Quantum-Query-Secure Deterministic Digital Signatures (\cref{sec:detsig})} We show how to construct deterministic signature schemes that satisfy Boneh-Zhandry (BZ) security. Note that the well-known trick of using a PRF to make a signature scheme deterministic does not preserve Boneh-Zhandry security: In fact, previous constructions (such as \cite{boneh2013secure}) crucially relied on independent randomness for each query by the adversary. Ours is the first deterministic scheme (which is useful in pure state applications) with BZ security.

\subsection{Related Work}
We refer the reader to \cref{sec:intro} for related work where we compare each of our results to the state-of-the-art existing work. Separately, We discuss the concurrent and independent work of \cite{AG26concurrent} below. 

The main common point between our work and \cite{AG26concurrent} is that we both achieve a generic compiler from collusion-resistant schemes to multi-copy schemes. Then, invoking known results, we both achieve quantum coins. We also both achieve multi-copy secure single-decryptor encryption (SDE) schemes, however, we achieve it from polynomially secure iO and one-way functions using our new collusion-resistant SDE result; whereas they invoke the collusion-resistant SDE scheme from \cite{TCC:CakGoy24} and thus they assume subexponential iO and LWE. Other than these, our results are quite different, as explained below in more detail.

In terms of techniques, while our purification compilers share some similarities (we both use pseudorandom states), there are two differing points. Our compiler is simpler but assumes that the underlying collusion-resistant scheme has quantum keys that are determined by some underlying classical key, whereas their compiler is slightly more general (at the expense of being slightly more complicated) and only assumes that the collusion-resistant scheme's keys are generated starting with a classical randomness (but allows some measurements during key generation). Though we note that all known schemes do satisfy our template. 

We also note that our second compiler for quantum money (the superposition compiler from \cref{sec:money_eqsup_compiler}, which works for arbitrary delocalized states instead of just pseudorandom-state purifications) is unique to our work.

Aside from the purification compilers, our works are significantly different. Their work shows some new results related to pseudorandom state generators and pseudorandom unitaries, which we do not consider. Our work develops various compiler and technical tools for achieving collusion-resistance from single-key schemes, which they do not consider (they invoke known collusion-resistant schemes to achieve multi-copy security). Thus, most of our results, such as generic black-box collusion-resistance compiler, the first multi-copy unclonable encryption, first multi-copy LOCC leakage-resilient encryption scheme, first collusion-resistant classical-lessor secure leasing scheme and so on are unique to our work.

\newpage
\section{Preliminaries}\label{sec:preliminaries}
In this section, we highlight some of our notations and conventions, and then recall some preliminaries on classical and quantum cryptography. We refer the reader to \cite{Goldreich_2001,Goldreich_2004} and \cite{NC10} for the basics of classical cryptography and quantum information.

\paragraph{\bf Notations.}
Let $[\ell]$ denote the set of integers $\{1, \cdots, \ell \}$, $\secp$ denote a security parameter, and $y \seteq z$ denote that $y$ is set to, defined to be, or substituted by $z$.
For a finite set $X$, $x \chosen X$ denotes selecting an element from $X$ uniformly at random.
For a distribution $D$, 
$x \chosen D$ denotes sampling an element $x$ according to $D$. 
For (classical or quantum) algorithm $\cA$, 
$y \gets \algo{A}(x)$ denotes assigning to $y$ the output of $\algo{A}$ on input $x$. When we explicitly show that $\cA$ uses randomness $r$, we write $y \gets \algo{A}(x;r)$. 
PPT and QPT algorithms stand for probabilistic  (classical)  polynomial-time algorithms and polynomial-time quantum algorithms, respectively.
Let $\negl$ denote a negligible function 
and $\poly$ denote a positive polynomial.  
We use gray sans serif font (e.g., $\qreg{R}$) to stand for a quantum register. 
For a quantum state $\rho$ over registers $\qreg{R}_1,\ldots,\qreg{R}_n$ and $i\in[n]$, we write $\rho[i]$  
to denote the portion of the state on register $\qreg{R}_i$. 
Similarly, for $i_1\ne i_2\in [n]$, we write $\rho[i_1,i_2]$   
to denote the portion of the state on registers $\qreg{R}_{i_1}$ and $\qreg{R}_{i_2}$.  We sometimes write in bold sans serif font to denote classical keys/variables (such as $\pk$) and write in  bold calligraphic font to denote quantum keys/variables (such as $\qsk$) to emphasize their distinction.

All of our cryptographic assumptions are implicitly post-quantum unless otherwise specified.

\subsection{Classical Cryptography}

\subsubsection{Pseudorandom Functions}
\begin{definition}[Pseudorandom Function (PRF)]\label{def:prf}
A pseudorandom function (PRF) is a pair of algorithms $\PRF = (\prfgen, \prfF)$, where 
$\{\prfF_{\prfkey} : \zo^{\ell_1} \ra \zo^{\ell_2} \mid \prfkey \in \zo^{\secp}\}$ 
is a family of efficiently computable functions. Here $\ell_1$ and $\ell_2$ are polynomials in the security parameter $\secp$.  
The PRF satisfies the following pseudorandomness condition.
    \paragraph{Pseudorandomness} For any QPT distinguisher $\cA$, it holds that
    \begin{align}
    \left|
    \Pr[\cA^{\prfF_{\prfkey}(\cdot)}(1^{\secp}) \ra 1] -
    \Pr[\cA^{\cU_{\ell_1,\ell_2}(\cdot)}(1^{\secp}) \ra 1]
    \right| \leq \negl(\secp),
    \end{align}
    where $\prfkey \lrun \prfgen(1^{\secp})$, and 
    $\cU_{\ell_1,\ell_2}$ denotes a truly random function mapping 
    $\zo^{\ell_1}$ to $\zo^{\ell_2}$.
\end{definition}

\begin{definition}[Puncturable PRF]\label{def:pprf}
A puncturable PRF (PPRF) is a tuple of algorithms $\PuncPRF = (\prfgen, \prfF,\Puncture)$ where $\{\prfF_{\prfkey}: \zo^{\ell_1} \ra \zo{\ell_2} \mid \prfkey \in \zo^{\secp}\}$ is a PRF family and satisfies the following two conditions. Note that $\ell_1$ and $\ell_2$ are polynomials of $\secp$.
   \begin{description}
       \item[Punctured correctness:] For any polynomial-size set $S \subseteq \zo^{\ell_1}$ and any $x\in \zo^{\ell_1} \setminus S$, it holds that
       \begin{align}
       \Pr[\prfF_{\prfkey}(x) = \prfF_{\prfkey_{\notin S}}(x)  \mid \prfkey \lrun \prfgen(1^{\secp}),
       \prfkey_{\notin S} \lrun \Puncture(\prfkey,S)]=1.
       \end{align}
       \item[Pseudorandom at punctured point:] For any polynomial-size set $S \subseteq\zo^{\ell_1}$
       and any QPT distinguisher $\cA$, it holds that
       \begin{align}
       \vert
       \Pr[\cA(\prfF_{\prfkey_{\notin S}},\{\prfF_{\prfkey}(x_i)\}_{x_i\in S}) \ra 1] -
       \Pr[\cA(\prfF_{\prfkey_{\notin S}}, (\cU_{\ell_2})^{\abs{S}}) \ra 1]
       \vert \leq \negl(\secp),
       \end{align}
       where $\prfkey\lrun \prfgen(1^{\secp})$,
       $\prfkey_{\notin S} \lrun \Puncture(\prfkey,S)$ and $\cU_{\ell_2}$ denotes the uniform distribution over $\zo^{\ell_2}$.
   \end{description}
   If $S = \setbk{x^\ast}$ (i.e., puncturing a single point), we simply write $\prfF_{\ne x^\ast}(\cdot)$ instead of $\prfF_{\prfkey_{\notin S}}(\cdot)$ and consider $\prfF_{\ne x^\ast}$ as a keyed function.
\end{definition}

It is easy to see that the Goldwasser-Goldreich-Micali tree-based construction of PRFs (GGM PRF)~\cite{JACM:GolGolMic86} from OWF yield puncturable PRFs where the size of the punctured key grows polynomially with the size of the set $S$ being punctured~\cite{AC:BonWat13,PKC:BoyGolIva14,CCS:KPTZ13}. Thus, we have:
\begin{theorem}[\cite{JACM:GolGolMic86,AC:BonWat13,PKC:BoyGolIva14,CCS:KPTZ13}]\label{thm:pprf-owf} If OWFs exist, then for any polynomials $\ell_1(\secp)$ and $\ell_2(\secp)$, there exists a PPRF that maps $\ell_1$-bits to $\ell_2$-bits.
\end{theorem}

\subsubsection{Indistinguishability Obfuscation}
In this section, we recall indistinguishability obfuscation (iO).
\begin{definition}[Indistinguishability Obfuscation]
    An indistinguishability obfuscation (iO) scheme for a class of circuits $\mathcal{C} = \{\mathcal{C}_\lambda\}_\lambda$ is an efficient algorithm $\io$ that satisfies the following.
    \paragraph{Correctness.} For all $\lambda \in \BN^+, C \in \mathcal{C}_\lambda$ and all inputs $x$ to $C$,
    $$\Pr[\Tilde{C}(x) = C(x): \Tilde{C} \samp \io(1^\lambda, C)] = 1.$$

    \paragraph{Security.} Let $\mathcal{B}$ be any QPT algorithm that outputs two circuits $C_0, C_1 \in \mathcal{C}$ of the same size, along with quantum auxiliary information $\regi{aux}$, such that $\Pr[\forall x ~ C_0(x)=C_1(x) : (C_0, C_1, \regi{aux}) \samp \mathcal{B}(1^\lambda)] \geq 1 - \negl(\lambda)$. Then, for any QPT adversary $\mathcal{A}$,
    \begin{align*}
      \bigg|&\Pr[\adve(\io(1^\lambda, C_0), \regi{aux}) = 1 :  (C_0, C_1, \regi{aux}) \samp \mathcal{B}(1^\lambda)] -\\ &\Pr[\adve(\io(1^\lambda, C_1), \regi{aux}) = 1 : (C_0, C_1, \regi{aux}) \samp \mathcal{B}(1^\lambda)]\bigg| \leq \negl(\lambda).  
    \end{align*}
\end{definition}

\subsubsection{Functional Encryption}
\begin{definition}[Functional Encryption]\label{def:FE}
An FE scheme $\FE$ is a tuple of PPT algorithms $(\Setup, \KG, \Enc, \Dec)$. 
\begin{description}
\item[$\Setup(1^\secp)\ra(\pk,\msk)$:] The setup algorithm takes a security parameter $1^\secp$ and outputs a public key $\pk$ and master secret key $\msk$.
\item[$\KG(\msk,f)\ra\fsk$:] The key generation algorithm $\KG$ takes a master secret key $\msk$ and a function $f$, and outputs a functional decryption key $\fsk$.

\item[$\Enc(\pk,x)\ra\ct$:] The encryption algorithm takes a public key $\pk$ and an input $x$, and outputs a ciphertext $\ct$.

\item[$\Dec(\fsk,\ct)\ra y$:] The decryption algorithm takes a functional decryption key $\fsk$ and a ciphertext $\ct$, and outputs $y$.

\item[Correctness:] We require 
we have that
\[
\Pr\left[
\Dec(\fsk, \ct) = f(x)
 \ \middle |
\begin{array}{rl}
 &(\pk,\msk) \lrun \Setup(1^\secp),\\
 & \fsk \lrun \KG(\msk,f), \\
 &\ct \lrun \Enc(\pk,x)
\end{array}
\right]=1 -\negl(\secp).
\]
\end{description}
\end{definition}

\begin{definition}[Adaptive Security for PKFE]\label{def:ad_ind_PKFE}
We formalize the experiment $\expb{\cA}{ada}{ind}(1^\secp,\coin)$ between an adversary $\cA$ and a challenger for PKFE scheme for $\cX,\cY$, and $\cF$ as follows:
        \begin{enumerate}
            \item The challenger runs $(\pk,\msk)\gets\Setup(1^\secp)$ and sends $\pk$ to $\cA$.
            \item $\cA$ sends arbitrary key queries. That is, $\cA$ sends function $f_{i}\in\cF$ to the challenger and the challenger responds with $\fsk_{i}\gets \KG(\msk,f_i)$ for the $i$-th query $f_{i}$.
            \item At some point, $\cA$ sends $(x_0,x_1)$ to the challenger. If $f_i(x_0)=f_i(x_1)$ for all $i$, the challenger generates a ciphertext $\ct^*\lrun\Enc(\pk,x_\coin)$. The challenger sends $\ct^*$ to $\cA$.
            \item Again, $\cA$ can sends function queries $f_i$ such that $f_i(x_0)=f_i(x_1)$.
            \item $\cA$ outputs a guess $\coin^\prime$ for $\coin$.
            \item The experiment outputs $\coin^\prime$.
        \end{enumerate}
        We say that $\PKFE$ is adaptively secure if, for any QPT $\cA$, it holds that
\begin{align}
\advb{\PKFE,\cA}{ada}{ind}(\secp) \seteq \abs{\Pr[\expb{\PKFE,\cA}{ada}{ind} (1^\secp,0) \ra 1] - \Pr[\expb{\PKFE,\cA}{ada}{ind} (1^\secp,1) \ra 1] }\leq \negl(\secp).
\end{align}
\end{definition}

\begin{theorem}[\cite{C:Waters15,C:ABSV15}]\label{thm:ada_PKFE}
    If there exist a secure iO for polynomial size classical circuits and OWFs,
then there exists an adaptively secure PKFE scheme.
\end{theorem}

\subsection{Quantum Cryptography}
\subsubsection{Quantum Money}
We recall the notion of quantum money mini-schemes.
\begin{definition}[Quantum Money Mini-Scheme]\label{def:qmm}
A quantum money mini-scheme is a tuple of QPT algorithms
$\bank = (\mathsf{GenBanknote}, \Ver)$ defined as follows:
\begin{itemize}
    \item $\mathsf{GenBanknote}(1^\lambda)$: Outputs a serial number $sn$ and a register $\reg$ containing a \emph{banknote}. 
    \item $\Ver({sn}, \reg)$: Outputs $1$ or $0$, denoting if the state in the register $\reg$ is a valid banknote for the serial number $\mathsf{sn}$.
\end{itemize}

We require the following.
\paragraph{Correctness}
    \begin{align}
    \Pr[\Ver({sn}, \reg) = 1] = 1,
    \end{align}
    where $({sn}, \reg) \lrun \mathsf{GenBanknote}(1^\lambda)$.

\paragraph{Unforgeability}
For public-key schemes, we require that for any QPT adversary $\cA$, it holds that
    \begin{equation*}
        \Pr[b_1 = b_2 = 1 : \begin{array}{c}
             (sn, \reg) \samp \mathsf{GenBanknote}(1^\lambda) \\
             \reg_1, \reg_2 \samp \adve(sn,\reg) \\
             b_1 \samp \Ver(sn, \reg_1) \\
             b_2 \samp \Ver(sn, \reg_2)
        \end{array}] \leq \negl(\lambda)
    \end{equation*}
    For private-key schemes, the requirement is the same except that the adversary does not receive the serial number $sn$.
\end{definition}

\subsubsection{Pseudorandom Quantum States}
In this section, we recall pseudorandom (quantum) states (PRS) \cite{C:JiLiuSon18}.
\begin{definition}[Pseudorandom Quantum States]\label{def:prs}
A pseudorandom quantum state (PRS) scheme $\prs$ is a family of quantum states 
$\{\ket{\phi_K}\}_{K \in \zo^{c(\lambda)}}$ over $m(\lambda)$ qubits along with the following efficient algorithms.
\begin{itemize}
    \item $\prs.\Setup(1^\lambda):$ Takes in the unary representation of the security parameter and outputs a key $K \in \zo^{c(\lambda)}$.
    \item $\prs.\genstate(K):$ On input $K \in \zo^{c(\lambda)}$, outputs a quantum state over $m(\lambda)$ qubits.
\end{itemize}

We require the following.

\paragraph{Correctness.}
For all $\lambda \in \BN^+$ and all $K \in \zo^{c(\lambda)}$, the algorithm $\prs.\genstate(K)$ outputs $\ket{\phi_K}$.

\paragraph{Security.}
We require that for any QPT adversary $\adve$ and any polynomial $t(\lambda)$, 
\begin{equation*}
    \left| 
    \Pr_{K \samp \prs.\Setup(1^\lambda)}\left[ \adve\left( \ket{\phi_K}^{\otimes t(\lambda)} \right) = 1 \right]
    - 
    \Pr_{\ket{\psi} \samp \mathsf{Haar}(m(\lambda))}\left[ \adve\left( \ket{\psi}^{\otimes t(\lambda)} \right) = 1 \right]
    \right| \leq \negl(\lambda),
\end{equation*}
where $\mathsf{Haar}(m(\lambda))$ denotes the Haar measure over $(\C^2)^{\otimes{m}}$
\end{definition}

\begin{theorem}[\cite{C:JiLiuSon18}]\label{thm:prsexists}
Let $m(\cdot)$ be polynomially bounded. If one-way functions exist, then there exists a pseudorandom quantum state scheme $\prs$ that outputs $m(\lambda)$-qubit states.
\end{theorem}

\subsubsection{BB84 States}
For strings $x=x_1\|x_2\|\ldots\|x_n \in\bit^n$ and $\theta=\theta_1\|\theta_2\|\ldots\|\theta_n \in\bit^n$, let $\ket{x_\theta} \seteq H^{\theta_1}\ket{x_1}\tensor\cdots \tensor H^{\theta_n}\ket{x_n}$ where  
$H$ is the Hadamard operator.   A state of the form $\ket{x^\theta}$ is called a BB84 state~\cite{BB84}.
\begin{lemma}[Augmented Search Security of BB84 States~{\cite[Lemmas 12 and 27]{C:AnaKalLiu23}}]\label{lem:BB84_aug}
    For a tuple of quantum algorithms $(\qA,\qB,\qC)$, consider the following experiment.
\begin{enumerate}
 \item Choose uniformly random $x\gets \bit^\secp$ and $\theta\gets \bit^\secp$.
 \item Let $O_x$ be a function that returns $1$ on $x$ and $0$ on all other $\secp$-bit inputs.  
 \item $\cA^{O_x}$ receives $\ket{x_\theta}$, creates a bipartite state $\rho$ over registers $\qreg{R}_\qB$ and $\qreg{R}_\qC$,
and sends $\qreg{R}_\qB$ and $\qreg{R}_\qC$ to $\qB$ and $\qC$, respectively. 
\item $\qB^{O_x}$ and $\qC^{O_x}$ receive $\theta$ and outputs $x_\qB$ and $x_\qC$, respectively.
\item The experiment outputs $1$ if $x_\qB=x_\qC=x$ and $0$ otherwise.
\end{enumerate}
For any $(\qA,\qB,\qC)$ that make $\poly(\secp)$ queries to $O_x$, 
the probability that the above experiment returns $1$ is $\negl(\secp)$. 
\end{lemma}

\subsubsection{Encryption Schemes with Quantum Keys}
In this section, we recall some notions related to encryption schemes with quantum decryption keys.

First, we recall PKE with quantum decryption keys. Note that here we only introduce the syntax of such schemes. There are various security notions that can be considered for such schemes, for example, one can consider regular CPA security, or the various quantum protection security notions such as copy-protection, or unbounded/LOCC leakage-resilience. Thus, we will introduce the actual security definitions in the relevant sections. 
\begin{definition}[PKE with Quantum Decryption Keys (Syntax)]\label{def:qpke}
A PKE scheme with quantum protected decryption keys (PKE-QDK) is a tuple of QPT algorithms $(\Setup,\KG,\Enc,\qDec)$ with plaintext space $\cM$.
\begin{description}
    \item[$\Setup (1^\secp) \ra (\pk,\msk)$:] The setup algorithm takes as input the security parameter $1^\secp$ and outputs a public encryption key $\pk$ and a master secret key $\msk$.
    \item[$\KG (\msk) \ra (\qsk,\mathsf{aux})$:] The protected key generation algorithm takes as input $\msk$ and outputs a quantum decryption key $\qsk$ and classical auxiliary information $\mathsf{aux}$ retained by the challenger and associated with this key instance.
    \item[$\Enc(\pk,m) \ra \ct$:] The encryption algorithm takes as input $\pk$ and a plaintext $m\in\cM$ and outputs a ciphertext $\ct$.
    \item[$\Dec(\qsk,\ct) \ra m^\prime$:] The decryption algorithm takes as input $\qsk$ and $\ct$ and outputs a plaintext $m^\prime \in \cM$ or $\bot$.

\item[Decryption correctness:] There exists a negligible function $\negl$ such that for any $m\in\cM$, 
\begin{align}
\Pr\left[
\Dec(\qsk,\ct)= m
\ \middle |
\begin{array}{ll}
(\pk,\msk)\lrun \Setup(1^\secp)\\
\qsk, \mathsf{aux}\lrun \KG(\msk)\\
\ct \lrun \Enc(\pk,m)
\end{array}
\right] 
\geq 1-\negl(\secp).
\end{align}

\end{description}
\end{definition}
The auxiliary output allows the same syntax to support different security notions. It is the deletion-verification key in SKL and is empty for SDE and LOCC leakage resilience. It is not given to the adversary unless the relevant security game explicitly provides it. When a game writes only $\qsk\lrun\KG(\msk)$, it means that the quantum component is returned and the auxiliary component is retained by the challenger (or discarded when empty). In the classical-vendor syntax below, the sender retains the corresponding auxiliary output.

We also consider schemes with \emph{classical vendors}, where the key generation is an interactive protocol between a quantum receiver $\mathsf{Receiver}$ and a classical sender $\mathsf{Sender}$. We will only provide the definition in the public-key setting, and the secret-key setting is defined the natural way.

\begin{definition}[PKE-QDK with a Classical Vendor]
 A PKE-QDK scheme with a classical vendor is a PKE-QDK scheme with an additional classical-message interactive protocol $(\mathsf{Send}, \mathsf{Receive})$  that satisfies the following.

\begin{align}
\Pr\left[
\Dec(\regi{Receiver},\ct)= m
\ \middle |
\begin{array}{ll}
(\pk,\msk)\lrun \Setup(1^\secp)\\
(\mathsf{out}_{\mathsf{Sender}}, \regi{Receiver}) \samp  (\langle\mathsf{Send}(msk), \mathsf{Receive}(1^\lambda)\rangle)\\
\ct \lrun \Enc(\pk,m)
\end{array}
\right] 
\geq 1-\negl(\secp).
\end{align}

In the case of a single-key scheme, we allow the setup and the interactive key-generation to be a single algorithm which outputs the public key $\pk$ as its output.  
\end{definition}

\subsubsection{Technical Lemmata}
We first recall miscellaneous definitions and results related to quantum information.

\begin{definition}[Quantum algorithms with classically determined outputs]\label{def:classically_determined}
We say that a quantum algorithm has classically determined outputs if it takes a classical input and classical randomness and outputs a pure state that is fully determined by its input and randomness.
\end{definition}

\begin{lemma}[\protect{\cite[Lemma~5 simplified]{zhandry2019record}}]\label{lem:qrom}
Let $H$ be a random function $H: \zo^{p(\lambda)} \to \zo^{q(\lambda)}$ where $k(\cdot), p(\cdot), q(\cdot)$ are polynomials. Consider a quantum algorithm $\adve$ that makes $k(\lambda)$ quantum-queries to $H$ and outputs $(x_1, y_1), \dots (x_{k+1}, y_{k+1})$. Then, $\Pr[y_1 = H(x_1) \wedge \dots \wedge y_{k+1} = H(x_{k+1})] \leq \negl(\lambda)$, where the probability is taken over the random choice of $H$ and the inherent randomness of the quantum algorithm $\adve$.
\end{lemma}

\begin{lemma}[\cite{C:BonZha13}]\label{lem:midmeasure}
Let $A$ be a quantum algorithm, and let $\Pr[x]$ be the probability that $A$ outputs x. Let $A'$ be another quantum algorithm obtained from $A$ by pausing $A$ at an arbitrary stage of execution, performing a partial measurement that obtains one of $k$ outcomes, and then resuming $A$. Let $\Pr'[x]$ be the probability $A'$ outputs $x$. Then $\Pr'[x] \geq \Pr[x]/k$.
\end{lemma}

\begin{lemma}[\cite{Z12b}]\label{thm:wise}
    Let $\mathcal{X}$ and $\mathcal{Y}$ be finite sets. Let $\mathcal{H} \subseteq \{h : \mathcal{X} \to \mathcal{Y}\}$ be a family of functions that is $2t$-wise independent. Let $\mathcal{U}$ denote the set of all functions from $\mathcal{X}$ to $\mathcal{Y}$. 
    
    For any quantum algorithm $\mathcal{A}$ making at most $t$ quantum queries,
    \begin{equation}
        \Pr_{h \leftarrow \mathcal{H}}[\mathcal{A}^h = 1] = \Pr_{f \leftarrow \mathcal{U}}[\mathcal{A}^f = 1]
    \end{equation}
    where the probabilities are taken over the choice of the function and the inherent randomness of $\mathcal{A}$.
\end{lemma}

\begin{lemma}[One-Way to Hiding Lemma~\cite{C:AmbHamUnr19}]\label{lem:O2H}
Let $H_0,H_1:X\ra Y$ be any functions and $\qstateq$ be any quantum state, where $(H_0,H_1,\qstateq)$ may have arbitrary joint distribution.
Let $S\subseteq X$ be the subset consisting of all $x\in X$ such that $H_0(x) \ne H_1(x)$. 
Let $\qA$ be a quantum oracle algorithm that makes at most $Q$ quantum queries.
Let ${\qA^\prime}^{H_0}$ be an algorithm that on input $\qstateq$, choose $i\gets[Q]$, runs $\qA^{H_0}(\qstateq)$, measures $\qA$'s $i$-th query, and outputs the measurement outcome.
Then, we have
\begin{align*}
\abs{\Pr[\qA^{H_0}(\qstateq)=1]-\Pr[\qA^{H_1}(\qstateq)=1]} \leq 2Q\cdot\sqrt{\Pr[{\qA^\prime}^{H_0}(\qstateq)\in S]}
\enspace.
\end{align*}
\end{lemma}

\subsubsection{Testing Quantum Adversaries}
We review the notions of projective and threshold implementations and related notions introduced in \cite{TCC:Zhandry20,C:ALLZZ21}. The following exposition is based on that of \cite{EPRINT:KitYam25}.
\begin{definition}[Projective Implementation~\cite{TCC:Zhandry20}]\label{def:projective_implementation}
Let $\cP$ be a binary-outcome POVM
and $\cE$ be a projective measurement indexed by a finite set consisting of elements of $[0,1]$.  
 We consider the following measurement procedure.
 \begin{enumerate}
 \item Apply $\cE$ to obtain an outcome $p\in[0,1]$
 \item Output $1$ with probability $p$ and $0$ otherwise. 
 \end{enumerate}
 We say $\cE$ is the projective implementation of $\cP$, denoted by $\projimp(\cP)$, if the measurement process above is equivalent to $\cP$.
\end{definition} 
\begin{remark}
\cite{TCC:Zhandry20} defines projective implementations for general POVMs, but we focus on \emph{binary-outcome} POVMs for simplicity. 
\end{remark}

\begin{theorem}[{\cite[Lemma 1]{TCC:Zhandry20}}]\label{lem:commutative_projective_implementation}
Any binary outcome POVM $\cP=(\mat{P},\mat{I}-\mat{P})$ has a unique projective implementation $\projimp(\cP)$.
\end{theorem}

\begin{definition}[Threshold Implementation~\cite{TCC:Zhandry20,C:ALLZZ21}]\label{def:threshold_implementation}
Let
$\cP$ be a binary-outcome POVM, 
$\cE$ be a projective implementation $\projimp(\cP)$ of $\cP$, and   $\gamma>0$.
A threshold implementation of $\cP$, denoted by $\TI_\gamma(\cP)$, is the following measurement procedure.
\begin{itemize}
\item Apply $\cE$ to a quantum state and obtain $p$ as an outcome. 
\item Output $1$ if $p\ge \gamma$, and $0$ otherwise.
\end{itemize}
\end{definition}

As a special case of binary-outcome POVMs, 
we introduce the notion of mixture of projective measurements following \cite{TCC:Zhandry20}. 
\begin{definition}[Mixture of Projective Measurement~\cite{TCC:Zhandry20}]\label{def:mixture_projective_measurement}
Let $\cR$ and $\cI$ be finite sets and $D: \cR \ra \cI$ be a function.  
Let $\setbk{(\mat{P}_i,\mat{Q}_i)}_{\in \cI}$ be a collection of binary projective measurements.
The mixture of projective measurements associated to $D$  and $\setbk{(\mat{P}_i,\mat{Q}_i)}_{\in \cI}$ is the binary POVM $\cP_D =(\mat{P}_D,\mat{Q}_D)$ defined as follows.
\begin{align}
& \mat{P}_D = \sum_{i\in\cI}\Pr[i \chosen D(R)]\mat{P}_i && \mat{Q}_D = \sum_{i\in\cI}\Pr[i \chosen D(R)]\mat{Q}_i,
\end{align}
where $R$ is uniformly distributed in $\cR$.
\end{definition}
\begin{remark}
We often treat $D$ as a (efficiently samplable) distribution over $\cI$ rather than as a function. Formally, this should be understood as a function that takes as input the randomness used by the sampler for $D$ and outputs the corresponding sample. 
\end{remark}

While projective and threshold implementations are not known to be efficiently doable, they are known to be efficiently approximable. 
\begin{definition}[Shift Distance]\label{def:shift_distance}
For two distributions $D_0,D_1$, the shift distance with parameter $\epsilon$, denoted by $\shiftdis{\epsilon}(D_0,D_1)$, is the smallest quantity $\delta$ such that for all $x \in \R$:
\begin{align*}
\Pr[D_0\le x] & \le \Pr[D_1\le x + \epsilon] + \delta,&& \Pr[D_0\ge x]  \le \Pr[D_1\ge x - \epsilon] + \delta,\\
\Pr[D_1\le x] & \le \Pr[D_0\le x + \epsilon] + \delta,&& \Pr[D_1\ge x]  \le \Pr[D_0\ge x - \epsilon] + \delta.
\end{align*}
For two real-valued measurements $\cM$ and $\cN$ over the same quantum system, the shift distance between $\cM$ and $\cN$ with parameter $\epsilon$ is
\[
\shiftdis{\epsilon}(\cM,\cN)\seteq \sup_{\ket{\psi}}\shiftdis{\epsilon}(\cM(\ket{\psi}),\cN(\ket{\psi})).
\]
\end{definition}
\begin{definition}[Almost Projective Measurements~\cite{TCC:Zhandry20}]
A real-valued measurement $\mathcal{M}=(M_i)_{i\in I}$ is $(\epsilon,\delta)$-almost projective if the following is true: for any quantum state $\ket{\psi}$, apply $\mathcal{M}$ twice in a row to $\ket{\psi}$, obtaining outcomes $x,y$. Then $\Pr[|x-y|\le \epsilon]\ge 1-\delta$. 
\end{definition}
\begin{lemma}[Approximate Projective Implementation~\cite{TCC:Zhandry20}]\label{lem:API}
For any $\epsilon,\delta\in (0,1)$, 
 any finite sets $\cR$ and $\cI$, 
  any collection of projective measurements $\cP=\setbk{(\mat{P}_i,\mat{I}-\mat{P}_i)}_{i\in\cI}$ and any function $D: \cR \ra \cI$, there exists a measurement procedure $\API_{\cP,D}^{\epsilon,\delta}$ that satisfies the following.
\begin{enumerate}
\item \label{item:API_shiftdis}
$\shiftdis{\epsilon}(\API_{\cP,D}^{\epsilon,\delta},\projimp(\mathcal{\cP_D}))\le \delta$.
\item  \label{item:API_almost_projective}
$\API_{\cP,D}^{\epsilon,\delta}$ is $(\epsilon,\delta)$-almost projective.
\item \label{item:API_run_time}
$\API_{\cP,D}^{\epsilon,\delta}$ makes black-box use of $D$ and $\cP$,  and its running time is $\poly(\epsilon^{-1},\log (\delta^{-1}))$, counting each access to $D$ and $\cP$ as unit time. 
\end{enumerate}
\end{lemma}
\begin{remark}
 The original proof in \cite{TCC:Zhandry20} only upper bounds the \emph{expected} running time of $\API$. However, \cite[Remark 6.4]{TCC:Zhandry20} also sketches an idea to upper bound the \emph{worst-case} running time by slightly modifying the algorithm, which is made formal in \cite{C:CLPY25}. 
\end{remark}

\begin{definition}[Approximate Threshold Implementation \cite{C:ALLZZ21}]\label{def:ATI}
For any $\epsilon,\delta,\cR,\cI,\cP,D$ as in \Cref{lem:API} and any $\gamma \in (0,1)$,  
$\ATI_{\cP,D,\gamma}^{\epsilon,\delta}$ is the measurement procedure that works similarly to $\TI_\gamma(\cP_D)$ except that it uses $\API_{\cP,D}^{\epsilon,\delta}$ instead of $\projimp(\cP_D)$.  
\end{definition}

\if0
\begin{theorem}[\cite{TCC:Zhandry20,C:ALLZZ21}]\label{thm:ind_distribution_TI}
Let
\begin{itemize}
\item $t>0$
\item $\cP$ be a collection of projective measurements indexed by some sets
\item $\rho$ be an efficiently constructible mixed state
\item $D_0$ and $D_1$ be two efficienctly samplable and computationally indistinguishable distributions over $\cI$.
\end{itemize}
For any inverse polynomial $\epsilon$, there exists a negligible function $\delta$ such that
\[
\Tr[\TI_{t -\epsilon}(\cP_{D_1})\rho] \ge \Tr[\TI_t(\cP_{D_0})\rho] - \delta,
\]
where $\cP_{D_\coin}$ is the mixture of projective measurements associated to $\cP$, $D_\coin$, and $\coin \in \zo$.
\end{theorem}
\fi

\begin{lemma}[Special case of {\cite[Lemma 5.1]{TCC:AnaHuHua24}}\footnote{This is proven based on a result of
\cite{TCC:Zhandry20}.}]\label{lem:ATI_gap_to_dist} 
Let $\mathcal{P}$ be a collection of projective measurements indexed by some set $\mathcal{X}$. 
Suppose $\mathcal{P}$ can be implemented by a quantum circuit of size $|\mathcal{P}|$. Let $D_0, D_1$ be two efficiently computable functions from $\cR$ to $\mathcal{X}$.  
Then there exists a quantum algorithm $\cC$ of running time
$\mathrm{poly}\bigl(\lambda, 1/\varepsilon, \log(1/\delta), \mu, |\mathcal{P}|\bigr)$ 
(which only uses $\mathcal{P}$ and $D_0, D_1$ as a black box) that satisfies the following.
For any polynomial $\mu$ and any quantum state $\rho$, %
if
\[
\left|
\Pr\!\left[
 \ATI^{\epsilon,\delta}_{\mathcal{P},D_0,\gamma}(\rho)=1
\right]
-
\Pr\!\left[
\ATI^{\epsilon,\delta}_{\mathcal{P},D_1,\gamma}(\rho)=1
\right]\right|
\ge \frac{1}{\mu(\lambda)},
\]
then 
\[
\Pr
\bigl[ \cC(\rho,x) = \coin : \coin \gets\{0,1\},x \leftarrow D_\coin \bigr]
\ge
\frac{1}{2}+
\frac{1}{(\mu(\lambda))^3 \cdot \mathrm{poly}(\lambda, 1/\varepsilon, \log(1/\delta))}. 
\]
\begin{remark}
In~\cite{TCC:AnaHuHua24}, the conclusion is stated with an absolute value, namely,
\[
\left|
\Pr\bigl[ \cC(\rho,x)=\coin : \coin\gets\{0,1\},\, x\leftarrow D_\coin \bigr]
-\frac{1}{2}
\right|
\ge 
\frac{1}{(\mu(\lambda))^3 \cdot \poly(\lambda,1/\varepsilon,\log(1/\delta))}.
\]
However, the absolute value can be removed by an easy transformation.
Suppose that $\cC$ satisfies the above inequality.
We construct an algorithm $\cC'$ as follows.
First, $\cC'$ uses $\API$ (with sufficiently small parameters $\epsilon=1/\poly(\secp)$ and $\delta=\negl(\secp)$) to estimate
\(
\Pr[ \cC(\rho,x)=\coin : \coin\gets\{0,1\},\, x\leftarrow D_\coin ].
\)
Then $\cC'$ runs $\cC$ on the post-measurement state and outputs the result of $\cC$ if the estimated probability is at least $1/2$, and its negation otherwise.
It is straightforward to verify that $\cC'$ satisfies the similar inequality without the absolute value, with a slightly weaker bound, which is absorbed into the hidden constant of the $\poly$ term. 
\end{remark}
\end{lemma}

\newpage
\part{Tools}
\section{Succinct Purification Compiler for State-Query Games}\label{sec:purification_compiler}
In this section, we give a purification compiler that transforms any \emph{state-query game} that involves adversary querying multiple samples from a mixed quantum state into a game where instead multiple identical copies of a pure state is used, while preserving security. Our compiler is succinct: The size of the keys held by the challenger for purification does not depend on the number of state-queries by the adversary each round.

We now define the notion of a state-query game.
\begin{definition}[State-Query Game]\label{def:state-query-game}
Let state-generator $\mathsf{GenState}$ be a QPT algorithm that has classically determined outputs (\Cref{def:classically_determined}) and let $\mathsf{Chal}$ be a challenger (i.e. a stateful QPT algorithm).

A \textbf{state-query game} defined by $\mathsf{GenState}$ and $\mathsf{Chal}$ is the following experiment (parametrized by a QPT adversary $\adve$).

\begin{mdframed}
    \underline{$\mathsf{StateQueryGame}_{\mathsf{GenState}, \mathsf{Chal}, \adve}(1^\lambda)$}
    \begin{enumerate}
        \item Until the challenger outputs a termination message:
        \begin{enumerate}[label=\arabic*.]
            \item The adversary $\adve$ sends a quantum register $\reg$ to the challenger $\mathsf{Chal}$.
            \item The challenger $\mathsf{Chal}$ outputs descriptions $desc_1, \dots, desc_t$ and random strings $r_1, \dots, r_t$ for some $t$, or a termination message along with a decision bit (in this case, end the loop).
            \item For $i \in [t]$, sample $\reg_i \samp \mathsf{GenState}(desc_i; r_i)$.
            \item Submit $\reg_1, \dots, \reg_t$ to the adversary $\adve$.
        \end{enumerate}
        \item Output the decision bit of the challenger.
    \end{enumerate}
\end{mdframed}
\end{definition}

Finally we state our main theorem.
\begin{theorem}[Purifying State Games]\label{thm:compiler} 
Consider a state-query game defined by some state-query generator $\mathsf{GenState}$ and a challenger $\mathsf{Chal}$, and let $r(\lambda), n(\lambda)$ denote the length of the randomness and output size of $\mathsf{GenState}$ respectively. Let $\mathsf{PRS}$ be a pseudorandom state scheme (\cref{def:prs}) with output size $\lambda$-qubits, and let $\prf$ be a pseudorandom function scheme (\cref{def:prf}) with input size $\lambda$ and output size $r(\lambda)$.

Define the modified state-generator $\mathsf{GenState}'$ and modified challenger $\mathsf{Chal}'$ as follows:
\begin{mdframed}
\underline{$\mathsf{GenState}'(desc)$}
\begin{enumerate}
    \item Parse $(K_{PRS}, K_{PRF}, desc') = desc$.
    \item Set $\reg_1 = \prs.\mathsf{GenState}(K_{PRS})$. Let $\sum_{x} \alpha_{K_{PRS},x}\ket{x}$ denote the state in the register.
    
    \item Set $\reg_2 = \ket{0^{r(\lambda)}}$.
    \item Coherently evaluate the PRF by running $\prf.\mathsf{Eval}(K_{PRF}, \cdot)$ with $\reg_1$ as the input register and $\reg_2$ as the output register to obtain the state $\sum_{x} \alpha_{K_{PRS},x}\ket{x}_{\reg_1}\otimes \ket{F(K_{PRF}, x)}_{\reg_2}$ over $\reg_1, \reg_2$.
    \item Coherently run $\mathsf{GenState}$ with $\reg_2$ as the input register, and let $\reg_3$ denote the output register. The state over $\reg_1, \reg_2, \reg_3$ is $\sum_{x} \alpha_{K_{PRS},x}\ket{x}_{\reg_1}\otimes$ $\ket{F(K_{PRF}, x)}_{\reg_2} \otimes$ $\ket{\mathsf{GenState}(desc'; F(K_{PRF},x))}_{\reg_3}$.
    \item Coherently evaluate the PRF to uncompute the register $\reg_2$.
    \item Output $\reg_1, \reg_3$ which has the state $\sum_{x} \alpha_{K_{PRS},x}\ket{x}_{\reg_1} \otimes \ket{\mathsf{GenState}(desc'; F(K_{PRF},x))}_{\reg_3}$
\end{enumerate}
\end{mdframed}

\begin{mdframed}
\underline{$\mathsf{Chal}'$}
\begin{enumerate}
\item Let $\regi{adv}$ be the input from the adversary.
\item Sample new keys $K_{PRF,i} \samp \prf.\mathsf{Gen}(1^\lambda)$ and $K_{PRS,i} \samp \prs.\mathsf{Setup}(1^\lambda)$ where $i$ denotes the round number.
\item Sample $(desc_1, \dots, desc_t), (r_1, \dots, r_t) \samp \mathsf{Chal}(\regi{state}', \regi{adv})$. 
\item Output $((K_{PRF,i}, K_{PRS,i}, desc_1)$, $\dots$, $(K_{PRF,i}, K_{PRS,i}, desc_t)), (\bot_j)_{j \in [t]}$.
\end{enumerate}
\end{mdframed}

Then, for any QPT adversary $\adve'$, there is a QPT adversary $\adve$ such that $$\left|\Pr[\mathsf{StateQueryGame}_{\mathsf{GenState}', \mathsf{Chal}', \adve'}=1] - \Pr[\mathsf{StateQueryGame}_{\mathsf{GenState}, \mathsf{Chal}, \adve} = 1]\right|\leq \negl(\lambda).$$  Thus, the original security guarantee holds even for the modified game.

     If the number of states queried the adversary each rounds is a-priori bounded by some fixed polynomial $q(\lambda)$ for all admissible adversaries, then the result holds unconditionally (i.e. without assuming pseudorandom functions and pseudorandom states) by replacing $\prs$ with a statistical $q(\lambda)$-design and replacing $\prf$ with a $2q(\lambda)$-wise independent function. In this case, the keys $K_{PRF}$ and $K_{PRS}$ will be the descriptions of the $2q(\lambda)$-wise independent function and the $q(\lambda)$-design respectively and thus their size will depend on $q(\lambda)$.

\end{theorem}
We prove this theorem in \cref{sec:haar-proof}.
\ifconference
\section{Proof of \cref{thm:compiler}}\label{sec:haar-proof}
\else
\subsection{Proof of \cref{thm:compiler}}\label{sec:haar-proof}
\fi

First, we introduce some notation and recall a technical lemma. Let $\mathcal{H}_{n}$ denote the Haar random state distribution on $n$ qubits and let $S_t$ denote the set of permutations on $[t]$. For any distinct $x_1,\dots,x_t \in \zo^n$, define $\ket{\mathsf{type}(x_1, \dots, x_t)}$ to be
\begin{equation*}
\ket{\mathsf{type}(x_1, \dots, x_t)} =\frac{1}{\sqrt{t!}}\sum_{\pi \in S_t}\bigotimes_{i \in [t]}\ket{x_{\pi(i)}}.
\end{equation*}

As observed by \cite{ananth2025simultaneous}, the following lemma follows by the well-known connection between type states (\cite{harrow2013church}) and a simple birthday bound.
\begin{lemma}\label{lem:typelem}
    For any $n,t \in \N^+$,
    \begin{equation*}
        \left| \E_{\ket{\psi} \samp \mathcal{H}_{n}} (\ketbra{\psi}{\psi})^{\otimes t} - \E_{\substack{x_1,\dots,x_t \in \zo^{n}\\ \text{all distinct}}}(\ketbra{\mathsf{type}(x_1, \dots, x_t)}{\mathsf{type}(x_1, \dots, x_t)})^{\otimes t} \right| \leq O\left(\frac{t^2}{2^n}\right).
    \end{equation*}
\end{lemma}

Now we move onto the proof. We will prove the general case, and the bounded case follows similarly.
    Fix a QPT adversary $\adve'$. Let $\hyb_0$ denote $\mathsf{StateQueryGame}_{\mathsf{GenState}', \mathsf{Chal}', \adve'}$. 

    Now define $\hyb_1$ by modifying $\hyb_0$ as follows: For each round $i$, we sample a new $2t$-wise independent function $H_i$ (which is included as part of the message from the challenger to state-generator) and the state-generator uses it instead of $\prf.\mathsf{Eval}$. We have $\hyb_0 \approx \hyb_1$ by the security of PRF and by  \cref{thm:wise}.
    
    Now define $\hyb_2$ by modifying $\hyb_1$ as follows: For all rounds $i$, sample $\ket{\psi_i'} \samp \mathcal{H}_\lambda$ and instead of running $\mathsf{GenState}'$, instead run the modified version $\mathsf{GenState}''$, where we set $\reg_1$ to $\ket{\psi_i'}$ instead of using $\prs.\mathsf{GenState}$. By security of the PRS scheme, we have $\hyb_1 \approx \hyb_2$. Note that since we sample Haar states, this hybrid is inefficient.
    
    We now define $\hyb_3$ by modifying $\hyb_2$ as follows. For each round $i$, sample $x^i_1, \dots, x^i_{t}$, each from $\zo^{\lambda}$ subject to all being distinct (otherwise independent). Then, instead of using $\mathsf{GenState}''$ for $t$ times, we simply submit 
    \begin{equation*}
        \sum_{\pi \in S_t} \ket{x^i_{\pi(1)}} \ket{\mathsf{GenState}(desc_1; H_i(x^i_{\pi(1)}))} \otimes \dots \ket{x^i_{\pi(t)}} \ket{\mathsf{GenState}(desc_1; H_i(x^i_{\pi(t)}))}
    \end{equation*}
    to the adversary.
    We have $\hyb_2 \approx \hyb_3$ by \cref{lem:typelem}. 

    We now modify $\hyb_3$ as follows to obtain $\hyb_4$: Instead of sampling $H_i$ as $2t$-wise independent functions, now sample them as truly random functions. Note that we can perform this sampling lazily since the function is only ever evaluated at the polynomial number of points $x^i_1, \dots, x^i_t$ for $i \in [\ell]$  where $\ell$ is the number of rounds in the game. $\hyb_3 \approx \hyb_4$ follows by \cref{thm:wise}.

    We define $\hyb_5$ by modifying $\hyb_4$ as follows: We sample a sequence of random strings $r^i_1, \dots, r^i_t$ for $i \in [\ell]$, and instead of using $H_i(x^i_j)$, we now use $r^i_j$. $\hyb_4 \approx \hyb_5$ follows immediately.

    Finally, we define the adversary $\adve$ as follows. During each round, when it receives the states $\ket{\phi^i_1}, \dots, \ket{\phi^i_t}$ from the challenger, it samples $x_1^i, \dots, x_t^i$ independently but distinct (as above) and it prepares the state \begin{equation*}
       \ket{\eta} = \sum_{\pi \in S_t} \ket{x^i_{\pi(1)}}\otimes\ket{\phi^i_{\pi(1)}}\otimes\dots\otimes \ket{x^i_{\pi(t)}}\otimes\ket{\phi^i_{\pi(t)}}
    \end{equation*} 
    and passes it to $\adve'$. It is easy to see that $\hyb_5 \equiv \mathsf{StateQueryGame}_{\mathsf{GenState}, \mathsf{Chal}, \adve}$. Combining this with above, we get $\mathsf{StateQueryGame}_{\mathsf{GenState}, \mathsf{Chal}, \adve} \approx \mathsf{StateQueryGame}_{\mathsf{GenState}', \mathsf{Chal}', \adve'}$

    We also note that $\adve$ is QPT. To see this, first, note that, as shown by \cite{berry2018improved}, we can efficiently create the state $\sum_{\pi \in S_t}\ket{\pi(1)}\dots \ket{\pi(t)}$ using the quantum version of the Fisher-Yates (FY) algorithm. Once we create the state $\ket{\eta_1} = \sum_{\pi \in S_t}\ket{\pi(1)}_{P_1}\dots \ket{\pi(t)}_{P_t}$, we can then coherently run a controlled SWAP algorithm to permute the registers carrying $\ket{x^i_1}\ket{\phi^i_1}, \dots, \ket{x^i_t}\ket{\phi^i_t}$ and obtain $$\ket{\eta_2} = \sum_{\pi \in S_t} \ket{\pi(1)}_{P_1}\otimes\ket{x^i_{\pi(1)}}\otimes\ket{\phi^i_{\pi(1)}}\otimes\dots\otimes \ket{\pi(t)}_{P_1}\otimes\ket{x^i_{\pi(t)}}\otimes\ket{\phi^i_{\pi(t)}}$$ Finally, since $x^i_1, \dots, x^i_t$ are distinct, we can map any $x_j$ to $j$, which means that we can read the $x$ registers and coherently uncompute the permutation registers ${P_1}, \dots, {P_t}$ to obtain $$\ket{\eta_2} = \sum_{\pi \in S_t} \ket{0^t}_{P_1}\otimes\ket{x^i_{\pi(1)}}\otimes\ket{\phi^i_{\pi(1)}}\otimes\dots\otimes \ket{0^t}_{P_1}\otimes\ket{x^i_{\pi(t)}}\otimes\ket{\phi^i_{\pi(t)}}$$ By tracing out the zero registers, we perfectly create the state $\ket{\eta}$.

    This completes the proof.

\newpage
\section{Collusion-Resistance Compiler for Quantum Decryption Keys}\label{sec:fe-compiler} 
In this section, we present a collusion-resistance compiler based on  classical FE for upgrading PKE schemes with quantum decryption keys (PKE-QDK) to collusion-resistant secure ones (under various security models). Here, PKE-QDK refers to the common syntax; SDE, SKL, and LOCC leakage resilience impose different security requirements on this syntax.
Later, in \cref{sec:CR-transformation}, \cref{sec:skl}, \cref{sec:lr}, we show that our compiler provides collusion-resistant security under copy-protection, key leasing and LOCC leakage-resilience models, respectively. The latter security notions will follow from \cref{thm:fe-coll_res}, whereas for copy-protection (\cref{sec:compiler_SDE}) we need a separate proof using ideas from the proof of \cref{thm:fe-coll_res}.

\subsection{Definitions}\label{sec:qprotdefn}

\subsection{Construction}\label{sec:compitself}
In this section, we present our compiler. See \cref{sec:CR-transformation}, \cref{sec:lr}, \cref{sec:skl} for instantiation of our scheme in the contexts of copy-protection, LOCC leakage-resilience and secure key-leasing.

Let $\One = \One.(\Setup,\qKG,\Enc,\qDec)$ be a \emph{single-key} PKE scheme with quantum decryption keys. We construct our scheme $\pke =(\Setup,\qKG,\Enc,\qDec)$ using the following building blocks.

\begin{itemize}
\item Adaptively secure classical public-key functional encryption scheme $\FE = \FE.(\Setup,\KG,\Enc,\Dec)$. 
\item Puncturable pseudorandom function $\PRF=(\prfgen,\prfF,\Puncture)$.
\end{itemize}
We set $\ell_\pk \seteq \abs{\one.\pk}$ and $\ell_\ct \seteq \abs{\one.\ct}$, where $\abs{\one.\pk}$ and $\abs{\one.\ct}$ are the public key and ciphertext length of $\One$, respectively. Let $\ell_\msg$ be the plaintext length.

\begin{description}
\item[$\Setup(1^\secp)$:] $ $
\begin{enumerate}
    \item Generate $(\sffe.\pk,\sffe.\msk) \lrun \FE.\Setup(1^\secp)$.
    \item Output $(\pk,\msk)\seteq (\sffe.\pk,\sffe.\msk)$.
\end{enumerate}
\item[$\qKG(\msk)$:] $ $
\begin{enumerate}
	\item Parse $\msk=\sffe.\msk$ and generate $(\one.\pk,\one.\msk)\lrun \One.\Setup(1^\secp)$.
	\item Generate $\one.\qsk \lrun \One.\qKG(\one.\msk)$. 
    \item Choose a tag $t\gets \bit^\secp$. 
	\item Construct a circuit $\RE[t,\one.\pk]$, which is described in~\cref{fig:func_reenc}.\footnote{The inequalities between tags means those in lexicographical order.} %
	\item Generate $\sffe.\fsk \lrun \FE.\KG(\sffe.\msk,\RE[t,\one.\pk])$.
	\item Output $\qsk \seteq (\one.\qsk,\sffe.\fsk), \mathsf{aux} \seteq \one.\msk$. 
\end{enumerate}
\item[$\Enc(\pk,\msg)$:] $ $
\begin{enumerate}
\item Parse $\sffe.\pk =\pk$.
\item Choose $\prfkey \chosen \prfgen(1^\secp)$.
\item Set $x \seteq \msg\concat \prfkey\concat 0\concat 0^\secp\concat 0^{\ell_\ct}$.
\item Compute $\sffe.\ct \lrun \FE.\Enc(\sffe.\pk,x)$.
\item Output $\ct \seteq \sffe.\ct$.
\end{enumerate}
\item[$\qDec(\qsk,\ct)$:] $ $
\begin{enumerate}
\item Parse $\qsk = (\one.\qsk,\sffe.\fsk)$ and $\ct =\sffe.\ct$.
\item Compute $\one.\ct \lrun \FE.\Dec(\sffe.\fsk,\sffe.\ct)$.
\item Compute and output $\msg^\prime \lrun \One.\qDec(\one.\qsk,\one.\ct)$.
\end{enumerate}
\end{description}

In the case where $\one$ is a classical-vendor scheme, our scheme can also be made classical-vendor by the following protocol. Let $\ell_{rounds}(\lambda)$ denote the number of rounds of the protocol of $\one$.
\begin{description}
\item[$\mathsf{Send}(msk)$:] $ $
\begin{enumerate}
\item \textbf{Until round $\ell_{rounds}(\lambda)$:} Emulate $\One.\mathsf{Send}$.
\item \textbf{Round $\ell_{rounds}(\lambda)$:} Emulate $\One.\mathsf{Send}$ to obtain the message $msg$ and the output $\pk^*$, then send $(msg, \sffe.\fsk)$ where $\sffe.\fsk$ is sampled as follows.
\begin{enumerate}[label=\arabic*.]
    \item Parse $\msk=\sffe.\msk$
    \item Choose a tag $t\gets \bit^\secp$. 
	\item Construct a circuit $\RE[t,\One.\pk]$, which is described in~\cref{fig:func_reenc}.
	\item Generate $\sffe.\fsk \lrun \FE.\KG(\sffe.\msk,\RE[t,\One.\pk])$.
\end{enumerate}
\end{enumerate}

\item[$\mathsf{Receive}(1^\lambda)$:] $ $
\begin{enumerate}
\item \textbf{Until round $\ell_{rounds}(\lambda)$:} Emulate $\One.\mathsf{Send}$.
\item \textbf{Round $\ell_{rounds}(\lambda)$:} Let $msg$ be the message received, parse it as $(msg', \fsk)$. Emulate $\one.\mathsf{Receive}$ with $msg'$, let $\regi{Receiver}^*$ be the output (which is output of the protocol of $\One$).
\item Output $(\regi{Receiver}, \fsk)$.
\end{enumerate}
\end{description}

\protocol
{Function $\RE[t,\one.\pk](\msg\concat \prfkey\concat \sfmode \concat t^\prime \concat \one.\ct^\ast)$}
{The description of $\RE[\one.\pk]$}
{fig:func_reenc}
{
\begin{description}
\setlength{\parskip}{0.3mm} %
\setlength{\itemsep}{0.3mm} %
\item[Constants:] Single-key SDE public key $\one.\pk$.
\item[Input:] Plaintext $\msg$, PRF key $\prfkey$, mode number $\sfmode$, public key $\one.\pk^\prime$, single-key SDE ciphertext $\one.\ct^\ast$.
\end{description}
\begin{enumerate}
	\item If $\sfmode=0$, do the following:
	\begin{itemize}
		\item Compute and output $\one.\ct \seteq \One.\Enc(\one.\pk,\msg;\prfF_\prfkey(t))$.
	\end{itemize}
	\item If $\sfmode =1$, do the following:
	\begin{itemize}
		\item If $t \le t^\prime$, compute and output $\one.\ct \seteq \One.\Enc(\one.\pk,0^{\ell_\msg};\prfF_\prfkey(t))$.
		\item If $t > t^\prime$, compute and output $\one.\ct \seteq \One.\Enc(\one.\pk,\msg;\prfF_\prfkey(t))$.
	\end{itemize}
	\item If $\sfmode =2$, do the following:
	\begin{itemize}
		\item If $t < t^\prime$, compute and output $\one.\ct \seteq \One.\Enc(\one.\pk,0^{\ell_\msg};\prfF_\prfkey(t))$.
		\item If $t = t^\prime$, output $\one.\ct^\ast$.
		\item If $t > t^\prime$, compute and output $\one.\ct \seteq \One.\Enc(\one.\pk,\msg;\prfF_\prfkey(t))$.
	\end{itemize}
\end{enumerate}
}

Correctness of the above scheme  immediately follows from the correctness of $\One$ and $\FE$. 

The following is a technical reduction statement, rather than a security definition for the common PKE-QDK syntax. It isolates one underlying key instance that accounts for a distinguisher's advantage; the SKL and leakage-resilience applications specify their challengers and then use the corresponding single-key security notion.
\begin{theorem}\label{thm:fe-coll_res}
    Let $\mathsf{Chal}$ be a challenger that simply invokes and receives the output of $\pke.\KG(\msk)$ (or $\pke.\mathsf{Send}(\msk)$ in case of classical-vendor schemes) for polynomially many times but does not have access to the master secret key $\msk$ otherwise. Let $\mathcal{D}(m)$ denote the ciphertext distribution where we sample $\ct \samp \pke.\Enc(\pk, m)$.
    
    Consider the following experiment between a QPT adversary $\adve = (\adve_1, \adve_2)$ and the challenger.

\begin{mdframed}
    \underline{$\mathsf{QPKGame}_{\pke, \mathsf{Chal}, \adve}(1^\lambda)$}
\begin{enumerate}
\item Sample $\pk, \msk \samp \pke.\Setup(1^\lambda)$. Submit $\pk$ to the challenger.
    \item The adversary $\adve_1$ and the challenger $\mathsf{Chal}$ play their game, and the challenger outputs a pair of messages $m_0, m_1$ or $\bot$ (in which case the experiment terminates with output $0$). The adversary $\adve_1$ outputs an internal state register $\regi{int}$.
    \item Sample $b \in \zo$.
    \item Sample $\ct \samp\mathcal{D}(m_b)$.
    \item Submit $\regi{int}, \ct$ to the adversary $\adve_2$, and it outputs a bit $b'$.
    \item Output $1$ if $b= b'$, and otherwise output $0$.
\end{enumerate}
    \end{mdframed}

    For a ciphertext distribution $\mathcal{D}^*$, let $\mathsf{QPKGame}^{\mathcal{D}^*}_{\pke, \mathsf{Chal}, \adve}(1^\lambda)$ denote the modified version of $\mathsf{QPKGame}_{\pke, \mathsf{Chal}, \adve}(1^\lambda)$ where we sample the challenge ciphertext from ${\mathcal{D}^*}$ instead of $\mathcal{D}$.

In the experiment, let $q(\lambda)$ denote the number of times the challenger invokes $\pke.\KG(\msk)$, and let $t_1, \dots, t_q$ and $\pk_1, \dots, \pk_q$ denote the tags and public-keys sampled during invocations of $\pke.\KG$. 

For $j \in [q(\lambda)]$, define the following distribution.
\begin{mdframed}
    \underline{$\mathcal{D}_{j}(m)$}
    \begin{enumerate}
        \item Sample $\ct_{inner} \samp \One.\Enc(\pk_{j}, m_b)$.
    \item Sample $\prfkey \chosen \prfgen(1^\secp)$.
    \item Sample $\prfkey_{\ne t_j} \lrun \Puncture(\prfkey,\setbk{t_j})$
    \item Sample $\ct \lrun \FE.\Enc(\sffe.\pk,m_b \concat \prfkey_{\ne t_j}\concat 2 \concat t_{j}\concat \ct_{inner})$
    \item Output $\ct$.
    \end{enumerate}
    \end{mdframed}

Suppose $\Pr[\mathsf{QPKGame}_{\pke, \mathsf{Chal}, \adve} = 1] \geq \frac{1}{2} + \epsilon(\lambda)$.

If $\FE$ satisfies collusion-resistant functional encryption security and $\PRF$ satisfies puncturable PRF security, then, there exists QPT $\adve_2'$ such that $\Pr[\mathsf{QPKGame}^{D_{j^*}}_{\pke, \mathsf{Chal}, (\adve_1, \adve_2')} = 1] \geq \frac{1}{2} + \frac{\epsilon(\lambda)}{16q(\lambda)}$ for some $j^* \in [q(\lambda)]$. 

Note that the ciphertexts from distributions $\mathcal{D}_j(m_b)$ can be simulated using a ciphertext $ct = \one.\Enc(pk_j, m_b)$, without knowing the message $m_b$ or the randomness for the ciphertext. For this, we simply simulate the  algorithm of $\mathcal{D}_j$ by using $ct$ as $ct_{inner}$.

If the number of key generations $q(\lambda)$ is a-priori bounded and $\FE$ satisfies $q(\lambda)$-bounded functional encryption security and $\PRF$ satisfies puncturable PRF security, then, $\Pr[\mathsf{QPKGame}^{D_{j^*}}_{\pke, \mathsf{Chal}, (\adve_1, \adve_2')} = 1] \geq \frac{1}{2} + \frac{\epsilon(\lambda)}{16q(\lambda)}$ for some $j^* \in [q(\lambda)]$.

\end{theorem}
We prove this theorem in \cref{sec:fe-based-proof}.

The fact that our compiler preserves secure key leasing and LOCC leakage-resilient security will follow from this theorem in a straightforward way (see \cref{sec:skl} and \cref{sec:lr}).

For copy-protection security of the compiler, we need a separate proof, which is given in \cref{sec:compiler_SDE}.
\ifconference
\section{Proof of \cref{thm:fe-coll_res}}\label{sec:fe-based-proof}
\else
\subsection{Proof of \cref{thm:fe-coll_res}}\label{sec:fe-based-proof}
\fi

In this section, we prove \cref{thm:fe-coll_res}.

For convenience, we define $t_0\seteq 0^\secp$ and $t_{q+1} = 1^\lambda$. With overwhelming probability, $t_0,t_1,\ldots,t_q,t_{q+1}$ are pairwise distinct, thus we assume so below. Without loss of generality, we assume $t_1<\cdots<t_{q}$, which can be ensured by an appropriate relabeling if necessary. Now we define the following distributions for $j \in \{0, \dots, q+1\}$.

\begin{mdframed}
    \underline{$\mathcal{D}'_{j}(m)$}
    \begin{enumerate}
    \item Sample $\prfkey \chosen \prfgen(1^\secp)$.
    \item Sample $\ct \lrun \FE.\Enc(\sffe.\pk,m \concat \prfkey\concat 1\concat t_{j}\concat 0^{\ell_{\ct}})$
    \item Output $\ct$.
    \end{enumerate}
    \end{mdframed}

We claim the following.
\begin{lemma}
$\mathsf{QPKGame}_{\pke, \mathsf{Chal}, \adve} \approx \mathsf{QPKGame}^{D'_{0}}_{\pke, \mathsf{Chal}, \adve}$    
\end{lemma}
\begin{proof}
By the adaptive security of $\FE$, we have $D \cind D_0'$ since $$\RE[t_k,\one.\pk_k](\msg \concat \prfkey\concat 1\concat 0^\secp\concat 0^{\ell_\ct})=\RE[t_k,\one.\pk_k](\msg\concat \prfkey\concat 0\concat 0^\secp\concat 0^{\ell_\ct})$$ for all $k\in[q]$.
\end{proof}

\begin{lemma}
For any $\adve$,
$\Pr[\mathsf{QPKGame}^{D'_{q}}_{\pke, \mathsf{Chal}, \adve} = 1] \leq 1/2 + \negl(\lambda)$.    
\end{lemma}
\begin{proof}
Define the following distribution $D''$. 
\begin{description}
	\item[$D''$:] On input $\msg$:
    \begin{enumerate}
    \item Sample $\prfkey \chosen \prfgen(1^\secp)$.
    \item Generate $\ct \lrun \FE.\Enc(\sffe.\pk,0^{\ell_\msg}\concat \prfkey\concat 1\concat 1^{\lambda} \concat 0^{\ell_\ct})$.
\end{enumerate}  
\end{description}

Then, the result follows by the adaptive security of $\FE$ since
$$\RE[t_k,\one.\pk_k](\msg\concat \prfkey\concat 1\concat t_q\concat 0^{\ell_\ct})=\RE[t_k,\one.\pk_k](0^{\ell_\msg}\concat \prfkey\concat 1\concat 1^\secp\concat 0^{\ell_\ct})$$  for all $k$ and since $D''$ does not use the message $\msg$ (thus any adversary has zero advantage over $1/2$).
\end{proof}

Now, by a simple hybrid argument, there must be $j^* \in [q]$ such that \begin{equation*}
    \left|\Pr[\mathsf{QPKGame}^{D'_{j}}_{\pke, \mathsf{Chal}, \adve} = 1] - \Pr[\mathsf{QPKGame}^{D'_{j-1}}_{\pke, \mathsf{Chal}, \adve} = 1]\right| \geq \frac{\epsilon}{4q}
\end{equation*}

Now we define the following distributions. For each $j\in [q]$, we consider the following distributions where $\cR_{\One}$ is the randomness space of $\One$.
\begin{description}
	\item[$D_{j-1}^{1}$:] On input $\msg$:
	\begin{enumerate}
	\item Sample  $\prfkey \lrun \prfgen(1^\secp)$, and  $\one.\ct^\ast \lrun \One.\Enc(\one.\pk_j,\msg;\prfF_{\prfkey}(t_j))$.
	\item Generate $\ct \lrun \FE.\Enc(\sffe.\pk,\msg\concat \prfkey\concat 2\concat t_{j} \concat \one.\ct^\ast)$.
	\end{enumerate}
	\item[$D_{j-1}^2$:] On input $\msg$:
	\begin{enumerate}
	\item Sample$\prfkey \lrun \prfgen(1^\secp)$, and  $\one.\ct^\ast \lrun \One.\Enc(\one.\pk_j,\msg;\prfF_{\prfkey}(t_j))$.
	\item Generate $\ct \lrun \FE.\Enc(\sffe.\pk,\msg \concat \prfkey_{\ne t_j}\concat 2 \concat t_{j}\concat \one.\ct^\ast)$,\\ where  $\prfkey_{\ne t_j} \lrun \Puncture(\prfkey,\setbk{t_j})$.
	\end{enumerate} 
	\item[$D_{j-1}^3$:] On input $\msg$:
	\begin{enumerate}
	\item Sample  $r \lrun \cR_{\One}$, and  $\one.\ct^\ast \lrun \One.\Enc(\one.\pk_j,\msg;r)$.
	\item  Generate $\ct \lrun \FE.\Enc(\sffe.\pk,\msg \concat \prfkey_{\ne t_j}\concat 2 \concat t_{j}\concat \one.\ct^\ast)$,\\ where  $\prfkey_{\ne t_j} \lrun \Puncture(\prfkey,\setbk{t_j})$.
	\end{enumerate} 
 	\item[$D_{j}^{-1}$:] On input $\msg$:
 	\begin{enumerate}
 	 	\item Sample $\prfkey \lrun \prfgen(1^\secp)$ and $\one.\ct^\ast \lrun \One.\Enc(\one.\pk_j,0^{\ell_\msg};\prfF_{\prfkey}(t_j))$.
 	 	\item  Generate $\ct \lrun \FE.\Enc(\sffe.\pk,\msg\concat \prfkey\concat 2\concat t_{j} \concat \one.\ct^\ast)$.
 	 \end{enumerate}
	\item[$D_{j}^{-2}$:] On input $\msg$:
	\begin{enumerate}
		\item Sample $\prfkey \lrun \prfgen(1^\secp)$ and $\one.\ct^\ast \lrun \One.\Enc(\one.\pk_j,0^{\ell_\msg};\prfF_{\prfkey}(t_j))$.
		\item Generate $\ct \lrun \FE.\Enc(\sffe.\pk,\msg \concat \prfkey_{\ne t_j}\concat 2 \concat t_j\concat \one.\ct^\ast)$,\\ where $\prfkey_{\ne t_j} \lrun \Puncture(\prfkey,\setbk{t_j})$.
	\end{enumerate}
	\item[$D_{j}^{-3}$:] On input $\msg$:
	\begin{enumerate}
	 	\item Sample $r\chosen \cR_{\One}$ and $\one.\ct^\ast \lrun \One.\Enc(\one.\pk_j,0^{\ell_\msg};r)$.
	 	\item  Generate $\ct \lrun \FE.\Enc(\sffe.\pk,\msg\concat \prfkey_{\ne t_j}\concat 2 \concat t_j\concat \one.\ct^\ast)$,\\ where $\prfkey_{\ne t_j} \lrun \Puncture(\prfkey,\setbk{t_j})$.
	 \end{enumerate} 
\end{description}

\begin{lemma}\label{lem:sub_hyb_ind_unbounded-a}
For $j\in[q]$, 
it holds that $\mathsf{QPKGame}^{D'_{j-1}}_{\pke, \mathsf{Chal}, \adve}  \cind \mathsf{QPKGame}^{D^3_{j-1}}_{\pke, \mathsf{Chal}, \adve} $ and $\mathsf{QPKGame}^{D'_{j}}_{\pke, \mathsf{Chal}, \adve}  \cind\mathsf{QPKGame}^{D^{-3}_{j}}_{\pke, \mathsf{Chal}, \adve} $
\end{lemma}
\begin{proof}
By adaptive security of $\FE$ we have $\mathsf{QPKGame}^{D'_{j-1}}_{\pke, \mathsf{Chal}, \adve} \cind \mathsf{QPKGame}^{D^1_{j-1}}_{\pke, \mathsf{Chal}, \adve} $.
By adaptive security of $\FE$ and punctured correctness of $\PRF$, we have $\mathsf{QPKGame}^{D^1_{j-1}}_{\pke, \mathsf{Chal}, \adve}  \cind \mathsf{QPKGame}^{D^2_{j-1}}_{\pke, \mathsf{Chal}, \adve}$. By puncturable security of $\PRF$, we have $\mathsf{QPKGame}^{D^2_{j-1}}_{\pke, \mathsf{Chal}, \adve} \cind \mathsf{QPKGame}^{D^3_{j-1}}_{\pke, \mathsf{Chal}, \adve} $.

By adaptive security of $\FE$ we have $\mathsf{QPKGame}^{D'_{j}}_{\pke, \mathsf{Chal}, \adve}  \cind \mathsf{QPKGame}^{D^{-1}_{j}}_{\pke, \mathsf{Chal}, \adve} $.
By adaptive security of $\FE$ and punctured correctness of $\PRF$, we have $\mathsf{QPKGame}^{D^{-1}_{j}}_{\pke, \mathsf{Chal}, \adve}  \cind \mathsf{QPKGame}^{D^{-2}_{j}}_{\pke, \mathsf{Chal}, \adve} $. By puncturable security $\PRF$, we have $\mathsf{QPKGame}^{D^{-2}_{j}}_{\pke, \mathsf{Chal}, \adve}  \cind \mathsf{QPKGame}^{D^{-3}_{j}}_{\pke, \mathsf{Chal}, \adve}$.
\end{proof}

Combining the last two lemmas, we get
\begin{equation*}
        \left|\Pr[\mathsf{QPKGame}^{D^{-3}_{j}}_{\pke, \mathsf{Chal}, \adve} = 1] - \Pr[\mathsf{QPKGame}^{D^{3}_{j-1}}_{\pke, \mathsf{Chal}, \adve} = 1]\right| \geq \frac{\epsilon}{8q}
\end{equation*}

However, observe that $\mathcal{D}_j \equiv \mathcal{D}^3_{j - 1}$ and the fact that $D^{-3}_{j}(m)$ does not use $m$ (and thus $ \Pr[\mathsf{QPKGame}^{D^{3}_{j-1}}_{\pke, \mathsf{Chal}, \adve} = 1] = 1/2$). Thus, through a standard argument, the result follows.

\newpage
\section{Quantum-Query-Secure Deterministic Signature Schemes}\label{sec:detsig}
In this section, we construct a \emph{deterministic} signature scheme that satisfies Boneh-Zhandry security, that is, it is secure against quantum queries.

We start by recalling Boneh-Zhandry (\cite{boneh2013secure}) security, also called \emph{plus-one security}, which requires that an adversary that makes $k$ signature quantum queries cannot output valid signatures for $k+1$ distinct messages.
\begin{definition}[$\mathsf{BZ}$ Security]\label{defn:bzsec}
Consider the following game between a challenger and an adversary $\adve$.
\paragraph{\underline{$\bzgame{\adve}(1^\lambda)$}}
\begin{enumerate}
    \item The challenger samples $vk, sk \samp \dss.\mathsf{Setup}(1^\lambda)$.
    \item For multiple rounds, the adversary queries for a signature by submitting a quantum state $\sum_{m,w} \alpha_{m,w}\ket{m}\ket{w}$, and the challenger responds with  $\sum_{m,w} \alpha_{m,w}\ket{m}\ket{w \oplus \dss.\mathsf{Sign}(sk, m)}$\footnote{Here, for randomized signature schemes, the randomness for $\dss.\mathsf{Sign}$ is supplied either by sampling a single random string $r$ per query, and using that for all signatures in the superposition, or by sampling a PRF key per query and using $F(K, m)$ for each signatures in the superposition. In our construction, there will be no per-query randomness, and the signature scheme will be fully deterministic.} which it computes by running the signing function in superposition.
    \item After the query phase, the adversary outputs $k+1$ pairs $(m_1, sig_1), \dots (m_{k+1}, sig_{k+1})$, where $k$ is the number of queries it made. 
    \item The challenger checks if all $m_i$ are distinct. If not, it outputs $0$ and terminates (adversary loses).
    \item The challenger verifies each signature as $\dss.\mathsf{Verify}(vk, m_i, sig_i)$. If all pass, the challenger outputs $1$, otherwise it outputs $0$.
\end{enumerate}
We say that a signature scheme is Boneh-Zhandry ($\mathsf{BZ}$) secure if for any QPT adversary $\adve$, we have
\begin{equation*}
    \Pr[\bzgame{\adve}(1^\lambda) = 1] \leq \negl(\lambda).
\end{equation*}
\end{definition}

We now state our main theorem.
\begin{theorem}\label{thm:bzexists}Assuming existence of any subexponentially secure collision-resistant hash function, there exists a \textbf{deterministic} signature scheme that satisfies $\mathsf{BZ}$ security.
\end{theorem}

\subsection{Construction}
In this section, we present our construction.

Let $m(\lambda)$ denote the length of the messages to be signed. Let $\cdss$ be a one-time \emph{classical-query-secure} signature scheme with the following properties:
\begin{itemize}
    \item It is deterministic,
    \item It satisfies subexponential security, i.e., advantage of any quantum polynomial time adversary is bounded by $\mu(\lambda)$ for some inverse subexponential $\mu(\lambda)$,
    \item It can sign messages of length two times its verification key length (these lengths can be any value, denote these as $q(\lambda)$ and $2q(\lambda)$).
\end{itemize}

Let $h$ be a (polynomially secure) collision-resistant hash function from $m(\lambda)$ bits to $n(\lambda)$ bits for some $n(\lambda)$ such that $\mu(\lambda) \leq  2^{-n(\lambda)}\cdot \negl(\lambda)$.

 Let $F, F_2$ be PRF schemes, where the input-output length of $F$ will be clear from context, and $F_2$ maps $n(\lambda)$ bits to $2\cdot n(\lambda)$ bits.  It is easy to see that all of these constructions can be obtained assuming any subexponential secure collision-resistant hash function.

For a message $m$, we will write $\biti{i}{m}$ to denote its $i$-th bit, and we will write $\prefii{i}{m}$ to denote prefix $\biti{1}{m} || \dots || \biti{i}{m}$, with the convention $\prefii{0}{m} = \emptyset$ (empty string).

\paragraph{\underline{$\dss.\Sign(sk, m)$}}
\begin{enumerate}
    \item Sample $vk_{\emptyset}, sk_{\emptyset} \samp \cdss.\mathsf{Setup}(1^\lambda)$.
    \item Sample $K \samp F.\mathsf{Setup}(1^\lambda)$.
    \item Sample $K_2 \samp F_2.\mathsf{Setup}(1^\lambda)$.
    \item Set $vk = vk_{\emptyset}$.
    \item Set $sk = (sk_{\emptyset}, K, K_2)$.
    \item Output $vk, sk$.
\end{enumerate}

Before describing our signing algorithm, we introduce some notation.  For any strings $a$ of length up to $n$, we write $vk^*_{a}, sk^*_{a}$ to denote the value computed as $vk^*_{a}, sk^*_{a} = \cdss.\mathsf{Setup}(1^\lambda; F(K, a))$. Similarly, we write $y^*_m = F_2(K_2, h(m))$ and $sig^*_{h(m)} = \cdss.\mathsf{Sign}(sk^*_{h(m)}, y^*_{h(m)})$ for all $m \in \zo^{m(\lambda)}$.
\paragraph{\underline{$\dss.\Sign(sk, m)$}}
\begin{enumerate}
    \item Parse $(sk_{\emptyset}, K, K_2) = sk$.
    \item For $t \in [n]$, compute $s_{t} = vk^*_{\prefii{t-1}{{h(m)}} || 0} || vk^*_{\prefii{t-1}{{h(m)}} || 1} || \cdss.\mathsf{Sign}(sk_{\prefii{t-1}{{h(m)}}},\allowbreak vk^*_{\prefii{t-1}{{h(m)}} || 0} ||\allowbreak vk^*_{\prefii{t-1}{{h(m)}} || 1})$.
    \item Output $(s_1, \dots s_n), y^*_{h(m)}, sig^*_{h(m)}$.
\end{enumerate}

\paragraph{\underline{$\dss.\Ver(vk, m, sig)$}}
\begin{enumerate}
    \item Parse $(s_1, \dots s_n), y, isig = sig$.
    \item For $t \in [n]$, parse $pl_t^0 || pl_t^1 || sigpl_t = s_t$ (lengths of components clear from below).
     \item Verify  $\cdss.\mathsf{Verify}(vk, pl_{1}^0 || pl_{1}^1, sigpl_{1})$ and then for $t \in [n - 1]$, verify $\cdss.\mathsf{Verify}(pl_{t}^{\biti{t}{{h(m)}}},\allowbreak pl_{t+1}^0 || pl_{t+1}^1, sigpl_{t+1})$. Output $0$ and terminate if any of the checks fails.
    \item Verify $\cdss.\mathsf{Verify}(pl_n^{\biti{n}{{h(m)}}}, y, isig)$.
    \item Output $1$ if all checks passed. Otherwise, output $0$.
\end{enumerate}

The correctness of the scheme is straightforward. We claim security.
\begin{theorem}\label{thm:bzproof}$\dss$ satisfies EUF-CMA security and $\mathsf{BZ}$ security.
\end{theorem}

\subsection{Proof of Security}
In this section, we prove the security of our scheme, that is, we prove \cref{thm:bzproof}.

The EUF-CMA security proof is standard, so we omit it. We will prove $\mathsf{BZ}$  security through a sequence of hybrids.

\paragraph{\underline{$\hyb_0$}:} The original game $\bzgame{\adve}(1^\lambda)$.

\paragraph{\underline{$\hyb_1$}:} In this hybrid, we modify the experiments so that all PRF queries to $F, F_2$ are replaced with queries to truly random functions $H, H_2$ sampled during setup.

\paragraph{\underline{$\hyb_{2, \alpha}$ for $\alpha \in [n]$}:} Let $(\tilde{sig}_j= ((\tilde{s}_{j,t})_{t \in [n]}, \tilde{y}_j, \tilde{isig}_j)), \tilde{m}_j$ for $j \in [k + 1]$ be the forgeries (signature and message) produced by the adversary. Let $\tilde{pl}_{j,t}^0 || \tilde{pl}_{j,t}^1 || \tilde{sigpl}_{j,t}$ be the parsing (as in $\dss.\mathsf{Verify}$) of $\tilde{s}_{j,t}$ for all $j \in [k + 1], t \in [n]$.

We add an additional check at the end of the game. For all $t \in \{1, \dots, \alpha\}$, for all $j \in [k + 1]$, we check if $\tilde{pl}_{j,t}^0 || \tilde{pl}_{j,t}^1 = vk^*_{\prefii{t-1}{h(\tilde{m}_j)} || 0} || vk^*_{\prefii{t-1}{h(\tilde{m}_j)} || 1}$. If any of the checks fails, we output $0$ (adversary fails) and terminate.

\paragraph{\underline{$\hyb_3$}:} At the end of the game, we add the following additional check. For $j \in [k+1]$, we check if $\tilde{y}_j = H_2(\tilde{m}_j)$. If any of the checks fails, we output $0$ (adversary fails) and terminate.

\paragraph{\underline{$\hyb_4$}:}At the end of the game, we add the following additional check: Check if $h(\tilde{m}_t)_{t \in [k+1]}$ are distinct. If not,  we output $0$ (adversary fails) and terminate.

\begin{lemma}
    $\hyb_0 \approx \hyb_1$
\end{lemma}
\begin{proof}
    This follows by the PRF security of $F, F_2$.
\end{proof}

\begin{lemma}
    $\hyb_1 \approx \hyb_{2,1}$ and $\hyb_{2,\alpha} \approx \hyb_{2,\alpha+1}$ for all $\alpha \in [n-1]$.
\end{lemma}
\begin{proof}
    We will prove $\hyb_{2,\alpha} \approx \hyb_{2,\alpha+1}$, and then $\hyb_1 \approx \hyb_{2,1}$ follows from essentially the same argument.

    Observe that in $\hyb_{2,\alpha}$, we already know that (conditioned on adversary passing the checks) we have $\tilde{pl}_{j,\alpha}^0 || \tilde{pl}_{j,\alpha}^1 = vk^*_{\prefii{\alpha-1}{h(\tilde{m}_j)} || 0} || vk^*_{\prefii{\alpha-1}{h(\tilde{m}_j)} || 1}$. Suppose that adversary's output is such that there exists some index $j' \in [k + 1]$ such that $\tilde{pl}_{j',\alpha+1}^0 || \tilde{pl}_{j',\alpha+1}^1 \neq vk^*_{\prefii{\alpha}{h(\tilde{m}_{j'})} || 0} || vk^*_{\prefii{\alpha}{h(\tilde{m}_{j'})} || 1}$. Observe that this constitutes a forgery for the instance of $\cdss$ with the key pair $(vk^*_{\prefii{\alpha}{h(\tilde{m}_{j'})}, },sk^*_{\prefii{\alpha}{h(\tilde{m}_{j'})}})$ since: 
    \begin{itemize}
            \item We already checked that $\tilde{pl}_{j',\alpha}^{\biti{\alpha}{h(\tilde{m}_{j'})}}= vk^*_{\prefii{\alpha}{h(\tilde{m}_{j'})}}$,
        \item We already checked that  $\cdss.\Ver(\tilde{pl}_{j',\alpha}^{\biti{\alpha}{h(\tilde{m}_{j'})}}, \tilde{pl}_{j',\alpha+1}^0 || \tilde{pl}_{j',\alpha+1}^1, \tilde{sigpl}_{j',\alpha'+1}) = 1$,
        \item Only message signed with $sk^*_{\prefii{\alpha}{h(\tilde{m}_{j'})}}$ in the experiment is $vk^*_{\prefii{\alpha}{h(\tilde{m}_{j'})} || 0} || vk^*_{\prefii{\alpha}{h(\tilde{m}_{j'})} || 1}$.
    \end{itemize}
    Thus, such $j'$ value cannot exist (except with negligible probability); or otherwise we can create an adversary for the one-time unforgeability game $\mathcal{G}$ of $\cdss$ as follows. We place the challenge from the challenger of $\mathcal{G}$ at the node associated with a random string $a$ of length $\alpha$, and simulate the experiment $\hyb_{2,\alpha}$. Then, in case we get $h(\tilde{m}_{j'}) = a$ (which happens with probability $2^{-\alpha}$), we get a forgery for $\mathcal{G}$ as argued above. Since $\cdss$ is at least $2^{-\alpha}\cdot\negl(\lambda)$ secure, we conclude that indeed such $j'$ value cannot exist except with negligible probability. This means that (except with negligible probability), it is already the case in $\hyb_{2,\alpha}$ that if the adversary wins we have $\tilde{pl}_{j,\alpha+1}^0 || \tilde{pl}_{j,\alpha+1}^1 = vk^*_{\prefii{\alpha}{h(\tilde{m}_{j})} || 0} || vk^*_{\prefii{\alpha}{h(\tilde{m}_{j})} || 1}$ for all $j \in [k+1]$. Thus, the newly added checks in $\hyb_{2,\alpha+1}$ do not change anything, and hence $\hyb_{2,\alpha}\approx\hyb_{2,\alpha+1}$.
\end{proof}

\begin{lemma}
    $\hyb_{2,t} \approx \hyb_{3}$ .
\end{lemma}
\begin{proof}
This follows from essentially the same argument as the one for $\hyb_{2,\alpha} \approx \hyb_{2,\alpha+1}$.
\end{proof}

\begin{lemma}
    $\hyb_3 \approx \hyb_4$.
\end{lemma}
\begin{proof}
    This follows by the collision-resistance of $h$.
\end{proof}

\begin{lemma}
    $\Pr[\hyb_4 = 1] \leq \negl(\lambda)$.
\end{lemma}
\begin{proof}
    Observe that in this hybrid, the adversary makes $k$ queries to the random function $H_2$, and at the end, (if it wins) it produces $k+1$ pairs $(\tilde{m}_t,\tilde{y}_t)$ such that $H_2(h(\tilde{m}_t)) = \tilde{y}_t$. We also have that $h(\tilde{m}_t)_{t \in [k+1]}$ are distinct.  Therefore, we get $\Pr[\hyb_3 = 1] \leq \negl(\lambda)$ by \cref{lem:qrom}.
\end{proof}
Combining the above lemmas completes the security proof.

\newpage
\section{Quantum-State Read-Once Small-Range Distributions}\label{sec:readonce}
In this section, we show essentially a quantum-state analogue of the well-known \emph{small-range distributions} lemma (\cite{Z12}); it will helps us simulate an adversary's queries using a single-copy of some state from the oracle's output.

We first recall \emph{small-range distributions}.
\begin{theorem}[Small Range Distributions \cite{Z12}]\label{thm:srd}
There is a universal constant $\csrd > 1$ such that, for any sets $\mathcal{X}$ and $\mathcal{Y}$, distribution $\mathcal{D}$ on $\mathcal{Y}$, any integer $\ell$, and any quantum algorithm $\adve$ that makes $k$ queries, we have
\begin{equation*}
    \left|\Pr_{\mathcal{O} \samp \mathfrak{O}_1}[\adve^{\mathcal{O}}() = 1] - \Pr_{\mathcal{O} \samp \mathfrak{O}_2}[\adve^{\mathcal{O}}() = 1] \right| \leq \csrd\frac{k^3}{\ell}
\end{equation*}
where the distributions $\mathfrak{O}_1, \mathfrak{O}_2$ are defined as follows.
\begin{itemize}
    \item For $\mathfrak{O}_1$, independently sample  $y_x \samp \mathcal{D}$ for each $x \in \mathcal{X}$ and set $\mathcal{O}(x) = {y}_x$.
    \item For $\mathfrak{O}_2$, sample a random function $P: \mathcal{X} \to [\ell]$, independently sample  $y_x \samp \mathcal{D}$ for each $x \in [\ell]$ and set $\mathcal{O}(x) = {y}_{P(x)}$.
\end{itemize}
\end{theorem}

We define two notions before stating our lemma.
\begin{definition}[Quantum State Oracle]
Let $\mathcal{X}$ be a set. A quantum state oracle $\mathcal{O}$ is defined by $\{\ket{\psi_x}\}_{x \in \mathcal{X}}$, where each $\ket{\psi_x}$ is a pure state over $(\mathds{C}^{2})^{\otimes n}$. A query adversary for a quantum state oracle outputs a quantum state $\sum_{x,w}\alpha_{x,w}\ket{x}\ket{w}\ket{0^n}$ each round and receives back the state $\sum_{x,w}\alpha_{x,w}\ket{x}\ket{w}\ket{\psi_x}$.

A \empty{non-adaptive} query adversary simply sends a state $\sum_{(x_i)_{i \in [k]},w}\alpha_{(x_i)_{i \in [k]},w}\ket{x_1, \dots, x_k}\ket{w}\ket{0^{nk}}$ and receives the state $\sum_{(x_i)_{i \in [k]},w}\alpha_{(x_i)_{i \in [k]},w}\ket{x_1, \dots, x_k}\ket{w}\ket{\psi_{x_1}}\otimes\dots\otimes \ket{\psi_{x_k}}$.
\end{definition}

\begin{definition}[Delocalized States and Queries]\label{def:delocalized}
A quantum state $\ket{\psi}$ over $n(\lambda)$ qubits (parametrized by a security parameter $\lambda$) is said to be \emph{delocalized} if for each $x \in \zo^n$, $|\braket{x}{\psi}| \leq \negl(\lambda)$.

Similarly, a query adversary for a quantum state oracle is said to make delocalized queries if the query weight of $x$ is negligible for all $x \in \mathcal{X}$.
\end{definition}

Now we now state our lemma, which shows that an adversary that makes delocalized queries  can be simulated with a state whose support consists only of states that have a single copy of (polynomially many) quantum state samples.

\begin{lemma}[Quantum-State Read-Once Small-Range Distributions]\label{thm:qsreadonce}
    There is a universal constant $\ocsrd > 1$ such that for any set $\mathcal{X}$, any distribution $\mathcal{D}$ over pure states, any polynomials $k(\cdot), p(\cdot)$ and any non-adaptive adversary $\adve$ that makes $k(\lambda)$ delocalized\footnote{For each query, each $x \in \mathcal{X}$ has negligible query weight negligible weight} queries to the oracle $\mathcal{O}$, we have 
    \begin{equation*}
        \left|\E_{\mathcal{O}} \ketbra{\phi^{\mathcal{O}}_{orig}}{\phi^\mathcal{O}_{orig}} - \E_{\mathcal{O}^{sr}} \ketbra{\phi^{\mathcal{O}^{sr}}_{sr}}{\phi^{\mathcal{O}^{sr}}_{sr}}\right|_{Tr}\leq \frac{1}{p(\lambda)}
    \end{equation*}

    where
    \begin{itemize}
       \item $\mathcal{O}$ is the oracle with input domain $\mathcal{X}$ that is constructed by sampling independently a state $\ket{\psi_x} \samp \mathcal{D}$ for each $x \in \mathcal{X}$ and setting $\mathcal{O}(x) = \ket{\psi_x}$.
\item $\mathcal{O}_{sr}$ is the oracle with input domain $\mathcal{X}$ that is constructed by sampling a random function $P: \mathcal{X} \to [\ell(\lambda)]$, then independently sampling a state $\ket{\Tilde{\psi_i}} \samp \mathcal{D}$ for each $i \in [\ell]$ and setting $\mathcal{O}^{sr}(x) = \ket{\Tilde{\psi}_{P(x)}}$.        
        \item $\ket{\phi_{query}} =  \sum_{\substack{(x_1, \dots, x_k) \in \mathcal{X}^k, w}} \alpha_{x_1}^{(1)}\dots \alpha_{x_k}^{(k)}   \ket{x_1}\dots\ket{x_k}\ket{w}\ket{0^{nk}}$ is the  query of $\adve$,

    \item $\ket{\phi^{\mathcal{O}}_{orig}} = \sum_{\substack{(x_1, \dots, x_k) \in \mathcal{X}^k, w}} \alpha_{x_1}^{(1)}\dots \alpha_{x_k}^{(k)}  \ket{x_1}\ket{{\psi}_{x_1}}\dots\ket{x_k}\ket{{\psi}_{x_k}}\ket{w}$ is the original output for the query game played by $\adve$ with the oracle $\mathcal{O}$,
\item $\ket{\phi^{\mathcal{O}_{sr}}_{sr}} = \sum_{\substack{(x_1, \dots, x_k) \in \mathcal{X}^k, w: \\ P(x_1), \dots, P(x_k) \textbf{ distinct}}} \alpha_{x_1}^{(1)}\dots \alpha_{x_k}^{(k)}  \ket{x_1}\ket{\Tilde{\psi}_{P(x_1)}}\dots\ket{x_k}\ket{\Tilde{\psi}_{P(x_k)}}\ket{w}$
        \item $\ell = \ocsrd\cdot p^2 k^3$
    \end{itemize}
\end{lemma}
Note that the support of $\ket{\phi_{sr}}$ consists entirely of states that have a single-copy of the underlying states $\ket{\psi_{\cdot}}$

\subsection{Proof of \cref{thm:qsreadonce}}
In this section, we prove our lemma. Throughout the proof, sometimes we will write an unnormalized state expression but we mean the normalized version, and sometimes we will explicitly be using the unnormalized version - the distinction will be clear from context.

We will first construct a sequence of indistinguishable hybrids. We set the constant $\ocsrd$ to be $128\cdot \csrd$. Let $n(\lambda)$ denote the number of qubits of any state on the support of $\mathcal{D}$.

\paragraph{$\underline{\hyb_0}$}: This is the original experiment, that is, $\E_{\mathcal{O}} \ketbra{\phi^{\mathcal{O}}_{orig}}{\phi^\mathcal{O}_{orig}}$.

\paragraph{$\underline{\hyb_1}$} Let $\mathcal{N}$ be an $\epsilon$-net of pure states for pure states from $(\mathds{C}^2)^{\otimes n}$ where $\epsilon = (\frac{1}{1024p(\lambda)\cdot k(\lambda)}\cdot\alpha_{min})^2$ and $\alpha_{min}$ is the smallest query-weight produced by $\adve$. Let $\mathcal{N}'$ denote the set where we replace each state $\ket{\gamma}$ from $\mathcal{N}$ with the classical canonical description of a canonical unitary that takes $\ket{0^n}$ to $\ket{\gamma}$. Let $\mathcal{D}'$ be the distribution where we sample $\ket{\psi} \samp \mathcal{D}$ and output the corresponding element in the net $\mathcal{N}'$. Note that this is just a way of describing the distribution, we do not actually convert the quantum states into their classical descriptions (which would be unphysical given only a single copy). Thus, information-theoretically, $\mathcal{D}'$ is a well-defined distribution. Let $\mathcal{O}'$ be the classical oracle where for each $x \in \mathcal{X}$, we sample independently from $\mathcal{D}'$. Now we consider the following experiment. We run $\adve$ to obtain the query register $\ket{\phi_{query}}$, then query the oracle $\mathcal{O}'$ to obtain the state 
\begin{equation*}
    \ket{\tau_0} =  \sum_{\substack{(x_1, \dots, x_k) \in \mathcal{X}^k, w}} \alpha_{x_1}^{(1)}\dots \alpha_{x_k}^{(k)}   \ket{x_1}\dots\ket{x_k}\ket{w}\ket{\mathcal{O}'(x_1)}\otimes\dots\otimes \ket{\mathcal{O}'(x_k)}
\end{equation*}
Then, we add empty registers initialized to zeroes and we coherently run the quantum operation that reads the description of a unitary and applies to the target register. Now we end up with the state
\begin{equation*}
    \ket{\tau_0'} =  \sum_{\substack{(x_1, \dots, x_k) \in \mathcal{X}^k, w}} \alpha_{x_1}^{(1)}\dots \alpha_{x_k}^{(k)}   \ket{x_1}\dots\ket{x_k}\ket{w}\ket{\mathcal{O}'(x_1)}\otimes\dots\otimes \ket{\mathcal{O}'(x_k)}\otimes(\mathcal{O}'(x_1)\ket{0^{n}})\dots \otimes(\mathcal{O}'(x_k)\ket{0^{n}})
\end{equation*}
Now we query  $\mathcal{O}'$ again to obtain the state
\begin{equation*}
    \ket{\tau_0''} =  \sum_{\substack{(x_1, \dots, x_k) \in \mathcal{X}^k, w}} \alpha_{x_1}^{(1)}\dots \alpha_{x_k}^{(k)}   \ket{x_1}\dots\ket{x_k}\ket{w}\otimes(\mathcal{O}'(x_1)\ket{0^{n}})\dots \otimes(\mathcal{O}'(x_k)\ket{0^{n}})
\end{equation*}
We output $\ket{\tau_0''}$ as the experiment outcome in this hybrid.

\paragraph{$\underline{\hyb_2}$}: We now sample $\mathcal{O}'$ as follows. We sample a truly random function $P: \mathcal{X} \to [\ell(\lambda)]$ at the beginning, and we sample $U_1, \dots, U_\ell$ independently from $\mathcal{D}'$, and for each $x \in \mathcal{X}$ we set $\mathcal{O}'(x) = U_{P(x)}$.

\paragraph{$\underline{\hyb_3}$}: We define this hybrid to be  $\E_{\mathcal{O},P}\ketbra{\tau^{\mathcal{O}''}_{3}}{\tau^{\mathcal{O}''}_{3}}$ where

$$
\ket{\tau^{\mathcal{O}''}_{3}} = \sum_{\substack{(x_1, \dots, x_k) \in \mathcal{X}^k, w: \\ P(x_1), \dots, P(x_k)}} \alpha_{x_1}^{(1)}\dots \alpha_{x_k}^{(k)}  \ket{x_1}\ket{\Tilde{\psi}_{P(x_1)}}\dots\ket{x_k}\ket{\Tilde{\psi}_{P(x_k)}}\ket{w}
$$
$\mathcal{O}''$ is the oracle with input domain $\mathcal{X}$ that is constructed by sampling a random function $P: \mathcal{X} \to [\ell(\lambda)]$, then independently sampling a state $\ket{\Tilde{\psi_i}} \samp \mathcal{D}$ for each $i \in [\ell]$ and setting $\mathcal{O}''(x) = \ket{\Tilde{\psi}_{P(x)}}$.        

\paragraph{$\underline{\hyb_4}$}: We define this hybrid to be $\E_{\mathcal{O}^{sr}} \ketbra{\phi^{\mathcal{O}^{sr}}_{sr}}{\phi^{\mathcal{O}^{sr}}_{sr}}$

\paragraph{}
Now we move onto proving the indistinguishability of the hybrids.

First, we have $\hyb_0 \approx_{\frac{1}{8p}} \hyb_1$ by the net property of $\mathcal{N}$ and by our choice of parameters. Then, observe that we have $|\hyb_1 - \hyb_2| \leq \frac{1}{4p}$ by \cref{thm:srd} since $\ell = \ocsrd\cdot p^2 k^3 \geq 8\csrd pk^3$. Also note that  $\mathcal{D}'$ has finite support and \cref{thm:srd} indeed applies to distributions with arbitrarily large (finite) supports. Then, we have $\hyb_2 \approx_{\frac{1}{8p}} \hyb_3$, again by the net property of $\mathcal{N}$ and by our choice of parameters.

Finally, we show $|\hyb_3 - \hyb_4| \leq \frac{1}{4p}$. We will show that for any states $\ket{\tilde{\psi}_1},\dots,\ket{\tilde{\psi}_\ell}$, we have $\E_{P}[|\braket{\tau_3}{\phi_{sr}}|^2] \geq 1 - \frac{k^2}{\ell} - \negl(\lambda)$. Since we have that the trace distance between the hybrids is upper bounded by $\sqrt{1 - f}$ where $f$ is the fidelity between the two states, and since the fidelity between the two inputs are $\E[|\braket{\tau_3}{\phi_{sr}}|^2]$, the above will imply that the trace distance between $\hyb_3$ and $\hyb_4$ is at most $\sqrt{1 - (1 - \frac{2k^2}{\ell})} \leq \frac{1}{4p}$ due to $\ell = \ocsrd p^2 k^3$ and $\ocsrd > 64$.

Fix the states $\ket{\tilde{\psi}_1},\dots,\ket{\tilde{\psi}_\ell}$, and the map $P$, and $\ket{\tau_3^P}, \ket{\phi^P_{sr}}$ will refer to the corresponding states based on this fixing. We define the unnormalized states
\begin{align*}
    &\ket{\zeta^P_\emptyset} =  \sum_{\substack{x_1, \dots, x_k \in \mathcal{X}\\P(x_1), \dots, P(x_k) \text{ distinct}}} \alpha_{x_1}^{(1)}\dots \alpha_{x_k}^{(k)}  \ket{x_1}\ket{\tilde{\psi}_{P(x_1)}}\dots\ket{x_k}\ket{\tilde{\psi}_{P(x_k)}} \\
    &\ket{\zeta^P_{S}} =  \sum_{\substack{x_1, \dots, x_k \in \mathcal{X}\\P(x_{i_1}), \dots, P(x_{i_{k - |S|}}) \text{ distinct} \\ P(x_{j_1}) = \dots = P(j_{|S|})}} \alpha_{x_1}^{(1)}\dots \alpha_{x_k}^{(k)}  \ket{x_1}\ket{\tilde{\psi}_{P(x_1)}}\dots\ket{x_k}\ket{\tilde{\psi}_{P(x_k)}}.
\end{align*}
for all $S \subseteq [k]$, $S \neq \emptyset$ where we write $S = \{j_1, \dots, j_k\}$ and $\overline{S} = \{i_1, \dots, i_{k-|S|}\}$.
Observe that we have $\ket{\phi^P_{sr}} = \frac{\ket{\zeta^P_\emptyset}}{\norm{\zeta^P_\emptyset}}$, and $\ket{\zeta^P_{S_1}} \perp \ket{\zeta^P_{S_2}}$ for all sets $S_1 \neq S_2$ (including empty sets). Define the unnormalized state $\ket{\zeta^P} = \sum_{S} \ket{\zeta^P_S}$. Then we have $\ket{\phi} = \frac{\ket{\zeta^P}}{\norm{\zeta^P}}$.

By above, we have that, $q^P = |\braket{\phi^P}{\phi^P_0}|^2 = \frac{|\braket{\zeta^P}{\zeta^P_\emptyset}|^2}{\norm{\zeta^P}^2 \norm{\zeta^P_\emptyset}^2}$. Also, since  $\ket{\zeta^P} = \sum_{S} \ket{\zeta^P_S}$ and $\ket{\zeta^P_{S_1}} \perp \ket{\zeta^P_{S_2}}$, we get $q = \frac{|\braket{\zeta_\emptyset}{\zeta^P_\emptyset}|^2}{\norm{\zeta^P}^2 \norm{\zeta^P_\emptyset}^2} = \frac{\norm{\ket{\zeta^P_\emptyset}}^2}{\norm{\ket{\zeta^P}}^2}$. Finally, again by pairwise orthogonality, we have $\norm{\ket{\zeta^P}}^2 = \sum_S {\norm{\ket{\zeta^P_S}}^2}$. Thus, $1 - q = \sum_{S \neq \emptyset} \frac{\norm{\ket{\zeta^P_S}^2}}{\norm{\ket{\zeta^P}}^2}$. Now define for each 
$\gamma,\beta \in [k]$ with $\gamma \neq \beta$
\begin{equation*}
  \ket{\eta^P_{\gamma,\beta}} =  \sum_{\substack{x_1, \dots, x_k \in \mathcal{X}\\ P(x_i) \text{ all distinct for } i \neq \gamma,\beta \\ P(x_{\gamma}) = P(x_\beta)}} \alpha_{x_1}^{(1)}\dots \alpha_{x_k}^{(k)}  \ket{x_1}\ket{\tilde{\psi}_{P(x_1)}}\dots\ket{x_k}\ket{\tilde{\psi}_{P(x_k)}}.
\end{equation*}
It is easy to see that $\sum_{S \neq \emptyset} {\norm{\ket{\zeta^P_S}^2}} \leq \sum_{\gamma \neq \beta} \norm{\ket{\eta^P_{\gamma,\beta}}}^2$ since the latter expression \emph{adds/counts with repetition} (i.e., same collisions are added more than once). Combining these, we get $1 - q^P \leq \sum_{\gamma \neq \beta} \frac{\norm{\ket{\eta^P_{\gamma,\beta}}}^2}{{\norm{\ket{\zeta^P}}^2}}$. This implies $\E_{P}[|\braket{\tau_3^P}{\phi^P_{sr}}|^2] \geq 1 - \E_{P}\left[ \sum_{\gamma \neq \beta} \frac{\norm{\ket{\eta^P_{\gamma,\beta}}}^2}{{\norm{\ket{\zeta^P}}^2}}\right] = 1 -  \sum_{\gamma \neq \beta} \E_{P}\left[ \frac{\norm{\ket{\eta^P_{\gamma,\beta}}}^2}{{\norm{\ket{\zeta^P}}^2}}\right]$ by linearity of expectation.

Now, we claim that for any $\gamma \neq \beta$, we have $\E_{P}\left[\frac{\norm{\ket{\eta^P_{\gamma,\beta}}}^2}{{\norm{\ket{\zeta^P}}^2}}\right] \leq \frac{1}{\ell} + \negl(\lambda)$, which implies $\sum_{\gamma \neq \beta} \E_{P}\left[\frac{\norm{\ket{\eta^P_{\gamma,\beta}}}^2}{{\norm{\ket{\zeta^P}}^2}}\right] \leq {k \choose 2}\cdot\left(\frac{1}{\ell} + \negl(\lambda) \right) \leq \frac{k^2}{\ell} + \negl(\lambda)$. To prove our final claim: It is easy to see that since each $\alpha^{(\gamma)}_x$ is negligible for all $\gamma \in [k], x \in \mathcal{X}$, we have that $\E_{P}\left[\frac{\norm{\ket{\eta^P_{\gamma,\beta}}}^2}{{\norm{\ket{\zeta^P}}^2}}\right] \leq \negl(\lambda) + \max_{x_1 \neq x_2}\Pr_P[P(x_1) = P(x_2)]$, and we also have $\Pr_P[P(x_1) = P(x_2)] = \frac{1}{\ell}$ for $x_1 \neq x_2$. This completes the proof that $|\hyb_3 - \hyb_4| \leq \frac{1}{4p}$. Then, by combining the hybrid indistinguishability results above, we get $|\hyb_0 - \hyb_4| \leq \frac{1}{p}$, which completes the proof of the theorem.

\newpage
\section{Quantum Pigeonhole Lemma for Entangled Adversaries, or, Simultaneous-Distinguisher Finding Lemma}\label{sec:lemma_collusion_resistance}
In this section, we prove a lemma about finding simultaneous distinguishers that is used in the security proofs of (both bounded and unbounded) collusion-resistant copy-protection schemes in \Cref{sec:CR-transformation}.

This lemma generalizes and modularizes a technique whose special cases were used in an ad-hoc manner in previous collusion-resistance security proofs\footnote{It was first used by \cite{TCC:LLQZ22} in the setting of bounded-collusion resistance for the special case of the SDE scheme from \cite{C:CLLZ21}, and later it was also used by \cite{TCC:CakGoy24} for a related construction.}.

We first define the notions used in our lemma.
\begin{definition}[Quantum Program]\label{def:quantum_program}
A quantum program (with classical inputs and classical outputs) is a pair consisting of a quantum state $\rho$ and a unitary $U$, an its execution on a classical input $x$ is by applying $U$ on $\rho\otimes \ket{x}\bra{x}$, then measuring designated output registers and outputting the measurement outcome. 
\end{definition}
\begin{definition}[Distinguishing test]\label{def:distinguishing_test}
Let $\tau\in[0,1/2]$. 
Let $D_0$ and $D_1$ be distributions over a finite set $\cX$. 
A $\tau$-distinguishing test for $D_0$ and $D_1$ is the following procedure. 
\begin{itemize}
\item The procedure takes as input a quantum program $\qD=(\rho,U)$ with input space $\cX$ and output space $\{0,1\}$.  
\item Define a distribution $\widetilde{D}$ over $\cX\times \bit$ and a family $\cP=\{\cP_{x,\coin}\}_{x\in\cX,\coin\in\bit}$ of projective measurements as follows:
\begin{itemize}
\item 
$\widetilde{D}$ is a distribution of $(x,\coin)$ where $\coin\gets \bit$ and $x\gets D_\coin$.    
\item 
$\cP_{x,\coin}$ is a projective measurement corresponding to the following procedure: On a quantum state $\rho'$,  
run $\coin^\prime \gets \qD'(x)$, where $\qD'=(\rho',U)$.  
If $\coin^\prime = \coin$, output $1$, otherwise output $0$.
\end{itemize}
\item Apply $\TI_{1/2+\tau}(\cP_{\widetilde{D}})$ to $\rho$. If the outcome is $1$, the test outputs $1$ (accepts) and otherwise outputs $0$ (rejects).
\end{itemize}
We define $(\epsilon,\delta)$-approximate $\tau$-distinguishing test similarly except that we replace $\TI_{1/2+\tau}(\cP_{\widetilde{D}})$ with $\ATI_{\cP,\widetilde{D},1/2+\tau}^{\epsilon,\delta}$. 
\end{definition}

We are now ready to state our lemma. It roughly states that given a sequence of distributions $\mathcal{D}_0, \dots, \mathcal{D}_{q}$ and $q+1$ quantum programs that are simultaneously \emph{successful} for the original distribution $\mathcal{D}_0$ and $\mathcal{D}_{q}$ is an \emph{inherently secure} distribution (i.e. no efficient adversary can succeed with respect to it), we can efficiently obtain a pair of distinguishers that can simultaneously distinguish between the same consecutive distribution pair $(\mathcal{D}_{j-1}, \mathcal{D}_{j})$.
\begin{theorem}[Quantum Pigeonhole Lemma for Entangled Adversaries, or, Simultaneous-Distinguisher Finding Lemma]\label{thm:collusion_resistant} 
Let $q,\ell\in\mathbb{N}$ be  polynomials in $\secp$. 
Let $\cR$ and $\cX$ be finite sets, $\cP_i=\setbk{(\mat{P}_{i,x},\mat{Q}_{i,x})}_{x \in \cX}$ be a collection of efficiently implementable binary projective measurements on $\ell$-qubit states for $i\in[q+1]$, and 
$D_{i,j}: \cR \ra \cX$ be efficiently sampleable distributions over $\cX$ for $i\in[q+1]$ and $j\in \{0,1,\ldots,q\}$. 
We write $\cP_{i,D_{i,j}}$ to mean a mixture of projective measurements associated to $D_{i,j}$ and $\cP_i$. 
Let $\rho$ be a quantum state over $q+1$ registers $\qreg{R}_{1},\ldots,\qreg{R}_{q+1}$ each of which is an 
$\ell$-qubit register, and let $\rho[i]$ be the state on $\qreg{R}_{i}$ for $i\in[q+1]$. 
    
    Suppose that the following assumptions hold for some $0<\alpha<\beta<1$ such that $\beta-\alpha=1/\poly(\secp)$.    
\begin{enumerate}
\item For any $i\in[q]$, \[\Pr[\projimp(\cP_{i,D_{i,0}}) \rho[i]\ge \beta]\ge 1-\negl(\secp).\] 
\item For any $i\in[q]$ and any  $\ell$-qubit state $\sigma$ that is efficiently generatable from $\rho[i]$,   \[\Pr[\projimp(\cP_{i,D_{i,q}}) \sigma\le \alpha]\ge 1-\negl(\secp).\] 
\end{enumerate}
Then, there exist $\epsilon=1/\poly(\secp)$, $\delta=\negl(\secp)$, $\tau=1/\poly(\secp)$,  and QPT algorithms $\Convert_1^{\epsilon,\delta}$ and $\Convert_2^{\epsilon,\delta}$ that satisfy  $\Pr[\mathsf{Exp}^{\epsilon,\delta}_{\tau}=1]\ge 1/\poly(\secp)$ where $\mathsf{Exp}^{\epsilon,\delta}_{\tau}$ is the following experiment.

\begin{mdframed}
    {\underline{$\mathsf{Exp}^{\epsilon,\delta}_{\tau}$}}
\begin{enumerate}
\item Choose uniformly random distinct  $i_1,i_2\in [q+1]$ and $j\in [q]$.
\item Apply $\Convert_1^{\epsilon,\delta}(i_1,j,\cdot)$ and  $\Convert_2^{\epsilon,\delta}(i_2,j,\cdot)$ on registers $\qreg{R}_{i_1}$ and $\qreg{R}_{i_2}$,  which output quantum distinguishers $\qD_1$ and $\qD_2$, respectively.   
\item Apply the $\tau$-distinguishing test for $D_{i_1,j-1}$ and $D_{i_1,j}$ on $\qD_1$, and the $\tau$-distinguishing test for $D_{i_2,j-1}$ and $D_{i_2,j}$ on $\qD_2$. If both tests accept, the experiment outputs $1$, otherwise outputs $0$. 
\end{enumerate}
\end{mdframed}
\end{theorem}
We prove this theorem in \cref{sec:proof_finder}.
\ifconference
\section{Proof of \cref{thm:collusion_resistant}}\label{sec:proof_finder}
\else
\subsection{Proof of \cref{thm:collusion_resistant}}\label{sec:proof_finder}
\fi

In this section, we prove our theorem.

Let $\epsilon:=(\beta-\alpha)/(3q+2)$, $\delta=2^{-\secp}$, 
$\mu_1=1/(4q^3)$, $\tau_1:=\frac{1}{2(\mu_1(\lambda))^3 \cdot \poly(\lambda, 1/\epsilon, \log(1/\delta))}$,
$\mu_2=\tau_1/(2q^3)$,  $\tau_2:=\frac{1}{2(\mu_2(\lambda))^3 \cdot \poly(\lambda, 1/\epsilon, \log(1/\delta))}$, and $\tau\seteq \min\{\tau_1,\tau_2\}$
where $\poly$ in the definitions of $\tau_1$ and $\tau_2$ are as in \Cref{lem:ATI_gap_to_dist}.   

We will now define $\Convert_b^{\epsilon,\delta}$ for $b\in \{1,2\}$. Let $\cC_{b,i,j}$ be the quantum circuit as in \Cref{lem:ATI_gap_to_dist},  where we set $\mu:=\mu_b$ and substitute $D_{i,j-1}$ and $D_{i,j}$ for $D_0$ and $D_1$ in the lemma (with the remaining variables $\cR$, $\cX$, $\cP$, $\epsilon$, $\delta$ substituted in the obvious way). Recalling \Cref{lem:ATI_gap_to_dist}, this roughly means that if $\ATI$ with respect to $D_{i,j-1}$ and $D_{i,j}$ on a state $\sigma$ disagree with probability at least $\mu_b$, then $\cC_{b,i,j}(\sigma,\cdot)$ distinguishes $D_{i,j-1}$ and $D_{i,j}$ with advantage at least $\tau_b$. Let $\API_{i,j} \seteq \API_{\cP_i,D_{i,j}}^{\epsilon,\delta}$ for $i\in[q+1]$ and $j\in \{0,1,\ldots,q\}$.  
Then, we define $\Convert_b^{\epsilon,\delta}$ as follows.
\begin{mdframed}
\underline{$\Convert_b^{\epsilon,\delta}(i,j,\rho[i])$}
    \begin{enumerate}
       \item Sequentially apply $\API_{i,0}, \API_{i,1}, \ldots, \API_{i,j-1}$ in this order on the state $\rho[i]$, and let $\rho^{*}[i]$ denote the resulting post-measurement state.
        \item Prepare the following quantum distinguisher $\qD$ with the state $\rho^{*}[i]$ hardwired: On input $x\in \cX$ it outputs $\cC_{b,i,j}(\rho^{*}[i],x)$.
        \item Output $\qD$.
    \end{enumerate}
\end{mdframed}

We show below that this satisfies the conclusion of the theorem.
Below, we just write $\API_{i,j}$ to mean $\API_{i,j}^{\epsilon,\delta}$ since we always use the same parameters $\epsilon,\delta$. 

Consider an experiment $\mathsf{MeasureAll}$ in which, for each $i\in[q+1]$, we apply $\API_{i,0}, \API_{i,1}, \ldots, \API_{i,q}$ sequentially on register $\qreg{R}_i$, starting from $\rho[i]$.
 Let $b_{i,j}$ be the result of $\API_{i,j}$ on $\qreg{R}_i$ for $i\in[q+1]$ and $j\in\{0,1,\ldots,q\}$.  
\begin{lemma}\label{lem:jump}
In $\mathsf{MeasureAll}$, 
for every $i \in [q+1]$, there exists $j \in [q]$ such that $b_{i,j-1} - b_{i,j} \ge 3\epsilon$ except for negligible probability. 
\end{lemma}
\begin{proof}
By the first assumption of \Cref{thm:collusion_resistant} and the first item of \Cref{lem:API}, 
we have $\Pr[b_{i,0}\ge \beta-\epsilon]\ge 1-\negl(\secp)$. On the other hand, by the second assumption of \Cref{thm:collusion_resistant} and the first item of \Cref{lem:API}, we have $\Pr[b_{i,q}\le \alpha+\epsilon]\ge 1-\negl(\secp)$. The lemma directly follows from those and $(\beta-\alpha-2\epsilon) /q=3\epsilon$. 
\end{proof}
For each $i\in[q+1]$, let $j_i$ be the first index $j\in [q]$ such that   $b_{i,j-1} - b_{i,j} \ge 3\epsilon$. If there is no such $j$, define $j_i:=\bot$. 
\begin{lemma}\label{lem:pigeonhole}
    In $\mathsf{MeasureAll}$, 
    there exist $i_1,i_2\in[q+1]$ such that $i_1\ne i_2$ and $j_{i_1}=j_{i_2}\ne \bot$ except for negligible probability. 
\end{lemma}
\begin{proof}
By \Cref{lem:jump}, we have $j_i\in [q]$ for all $i\in[q+1]$ except for negligible probability. 
Thus, the lemma follows by the pigeonhole principle. 
\end{proof}

Now, we analyze the experiment $\mathsf{Exp}^{\epsilon,\delta}_{\tau}$ in the theorem. 
\begin{lemma}\label{lem:two_good}
In the experiment $\mathsf{Exp}^{\epsilon,\delta}_{\tau}$, consider running the execution until right after the initial step of $\Convert_1^{\epsilon,\delta}$ and $\Convert_2^{\epsilon,\delta}$, that is, until right after their repeated applications of the $\API$s.
Let $\rho^*[i_1,i_2]$ denote the joint state of the registers $\qreg{R}_{i_1}$ and $\qreg{R}_{i_2}$ at this point.
Let $b_{i_1,j-1}$ and $b_{i_2,j-1}$ be the respective outcomes of applying $\API_{i_1,j-1}$ to $\qreg{R}_{i_1}$ in $\Convert_1^{\epsilon,\delta}$ and $\API_{i_2,j-1}$ to $\qreg{R}_{i_2}$ in $\Convert_2^{\epsilon,\delta}$.

 With probability at least $1/(2q^3)-\negl(\secp)$ over the choice of 
 $i_1,i_2,j$ and $\rho^*[i_1,i_2]$ (when they are sampled as in $\mathsf{Exp}^{\epsilon,\delta}_{\tau}$), the following hold: 

\begin{enumerate}
	\item 
    Suppose that we apply $\API_{i_1,j-1}$ and $\API_{i_2,j-1}$ on registers $\qreg{R}_{i_1}$ and $\qreg{R}_{i_2}$  of $\rho^*[i_1,i_2]$ and let $b'_{i_1,j-1}$ and $b'_{i_2,j-1}$ be the outcomes. 
    Then it holds that
    \[
		\Pr[|b_{i_1,j-1} - b'_{i_1,j-1}| \le \epsilon \land |b_{i_2,j-1} - b'_{i_2,j-1}| \le \epsilon] \ge 1-\negl(\secp).
	\]
	\item Suppose that we apply $\API_{i_1,j}$ and $\API_{i_2,j}$ on registers $\qreg{R}_{i_1}$ and $\qreg{R}_{i_2}$ of $\rho^*[i_1,i_2]$ and let $b_{i_1,j}$ and $b_{i_2,j}$ be the outcomes. Then it holds that  
	\[
		\Pr[b_{i_1,j-1} - b_{i_1,j} \ge 3\epsilon \land b_{i_2,j-1} - b_{i_2,j} \ge 3\epsilon] \ge \frac{1}{2q^3}-\negl(\secp).
	\]
\end{enumerate}
\end{lemma}
\begin{proof}
    The first item directly follows from the second item of \Cref{lem:API}.  

    By \Cref{lem:pigeonhole}, 
    if we uniformly choose $i_1\ne i_2$ and $j$
    after the execution of $\mathsf{MeasureAll}$,
    \begin{align*}
\Pr_{
\substack{
\mathsf{MeasureAll}\\
i_1,i_2,j
}}\left[b_{i_1,j-1}-b_{i_1,j} \ge 3\epsilon \land b_{i_2,j-1}-b_{i_2,j} \ge 3\epsilon\right] \ge \frac{1-\negl(\secp)}{\frac{q(q-1)}{2} \cdot q} \ge \frac{1}{q^3}-\negl(\secp).
\end{align*}
Since the repeated approximate projective implementations on disjoint registers commute, 
we achieve the same probability bound even if we 
first choose $i_1,i_2,j$ and 
restrict the measurements to only on $\qreg{R}_{i_1}$ and $\qreg{R}_{i_2}$ omitting those on other $q-1$ registers. 
That is, the above inequality holds over execution of $\mathsf{Exp}^{\epsilon,\delta}_{\tau}$ as well. 
Thus, by a standard averaging argument, 
the second item in \Cref{lem:two_good} holds with probability at least $1/(2q^3)-\negl(\secp)$.  
\end{proof}

We say that $(i_1,i_2,j,\rho^*[i_1,i_2])$ is good if it satisfies the two items of \Cref{lem:two_good}. \Cref{lem:two_good} means that $(i_1,i_2,j,\rho^*[i_1,i_2])$ is good with probability at least $1/(2q^3)-\negl(\secp)$ over the randomness of $\mathsf{Exp}^{\epsilon,\delta}_{\tau}$. Below, we fix good $(i_1,i_2,j,\rho^*[i_1,i_2])$ until \Cref{lem:Acc_1_and_Acc_2}. 

We define the following events for $b\in\{1,2\}$. Here, we note that $\event{G}_b$ is an event in a thought experiment which is a modified version of $\mathsf{Exp}^{\epsilon,\delta}_{\tau}$ where we apply $\API_{i_b,j}$ on $\rho^{\ast}[i_b]$. While it is different from the actual experiment $\mathsf{Exp}^{\epsilon,\delta}_{\tau}$, as we will see below, it helps us analyze $\mathsf{Exp}^{\epsilon,\delta}_{\tau}$.
\begin{description}
	\item[$\event{Acc}_b$:] The $\tau$-distinguishing test on $\qD_b$ accepts in $\mathsf{Exp}^{\epsilon,\delta}_{\tau}$. 
	\item[$\event{G}_b$:] We apply $\API_{i_b,j}$ on $\rho^{\ast}[i_b]$ to get $b_{i_b,j}$ and it holds that $b_{i_b,j-1}-b_{i_b,j} \ge 3\epsilon$. %
\end{description}
The second item of \Cref{lem:two_good} can be rephrased as 
\begin{align}
\Pr[\event{G}_1\land \event{G}_2]\ge \frac{1}{2q^3}-\negl(\secp). \label{eq:G_1_and_G_2} 
\end{align}
\begin{lemma}\label{lem:Acc_1}
It holds that $\Pr[\event{Acc}_1]\ge \tau_1$.
\end{lemma}
\begin{proof}
Let $\gamma:=b_{i_1,j-1}+2\epsilon$. Then the first item of \Cref{lem:two_good} implies 
\[
\Pr\!\left[
 \ATI^{\epsilon,\delta}_{\mathcal{P}_{i_1},D_{i_1,j-1},\gamma}(\rho^*[i_1])=1
\right]\le \negl(\secp). 
\]
On the other hand,  the second item of \Cref{lem:two_good} implies 
\[
\Pr\!\left[
 \ATI^{\epsilon,\delta}_{\mathcal{P}_{i_1},D_{i_1,j},\gamma}(\rho^*[i_1])=1
\right]\ge \frac{1}{2q^3}-\negl(\secp). 
\] 

Thus, \Cref{lem:ATI_gap_to_dist} with $\mu=\mu_1=1/(4q^3)$ implies   
\[
\Pr
\bigl[ \cC(\rho^*[i_1],x) = \coin : \coin \gets\{0,1\},x \leftarrow D_{i_1,j-1+\coin} \bigr]
\ge \frac{1}{2}
+
\frac{1}{(\mu_1(\lambda))^3 \cdot \poly(\lambda, 1/\epsilon, \log(1/\delta))}=\frac{1}{2}
+2\tau_1. 
\]
Thus, by a standard averaging argument, the $\tau_1$-distinguishing test for $D_{i_1,j-1}$ and $D_{i_1,j}$ accepts  with probability at least $\tau_1$. 
Since $\tau\le \tau_1$, the lemma follows.
\end{proof}
\begin{lemma}\label{lem:G_2_conditioned_Acc_1}
It holds that $\Pr[\event{G}_2\mid \event{Acc}_1]\ge \frac{\tau_1}{2q^3}-\negl(\secp)$. 
\end{lemma}
\begin{proof}
We have 
$\Pr[\event{G}_2 \mid \event{Acc}_1]\ge  \Pr[\event{G}_2 \land \event{Acc}_1]=\Pr[ \event{Acc}_1 \mid \event{G}_2] \Pr[\event{G}_2]$.
We have $\Pr[\event{G}_2]\ge 1/(2q^3)-\negl(\secp)$ by \Cref{eq:G_1_and_G_2}. Thus, it suffices to show  $\Pr[ \event{Acc}_1 \mid \event{G}_2]\ge \tau_1$. 
To show this, consider the state $\rho^{\ast\ast}[i_1]$ on $\qreg{R}_{i_1}$, conditioned on the event $\event{G}_2$. 
The first item of \Cref{lem:two_good} continues to hold after conditioning since $\Pr[\event{G}_2]$ is inverse polynomial.\footnote{This follows from a standard conditioning argument: for events $E$ and $F$ with $\Pr[F]\ge 1/\poly(\secp)$, we have
$\Pr[\neg E\mid F]=\Pr[\neg E\land F]/\Pr[F]\le \Pr[\neg E]/\Pr[F]$.
Thus, if $\Pr[\neg E]\le \negl(\secp)$, then $\Pr[\neg E\mid F]\le \negl(\secp)$ as well, and the claimed property continues to hold after conditioning.} 
Moreover, the second item of \Cref{lem:two_good} also holds after conditioning since 
we have 
$\Pr[\event{G}_1 \mid \event{G}_2]\ge \Pr[\event{G}_1\land \event{G}_2]\ge 1/(2q^3)-\negl(\secp)$ where the last inequality follows from  \Cref{eq:G_1_and_G_2}. 
That is, we have the following: 
\begin{enumerate}
\item Suppose that we apply $\API_{i_1,j-1}$ on $\rho^{**}[i_1]$ and let $b'_{i_1,j-1}$ be the outcome. 
    Then it hold that
    $$\Pr[|b_{i_1,j-1}-b'_{i_1,j-1}|\le \epsilon]\ge 1-\negl(\secp).$$ 
\item Suppose that we apply $\API_{i_1,j}$ on $\rho^{**}[i_1]$ and let $b_{i_1,j}$ be the outcome. Then it holds that  
	\[
		\Pr[b_{i_1,j-1} - b_{i_1,j} \ge 3\epsilon] \ge \frac{1}{2q^3}-\negl(\secp).
	\]
\end{enumerate}
Therefore, by the same argument as in the proof of \Cref{lem:Acc_1}, 
the quantum distinguisher $\qD_1$ generated by $\Convert_1^{\epsilon,\delta}(i_1,j,\rho^{**}[i_1])$ passes the $\tau_1$-distinguishing test for $D_{i_1,j-1}$ and $D_{i_1,j}$ accepts  with probability at least $\tau_1$. 
Since $\tau\le \tau_1$, this means that $\Pr[ \event{Acc}_1 \mid \event{G}_2]\ge \tau_1$. This completes the proof. 
\end{proof}
\begin{lemma}\label{lem:Acc_2_conditioned_Acc_1}
It holds that $\Pr[\event{Acc}_2\mid \event{Acc}_1]\ge \tau_2$. 
\end{lemma}
\begin{proof}
The proof is very similar to that of \Cref{lem:G_2_conditioned_Acc_1}. 
Consider the state $\rho^{\ast\ast\ast}[i_2]$ on $\qreg{R}_{i_2}$, conditioned on the event $\event{Acc}_1$. 
The first item of \Cref{lem:two_good} continues to hold after conditioning since $\Pr[\event{Acc}_1]$ is inverse polynomial by \Cref{lem:Acc_1}. 
Moreover, the second item of \Cref{lem:two_good}, with the bound replaced by $\tau_1/(2q^3)-\negl(\secp)$, also holds 
after conditioning  by  \Cref{lem:G_2_conditioned_Acc_1}.  
That is, we have the following: 
\begin{enumerate}
\item Suppose that we apply $\API_{i_2,j-1}$ on $\rho^{***}[i_2]$ and let $b'_{i_2,j-1}$ be the outcome. 
    Then it hold that
    $$|b_{i_2,j-1}-b'_{i_2,j-1}|\le \epsilon$$ 
    except for negligible probability. 
\item Suppose that we apply $\API_{i_2,j}$ on $\rho^{***}[i_2]$ and let $b_{i_2,j}$ be the outcome. Then it holds that  
	\[
		\Pr[b_{i_2,j-1} - b_{i_2,j} \ge 3\epsilon] \ge \frac{\tau_1}{2q^3}-\negl(\secp).
	\]
\end{enumerate}
Therefore, by the same argument as in the proof of \Cref{lem:Acc_1}, 
the quantum distinguisher $\qD_2$ generated by $\Convert_2^{\epsilon,\delta}(i_2,j,\rho^{***}[i_2])$ passes the $\tau_2$-distinguishing test for $D_{i_2,j-1}$ and $D_{i_2,j}$  with probability at least $\tau_2$. 
Since $\tau\le \tau_2$, this means that $\Pr[ \event{Acc}_2 \mid \event{Acc}_1]\ge \tau_2$. This completes the proof. 
\end{proof}
\begin{lemma}\label{lem:Acc_1_and_Acc_2}
It holds that $\Pr[\event{Acc}_1\land\event{Acc}_2]\ge \tau_1\tau_2$. 
\end{lemma}
\begin{proof}
This immediately follows from \Cref{lem:Acc_1,lem:Acc_2_conditioned_Acc_1}.  
\end{proof}

\Cref{lem:Acc_1_and_Acc_2} means that the probability that $\mathsf{Exp}^{\epsilon,\delta}_{\tau}$ outputs $1$ is at least $\tau_1\tau_2$ conditioned on any fixed good $(i_1,i_2,j,\rho^*[i_1,i_2])$. Moreover, by \Cref{lem:two_good}, $(i_1,i_2,j,\rho^*[i_1,i_2])$ is good with probability at least $1/(2q^3)-\negl(\secp)\ge 1/\poly(\secp)$. Since both $\tau_1$ and $\tau_2$ are inverse polynomial, the overall probability that $\mathsf{Exp}^{\epsilon,\delta}_{\tau}$ returns $1$ is inverse polynomial. This completes the proof of \Cref{thm:collusion_resistant}.

\newpage
\section{Collusion-Resistant One-Way to Hiding Lemma}\label{sec:o2h}
In this section, we prove a new variant of the one-way to hiding lemma~\cite{JACM:Unruh15,C:AmbHamUnr19} (\cref{lem:O2H}), which we call the \emph{collusion-resistant one-way to hiding lemma}. This generalizes both the original  one-way to hiding lemma and the recent simultaneous double-sided one-way to hiding lemma shown in \cite{EPRINT:KitYam25} (which considers only two separate parties) to the case of arbitrary polynomially many parties. 

We now state our lemma.
\begin{lemma}[Collusion-Resistant  One-Way to Hiding Lemma]\label{lem:cr_O2H}
There is a QPT oracle algorithm $\qalgo{Ext}$  that satisfies the following.
Let $H_0:\cX\ra \cY$ and $H_1:\cX\ra \cY$ be functions and $\rho$ be a quantum state over $q+1$ registers $\qreg{R}_1,\ldots,\qreg{R}_{q+1}$,
where 
$q=\poly(\secp)$ and 
$(H_0,H_1,\rho)$ may have arbitrary joint distribution. 
Let $\{S_j\}_{j\in[q]}$ be a partition of the set $\{x\in \cX:H_0(x)\ne H_1(x)\}$.\footnote{That is, $S_1,\ldots,S_q$ are pairwise disjoint subsets such that $\bigcup_{j\in[q]}S_j=\{x\in \cX:H_0(x)\ne H_1(x)\}$.} 
For $i\in[q+1]$, let $\qA_i$ be a quantum oracle algorithm that 
takes a state in $\qreg{R}_i$ as input, 
makes at most $Q=\poly(\secp)$ quantum queries, and outputs a bit. 
Suppose that we have 
\[
\Pr\left[\forall i\in[q+1]~\coin^\prime_i=\coin:
\begin{array}{ll}
\coin\gets \bit\\
\coin_i^\prime \gets \qA_i^{H_\coin}(\rho[i]) \text{~for~}i\in[q+1]
\end{array}
\right]
\ge \frac{1}{2}+\epsilon
\]
for $\epsilon=1/\poly(\secp)$.  
Then it holds that
\begin{align*}
\Pr[x_{i_1}\in S_j\land x_{i_2}\in S_j:
\begin{array}{ll}
i_1,i_2\gets [q+1]^2 \text{~such~that~}i_1\ne i_2\\ 
j\gets [q]\\
x_{i_1} \gets \qalgo{Ext}^{H_0,H_1,O_S}(1^\secp,i_1,j,\qA_{i_1},\rho[i_1])\\
x_{i_2} \gets \qalgo{Ext}^{H_0,H_1,O_S}(1^\secp,i_2,j,\qA_{i_2},\rho[i_2])
\end{array}
]\ge \frac{1}{\poly(\secp)}
\end{align*}
where 
$O_S$ is an oracle that takes $(j,x)\in[q]\times \cX$ as input and outputs $1$ if $x\in S_j$ and outputs $0$ otherwise, and 
$\qA_{i_1}$ and $\qA_{i_2}$ provided as input of $\qalgo{Ext}$ means their descriptions.  
\end{lemma}
We prove this lemma in \cref{sec:o2proof}.
\ifconference
\section{Proof of \cref{lem:cr_O2H}}\label{sec:o2proof}
\else
\subsection{Proof of \cref{lem:cr_O2H}}\label{sec:o2proof}
\fi

In this section, we prove our lemma. The proof combines ideas from the proof of the simultaneous double-sided one-way to hiding lemma of \cite{EPRINT:KitYam25}, which builds on \cite{C:AnaKalLiu23}, and from the proof of \Cref{thm:collusion_resistant}.

Now we move onto our proof. We note that the proof has large similarities to that of \Cref{thm:collusion_resistant}, nevertheless, we present it in full for completeness. 

By the deferred measurement principle, we can assume without loss of generality that $\qA_1,\ldots,\qA_{q+1}$ do not perform any intermediate measurements and do not use ancilla qubits by augmenting registers $\qreg{R}_1,...,\qreg{R}_{q+1}$ with a sufficient number of $0$s.

For $j\in \{0,1,\ldots,q\}$, define a function $H^j:\cX\rightarrow \cY$ as  
\[
H^j(x)\seteq 
\begin{cases}
H_1(x)~\text{if}~x\in\bigcup_{j'\in[j]}S_{j'}\\
H_0(x)~\text{else}
\end{cases}
\]
Note that $H^0=H_0$ and $H^q=H_1$. 
Also note that $H^j$ for any $j\in \{0,1,\ldots,q\}$ can be computed given polynomial time oracle access to $H_0$, $H_1$, and $O_S$.

For $i\in[q+1]$ and $j\in \{0,1,\ldots,q\}$, 
let $\cP_{i,j,D}$ be a mixture of projective measurements associated to a distribution $D$ over $\bit$ and a family $\cP_{i,j}=\{\cP_{i,j,\coin}\}_{\coin\in\bit}$ of projective measurements defined as follows: 
\begin{itemize}
\item 
$D$ is a uniform distribution over $\coin \in \bit$.    
\item 
$\cP_{i,j,\coin}$ is a projective measurement corresponding to the following procedure: On a quantum state $\rho'$ on register $\qreg{R}_i$,  
run $\coin^\prime_i \gets \qA_i^{H^{j}}(\rho')$ if $\coin=0$
and 
$\coin^\prime_i \gets \qA_i^{H^{q}}(\rho')$ if $\coin=1$,  
and output $1$ if $\coin^\prime_i=\coin$ and output $0$ otherwise. 
\end{itemize} 

 Let $\API_{i,j} \seteq\API_{\cP_{i,j},D}^{\epsilon',\delta}$ for
$i \in [q+1]$ and $j \in \{0,1,\ldots,q\}$, where $\epsilon'=\epsilon/(8q)$ and $\delta=2^{-\secp}$.
By the third item of \Cref{lem:API}, each $\API_{i,j}$ can be viewed as an oracle algorithm with access to oracles $H^j$ and $H^q$.
Under this interpretation, for each fixed $i \in [q+1]$, the algorithms
$\API_{i,0}, \API_{i,1}, \ldots, \API_{i,q}$ are identical except for their first oracle, namely $H^0, H^1, \ldots, H^q$. 
Without loss of generality, we assume that each $\API_{i,j}$ makes exactly the same number of queries to its first oracle; we denote this number by $Q_{\API}$.
Since $\epsilon' = 1/\poly(\secp)$ and $\log \delta = \poly(\secp)$, it follows that $Q_{\API} = \poly(\secp)$.

We construct 
$\qalgo{Ext}$ as follows. 
\begin{description}
\item[$\qalgo{Ext}^{H_0,H_1,O_S}(1^\secp,i,j,\qA_i,\rho'_i)$:] 
$\qalgo{Ext}$ applies $\API_{i,0}, \API_{i,1}, \ldots, \API_{i,j-1}$ in this order to the state $\rho'_i$, obtaining a residual state $\rho''_i$, and then applies $\widetilde{\API}_{i,j}$ to $\rho''_i$, where $\widetilde{\API}_{i,j}$ is defined below.

\begin{description}
\item[$\widetilde{\API}_{i,j}$:]
On input a state $\rho''_i$, sample $k \gets [Q_\API]$, run $\API_{i,j}$ on $\rho''_i$, measure the $k$-th query made by $\API_{i,j}$ to its first oracle (i.e., $H^j$), and output the measurement outcome $x$.
\end{description}

Finally, $\qalgo{Ext}$ outputs $x$ output by $\widetilde{\API}_{i,j}$.

\end{description}

Consider an experiment $\mathsf{MeasureAll}$ in which, for each $i\in[q+1]$, we apply $\API_{i,0}, \API_{i,1}, \ldots, \API_{i,q}$ sequentially on register $\qreg{R}_i$ (starting from $\rho[i]$).
 Let $b_{i,j}$ be the result of $\API_{i,j}$ on $\qreg{R}_i$ for $i\in[q+1]$ and $j\in\{0,1,\ldots,q\}$. 
\begin{lemma}\label{lem:success_prob}
In $\mathsf{MeasureAll}$, let $\event{E}$ be the event that $|b_{i,0}-1/2|\ge 3\epsilon/8$ holds for all $i\in[q+1]$. 
Then it holds that 
$\Pr[\event{E}]\ge \epsilon/4$.
\end{lemma}
The proof of \Cref{lem:success_prob} is almost identical to the proof of \cite[Lemma 7.6]{EPRINT:KitYam25} where they prove a similar statement in the special case of $q=1$.
Thus, we defer the proof to \Cref{app:success_prob}. 
\begin{lemma}\label{lem:q_always_half}
In $\mathsf{MeasureAll}$,
for $i \in [q+1]$, 
$\Pr[|b_{i,q}-1/2|\le \epsilon/8]\ge 1-\negl(\secp)$.
\end{lemma}
\begin{proof}
In $\cP_{i,q,\coin}$, $\cA_i$'s oracle is $H^q$ regardless of $\coin$ and thus $\cA_i$ can correctly guess $\coin$ only with probability $1/2$. This implies that 
the projective implementation 
$\projimp(\cP_{i,q,D})$ associated to 
$D$ and 
$\cP_{i,q}$ always return $1/2$ on any state. 
\Cref{lem:q_always_half} follows from this and the first item of \Cref{lem:API} noting that $\epsilon'\le \epsilon/8$ and $\delta=\negl(\secp)$. 
\end{proof}

By the definition of $\event{E}$ and \Cref{lem:q_always_half}, 
when $\event{E}$ occurs, for every $i\in[q+1]$, there exists $j\in [q]$ such that $|b_{i,j-1}-b_{i,j}|\ge \frac{\epsilon}{4q}$ except for negligible probability. For $i\in[q+1]$, let $j_i$ be the smallest such index. (If there is no such index, we define $j_i\seteq \bot$.) 
\begin{lemma}\label{lem:pigeonhole_two}
 In $\mathsf{MeasureAll}$, when $\event{E}$ occurs, there exist $i_1,i_2\in[q+1]$  such that $i_1\ne i_2$ and $j_{i_1}=j_{i_2}\ne \bot$ except for negligible probability.  
\end{lemma}
\begin{proof}
    This immediately follows from the pigeonhole principle. 
\end{proof}

Consider an experiment $\mathsf{MeasureRandom}$ that works as follows:
\begin{itemize}
   \item Choose uniformly random distinct  $i_1,i_2\in [q+1]$ and $j\in [q]$.
\item Apply $\API_{i_1,0}, \API_{i_1,1}, \ldots, \API_{i_1,j-1}$ sequentially on $\rho[i_1]$, and apply $\API_{i_2,0}, \API_{i_2,1}, \ldots, \API_{i_2,j-1}$ sequentially on $\rho[i_2]$. 
\end{itemize}

\begin{lemma}\label{lem:two_good_o2h}
In $\mathsf{MeasureRandom}$, let $\rho^*[i_1,i_2]$ be the state on registers $\qreg{R}_{i_1}$ and $\qreg{R}_{i_2}$ after the experiment, and $b_{i_1,j-1}$ and $b_{i_2,j-1}$ be the outcomes of $\API_{i_1,j-1}$ on $\qreg{R}_{i_1}$ and 
$\API_{i_2,j-1}$ on
$\qreg{R}_{i_2}$, respectively (where $i_1,i_2,j$ are randomly chosen in $\mathsf{MeasureRandom}$). 

 With probability at least $\epsilon/(8q^3)-\negl(\secp)$ over the execution  of $\mathsf{MeasureRandom}$, the following hold: 
\begin{enumerate}
	\item 
    Suppose that we apply $\API_{i_1,j-1}$ and $\API_{i_2,j-1}$ on registers $\qreg{R}_{i_1}$ and $\qreg{R}_{i_2}$  of $\rho^*[i_1,i_2]$ and let $b'_{i_1,j-1}$ and $b'_{i_2,j-1}$ be the outcomes. 
    Then it holds that
       \[
		\Pr[|b_{i_1,j-1} - b'_{i_1,j-1}| \le \epsilon' \land |b_{i_2,j-1} - b'_{i_2,j-1}| \le \epsilon'] \ge 1-\negl(\secp).
	\]
	\item Suppose that we apply $\API_{i_1,j}$ and $\API_{i_2,j}$ on registers $\qreg{R}_{i_1}$ and $\qreg{R}_{i_2}$ of $\rho^*[i_1,i_2]$ and let $b_{i_1,j}$ and $b_{i_2,j}$ be the outcomes. Then it holds that  
	\[
		\Pr[|b_{i_1,j-1} - b_{i_1,j}| \ge \epsilon/(4q) \land |b_{i_2,j-1} - b_{i_2,j}| \ge \epsilon/(4q)] \ge \frac{\epsilon}{8q^3}-\negl(\secp)
	\]
\end{enumerate}
\end{lemma}
\begin{proof}
    The first item follows from the second item of \Cref{lem:API} noting that $\delta=\negl(\secp)$.  

    By \Cref{lem:pigeonhole_two}, 
    if we uniformly choose $i_1\ne i_2$ and $j$
    after the execution of $\mathsf{MeasureAll}$,
    \begin{align*}
\Pr_{
\substack{
\mathsf{MeasureAll}\\
i_1,i_2,j
}}[|b_{i_1,j-1} - b_{i_1,j}| \ge \epsilon/(4q) \land |b_{i_2,j-1} - b_{i_2,j}| \ge \epsilon/(4q)] \ge \frac{\Pr[\event{E}]-\negl(\secp)}{\frac{q(q-1)}{2} \cdot q} \ge \frac{\epsilon}{4q^3} -\negl(\secp)
\end{align*}
where the final inequality follows from \Cref{lem:success_prob}. 
Since the repeated approximate projective implementations on disjoint registers commute, 
we achieve the same probability bound even if we 
first choose $i_1,i_2,j$ and 
restrict the measurements to only on $\qreg{R}_{i_1}$ and $\qreg{R}_{i_2}$ omitting those on other $q-1$ registers. 
That is, the above inequality holds over execution of $\mathsf{MeasureRandom}$ as well. 
Thus, by a standard averaging argument, 
the second item in \Cref{lem:two_good_o2h} holds with probability at least $\epsilon/(8q^3)-\negl(\secp)$. 
\end{proof}

We say that $(i_1,i_2,j,\rho^*[i_1,i_2])$ is good if it satisfies the two conditions of \Cref{lem:two_good_o2h}. \Cref{lem:two_good_o2h} means that $(i_1,i_2,j,\rho^*[i_1,i_2])$ is good with probability at least $\epsilon/(8q^3)-\negl(\secp)$ over the randomness of $\mathsf{MeasureRandom}$. Below, we fix good $(i_1,i_2,j,\rho^*[i_1,i_2])$ until \Cref{lem:Acc_1_and_Acc_2_o2h}. 

We define the following events for $b\in\{1,2\}$:
\begin{description}
	\item[$\event{Acc}_b$:] 
    We run $x_{i_b}\gets \widetilde{\API}_{i,j}$ on $\rho^*[i_b]$ where $\widetilde{\API}_{i,j}$ is as in the description of $\mathsf{Ext}$ and it holds that $x_{i_b}\in S_j$.  
	\item[$\event{G}_b$:] We apply $\API_{i_b,j}$ on $\rho^{\ast}[i_b]$  to get $b_{i_b,j}$ and it holds that $|b_{i_b,j-1}-b_{i_b,j}| \ge \epsilon/(4q)$. 
\end{description}
The second item of \Cref{lem:two_good_o2h} can be rephrased as 
\begin{align}
\Pr[\event{G}_1\land \event{G}_2]\ge \frac{\epsilon}{8q^3}-\negl(\secp). \label{eq:G_1_and_G_2_o2h} 
\end{align}
\begin{lemma}\label{lem:Acc_1_o2h}
It holds that $\Pr[\event{Acc}_1]\ge \frac{\epsilon^2}{2^8 q^6 Q_\API^2}-\negl(\secp)$. 
\end{lemma}
\begin{proof} 
Recall that $\API_{i_1,j-1}$ and $\API_{i_1,j}$ work identically, except that
the oracle $H^{j-1}$ in the former is replaced by $H^j$ in the latter. 
Noting that $\epsilon'=\epsilon/(8q)<\epsilon/(4q)$,  the two items of \Cref{lem:two_good_o2h} gives us a distinguisher that distinguishes $H^{j-1}$ and $H^j$ with advantage at least $\epsilon/(8q^3)-\negl(\secp)$. 
Noting that $H^{j-1}$ and $H^j$ differ only on $S_j$,  \Cref{lem:O2H} gives us that
\[
\frac{\epsilon}{8q^3}-\negl(\secp)\le 2Q_\API\sqrt{\Pr[\event{Acc}_1]}. 
\]
This immediately gives \Cref{lem:Acc_1_o2h}. 
\end{proof}
\begin{lemma}\label{lem:G_2_conditioned_Acc_1_o2h}
It holds that $\Pr[\event{G}_2\mid \event{Acc}_1]\ge \frac{\epsilon^3}{2^{11} q^9 Q_\API^2}-\negl(\secp)$. 
\end{lemma}
\begin{proof}
We have 
$\Pr[\event{G}_2 \mid \event{Acc}_1]\ge  \Pr[\event{G}_2 \land \event{Acc}_1]=\Pr[ \event{Acc}_1 \mid \event{G}_2] \Pr[\event{G}_2]$.
We have $\Pr[\event{G}_2]\ge \epsilon/(8q^3)-\negl(\secp)$ by \Cref{eq:G_1_and_G_2_o2h}. Thus, it suffices to show  $\Pr[ \event{Acc}_1 \mid \event{G}_2]\ge \frac{\epsilon^2}{2^8 q^6 Q_\API^2}-\negl(\secp)$.  
To show this, consider the state $\rho^{\ast\ast}[i_1]$ on $\qreg{R}_{i_1}$, conditioned on the event $\event{G}_2$.  
The first item of \Cref{lem:two_good_o2h} continues to hold after conditioning since $\Pr[\event{G}_2]$ is inverse polynomial by \Cref{eq:G_1_and_G_2_o2h}.  
Moreover, the second item of \Cref{lem:two_good_o2h} also holds after conditioning since 
we have 
$\Pr[\event{G}_1 \mid \event{G}_2]\ge \Pr[\event{G}_1\land \event{G}_2]\ge \epsilon/(8q^3)$ by  \Cref{eq:G_1_and_G_2_o2h}. 
That is, we have the following: 
\begin{enumerate}
\item Suppose that we apply $\API_{i_1,j-1}$ on $\rho^{**}[i_1]$ and let $b'_{i_1,j-1}$ be the outcome. 
    Then it holds that
    $$\Pr[|b_{i_1,j-1}-b'_{i_1,j-1}|\le \epsilon']\ge 1-\negl(\secp).$$ 
\item Suppose that we apply $\API_{i_1,j}$ on $\rho^{**}[i_1]$ and let $b_{i_1,j}$ be the outcome. Then it holds that  
	\[
		\Pr[|b_{i_1,j-1} - b_{i_1,j}| \ge \epsilon/(4q)] \ge \frac{\epsilon}{8q^3}-\negl(\secp).
	\]
\end{enumerate}
Therefore, by applying \Cref{lem:O2H} in the same way as in the proof of
\Cref{lem:Acc_1_o2h}, we obtain  
\[
\Pr[ \event{Acc}_1 \mid \event{G}_2]\ge \frac{\epsilon^2}{2^8 q^6 Q_\API^2}-\negl(\secp).
\]
This completes the proof. 
\end{proof}
\begin{lemma}\label{lem:Acc_2_conditioned_Acc_1_o2h}
It holds that $\Pr[\event{Acc}_2\mid \event{Acc}_1]\ge \frac{\epsilon^6}{2^{24} q^{18} Q_\API^6}-\negl(\secp)$. 
\end{lemma}
\begin{proof}
The proof is very similar to that of \Cref{lem:G_2_conditioned_Acc_1_o2h}. 
Consider the state $\rho^{\ast\ast\ast}[i_2]$ on $\qreg{R}_{i_2}$, conditioned on the event $\event{Acc}_1$.  
The first item of \Cref{lem:two_good_o2h} continues to hold after  conditioning since $\Pr[\event{Acc}_1]$ is inverse polynomial by \Cref{lem:Acc_1_o2h}. 
Moreover,  by \Cref{lem:G_2_conditioned_Acc_1_o2h}, the second item of \Cref{lem:two_good_o2h}, with the bound replaced with $\frac{\epsilon^3}{2^{11} q^9 Q_\API^2}-\negl(\secp)$, also holds 
after conditioning.  
That is, we have the following: 
\begin{enumerate}
\item Suppose that we apply $\API_{i_2,j-1}$ on $\rho^{***}[i_2]$ and let $b'_{i_2,j-1}$ be the outcome. 
   Then it holds that
    $$\Pr[|b_{i_2,j-1}-b'_{i_2,j-1}|\le \epsilon']\ge 1-\negl(\secp).$$ 
\item Suppose that we apply $\API_{i_2,j}$ on $\rho^{***}[i_2]$ and let $b_{i_2,j}$ be the outcome. Then it holds that  
	\[
		\Pr[|b_{i_2,j-1} - b_{i_2,j}| \ge \epsilon/(4q)] \ge \frac{\epsilon^3}{2^{11} q^9 Q_\API^2}-\negl(\secp).
	\]
\end{enumerate}
Therefore, by applying \Cref{lem:O2H} in the same way as in the proof of
\Cref{lem:Acc_1_o2h}, we obtain
\[
\Pr[ \event{Acc}_2 \mid \event{Acc}_1]\ge \frac{\epsilon^6}{2^{24} q^{18} Q_\API^6}-\negl(\secp).
\]
This completes the proof. 
\end{proof}
\begin{lemma}\label{lem:Acc_1_and_Acc_2_o2h}
It holds that $\Pr[\event{Acc}_1\land\event{Acc}_2]\ge \frac{\epsilon^8}{2^{32} q^{24} Q_\API^8}-\negl(\secp)$. 
\end{lemma}
\begin{proof}
This immediately follows from \Cref{lem:Acc_1_o2h,lem:Acc_2_conditioned_Acc_1_o2h}.  
\end{proof}

\Cref{lem:Acc_1_and_Acc_2_o2h} means that $\mathsf{Ext}$ succeeds (i.e., $x_{i_1}\in S_j$ and $x_{i_2}\in S_j$) with probability at least $\frac{\epsilon^8}{2^{32} q^{24} Q_\API^8}-\negl(\secp)$ conditioned that $(i_1,i_2,j,\rho^*[i_1,i_2])$  is good. As already shown, $(i_1,i_2,j,\rho^*[i_1,i_2])$  is good with probability at least  $\epsilon/(8q^3)-\negl(\secp)$. Thus, the overall probability that $\mathsf{Ext}$ succeeds is at least $\frac{\epsilon^9}{2^{35} q^{27} Q_\API^8}-\negl(\secp)=1/\poly(\secp)$. 
This completes the proof of \Cref{lem:cr_O2H}. 

\newpage
\section{Relating Identical-Challenge and Strong-Independent-Challenge Search Security}\label{sec:relating}
In this section, we show that strong search security implies independent-challenge search security, even in the setting of arbitrary (polynomial) many adversaries, generalizing the result of \cite{C:AnaKalLiu23} who showed a similar statement for only two adversaries. This theorem helps us achieve collusion-resistant identical-secure single-decryptor encryption, which in turn helps us achieve collusion-resistant unclonable encryption in the plain model.

\begin{lemma}\label{thm:ind-to-iden}
    Let $\mathcal{X}$ be an index set, $k(\cdot): \mathds{N}^+ \to \mathds{N}^+$ be a function, and $\{ \{M_{x}^{(i)}, I - M_{x}^{(i)}\}  \}_{x \in \mathcal{X},i \in [k(\lambda)]}$ be binary POVM measurements. Let $\mathcal{D}$ be a distribution supported on $\mathcal{X}$, and let $\mathcal{R}$ denote the random coin set\footnote{That is, $\mathcal{D}$ is a function and sampling from $\mathcal{D}$ is simply sampling $r$ from $\mathcal{R}$ uniformly at random and outputting $\mathcal{D}(r)$.} of $\mathcal{D}$. Let $\TI_{i,\gamma}$ denote the threshold implementation (\cref{def:threshold_implementation}) with threshold $\gamma \in \zo$ corresponding to the mixture of measurements where sample uniformly random $r$ from $\mathcal{R}$ and apply the measurement $\{M_{\mathcal{D}(r)}^{(i)}, I - M_{\mathcal{D}(r)}^{(i)}\}$. Finally, let $\rho$ be a mixed state over $k(\lambda)$ registers.

    Define
    \begin{align*}
        &p_{\mathsf{iden}} =  \E_{x \samp \mathcal{{D}}}[\Tr[(M_{x}^{(1)}\otimes \dots \otimes M_{x}^{(k)})\rho]] \\
        &p_{\mathsf{strong},\gamma} = \Tr[(\TI_{1,\gamma}\otimes\dots\otimes\TI_{k,\gamma})\rho]
    \end{align*}

    Then,  $$p_{\mathsf{iden}} \leq (k+1)\cdot(\gamma + p_{\mathsf{strong},\gamma}).$$

\end{lemma}
We prove this theorem in \cref{sec:proof-of-relating}.
\ifconference
\section{Proof of \cref{thm:ind-to-iden}}\label{sec:proof-of-relating}
\else
\subsection{Proof of \cref{thm:ind-to-iden}}\label{sec:proof-of-relating}
\fi
By the Naimark dilation theorem, we may assume without loss of generality that the measurements $M_{x}^{(i)}$ are projective for all $x,i$. We can also assume that $\rho=\ket{\psi}\bra{\psi}$ is pure. The general (mixed-state) case follows by convexity.    

First, for each $i \in [k]$ we define
\[
P^{(i)} \seteq \Exp_{x\chosen \mathcal{D}}M^{(i)}_{x}.
\]
Observe that $P^{(i)}$ is positive semidefinite, and let $\ket{\psi} =\sum_{i_1,\ldots,i_{k}} \alpha_{i_1,\ldots,i_{k}} \ket{\phi_{i_1}^{(1)}}\cdots\ket{\phi_{i_{k}}^{(k)}}$
be the spectral decomposition of $\ket{\psi}$ with respect to $P^{(1)}\otimes \cdots \otimes P^{(k)}$, where for each $j\in[k]$, $\ket{\phi_{i_j}^{(j)}}$ is an eigenvector of $P^{(j)}$ corresponding to eigenvalue $\mu_{i_j}^{(j)}$.

We now set $\eta \seteq \gamma(\secp)$ and define
\begin{align}
\ket{\psi^{(1)}}& =\sum_{\substack{i_1: \mu_{i_1}^{(1)} \le \eta \\i_2: \mu_{i_2}^{(2)}> \eta \\\vdots\\ i_{k}:\mu_{i_{k}}^{(k)}>\eta}} \alpha_{i_1,\ldots,i_{k}}\ket{\phi_{i_1}^{(1)}}\cdots\ket{\phi_{i_{k}}^{(k)}},\nonumber\\
\vdots&\nonumber\\
\ket{\psi^{(j)}}& =\sum_{\substack{i_1,\ldots, i_{j-1}\\ i_{j}: \mu_{i_j}^{(j)}\le \eta\\ i_{j+1}:\mu_{i_{j+1}}^{(j+1)}>\eta \\\vdots\\ i_{k}:\mu_{i_{k}}^{(k)}>\eta}} \alpha_{i_1,\ldots,i_{k}}\ket{\phi_{i_1}^{(1)}}\cdots\ket{\phi_{i_{k}}^{(k)}},\nonumber\\
\vdots&\nonumber\\
\ket{\psi^{(k)}}& =\sum_{\substack{i_1,\ldots,i_{k-1} \\ i_{k}:\mu_{i_{k}}^{(k)}\le\eta}} \alpha_{i_1,\ldots,i_{k}}\ket{\phi_{i_1}^{(1)}}\cdots\ket{\phi_{i_{k}}^{(k)}},\nonumber\\
\ket{\psi_S}&=\sum_{\substack{i_1: \mu_{i_1}^{(1)} > \eta \\\vdots\\ i_{k}:\mu_{i_{k}}^{(k)}>\eta}} \alpha_{i_1,\ldots,i_{k}}\ket{\phi_{i_1}^{(1)}}\cdots\ket{\phi_{i_{k}}^{(k)}},\nonumber
\end{align}
so that $\ket{\psi}=\sum_{j=1}^{k}\ket{\psi^{(j)}} +\ket{\psi_{S}}$. Observe that we have $\norm{\psi_S}^2 = p_{\mathsf{strong},\gamma}$.

For notational convenience, let
\begin{align}
M^\ast_x &\seteq M^{(1)}_x\otimes \cdots \otimes M^{(k)}_x,\\
M^{(-i)}_x &\seteq \overbrace{\mat{I} \otimes \cdots \otimes \mat{I}}^{i-1} \otimes M^{(i)}_x \otimes\overbrace{\mat{I} \cdots \otimes \mat{I}}^{k-i},\\
P^{(-i)} &\seteq \overbrace{\mat{I} \otimes \cdots \otimes \mat{I}}^{i-1} \otimes P^{(i)} \otimes\overbrace{\mat{I} \cdots \otimes \mat{I}}^{k-i}.
\end{align}

Then, we have
\begin{align}
p_{\mathsf{iden}}&= \Exp_{x\chosen \mathcal{D}} \norm{(M^{\ast}_x) \left(\ket{\psi^{(1)}}+\cdots + \ket{\psi^{(k)}} + \ket{\psi_S}\right)}^{2}\nonumber\\
&\le \Exp_{x\chosen \mathcal{D}} (k+1)\left(\norm{(M^{\ast}_x )\ket{\psi^{(1)}}}^2 +\cdots + \norm{(M^{\ast}_x )\ket{\psi^{(k)}}}^2 + \norm{(M^{\ast}_x )\ket{\psi_S}}^2\right)\label{eq:CSieq}\\
&\le \Exp_{x\chosen \mathcal{D}} (k+1)\left(\norm{(M^{(-1)}_x)\ket{\psi^{(1)}}}^2 +\cdots + \norm{(M^{(-k)}_r )\ket{\psi^{(k)}}}^2 + \norm{\ket{\psi_S}}^2\right)\nonumber\\
&=\Exp_{x\chosen \mathcal{D}} (k+1)\left(\norm{(M^{(-1)}_x)\ket{\psi^{(1)}}}^2 +\cdots + \norm{(M^{(-k)}_x )\ket{\psi^{(k)}}}^2 + \norm{\ket{\psi_S}}^2\right)\nonumber\\
&=(k+1)\left(\bra{\psi^{(1)}} P^{(-1)}\ket{\psi^{(1)}} + \cdots + \bra{\psi^{(k)}} P^{(-k)}\ket{\psi^{(k)}} +\norm{\ket{\psi_S}}^2\right)\nonumber\\
&= (k+1)\left(\eta \left(\norm{\psi^{(1)}}^2 +\cdots +\norm{\psi^{(k)}}^2\right) + \norm{\ket{\psi_S}}^2\right)\nonumber\\
&= (k+1)\left(\eta \left(\norm{\psi^{(1)}}^2 +\cdots +\norm{\psi^{(k)}}^2\right) + p_{\mathsf{strong},\gamma} \right)\nonumber\\
&\le (k+1)(\gamma(\secp)+p_{\mathsf{strong},\gamma})\nonumber
\end{align}
Above, \Cref{eq:CSieq} is derived by Cauchy--Schwarz inequality. This completes the proof.

\newpage
\part{Quantum Money}
\section{Quantum Coins}\label{sec:coincomp}
In this section, we give two compilers (\cref{sec:money_PR_compiler} and \cref{sec:money_eqsup_compiler}) that transforms any public-key quantum money mini-scheme (i.e. a single-banknote scheme - see \cref{def:qmm}) into a public-key quantum \emph{coin} scheme, that is, a quantum money scheme where all honest banknotes are identical pure states.

We now define quantum money and quantum coin schemes. Since our focus is public-key schemes, we only define the public-key versions, and the private-key versions are defined similarly as in the case of mini-schemes.
\begin{definition}[Quantum Money and Quantum Coins]\label{def:qcoins}
A (public-key) quantum money scheme consists of the following efficient algorithms.
\begin{itemize}
    \item $\Setup(1^{\secp})$ outputs a public verification key $vk$ and a secret key $sk$.
    \item $\mathsf{GenBanknote}(sk)$ outputs a register.
    \item $\Ver(vk, \reg)$ outputs $1$ or $0$, denoting if the input is a valid banknote.
\end{itemize}

If $\mathsf{GenBanknote}(sk)$ outputs a fixed pure state $\ket{\psi_{sk}}$, then the scheme is called a quantum coin scheme.

We require correctness as before and also require the following security guarantee.

\paragraph{Unforgeability}
We require that for any polynomial $t(\cdot)$ and for any QPT adversary $\cA$ 
    \begin{equation*}
        \Pr[\forall i \in [t(\lambda)+1] b_i = 1 : \begin{array}{c}
            sk \samp \mathsf{Setup}(1^\lambda) \\
             (\reg_i) \samp \mathsf{GenBanknote}(sk)~\text{for}~i \in [t(\lambda)] \\
             (\reg^{adv}_{i})_{i \in [t(\lambda)+1]} \samp \adve(vk,(\reg_i)_{i \in [t(\lambda)]}) \\
             (b_i) \samp \mathsf{Verify}(vk, \reg^{adv}_i)~\text{for}~i \in [t(\lambda)+1]
        \end{array}] \leq \negl(\lambda)
    \end{equation*}
\end{definition}

\subsection{Pseudorandom State Compiler}\label{sec:money_PR_compiler}
Let $\minibank=(\minibank.\Gen,\minibank.\mathsf{Verify})$ be a public-key mini-scheme. 
We assume that $\minibank.\Gen$ has classically determined outputs (\Cref{def:classically_determined}); namely, its money state is a pure state fully determined by its classical randomness (and the security parameter).  
This is a very natural requirement satisfied by the existing schemes.

Let $\dss$ be a classical deterministic signature scheme that satisfies unforgeability security with classical queries. Let $F$ a pseudorandom function - its input-output size will be clear from context. Let $\mathsf{PRS}$ be a pseudorandom state scheme.

\paragraph{\underline{$\bank.\mathsf{Setup}(1^\lambda)$}}
\begin{enumerate}
    \item Sample $K \samp F.\mathsf{Setup}(1^\lambda)$.
\item Sample $vk, sgk \samp \dss.\mathsf{Setup}(1^\lambda)$.
\item Sample $k \samp \mathsf{PRS.Setup}(1^\lambda)$
\item Set $sk = (sgk, K, k)$
\item Output $vk, sk$.
\end{enumerate}

\paragraph{\underline{$\bank.\mathsf{GenBanknote}(sk)$}}
\begin{enumerate}
\item Parse $(sgk, K, k) = sk$.
    \item Prepare the state $\ket{\$} = \sum_{x}\alpha_x\ket{x}\ket{sn_{x}}\ket{\dss.\mathsf{Sign}(sgk, sn_{x})}\ket{\$_{x}}$ where $sn_{x},\ket{\$_{x}} = \minibank.\mathsf{Gen}(1^\lambda; r_{x})$ and $r_{x} = F(K, x)$ and $\ket{\psi_k} = \mathsf{PRS.StateGen}(k) = \sum_x \alpha_x\ket{x}$. $\ket{\$}$ can be prepared efficiently by running $\mathsf{PRS.StateGen}(k)$ and then running $F(K,\cdot)$,  $\minibank.\mathsf{Gen}$ and the signing algorithms coherently.
    
    \item Output $\ket{\$}$.
\end{enumerate}

\paragraph{\underline{$\bank.\mathsf{Verify}(vk, \reg)$}}

\begin{enumerate}
\item Perform the following operations coherently (as in Gentle Measurement Lemma \cite{aarlemma}), then rewind.
\begin{enumerate}[label=\arabic*.]
    \item Measure all but the last register (the mini-banknote register) of $\reg$ to obtain $id, sn, sig, \reg'$.
    \item Verify the signature by $\dss.\mathsf{Verify}(vk, sn, sig)$.
    \item Verify $\minibank.\mathsf{Verify}(sn, \reg')$
\end{enumerate}
\item If all verifications passed, output $1$. Otherwise, output $0$.
\end{enumerate}
\begin{theorem}
    $\mathsf{Bank}$ is a quantum coin scheme that satisfies correctness and unclonability (counterfeiting) security.
\end{theorem}
\begin{proof}
    The correctness is immediate. By \cref{thm:compiler}, security reduces to the mixed state case where the adversary receives independent samples from the mini-scheme with signed serial numbers, and the security in that case was proven by \cite{TOC:AarChr13}.
\end{proof}

Since the PRS, PRF and signature schemes can be constructed from one-way functions,
\begin{corollary}\label{cor:qcoin}
    Assuming the existence of a public-key mini-scheme whose banknote generation algorithm has classically determined outputs and one-way functions, there exists a public-key quantum coin scheme.
\end{corollary}

Since a public-key mini-scheme can be constructed from subspace-hiding obfuscation (\cite{EC:Zhandry19b}), we also get the following corollary.
\begin{corollary}
    Assuming the existence of subspace-hiding obfuscation and one-way functions, there exists a public-key quantum coin scheme.
\end{corollary}

\subsection{Equal Superposition Compiler for Quantum Money}\label{sec:money_eqsup_compiler}
In this section, we show that the natural scheme that uses the equal superposition of all possible banknotes as the purified banknote gives a secure scheme.

Let $\minibank=(\minibank.\Gen,\minibank.\mathsf{Verify})$ be a public-key mini-scheme. 
We assume that $\minibank.\Gen$ has classically determined outputs (\Cref{def:classically_determined}).

Let $\dss$ be a classical deterministic signature scheme that satisfies $\mathsf{BZ}$-security (\cref{defn:bzsec}). Let $F$ a pseudorandom function - its input-output size will be clear from context. Finally, let $\nu(\lambda)$ be a function such that $1/2^{\nu(\lambda)}$ is negligible (i.e., $\nu(\lambda)$ is superlogarithmic).

\paragraph{\underline{$\bank.\mathsf{Setup}(1^\lambda)$}}
\begin{enumerate}
    \item Sample $K \samp F.\mathsf{Setup}(1^\lambda)$.
\item Sample $vk, sgk \samp \dss.\mathsf{Setup}(1^\lambda)$.
\item Set $sk = (sgk, K)$
\item Output $vk, sk$.
\end{enumerate}

\paragraph{\underline{$\bank.\mathsf{GenBanknote}(sk)$}}
\begin{enumerate}
\item Parse $(sgk, K) = sk$.
    \item Prepare the state $\ket{\$} = \sum_{id \in \zo^{\nu(\lambda)}}\ket{id}\ket{sn_{id}}\ket{\dss.\mathsf{Sign}(sgk, sn_{id})}\ket{\$_{id}}$ where $sn_{id},\ket{\$_{id}} = \minibank.\mathsf{Gen}(1^\lambda; r_{id})$ and $r_{id} = F(K, id)$.
    
    \item Output $\ket{\$}$.
\end{enumerate}

\paragraph{\underline{$\bank.\mathsf{Verify}(vk, \reg)$}}

\begin{enumerate}
\item Perform the following operations coherently (as in Gentle Measurement Lemma \cite{aarlemma}), then rewind.
\begin{enumerate}[label=\arabic*.]
    \item Measure all but the last register (the mini-banknote register) of $\reg$ to obtain $id, sn, sig, \reg'$.
    \item Verify the signature by $\dss.\mathsf{Verify}(vk, sn, sig)$.
    \item Verify $\minibank.\mathsf{Verify}(sn, \reg')$
\end{enumerate}
\item If all verifications passed, output $1$. Otherwise, output $0$.
\end{enumerate}
\begin{theorem}\label{thm:coinsecure}
    $\mathsf{Bank}$ is a quantum coin scheme that satisfies correctness and unclonability (counterfeiting) security.
\end{theorem}
We note that our compiler actually works with arbitrary delocalized superpositions (\cref{def:delocalized}), rather than just equal superpositions.

Since the required signature scheme can be constructed from subexponentially secure collision-resistant hash functions (\cref{thm:bzexists}), we obtain the following corollary.
\begin{corollary}
    Assuming the existence of a public-key mini-scheme and subexponentially secure collision-resistant hash functions, there exists a public-key quantum coin scheme.
\end{corollary}

Since a public-key mini-scheme can be constructed from subspace-hiding obfuscation (\cite{EC:Zhandry19b}), we also get the following corollary.
\begin{corollary}
    Assuming the existence of subspace-hiding obfuscation and subexponentially secure collision-resistant hash functions, there exists a public-key quantum coin scheme.
\end{corollary}

\subsubsection{Proof of Security}
In this section, we prove \cref{thm:coinsecure}. Correctness follows in a straightforward manner from the correctness of the mini-scheme and the other underlying schemes. It is also evident that the scheme always produces identical copies of the pure state $\ket{\$}$ as banknotes.

We move onto counterfeiting (unclonability) security, which we will prove through a sequence of hybrids (each of which is constructed by modifying the previous one). Suppose for a contradiction that there exists a QPT adversary $\adve$ that wins the counterfeiting game with probability $> 1/p(\lambda)$ for infinitely many values of $\lambda > 0$, where $p(\cdot)$ is a polynomial. Let $k(\cdot)$ denote the number of banknotes requested by the adversary\footnote{Note that this does \textbf{not} make the construction \emph{bounded collusion}: This $k$ value is only used in the reduction, not in the construction.}.

\paragraph{$\underline{\hyb_0}$}: The original game $\cfgame{\adve}(1^\lambda)$.

\paragraph{$\underline{\hyb_1}$}: During setup, we sample a random function $H$ with the same input-output size as $F$. When answering the banknote generation queries, we now output $\ket{\$'}$ where
\begin{equation*}
    \ket{\$'} = \sum_{id \in \zo^{\nu(\lambda)}}\ket{id}\ket{sn_{id}}\ket{\dss.\mathsf{Sign}(sgk, sn_{id})}\ket{\$_{id}}
\end{equation*}
where $sn_{id},\ket{\$_{id}} = \minibank.\mathsf{Gen}(1^\lambda; H(id))$.

\paragraph{$\underline{\hyb_2}$}:
During setup, we sample a truly random function $Q: \zo^{\nu(\lambda)} \to [\ell]$ where $\ell = 16\ocsrd p^2k^3$ and $\ocsrd$ is the universal constant from \cref{thm:qsreadonce}. Then, we sample  $\tilde{sn}_{i}, \ket{\tilde{\$}_{i}} \samp \minibank.\mathsf{Gen}(1^\lambda)$ for $i \in [\ell]$. Further,
instead of providing the $k$ banknotes as $\ket{\psi} = \ket{\$'}^{\otimes{k}}$ to the adversary, we output $\ket{\phi}$ where 
\begin{align*}
 \ket{\phi} = \sum_{\substack{id_1, \dots, id_k \in \zo^{\nu(\lambda)}:\\ Q(id_1), \dots, Q(id_k) \text{ distinct}}}\left(\ket{id_1}_{\mathsf{ID}_1}\ket{\tilde{sn}_{Q(id_1)}}\ket{\dss.\mathsf{Sign}(sgk, \tilde{sn}_{Q(id_1)})}\ket{\tilde{\$}_{Q(id_1)}}\right)\otimes \cdots \\ \otimes \left(\ket{id_k}_{{\mathsf{ID}_k}}\ket{\tilde{sn}_{Q(id_k)}}\ket{\dss.\mathsf{Sign}(sgk, \tilde{sn}_{Q(id_k)})}\ket{\tilde{\$}_{Q(id_k)}}\right)
\end{align*}

\paragraph{$\underline{\hyb_3}$}: We add an extra check at the end of the experiment, in addition to the banknote verification using $\bank.\mathsf{Verify}$. Let $sn^*_1, \dots, sn^*_{k+1}$ be the serial numbers of the forged banknotes\footnote{More formally, let these be the serial numbers obtained from measuring the serial number registers during verification}. We check if $sn^*_i \in \{\tilde{sn}_1, \dots, \tilde{sn}_\ell\}$ for all $i \in [k+1]$. If any of the checks fail, the adversary loses and experiment outputs $0$.

\paragraph{$\underline{\hyb_4}$}: We add yet another additional check at the end of the experiment. We check if the serial numbers $sn^*_1, \dots, sn^*_{k+1}$ are distinct, and if not, the adversary loses and experiment outputs $0$.

\paragraph{$\underline{\hyb_5}$}: We add yet another additional check at the end of the experiment. We sample a random value $i^* \samp [\ell]$ at the beginning of the game. At the end, if $\tilde{sn}_{i^*}$ does not appear at least twice in $sn^*_1, \dots, sn^*_{k+1}$, the adversary loses and experiment outputs $0$.

\paragraph{$\underline{\hyb_6}$}: Instead of giving the adversary the state $\ket{\phi}$, instead give it the state $\ket{\phi'}$ constructed as follows. First add an additional auxiliary register whose space is $\C^{\Z_{k+1}}$ to the end of the state $\ket{\phi}$, and initialize it to $\ket{0}$. Then, for each $j \in {1,\dots,k}$, apply a controlled add-$j$ gate from the register $\mathsf{ID}_j$ to the auxiliary register with the control bit being $1$ if the value $id$ in $\mathsf{ID}_j$ is such that $Q(id) = i^*$. Then, give the original registers (i.e. the registers except the measured auxiliary register) to the adversary.

\paragraph{}
We now show hybrid indistinguishability results. Throughout the remainder of the section, when needed, security reductions will implicitly use PRFs to simulate any random functions (such as $Q$) involved in the experiments to make them efficient.

\begin{lemma}
    $\hyb_0 \approx \hyb_1$
\end{lemma}
\begin{proof}
    We simply replaced the PRF queries with a truly random function. The result follows by the PRF security of $F$.
\end{proof}

\begin{lemma}
    $|\Pr[\hyb_1 = 1] - \Pr[\hyb_2 = 1]| < \frac{1}{4p}$.
\end{lemma}
\begin{proof}
Observe that here we are simply using the simulator from the quantum-state read-once small-range distribution lemma (\cref{thm:qsreadonce}).
\end{proof}

\begin{lemma}
    $\hyb_2 \approx \hyb_3$.
\end{lemma}
\begin{proof}
    Observe to simulate these experiments, we only need the signatures $\dss.\mathsf{Sign}(sgk, \tilde{sn}_{i})$ for $i \in [\ell]$. By $\mathsf{EUF-CMA}$ security of $\dss$, the result follows.
\end{proof}

\begin{lemma}
    $\hyb_3 \approx \hyb_4$.
\end{lemma}
\begin{proof}
    Due to $\mathsf{BZ}$ security of $\dss$, we know that (except with negligible probability) the adversary cannot output signatures on $k + 1$ distinct messages. Thus the result follows.
\end{proof}

\begin{lemma}
    $\Pr[\hyb_5 = 1] \geq \frac{\Pr[\hyb_4 = 1]}{\ell}$.
\end{lemma}
\begin{proof}
    We know that there is a repeating value among $sn^*_1, \dots, sn^*_{k+1}$, and all these values are in the set $\{\tilde{sn}_1, \dots, \tilde{sn}_\ell\}$. Thus, with probability $1/\ell$, the index $i^*$ we sampled randomly will be the repeating one.
\end{proof}

\begin{lemma}
    $\Pr[\hyb_6 = 1] \geq \frac{\Pr[\hyb_5 = 1]}{k+1}$.
\end{lemma}
\begin{proof}
Since there are $k+1$ measurement outcomes, the result follows by \cref{lem:midmeasure}.
\end{proof}

\begin{lemma}
    $\Pr[\hyb_6 = 1] \leq \negl(\lambda).$
\end{lemma}
\begin{proof}
First, note that due to the distinctness of these values in the support of the state $\ket{\phi}$, when we are constructing $\ket{\phi'}$, each state in the support of $\ket{\phi}$ will have its auxiliary register modified at most once, and the value in the auxiliary state will exactly show the index at which the state $\ket{\tilde{\$}_{i^*}}$ is used. Thus, once we measure the auxiliary register, no matter what the measurement outcome is, the resulting state either does not involve at all or uses  $\ket{\tilde{\$}_{i^*}}$ exactly once (at the the index equal to the auxiliary register measurement outcome). Thus, we can simulate $\ket{\phi'}$ using one copy of $\ket{\tilde{\$}_{i^*}}$.

Using this, we obtain a reduction to the security of the mini scheme, by construct the following adversary $\adve'$ for its security game. $\adve'$ simulates $\adve$ by sampling independent instance of the mini scheme on its own for all $i \neq i^*$ as $\tilde{sn}_{i}, \ket{\tilde{\$}_{i}} \samp \minibank.\mathsf{Gen}(1^\lambda)$, and setting $\tilde{sn}_{i^*}, \ket{\tilde{\$}_{i^*}}$ to be the challenge banknote received from its own challenger. At the end, we know that if $\adve$ wins, it means that it has output two banknotes that pass verification with the serial number $\tilde{sn}_{i^*}$. Thus, $\adve'$ wins the mini scheme security game with probability $\Pr[\hyb_5 = 1]$. By security of the mini-scheme, we get $\Pr[\hyb_5 = 1] \leq \negl(\lambda)$.
\end{proof}

Combining all but the last lemma gives us $\Pr[\hyb_6 = 1] \geq \frac{1}{4p\ell(k+1)} \geq \frac{1}{64\ocsrd p^2 k^3\cdot(k+1)}$,  whereas the last lemma says $\Pr[\hyb_6 = 1] \leq \negl(\lambda)$. Thus, we arrive at a contradiction, which completes the proof of security.
\newpage
\section{Upgradable Quantum Coins}\label{sec:upgrade}
In this section, we introduce the notion of upgradable quantum coins. 
Then, we give a construction that transforms any public-key mini-scheme into an upgradable quantum coin.

We first introduce the relevant models. We will assume that for quantum money mini-schemes, the banknote generation proceeds in two steps: $\minibank.\mathsf{Gen}$ outputs a private string $s$ (e.g. the classical description of a quantum state), and $\minibank.\mathsf{State}(s)$ outputs the related banknote state and $\minibank.\mathsf{SerialNumber}(s)$ outputs the related serial number. This is a natural model and indeed applies to all existing constructions. Note that in the private-key setting, we can take the serial number to be the same as $s$.

\begin{definition}[Upgradable Quantum Money Mini-Scheme]
        Let $\crcons_1, \crcons_2$ be sets of cryptographic constructions, e.g., $\crcons_1 = \{\mathsf{one-way-function}\}$ and $\crcons_2 = \{\mathsf{one-way-function},\allowbreak \mathsf{subspace}-\mathsf{hiding}-\mathsf{obfuscation}\}$. A quantum money mini-scheme is said satisfy upgradable correctness and security with respect to  $\crcons_1, \crcons_2$ if it satisfies the following.
        \begin{itemize}
            \item $\minibank.\mathsf{Gen}$ and $\minibank.\mathsf{State}$ only uses $\crcons_1$, and private-key security (i.e., unclonability without serial number revealed) is satisfied only assuming $\crcons_1$.
            \item $\minibank.\mathsf{SerialNumber}(s)$ uses constructions from $\crcons_2$, and the mini-scheme satisfies public-key security (i.e., unclonability with serial number given to adversary) assuming $\crcons_2$.
        \end{itemize}
\end{definition}
An example is the subspace-based quantum money construction of \cite{TOC:AarChr13,EC:Zhandry19b}. It satisfies upgradable security with $\crcons_1 = \{\mathsf{one-way-function}\}$ and $\crcons_2 = \{\mathsf{subspace-hiding-obfuscation}\}$ where $\minibank.\mathsf{SerialNumber}$ would be the subspace-hiding obfuscation construction.

\begin{definition}[Upgradable Quantum Coin]
    Let $\crcons_1, \crcons_2$ be sets of cryptographic constructions. A quantum coin scheme with upgradable correctness and security with respect to $(\crcons_1, \crcons_2)$ is defined to be a quantum coin scheme with an additional algorithm $\bank.\mathsf{Reveal}$ that satisfies the following. 

    \begin{itemize}
        \item \textbf{Before Reveal:} Banknotes can be generated assuming only $\crcons_1$, without using any information about $\crcons_2$. No key is published initially. $\bank$ satisfies correctness and counterfeiting (unclonability) security assuming only $\crcons_1$. Banknote verification by users is done through comparison-based-verification (see \cite{bs21}). 
        \item \textbf{After Reveal:} $\bank.\mathsf{Reveal}$ takes as input bank's secret key, and a randomized program $P$ (which uses constructions from $\crcons_2$), and outputs a classical public verification key $vk$ which is published. The scheme (with $vk$ published) satisfies full-fledged public-key quantum money correctness and security assuming $\crcons_2$.
    \end{itemize}
\end{definition}

\subsection{Construction}
Let $\minibank$ be an upgradable quantum-money mini-scheme. Let $\crcons_1, \crcons_2$ be sets of cryptographic assumptions. Assume that $\minibank$ is secure assuming $\crcons_1$ in the private mode, and it is correct and secure assuming $\crcons_2$ in the public mode. 

Let $\dss$ be a classical deterministic signature scheme that satisfies unforgeability security with classical queries.  Let $F$ a pseudorandom function - its input-output size will be clear from context. Let $\npke$ be a non-committing public-key encryption scheme.  Let $\nu(\lambda)$ be a function such that $2^{\nu(\lambda)}$ is superpolynomial (i.e., $\nu(\lambda)$ is superlogarithmic). Finally, let $\ofe$ be a public-key functional encryption scheme that satisfies simulation security for a single functional-key with $2^{-\nu(\lambda)}\cdot \negl(\lambda)$ security. 

We move onto our construction below. The construction will be upgraded by publishing $\bank.\mathsf{Reveal}(sk, P = \minibank.\mathsf{SerialNumber})$ once an implementation of $\minibank.\mathsf{SerialNumber}$ with $\crcons_2$ becomes available.
\paragraph{\underline{$\bank.\mathsf{Setup}(1^\lambda)$}}
\begin{enumerate}
    \item Sample $K \samp F.\mathsf{Setup}(1^\lambda)$.
    \item Sample $pk, msk \samp \ofe.\mathsf{Setup}(1^\lambda)$.
    \item Sample $npk, nsk \samp \npke.\mathsf{Setup}(1^\lambda)$.
    \item Sample $svk, sgk \samp \dss.\mathsf{Setup}(1^\lambda)$.
    \item Sample $k \samp \mathsf{PRS.Setup}(1^\lambda)$
\item Set $sk = (K, pk, msk, npk, nsk, svk, sgk, k)$.
\item Output $sk$.
\end{enumerate}

\paragraph{\underline{$\bank.\mathsf{GenBanknote}(sk)$}}
\begin{enumerate}
\item Parse $(K, pk, msk, npk, nsk, svk, sgk, k) = sk$.
    \item Prepare the state $\ket{\$} = \sum_{id \in \zo^{\nu(\lambda)}}\alpha_{id}\ket{id}\ket{\npke.\enc(npk, \ofe.\enc(pk, s_{id} || r_{id}^4 || sgk; r^2_{id}); r^3_{id})}\ket{\$_{id}}$ where $r^1_{id} || r^2_{id} || r^3_{id} || r^4_{id} = F(K,id)$, $s_{id} = \minibank.\mathsf{Gen}(1^\lambda; r^1_{id})$ and $\ket{\$_{id}} = \minibank.\mathsf{State}(s_{id})$ and $\ket{\psi_k} = \mathsf{PRS.StateGen}(k) = \sum_{id} \alpha_{id}\ket{id}$.
    
\item Output $\ket{\$}$.
\end{enumerate}

\paragraph{\underline{$\bank.\mathsf{Reveal}(sk, P)$}}
\begin{enumerate}
    \item Parse $(K, pk, msk, npk, nsk, svk, sgk) = sk$.
    \item Let $f_{P}$ be the circuit that takes as input $s || r || k$ and outputs $\dss.\Sign(k, P(s; r)) || P(s; r)$.
\item Sample $fsk \samp \ofe.\mathsf{KeyGen}(msk, f_{P})$
\item Output $vk = (nsk, svk, fsk)$.
\end{enumerate}

\paragraph{\underline{$\bank.\mathsf{Verify}(vk, \reg)$}}
\begin{enumerate}
\item Parse $(nsk, svk, fsk) = vk$.
\item Measure all but the last register (the mini-banknote register) of $\reg$ to obtain $id, ct, \reg'$.
\item Decrypt $ct' \samp \npke.\dec(nsk, ct)$.
\item Decrypt $sig || sn \samp \ofe.\dec(fsk, ct')$.
\item Verify the signature by $\dss.\mathsf{Verify}(svk, sn, sig)$; output $0$ and terminate if verification fails.
\item Execute the verification code inside $sn$ on $\reg'$, output its output.
\end{enumerate}
Since the required PRF (by \cite{Z12}), PRS, non-committing encryption (by \cite{hiroka2021quantum}) and functional encryption schemes (by \cite{sahai2010worry,gorbunov2012functional}) can be constructed from PKE with subexponential security, we get the following result.
\begin{theorem}
    $\mathsf{Bank}$ is an upgradable quantum coin scheme that satisfies correctness and upgradable security with respect to $(\crcons_1 \bigcup \allowbreak \{\mathsf{subexp-PKE}\},\crcons_2 \bigcup \allowbreak \{\mathsf{subexp-PKE}\})$.
\end{theorem}
Since we can instantiate the mini-scheme with $\crcons_1 = \{\mathsf{one-way-function}\}$ and  $\crcons_2 = \{\mathsf{subspace}-\mathsf{hiding}-\mathsf{obfuscation}\}$ by \cite{EC:Zhandry19b}, we get the following corollary.
\begin{corollary}
    There exists an upgradable quantum coin scheme that satisfies correctness and upgradable security with respect to $\crcons_1 = \{\mathsf{\mathsf{subexp-PKE}}\}$ and  $\crcons_2 = \{\mathsf{\mathsf{subexp-PKE}},\allowbreak\mathsf{subspace}-\mathsf{hiding}-\mathsf{obfuscation}\}$.
\end{corollary}

\subsection{Proof of Security}
Proof of security before the reveal phase simply follows by \cref{thm:compiler}.
We will prove the counterfeiting security of the scheme after the reveal phase, that is, once $vk = (nsk, svk, fsk)$ is public.

\paragraph{$\underline{\hyb_0}$}: The original counterfeiting security game.

\paragraph{$\underline{\hyb_1}$}: When answering the banknote queries, instead of computing $r^1_{id} || r^2_{id} || r^3_{id} || r^4_{id} = F(K, id)$, we now compute $r^1_{id} || r^2_{id} || r^3_{id} || r^4_{id} = H(id)$, where $H$ is a truly random function sampled at setup.

\paragraph{}
Let $\mathcal{S}$ be a simulator, promised by the simulation-security of $\ofe$, such that $(pk, fsk, \ofe.\enc(pk, pl)) \approx_{2^{-\nu(\lambda)}\cdot\negl(\lambda)} (pk, fsk, \mathcal{S}(pk, fsk, f_P(pl)))$. 
\paragraph{\underline{${\hyb_{2,j}}$  for $j \in \zo^{v(\lambda)}$}}: 

When answering the banknote queries, we now output 
\begin{equation*}
    \ket{\$^{(j)}} = \sum_{id \in \zo^{\nu(\lambda)}}\ket{id}\ket{a_{id}}\ket{\$_{id}}
\end{equation*}
where we define $a_{id} = \npke.\enc(npk, \ofe.\enc(pk, s_{id} || r_{id}^4 || sgk; r^2_{id}); r^3_{id})$ for $id > j$ and define 
\begin{equation*}
    a_{id} = \npke.\enc(S(pk, fsk, \dss.\Sign(sgk, P(s_{id}; r^4_{id})) ||  P(s_{id}; r^4_{id}); H'(id)); r^3_{id})
\end{equation*}
for $id \leq j$.

\begin{lemma}
    $\hyb_0 \approx \hyb_1$
\end{lemma}
\begin{proof}
    We simply replaced the PRF queries with a truly random function, the result follows by the PRF security.
\end{proof}
\begin{lemma}
    $\hyb_{1} \approx_{2^{-\nu(\lambda)}\cdot\negl(\lambda)} \hyb_{2,0}$ and $\hyb_{2,j} \approx_{2^{-\nu(\lambda)}\cdot\negl(\lambda)} \hyb_{2,j+1}$ for all $j \in \zo^{\nu(\lambda)}$.\footnote{Here, the strings are associated with the numbers $0, 1, \dots, 2^{\nu(\lambda)}$ in the natural way.} 
\end{lemma}
\begin{proof}
    The only difference between the two neighboring hybrids is switching from the real ciphertext to the simulated ciphertext for $id = j + 1$. The result thus follows by the simulation security of $\ofe$.
\end{proof}

\begin{corollary}
    $\hyb_0 \approx \hyb_{2, 2^{\nu(\lambda)}}$.
\end{corollary}
\begin{proof}
    This follows by combining the above lemmas.
\end{proof}

\begin{lemma}
    $\Pr[\hyb_{2, 2^{\nu(\lambda)}} = 1] = \negl(\lambda)$.
\end{lemma}
\begin{proof}
    In this hybrid, we answer the banknote queries with the state
    \begin{equation*}
        \ket{\phi} = \sum_{id}\alpha_{id}\ket{id}\ket{ \npke.\enc(S(pk, fsk, \dss.\Sign(sgk, P(s_{id}; r^4_{id})) ||  P(s_{id}; r^4_{id}); H'(id)); r^3_{id})}\ket{\$_{id}}
    \end{equation*}
    Observe that this state can be simulated using the state
    \begin{equation*}
        \ket{\phi} = \sum_{id}\alpha_{id}\ket{id}\ket{P(s_{id}; r^4_{id})}\ket{\dss.\Sign(sgk, P(s_{id}; r^4_{id}))}\ket{\$_{id}}
    \end{equation*}
    by sampling the instances of $\ofe$ and $\npke$ ourselves. Finally, observe that $\ket{\phi}$ is exactly the banknote states that the compiler in \cref{sec:coincomp}. Since the compiled scheme's security is proven in \cref{sec:coincomp} assuming the security of the mini scheme and $\dss$, it follows that $\Pr[\hyb_{2, 2^{\nu(\lambda)}} = 1] = \negl(\lambda)$ assuming the security of $\dss$ and the constructions $\crcons_2$
\end{proof}

By above we conclude $\Pr[\hyb_0 = 1] = \negl(\lambda)$, which completes the security proof.

\newpage
\part{Quantum Protection}
\section{Copy-Protected Decryption Keys (Single-Decryptor Encryption)}\label{sec:CR-transformation}
In this section, we construct multi-copy secure copy-protection schemes for decryption keys, that is, single-decryptor encryption (SDE) schemes.

In \Cref{sec:SDE_def}, we recall the definitions of collusion-resistant SDE. In \Cref{sec:compiler_SDE}, we obtain collusion-resistant schemes using our generic compiler (\cref{sec:fe-compiler}) and consider its instantiations. In \Cref{sec:compiler_SDE_bounded}, we give an alternative compiler that yields bounded collusion resistance without any additional assumptions.
In \Cref{sec:identical_SDE}, we show that collusion-resistant strong security implies search-based security in the identical-challenge setting, which will be crucial for our compiler for unclonable encryption in \Cref{Sec:UE}. Finally, in \Cref{sec:multi-copy_SDE}, we apply our purification compiler (\Cref{thm:compiler}) to obtain multi-copy secure schemes from our collusion-resistant schemes.

\subsection{Definitions of Collusion-Resistant SDE}\label{sec:SDE_def}
In this section, we provide some definitions related to \emph{single-decryptor encryption} (SDE) (i.e. copy-protected decryption keys). 

\paragraph{\bf Public key SDE.} 
We note that a public-key SDE scheme is simply a public-key encryption scheme with quantum protected keys (\cref{def:qpke}) with no additional algorithms. Its auxiliary output is empty; throughout this section, $\qsk\lrun\KG(\msk)$ denotes the quantum component of key generation, omitting the empty auxiliary component. Thus, we move onto the security notions for SDE. For conciseness, we will only explicitly write the definitions for the case of quantum key generation - the classical vendor case is defined the natural way\footnote{That is, the adversary $\adve$ interacts with $\mathsf{Send}$ instead of receiving the output of $\KG$.}. In particular, we allow the adversary to run the multiple generation protocols in any order (e.g., parallel, sequential, interleaved etc.).

For defining security of SDE, we introduce the notions of quantum decryptors following \cite{C:CLLZ21}. 

\if0
\begin{definition}[Quantum Decryptor~\cite{C:CLLZ21}]\label{def:qunatum_decryptor}
Let $\rho$ and $U$ be a quantum state and a general quantum circuit acting on $n+m$ qubits, where $m$ is the number of qubits of $\rho$ and $n$ is the length of ciphertexts. A quantum decryptor for ciphertexts is a pair $(\rho,U)$.

When we say that we run the quantum decryptor $(\rho,U)$ on ciphertext $\ct$, we execute the circuit $U$ on inputs $\ket{\ct} \tensor \delta$. 
\end{definition}
\fi

Following \cite{C:CLLZ21}, we define a \emph{quantum decryptor} as a quantum program (\Cref{def:quantum_program}) that takes as input a ciphertext of an SDE scheme and outputs a classical string.  

\begin{definition}[Testing a (CPA-Style) Quantum Decryptor~\cite{C:CLLZ21}]\label{def:test_quantum_decryptor}
Let $\gamma \in [0,1/2]$. Let $\pk$ and $(\msg_0,\msg_1)$ be a public key and a pair of messages of an SDE scheme with ciphertext space $\cC$, respectively.
A test for a $\gamma$-good quantum decryptor with respect to $\pk$ and $(\msg_0,\msg_1)$ is the following procedure. 
\begin{itemize}
\item The procedure takes as input a quantum decryptor $\qD=(\rho,U)$ with output space $\bit$.
\item Define a distribution $D$ over $\cC\times \bit$ and a family $\cP=\{\cP_{\ct,b}\}_{\ct\in\cC,b\in\bit}$ of projective measurements as follows:
\begin{itemize}
\item 
$D$ is a distribution of $(\ct,\coin)$ where $\coin\gets \bit$ and $\ct\gets \Enc(\pk,\msg_\coin)$.    
\item 
$\cP_{\ct,\coin}$ is a projective measurement corresponding to the following procedure: On a quantum state $\rho'$,  
run $\coin^\prime \gets \qD'(\ct)$, where $\qD'=(\rho',U)$.  
If $\coin^\prime = \coin$, output $1$, otherwise output $0$.
\end{itemize}
\item Apply $\TI_{1/2+\gamma}(\cP_{D})$ to $\rho$. If the outcome is $1$, the test outputs $1$ (accepts) and otherwise outputs $0$ (rejects).
\end{itemize}
\if0
\begin{itemize}
\item The procedure takes as input a quantum decryptor $\qD =(\rho,U)$.
\item Let $\cP = (\mat{P}, \mat{I} - \mat{P})$ be the following mixture of projective measurements acting on a quantum state $\delta^\prime$: 
\begin{itemize}
\item Sample a uniformly random $\coin\chosen \zo$ and compute $\ct \lrun \Enc(\pk,\msg_\coin)$.
\item Run $\coin^\prime \lrun \qD^{\prime}(\ct)$, where $\qD^\prime = (\rho^\prime,U)$. If $\coin^\prime = \coin$, output $1$, otherwise output $0$.
\end{itemize}
Let $\TI_{1/2+\gamma}(\cP)$ be the threshold implementation of $\cP$ with threshold $\frac{1}{2}+\gamma$.
Apply $\TI_{1/2+\gamma}(\cP)$ to $\delta$. If the outcome is $1$, we say that the test passed, otherwise the test failed.
\end{itemize}
\fi
\end{definition}
\begin{remark}
Using \Cref{def:distinguishing_test}, the above procedure can be phrased as 
a $\gamma$-distinguishing test  for distributions $\{\ct_0\gets \Enc(\pk,\msg_0)\}$ and  $\{\ct_1\gets \Enc(\pk,\msg_1)\}$.
\end{remark}

\begin{definition}[Testing a Search-Type Quantum Decryptor~\cite{C:CLLZ21}]\label{def:test_search_quantum_decryptor}
Let $\gamma \in [0,1/2]$. Let $\pk$ be a public key of an SDE scheme with message space $\cM$ and ciphertext space $\cC$.
A test for a $\gamma$-good search-type quantum decryptor with respect to $\pk$ is the following procedure. 
\begin{itemize}
\item The procedure takes as input a quantum decryptor $\qD=(\rho,U)$ with output space $\cM$.
\item Define a distribution $D$ over $\cC\times \cM$ and a family $\cP=\{\cP_{\ct,\msg}\}_{\ct\in\cC,\msg\in\cM}$ of projective measurements as follows:
\begin{itemize}
\item 
$D$ is a distribution of $(\ct,\msg)$ where $\msg\gets \cM$ and $\ct\gets \Enc(\pk,\msg)$.    
\item 
$\cP_{\ct,\msg}$ is a projective measurement corresponding to the following procedure: On a quantum state $\rho'$,  
run $\msg^\prime \gets \qD'(\ct)$, where $\qD'=(\rho',U)$.  
If $\msg^\prime = \msg$, output $1$, otherwise output $0$.
\end{itemize}
\item Apply $\TI_{1/|\cM|+\gamma}(\cP_{D})$ to $\rho$. If the outcome is $1$, the test outputs $1$ (accepts) and otherwise outputs $0$ (rejects).
\end{itemize}
\if0
\begin{itemize}
\item The procedure takes as input a quantum decryptor $\qD =(\rho,U)$.
\item Let $\cP = (\mat{P}, \mat{I} - \mat{P})$ be the following mixture of projective measurements acting on a quantum state $\delta^\prime$:
\begin{itemize}
\item Sample a uniformly random $\msg\chosen \cM$ and compute $\ct \lrun \Enc(\pk,\msg)$.
\item Run $\msg^\prime \lrun \qD^{\prime}(\ct)$, where $\qD^\prime = (\rho^\prime,U)$. If $\msg^\prime = \msg$, output $1$, otherwise output $0$.
\end{itemize}
Let $\TI_{1/\abs{\cM}+\gamma}(\cP)$ be the threshold implementation of $\cP$ with threshold $\frac{1}{\abs{\cM}}+\gamma$.
Apply $\TI_{1/\abs{\cM}+\gamma}(\cP)$ to $\delta$. If the outcome is $1$, we say that the test passed, otherwise the test failed.
\end{itemize}
We say that a quantum decryptor $\qD=(\rho,U)$ is a $\gamma$-good search-type quantum decryptor with respect to $\pk$ if it passes the above test with probability $1$.
\fi
\end{definition}

Now we are ready to state the security notions of SDE. Specifically, we define two notions: collusion-resistant strong anti-piracy and collusion-resistant strong search anti-piracy, following~\cite{C:CLLZ21,TCC:CakGoy24}. Intuitively, the former corresponds to indistinguishability-based security, while the latter corresponds to search-based security. Unlike in standard public-key encryption, no relationship between these two notions is currently known.\footnote{A recent work \cite{EPRINT:KitYam25} showed that strong anti-piracy implies strong search anti-piracy in the single-key setting, but their proof does not extend to the collusion-resistant setting.}

\begin{definition}[Collusion-Resistant Strong Anti-Piracy]\label{def:cr_strong_anti_piracy}
Let $\gamma\in[0,1]$. Let $\SDE=(\Setup,\KG,\Enc,\Dec)$ be a public-key SDE scheme. We consider the collusion-resistant strong anti-piracy game $\expc{\SDE,\cA}{strong}{anti}{piracy}(\secp,\gamma)$ between the challenger and an adversary $\cA$ below.
\begin{enumerate}
\item The challenger generates $(\pk,\msk)\lrun \Setup(1^\secp)$ and sends $(1^\secp,\pk)$ to $\cA$.
\item $\cA$ sends $1^q$ to the challenger, then the challenger generates $\qsk_i \lrun \KG(\msk)$ for $i\in[q]$, and sends $(\qsk_1,\ldots,\qsk_q)$ to $\cA$.
\item $\cA$ outputs $\setbk{(\msg_{i,0},\msg_{i,1})}_{i\in[q+1]} \in \cM^{2(q+1)}$ and $(q+1)$ (possibly entangled) quantum decryptors $\qD_1=(\rho[\qreg{R}_1],U_1) ,\ldots, \qD_{q+1}=(\rho[\qreg{R}_{q+1}],U_{q+1})$, where $\abs{\msg_{i,0}}=\abs{\msg_{i,1}}$ for all $i\in[q+1]$, $\rho$ is a quantum state over registers $\qreg{R}_1,\ldots, \qreg{R}_{q+1}$, and each $U_i$ is a unitary quantum circuit.
\item The challenger runs the tests for a $\gamma$-good quantum decryptor with respect to $\pk$ and $(\msg_{i,0},\msg_{i,1})$ on $\qD_i$ for all $i\in[q+1]$. The challenger outputs $1$ if all tests pass; otherwise outputs $0$.
\end{enumerate}

We say that $\SDE$ is collusion-resistant strong $\gamma$-anti-piracy secure if for any QPT adversary $\cA$, it satisfies that
\begin{align}
\advc{\SDE,\cA}{strong}{anti}{piracy}(\secp,\gamma)\seteq\Pr[ \expc{\SDE,\cA}{strong}{anti}{piracy}(\secp, \gamma)=1]= \negl(\secp).
\end{align}
We say that $\SDE$ is collusion-resistant strong anti-piracy secure if it is collusion-resistant strong $\gamma$-anti-piracy secure for any inverse polynomial $\gamma$.

We say that $\SDE$ is $q^*$-bounded collusion-resistant strong anti-piracy secure if the above holds when $q$ is limited to be at most $q^*$. When $q^*=1$, we also call it single-key strong anti-piracy security. 
\end{definition}

\begin{definition}[Collusion-Resistant Strong-Search Anti-Piracy]\label{def:cr_strong_search_anti_piracy}
Let $\gamma\in[0,1]$. Let $\SDE=(\Setup,\KG,\Enc,\Dec)$ be a public-key SDE scheme. We consider the collusion-resistant strong-search anti-piracy game $\expd{\SDE,\cA}{strong}{search}{anti}{piracy}(\secp,\gamma)$ between the challenger and an adversary $\cA$ below.
\begin{enumerate}
\item The challenger generates $(\pk,\msk)\lrun \Setup(1^\secp)$ and sends $(1^\secp,\pk)$ to $\cA$.
\item $\cA$ sends $1^q$ to the challenger, then the challenger generates $\sk_i \lrun \KG(\msk)$ for $i\in[q]$, and sends $(\qsk_1,\ldots,\qsk_q)$ to $\cA$.
\item $\cA$ outputs $(q+1)$ (possibly entangled) quantum decryptors $\qD_1=(\rho[\qreg{R}_1],U_1) ,\ldots, \qD_{q+1}=(\rho[\qreg{R}_{q+1}],U_{q+1})$, where $\rho$ is a quantum state over registers $\qreg{R}_1,\ldots, \qreg{R}_{q+1}$, and each $U_i$ is a unitary quantum circuit.
\item The challenger runs the tests for a $\gamma$-good search-type quantum decryptor with respect to $\pk$ on $\qD_i$ for all $i\in[q+1]$. The challenger outputs $1$ if all tests pass; otherwise outputs $0$.
\end{enumerate}

We say that $\SDE$ is collusion-resistant strong-search $\gamma$-anti-piracy secure if for any QPT adversary $\cA$, it satisfies that
\begin{align}
\advd{\SDE,\cA}{strong}{search}{anti}{piracy}(\secp,\gamma)\seteq\Pr[ \expd{\SDE,\cA}{strong}{search}{anti}{piracy}(\secp, \gamma)=1]= \negl(\secp).
\end{align}
We say that $\SDE$ is collusion-resistant strong-search anti-piracy secure if it is collusion-resistant strong-search $\gamma$-anti-piracy secure for any inverse polynomial $\gamma$.

We say that $\SDE$ is $q^*$-bounded collusion-resistant strong-search anti-piracy secure if the above holds  when $q$ is limited to be at most $q^*$. 
\end{definition}

We recall the following theorems about the existence of single-key SDE. 
\begin{theorem}[\cite{EPRINT:KitYam25}]\label{thm:SDE_from_IO}
If there exist a secure iO for polynomial size classical circuits and OWFs, then there exists an SDE scheme that satisfies single-key strong anti-piracy security with message space $\zo^{\ell_\msg(\secp)}$ for any polynomial $\ell_\msg$. Moreover, its key generation algorithm has classically determined outputs (\Cref{def:classically_determined}). 
\end{theorem}

\begin{theorem}[\cite{chevalier2023semi}]\label{thm:single-classical-vendor-sde}
If there exist a subexponentially secure iO for polynomial size classical circuits and LWE is subexponentially hard for quantum adversaries, then there exists a classical-vendor SDE scheme that satisfies single-key strong anti-piracy security with message space $\zo^{\ell_\msg(\secp)}$ for any polynomial $\ell_\msg$. 
\end{theorem}

\paragraph{\bf Secret key SDE.}
\begin{definition}[Secret-Key SDE]\label{def:sksde}
A secret-key SDE (SK-SDE) scheme is a tuple of QPT algorithms
$(\Setup,\KG,\Enc,\qDec)$
with the same syntax as a (public-key) SDE scheme, except that the public key $\pk$ is replaced by an encryption key $\ek$ (since it is no longer intended to be public.) 
\end{definition}

In the security notions for SK-SDE considered in this paper, we do not provide $\ek$ to the adversary in the first stage, but reveal it in the second stage. This choice is made for a technical reason: it is needed for our bounded collusion-resistant compiler (\Cref{sec:compiler_SDE_bounded}) to work in the secret-key setting (see \Cref{rem:bounded_SK-SDE}).

More specifically, we define collusion-resistant strong anti-piracy and collusion-resistant strong-search anti-piracy security analogously to the corresponding notions for (public-key) SDE (\Cref{def:cr_strong_anti_piracy,def:cr_strong_search_anti_piracy}), with the following modification: $\ek$ is not given to the adversary at the beginning of the game, but only after the adversary splits its state into $q+1$ registers. Concretely, after $\qA$ outputs $(q+1)$ possibly entangled quantum states
$\sigma[\qreg{R}_1],\ldots,\sigma[\qreg{R}_{q+1}]$, 
for each $i \in [q+1]$, we allow any local operation on $\qreg{R}_i$ that depends on $\ek$ in order to generate a quantum decryptor
$\qD_i = (\rho[\qreg{R}_i], U_i).$

\begin{remark}
In the security notions for SK-SDE, we do not provide the adversary with an encryption oracle. Thus, these notions correspond to one-time security. At present, there is no known construction of an SK-SDE scheme with reusable security (that is, security in the presence of an encryption oracle) that is not already a public-key SDE scheme.
\end{remark}

\begin{theorem}[\cite{TCC:KitNis23}]\label{thm:SKSDE_from_LWE}
Assuming the LWE assumption, there exists an SK-SDE scheme that satisfies single-key strong anti-piracy security. Moreover, its key-generation algorithm has classically determined outputs (\Cref{def:classically_determined}).
\end{theorem}

\begin{remark}
In the security definition of \cite{TCC:KitNis23}, the encryption key $\ek$ is not explicitly given to the adversary in the second stage. However, it is easy to see that their proof still goes through even if $\ek$ is revealed at that point. Roughly speaking, this is because their scheme uses an encryption key of the form
$\ek = (\theta,\mathsf{mpk})$, 
where $\theta$ is the basis choice of a BB84 state and $\mathsf{mpk}$ is the public key of a public-key functional encryption scheme. The latter component, $\mathsf{mpk}$, can safely be revealed to the adversary even in the first stage, since the security of FE holds in the public-key setting. The former component, $\theta$, can be revealed in the second stage, since it is already given to the adversary in the BB84 monogamy game to which their security proof reduces.
\end{remark}

\subsection{Collusion-Resistant Schemes}\label{sec:compiler_SDE}
In this section, we show that the compiler from \cref{sec:fe-compiler} yields a collusion-resistant SDE scheme when the single-key scheme $\One$ is a secure SDE scheme, and then instantiate it with single-key schemes to obtain collusion-resistant schemes in the plain model.

\begin{theorem}\label{thm:collusion-resistant_conversion}
Assume that $\One$ is single-key  strong anti-piracy secure, $\FE$ is adaptively secure, and $\PRF$ is a puncturable PRF. Let $\SDE$ be the compiled scheme from \cref{sec:compitself}.
Then, $\SDE$ satisfies \begin{itemize}
    \item 
    collusion-resistant 
    strong anti-piracy security (\cref{def:cr_strong_anti_piracy}),
    \item 
    collusion-resistant 
    strong search anti-piracy security (\cref{def:cr_strong_search_anti_piracy}).
\end{itemize}
\end{theorem}
We prove this theorem below.

By applying the above compiler to the single-key SDE schemes of \Cref{thm:SDE_from_IO,thm:single-classical-vendor-sde}, we obtain the following corollary.
\begin{corollary}\label{cor:cr_SDE_from_IO}
If there exist a secure iO for polynomial size classical circuits and OWFs, then there exists an SDE scheme that satisfies collusion-resistant strong anti-piracy security and collusion-resistant strong search anti-piracy security. Moreover, its key generation algorithm has classically determined outputs (\Cref{def:classically_determined}). 
\end{corollary}

\begin{corollary}\label{cor:classical-vendor-sde}
If there exist a subexponentially secure iO for polynomial size classical circuits and LWE is subexponentially hard for quantum adversaries, then there exists a classical-vendor SDE scheme that satisfies collusion-resistant strong anti-piracy security and collusion-resistant strong search anti-piracy security.
\end{corollary}

\subsubsection{Proof of \cref{thm:collusion-resistant_conversion}}
In this section, we prove \cref{thm:collusion-resistant_conversion}.

We only give a full proof of collusion-resistant strong anti-piracy security, since collusion-resistant strong search anti-piracy security follows similarly; see \cref{rem:res_cr_strong_search_anti_piracy}. We also do not give a separate proof for classical-vendor schemes since it follows the same way.

Assume, for the sake of contradiction, that there exists an inverse polynomial $\gamma$ and a QPT adversary $\qA$ that breaks collusion-resistant strong $\gamma$-anti-piracy security. 
Let 
$\pk=\fe.\pk$ be the public key given to the security experiment, 
$\qD_1=(\rho_1,U_1),...,\qD_{q+1}=(\rho_{q+1},U_{q+1})$ be (possibly entangled) $q+1$ quantum decryptors generated by $\qA$, and 
$(\msg_{i,0},\msg_{i,1})_{i\in[q+1]}$ be the challenge messages chosen by $\qA$.  
Define a distribution $D_i$ over $\cC\times \bit$ where $\cC$ is the ciphertext space of $\SDE$  and a family $\mathcal{P}_i=\{\mathcal{P}_{i,\ct,\coin}\}_{\ct\in\cC,\coin\in\bit}$ as in the test for a $\gamma$-good decryptor (\Cref{def:test_quantum_decryptor}) on the $i$-th quantum decryptor $\qD_{i}$. 
That is, they are defined as follows:
\begin{itemize}
\item 
$D_{i}$ is a distribution of $(\ct,\coin)$
where $\coin\gets\bit$ and $\ct\gets\Enc(\pk,\msg_{i,\coin})$.
More precisely, they are generated as follows:
\begin{enumerate}
    \item Sample $\coin\gets \bit$ and $\prfkey \chosen \prfgen(1^\secp)$.
    \item Generate $\ct \lrun \FE.\Enc(\sffe.\pk,\msg_{i,\coin}\concat \prfkey\concat 0\concat 0^\secp\concat 0^{\ell_\ct})$.
\end{enumerate}
\item 
$\cP_{i,\ct,\coin}$ is a projective measurement corresponding to the following procedure: On a quantum state $\rho'$,  
run $\coin^\prime \gets \qD'_i(\ct)$, where $\qD'_i=(\rho',U_i)$.  
If $\coin^\prime = \coin$, output $1$, otherwise output $0$.
\end{itemize}

Let $\cP_{i,D_{i}}$ be a mixture of projective measurements corresponding to $\cP_i$ and $D_i$. 
Then, our assumption means that if we apply $\TI_{1/2+\gamma}(\cP_{i,D_{i}})$ on $\qD_i$ for all $i\in [q+1]$, then the probability that all of them return $1$ is non-negligible. 
Let $\epsilon'\seteq \gamma/8$ and $\delta'\seteq 2^{-\secp}$. 
By the first item of \Cref{lem:API}, if we apply $\ATI^{\epsilon',\delta'}_{\cP_i,D_i,1/2+\gamma-\epsilon'}$ on $\rho_i$ for $i\in[q+1]$, all of them return $1$ with a non-negligible probability.  
Suppose that it occurs, and let $\tilde{\rho}_1,...,\tilde{\rho}_{q+1}$ be the residual states. 

For $k\in[q]$, 
let $t_k\in\bit^\secp$ be the tag chosen in generation of the $k$-th secret key. 
For convenience, we define $t_0\seteq 0^\secp$.  
With overwhelming probability, $t_0,t_1,\ldots,t_q$ are pairwise distinct, thus we assume so below. 
Without loss of generality, we assume $t_0<t_1<\cdots<t_{q}$, which can be ensured by an appropriate relabeling if necessary.  

For $i\in[q+1]$ and $j\in\{0,1,\ldots,q\}$, we define the following distributions $D_{i,j}$ over $\cC\times \bit$. %

\begin{description}
	\item[$D_{i,j}$:] It is a distribution of $(\ct,\coin)$ which are generated as follows:
    \begin{enumerate}
    \item Sample $\coin\gets \bit$ and $\prfkey \chosen \prfgen(1^\secp)$.
    \item Generate $\ct \lrun \FE.\Enc(\sffe.\pk,\msg_{i,\coin}\concat \prfkey\concat 1\concat t_j\concat 0^{\ell_\ct})$.
\end{enumerate}  
\end{description}
Below, we use the notation $D\cind D'$ to mean that  distributions $D$ and $D'$ are computationally indistinguishable against the adversary in the security experiment of collusion-resistant strong anti-piracy security.

To apply \Cref{thm:collusion_resistant}, we show the following lemmas. 
\begin{lemma}\label{lem:unbounded_cr_assumption_one}
    For $i\in[q]$, \[
\Pr[\projimp(\cP_{i,D_{i,0}}) \tilde{\rho}_i\ge 1/2+5\gamma/8]\ge 1-\negl(\secp).
\]
\end{lemma}
\begin{proof}
By the adaptive security of $\FE$, we have $D_{i,0}\cind D_i$ since $$\RE[t_k,\one.\pk_k](\msg_{i,\coin}\concat \prfkey\concat 1\concat 0^\secp\concat 0^{\ell_\ct})=\RE[t_k,\one.\pk_k](\msg_{i,\coin}\concat \prfkey\concat 0\concat 0^\secp\concat 0^{\ell_\ct})$$ for all $k\in[q]$.

Thus, by \Cref{lem:ATI_gap_to_dist}, 
\[\left|\Pr[\ATI^{\epsilon',\delta'}_{\cP_i,D_{i,0},1/2+\gamma-2\epsilon'} \tilde{\rho}_i=1]-\Pr[\ATI^{\epsilon',\delta'}_{\cP_i,D_i,1/2+\gamma-2\epsilon'} \tilde{\rho}_i=1]\right|\le \negl(\secp).\]

Note that $\tilde{\rho}_i$ is a residual state after $\ATI^{\epsilon',\delta'}_{\cP_i,D_i,1/2+\gamma-\epsilon'}$ outputs $1$. Thus, by the second item of \Cref{lem:API}, we have 

\[\Pr[\ATI^{\epsilon',\delta'}_{\cP_i,D_i,1/2+\gamma-2\epsilon'} \tilde{\rho}_i=1]\ge 1-\negl(\secp).\] 

Combining above, we have 
\[\Pr[\ATI^{\epsilon',\delta'}_{\cP_i,D_{i,0},1/2+\gamma-2\epsilon'} \tilde{\rho}_i=1]\ge 1-\negl(\secp).\] 
Then \Cref{lem:unbounded_cr_assumption_one} follows from the above and  the first item of \Cref{lem:API} noting that $1/2+\gamma-2\epsilon'-\epsilon'=1/2+5\gamma/8$. 
\end{proof}

\begin{lemma}\label{lem:unbounded_cr_assumption_two}
For $i\in[q]$ and any  state $\sigma$ that is efficiently generatable from $\tilde{\rho}_i$,   \[
\Pr[\projimp(\cP_{i,D_{i,q}}) \sigma\le 1/2+\gamma/4]\ge  1-\negl(\secp).
\]
\end{lemma}
\begin{proof} 
Define the following distribution $D_{i,q+1}$. 
\begin{description}
	\item[$D_{i,q+1}$:] It is a distribution of $(\ct,\coin)$ which are generated as follows:
    \begin{enumerate}
    \item Sample 
    $\coin\gets \bit$ and 
    $\prfkey \chosen \prfgen(1^\secp)$.
    \item Generate $\ct \lrun \FE.\Enc(\sffe.\pk,0^{\ell_\msg}\concat \prfkey\concat 1\concat 1^{\lambda} \concat 0^{\ell_\ct})$.
\end{enumerate}  
\end{description}

By the adaptive security of $\FE$, we have $D_{i,q}\cind D_{i,q+1}$ since $$\RE[t_k,\one.\pk_k](\msg_{i,\coin}\concat \prfkey\concat 1\concat t_q\concat 0^{\ell_\ct})=\RE[t_k,\one.\pk_k](0^{\ell_\msg}\concat \prfkey\concat 1\concat 1^\secp\concat 0^{\ell_\ct})$$ for all $k\in[q]$.

Thus, by \Cref{lem:ATI_gap_to_dist}, 
\[\left|\Pr[\ATI^{\epsilon',\delta'}_{\cP_i,D_{i,q},1/2+\epsilon'} \tilde{\rho}_i=1]-\Pr[\ATI^{\epsilon',\delta'}_{\cP_i,D_{i,q+1},1/2+\epsilon'} \tilde{\rho}_i=1]\right|\le \negl(\secp).\] 

Since $\ct$ generated by $D_{i,q+1}$ is independent of $\coin$, we have 
\[
\Pr[\projimp(\cP_{i,D_{i,q+1}})=1/2]=1.
\]
Thus, by the first item of \Cref{lem:API}, 
\[
\Pr[\ATI^{\epsilon',\delta'}_{\cP_i,D_{i,q+1},1/2+\epsilon'} \tilde{\rho}_i=1]\le \negl(\secp).
\]

Combining above, we have 
\[\Pr[\ATI^{\epsilon',\delta'}_{\cP_i,D_{i,q},1/2+\epsilon'} \tilde{\rho}_i=1]\le\negl(\secp).\] 
Then \Cref{lem:unbounded_cr_assumption_two} follows from the above and  the first item of \Cref{lem:API} noting that $1/2+\epsilon'+\epsilon'=1/2+\gamma/4$. 
\end{proof}

Since $(1/2+5\gamma/8)-(1/2+\gamma/4)=3\gamma/8=1/\poly(\secp)$, 
\Cref{lem:unbounded_cr_assumption_one,lem:unbounded_cr_assumption_two} imply that the assumptions of \Cref{thm:collusion_resistant} are satisfied. Thus, the following lemma holds.
\begin{lemma}\label{lem:apply_the_theorem_unbounded}
There exist $\epsilon=1/\poly(\secp)$, $\delta=\negl(\secp)$, $\tau=1/\poly(\secp)$, and QPT algorithms $\Convert_1^{\epsilon,\delta}$ and $\Convert_2^{\epsilon,\delta}$ that satisfy the conclusion of \Cref{thm:collusion_resistant}. 
That is, consider the following experiment $\mathsf{Exp}_\tau^{\epsilon,\delta}$ where %
$\qreg{R}_i$ is the register for $\tilde{\rho}_i$ for $i\in[q+1]$. 
\paragraph{\underline{$\mathsf{Exp}_\tau^{\epsilon,\delta}$}}
\begin{enumerate}
\item Choose uniformly random distinct  $i_1,i_2\in [q+1]$ and $j\in [q]$.
\item Apply $\Convert_1^{\epsilon,\delta}(i_1,j,\cdot)$ and  $\Convert_2^{\epsilon,\delta}(i_2,j,\cdot)$ on $\qreg{R}_{i_1}$ and $\qreg{R}_{i_2}$,  which output quantum distinguishers $\qD'_1$ and $\qD'_2$, respectively.   
\item Apply the $\tau$-distinguishing test between $D_{i_1,j-1}$ and $D_{i_1,j}$ on $\qD'_1$ and the $\tau$-distinguishing test between $D_{i_2,j-1}$ and $D_{i_2,j}$ on $\qD'_2$ . If both tests accept, the experiment outputs $1$, otherwise outputs $0$. 
\end{enumerate}
Then $\Pr[\mathsf{Exp}_\tau^{\epsilon,\delta}=1]\ge 1/\poly(\secp)$. 
\end{lemma}

For each 
$i\in[q+1]$ and 
$j\in [q]$, we consider the following distributions where $\cR_{\One}$ is the randomness space of $\One$ and $\one.\pk_j$ is a public key of $\One$ generated in $j$-th key generation.  

\begin{description}
	\item[$D_{i,j-1}^{1}$:] It is a distribution of $(\ct,\coin)$ which are generated as follows:
	\begin{enumerate}
	\item Sample $\coin\chosen\zo$, $\prfkey \lrun \prfgen(1^\secp)$, and  $\one.\ct^\ast \lrun \One.\Enc(\one.\pk_j,\msg_{i,\coin};\prfF_{\prfkey}(t_j))$.
	\item Generate $\ct \lrun \FE.\Enc(\sffe.\pk,\msg_{i,\coin} \concat \prfkey\concat 2\concat t_{j} \concat \one.\ct^\ast)$.
	\end{enumerate}
	\item[$D_{i,j-1}^2$:] It is a distribution of $(\ct,\coin)$ which are generated as follows:
	\begin{enumerate}
	\item Sample $\coin\chosen\zo$, $\prfkey \lrun \prfgen(1^\secp)$, and  $\one.\ct^\ast \lrun \One.\Enc(\one.\pk_j,\msg_{i,\coin};\prfF_{\prfkey}(t_j))$.
	\item Generate $\ct \lrun \FE.\Enc(\sffe.\pk,\msg_{i,\coin} \concat \prfkey_{\ne t_j}\concat 2 \concat t_{j}\concat \one.\ct^\ast)$,\\ where  $\prfkey_{\ne t_j} \lrun \Puncture(\prfkey,\setbk{t_j})$.
	\end{enumerate} 
	\item[$D_{i,j-1}^3$:] It is a distribution of $(\ct,\coin)$ which are generated as follows:
	\begin{enumerate}
	\item Sample $\coin\chosen\zo$, $r \lrun \cR_{\One}$, and  $\one.\ct^\ast \lrun \One.\Enc(\one.\pk_j,\msg_{i,\coin};r)$.
	\item  Generate $\ct \lrun \FE.\Enc(\sffe.\pk,\msg_{i,\coin} \concat \prfkey_{\ne t_j}\concat 2 \concat t_{j}\concat \one.\ct^\ast)$,\\ where  $\prfkey_{\ne t_j} \lrun \Puncture(\prfkey,\setbk{t_j})$.
	\end{enumerate} 
 	\item[$D_{i,j}^{-1}$:] It is a distribution of $(\ct,\coin)$ which are generated as follows:
 	\begin{enumerate}
 	 	\item Sample $\prfkey \lrun \prfgen(1^\secp)$ and $\one.\ct^\ast \lrun \One.\Enc(\one.\pk_j,0^{\ell_\msg};\prfF_{\prfkey}(t_j))$.
 	 	\item  Generate $\ct \lrun \FE.\Enc(\sffe.\pk,\msg_{i,\coin} \concat \prfkey\concat 2\concat t_{j} \concat \one.\ct^\ast)$.
 	 \end{enumerate}
	\item[$D_{i,j}^{-2}$:] It is a distribution of $(\ct,\coin)$ which are generated as follows:
	\begin{enumerate}
		\item Sample $\prfkey \lrun \prfgen(1^\secp)$ and $\one.\ct^\ast \lrun \One.\Enc(\one.\pk_j,0^{\ell_\msg};\prfF_{\prfkey}(t_j))$.
		\item Generate $\ct \lrun \FE.\Enc(\sffe.\pk,\msg_{i,\coin} \concat \prfkey_{\ne t_j}\concat 2 \concat t_j\concat \one.\ct^\ast)$,\\ where $\prfkey_{\ne t_j} \lrun \Puncture(\prfkey,\setbk{t_j})$.
	\end{enumerate}
	\item[$D_{i,j}^{-3}$:] It is a distribution of $(\ct,\coin)$ which are generated as follows:
	\begin{enumerate}
	 	\item Sample $r\chosen \cR_{\One}$ and $\one.\ct^\ast \lrun \One.\Enc(\one.\pk_j,0^{\ell_\msg};r)$.
	 	\item  Generate $\ct \lrun \FE.\Enc(\sffe.\pk,\msg_{i,\coin} \concat \prfkey_{\ne t_j}\concat 2 \concat t_j\concat \one.\ct^\ast)$,\\ where $\prfkey_{\ne t_j} \lrun \Puncture(\prfkey,\setbk{t_j})$.
	 \end{enumerate} 
\end{description}

\begin{lemma}\label{lem:sub_hyb_ind_unbounded}
For $i\in[q+1]$ and $j\in[q]$, 
it hold that $D_{i,j-1} \cind D_{i,j-1}^3$ and $D_{i,j} \cind D_{i,j}^{-3}$
\end{lemma}
\begin{proof}
By adaptive security of $\FE$ we have $D_{i,j-1} \cind D_{i,j-1}^1$.
By adaptive security of $\FE$ and punctured correctness of $\PRF$, we have $D_{i,j-1}^1 \cind D_{i,j-1}^2$. By puncturable security of $\PRF$, we have $D_{i,j-1}^2 \cind D_{i,j-1}^3$.

By adaptive security of $\FE$ we have $D_{i,j} \cind D_{i,j}^{-1}$.
By adaptive security of $\FE$ and punctured correctness of $\PRF$, we have $D_{i,j}^{-1} \cind D_{i,j}^{-2}$. By puncturable security $\PRF$, we have $D_{i,j}^{-2} \cind D_{i,j}^{-3}$.
\end{proof}

\Cref{lem:sub_hyb_ind_unbounded} implies that a distinguisher that distinguishes $D_{i,j-1}$ and $D_{i,j}$ also distinguishes $D_{i,j-1}^3$ and $D_{i,j}^{-3}$. 
Moreover, the only difference between $D_{i,j-1}^3$ and $D_{i,j}^{-3}$ is that  $\one.\ct^*$ encrypts $\msg_{i,\coin}$ in the former but encrypts $0^{\ell_\msg}$ in the latter. 
Thus, a distinguisher that distinguishes $D_{i,j-1}$ and $D_{i,j}$, obtained in the experiment in \Cref{lem:apply_the_theorem_unbounded}, can be used to construct a distinguisher that distinguishes encryptions of $\msg_{i,\coin}$ and $0^{\ell_\msg}$ in $\One$. 

Based on this observation, we construct an adversary $\qB$ that breaks single-key strong anti-piracy security of $\One$ that works as follows: 
Let $\epsilon$, $\delta$, $\tau$, $\Convert_1^{\epsilon,\delta}$, and  $\Convert_2^{\epsilon,\delta}$ %
be as in \Cref{lem:apply_the_theorem_unbounded}.  
\paragraph{\underline{$\mathcal{B}$}}
\begin{enumerate}
\item After receiving $(1^\secp,\one.\pk,\one.\sk)$ from the challenger\footnote{In the case of classical-vendors, here $\mathcal{B}$ will instead simulate the receiver part of the adversary $\adve$, while simulating both sides (receiver and sender) of the protocols for non-challenge instances by sampling them himself (as done in the case of quantum vendor). Rest of the adversary is the same as the quantum vendor case.},
do the following.
\begin{enumerate}[label=\arabic*.]
\item Generate $(\sffe.\pk,\sffe.\msk) \lrun \FE.\Setup(1^\secp)$ and send $(1^\secp,\pk\seteq \sffe.\pk)$ to $\cA$. 
\item Receive $1^q$ from  $\cA$.  
\item Choose uniformly random distinct  $i_1,i_2\in [q+1]$ and $j\in [q]$. 
\item Choose uniformly random distinct $\secp$-bit tags $t_0<t_1<t_2<\cdots<t_q$ where $t_0=0^\secp$. 
\item  Generate $(\one.\pk_k,\one.\msk_k)\lrun \One.\Setup(1^\secp)$ 
and $\one.\sk_k\gets \One.\KG(\one.\msk_k)$ for $k\in [q]\setminus\{j\}$ and let $\one.\pk_j\seteq \one.\pk$
and $\one.\sk_j\seteq \one.\sk$. 
\item Generate $\sffe.\fsk_k \lrun \FE.\KG(\sffe.\msk,\RE[t_k,\one.\pk_k])$ for $k\in[q]$. 
\item Set $\qsk_k \seteq (\one.\qsk_k,\sffe.\fsk_k)$ for $k\in[q]$ and send $(\qsk_1,...,\qsk_q)$ to $\cA$.
\end{enumerate}
\item When $\cA$ outputs $\setbk{(\msg_{i,0},\msg_{i,1})}_{i\in[q+1]}$ and $(q+1)$ (possibly entangled) quantum decryptors $\qD_1=(\rho[\qreg{R}_1],U_1) ,\ldots, \qD_{q+1}=(\rho[\qreg{R}_{q+1}],U_{q+1})$, do the following: 
\begin{enumerate}[label=\arabic*.]
\item Apply $\Convert_1^{\epsilon,\delta}(i_1,j,\cdot)$ and  $\Convert_2^{\epsilon,\delta}(i_2,j,\cdot)$ on $\qreg{R}_{i_1}$ and $\qreg{R}_{i_2}$,  which output quantum distinguishers $\qD'_1$ and $\qD'_2$, respectively.  
\item Apply the $(\epsilon'',\delta'')$-approximate  $(\tau-\epsilon'')$-distinguishing test between $D_{i_1,j-1}$ and $D_{i_1,j}$ on $\qD'_1$ and the 
$(\epsilon'',\delta'')$-approximate 
$(\tau-\epsilon'')$-distinguishing test between $D_{i_2,j-1}$ and $D_{i_2,j}$ on $\qD'_2$
 where $\epsilon''=\tau/4$ and $\delta''=2^{-\secp}$.\footnote{See \Cref{def:distinguishing_test} for the definition of approximate distinguishing tests.} If either test rejects, $\qB$ aborts. Otherwise, let  $\qD''_1$ 
 and $\qD''_2$ be the quantum distinguishers after the tests.  \label{item:approximate_distinguish_test_by_B-1}
\item Choose $\coin\gets \bit$  and define quantum decryptors $\One.\qD_1$ and  $\One.\qD_2$ as follows:
\begin{description}
    \item[$\One.\qD_1$:] 
    Upon receiving $\one.\ct^*$, 
    generate $\ct \lrun \FE.\Enc(\sffe.\pk,\msg_{i,\coin} \concat \prfkey_{\ne t_j}\concat 2 \concat t_{j}\concat \one.\ct^\ast)$ where  $\prfkey_{\ne t_j} \lrun \Puncture(\prfkey,\setbk{t_j})$, run $\qD''_1$ on $\ct$, and output whatever $\qD''_1$ outputs. 
    \item[$\One.\qD_2$:]   Upon receiving $\one.\ct^*$, 
   generate $\ct \lrun \FE.\Enc(\sffe.\pk,\msg_{i,\coin} \concat \prfkey_{\ne t_j}\concat 2 \concat t_{j}\concat \one.\ct^\ast)$ where  $\prfkey_{\ne t_j} \lrun \Puncture(\prfkey,\setbk{t_j})$,  run $\qD''_2$ on $\ct$, and output whatever $\qD''_2$ outputs. 
\end{description}
\item Output challenge messages $(\msg_{i_1,\coin},0^{\ell_\msg})$ and $(\msg_{i_2,\coin},0^{\ell_\msg})$ and quantum decryptors $\One.\qD_1$ and  $\One.\qD_2$. 
\end{enumerate}
\end{enumerate}

By \Cref{lem:apply_the_theorem_unbounded} and the first item of \Cref{lem:API}, $\qB$ does not abort with an inverse polynomial probability. 
Below, we assume that $\qB$ does not abort. 

Let $\gamma'\seteq \tau/6$. 
Suppose that we apply the test for a $\gamma'$-good quantum decryptor with respect to $\one.\pk$ and $(\msg_{i_1,\coin},0^\ell)$ on $\One.\qD_1$
and the test for a $\gamma'$-good quantum decryptor with respect to $\one.\pk$ and $(\msg_{i_2,\coin},0^\ell)$ on $\One.\qD_2$. 
Let $\event{Acc}_1$ and $\event{Acc}_2$ be the events that the first and second tests accept, respectively.
\begin{lemma}\label{lem:unbounded_cr_acc_one}
 It holds that $\Pr[\event{Acc}_1]\ge \gamma'$.   
\end{lemma}
\begin{proof}
 Since $\qD_1''$ is a quantum distinguisher after passing the  $(\epsilon'',\delta'')$-approximate  $(\tau-\epsilon'')$-distinguishing test between $D_{i_1,j-1}$ and $D_{i_1,j}$, by the first item of \Cref{lem:API}, if we apply the  $(\tau-2\epsilon'')$-distinguishing test between $D_{i_1,j-1}$ and $D_{i_1,j}$ on $\qD_1''$, it accepts with overwhelming probability. This means that $\qD_1''$ distinguishes $D_{i_1,j-1}$ and $D_{i_1,j}$ with advantage at least $\tau-2\epsilon''-\negl(\secp).$
 By \Cref{lem:sub_hyb_ind_unbounded}, $\qD_1''$ distinguishes $D_{i_1,j-1}^3$ and $D_{i_1,j}^{-3}$ with advantage at least $\tau-2\epsilon''-\negl(\secp)\ge \tau/3$. This directly translates into that $\One.\qD_1$ distinguishes ciphertexts on messages $\msg_{i_1,\coin}$ and $0^{\ell_\msg}$ with advantage at least $\tau/3$. Thus, by a standard averaging argument, $\One.\qD_1$ passes the test for a $\gamma'$-good quantum decryptor with respect to $\one.\pk$ and $(\msg_{i_1,\coin},0^{\ell_\msg})$ with probability at least $\gamma'$ noting that $\gamma'=\tau/6$.  
\end{proof}

\begin{lemma}\label{lem:unbounded_cr_acc_two}
 It holds that $\Pr[\event{Acc}_2\mid \event{Acc}_1]\ge \gamma'$.   
\end{lemma}
\begin{proof}
    Let $\qD_2'''$ be the quantum distinguisher on $\qreg{R}_{i_2}$ conditioned on the event $\event{Acc}_1$. Since $\Pr[\event{Acc}_1]=1/\poly(\secp)$, due to a similar reason to that in the proof of \Cref{lem:unbounded_cr_acc_one}, if we apply the  $(\tau-2\epsilon'')$-distinguishing test between $D_{i_2,j-1}$ and $D_{i_2,j}$ on $\qD_2'''$, it accepts with overwhelming probability. Then the rest of the proof can be done similarly to the proof of \Cref{lem:unbounded_cr_acc_one}. 
\end{proof}
By \Cref{lem:unbounded_cr_acc_one,lem:unbounded_cr_acc_two}, we have $\Pr[\event{Acc}_1\land \event{Acc}_2]\ge {\gamma'}^2=1/\poly(\secp)$. This means that $\qB$ breaks single-key $\gamma'$-strong anti-piracy security of $\One$, which contradicts the assumption. This completes the proof of collusion-resistant strong anti-piracy security of $\SDE$.

\begin{remark}\label{rem:res_cr_strong_search_anti_piracy}
We can prove collusion-resistant strong search anti-piracy security similarly to the proof of collusion-resistant strong anti-piracy security given above.
The main difference is that we modify the definition of $\cP_{i,D_{i}}$ so that it measures the search advantage rather than the distinguishing advantage.   
Then, instead of \cref{lem:unbounded_cr_assumption_one,lem:unbounded_cr_assumption_two}, we show
\begin{align*}
\Pr[\projimp(\cP_{i,D_{i,0}}) \tilde{\rho}_i\ge 1/\poly(\secp)]\ge 1-\negl(\secp)\\
	\Pr[\projimp(\cP_{i,D_{i,q}}) \sigma=\negl(\secp)]= 1-\negl(\secp). 
\end{align*}
Since we have a $1/\poly(\secp)$ gap between the threshold in the above inequalities, the rest of the proof works similarly. 
\end{remark}

\subsection{Generic Compiler from Single-Key to Bounded Collusion-Resistant SDE}\label{sec:compiler_SDE_bounded}
In this section, we show a generic compiler from single-key SDE scheme into bounded collusion-resistant SDE scheme. The advantage of the scheme in this section over the one given in \Cref{sec:compiler_SDE} is that this compiler works in both public key and secret key settings and does not rely on any additional assumptions. 

We now move onto our construction, which is inspired by a similar compiler for (classical) attribute-based encryption \cite{PKC:ISVWY17}. Since the construction and security proofs are almost identical in both public key and secret key settings, we will focus on the public key setting here. 
 
Let $\One = \One.(\Setup,\qKG,\Enc,\qDec)$ be a single-key SDE scheme with $\ell$-bit messages. For a polynomial $q$, 
we construct a $q$-bounded collusion-resistant SDE scheme $\SDE =(\Setup,\qKG,\Enc,\qDec)$ with $\ell$-bit messages as follows. 
\begin{description}
\item[$\Setup(1^\secp)$:] $ $
\begin{enumerate}
    \item Generate $(\one.\pk_{a,b},\one.\msk_{a,b})\lrun \One.\Setup(1^\secp)$ for $a\in[\secp]$ and $b\in[q^2]$. 
    \item Output $\pk\seteq (\one.\pk_{a,b})_{a\in[\secp],b\in[q^2]}$ and $\msk\seteq (\one.\msk_{a,b})_{a\in[\secp],b\in[q^2]}$.
\end{enumerate}
\item[$\qKG(\msk)$:] $ $
\begin{enumerate}
\item Parse $\msk= (\one.\msk_{a,b})_{a\in[\secp],b\in[q^2]}$. 
\item Choose $b_a\gets [q^2]$ for $a\in[\secp]$. 
	\item Generate $\one.\qsk_{a} \lrun \One.\qKG(\one.\msk_{a,b_a})$ for $a\in[\secp]$. 
    \item Output $\sk\seteq (b_a,\one.\qsk_{a})_{a\in[\secp]}$.
\end{enumerate}
\item[$\Enc(\pk,\msg)$:] $ $
\begin{enumerate}
\item Parse $\pk= (\one.\pk_{a,b})_{a\in[\secp],b\in[q^2]}$.
\item Generate $\secp$-out-of-$\secp$ XOR secret sharing $(\msg_1,....,\msg_{\secp})\in \bit^{\ell\times \secp}$ of $\msg$, i.e., they are uniformly random conditioned on 
$\bigoplus_{a\in[\secp]}\msg_a=\msg$. 
\item Generate $\one.\ct_{a,b}\gets \One.\Enc(\one.\pk_{a,b},\msg_{a})$ for $a\in[\secp]$ and $b\in[q^2]$. 
\item Output $\ct\seteq(\one.\ct_{a,b})_{a\in[\secp],b\in[q^2]}$.
\end{enumerate}
\item[$\qDec(\qsk,\ct)$:] $ $
\begin{enumerate}
\item Parse $\qsk= (b_a,\one.\qsk_{a})_{a\in[\secp]}$ and $\ct=(\one.\ct_{a,b})_{a\in[\secp],b\in[q^2]}$.
\item Compute $\msg'_a\gets \One.\qDec(\one.\qsk_a,\one.\ct_{a,b_a})$ for $a\in[\secp]$
\item Compute and output $\msg'=\bigoplus_{a\in[\secp]}\msg'_a$.
\end{enumerate}
\end{description}

In the case where $\one$ is a classical-vendor scheme, our scheme can also be made classical-vendor by the following protocol. Let $\ell_{rounds}(\lambda)$ denote the number of rounds of the protocol of $\one$ and let $\One.\mathsf{Send}^{\times \lambda}, \One.\mathsf{Receive}^{\times \lambda}$ denote $\lambda$-parallel repetitions of $\One.\mathsf{Send}, \One.\mathsf{Receive}$.
\begin{description}
\item[$\mathsf{Send}(msk)$:] $ $
\begin{enumerate}
\item Sample $b_a \samp [q^2]$ for $a \in [\lambda]$.
\item Parse $\msk= (\one.\msk_{a,b})_{a\in[\secp],b\in[q^2]}$. 
\item \textbf{First Round:} Emulate $\One.\mathsf{Send}^{\times \lambda}$ with initial inputs $\msk_{1,b_1}, \dots, \msk_{1,b_\lambda}$. Let $msg$ to be send, and instead send $(msg, (b_a)_{a \in [\lambda]})$.
\item \textbf{Other Rounds:} Emulate $\One.\mathsf{Send}^{\times \lambda}$.
\end{enumerate}

\item[$\mathsf{Receive}(1^\lambda)$:] $ $
\begin{enumerate}
\item \textbf{First Round:} Let $msg$ be the message received, parse it as $(msg', (b_a)_{a \in [\lambda]})$. Emulate $\one.\mathsf{Receive}$ with $msg'$.
\item \textbf{Other Rounds:} Emulate $\One.\mathsf{Receive}^{\times \lambda}$.
\item Output $((b_a)_{a \in [\lambda]}, \regi{Receiver})$.
\end{enumerate}
\end{description}

\begin{theorem}\label{thm:bounded-collusion-resistant_conversion}
Assume that $\One$ is single-key strong anti-piracy secure, then $\SDE$ satisfies \begin{itemize}
    \item 
    $q$-bounded collusion-resistant 
    strong anti-piracy security (\cref{def:cr_strong_anti_piracy}),
    \item 
    $q$-bounded collusion-resistant 
    strong search anti-piracy security (\cref{def:cr_strong_search_anti_piracy}).
\end{itemize}
\end{theorem}
We prove this theorem below.

Thus, we get the following result.
\begin{corollary}
Assuming the existence of a single-key SDE scheme, for any polynomial $q$ there exists a $q$-collusion resistant SDE scheme. Moreover, if the single-key scheme supports classical-vendors, then the bounded collusion-resistant scheme also supports it.
\end{corollary}

As already mentioned, the above compiler works in the secret key setting as well. We omit the proof since it is almost identical. Applying the compiler to the single-key SK-SDE scheme of \Cref{thm:SKSDE_from_LWE}, we obtain the following corollary. 
\begin{corollary}\label{cor:cr_SKSDE_from_LWE}
Assuming the LWE assumption, 
for any polynomial $q$, 
there exists an SK-SDE scheme that satisfies $q$-bounded collusion-resistant strong anti-piracy security and $q$-bounded collusion-resistant strong search anti-piracy security. Moreover, its key generation algorithm has classically determined outputs (\Cref{def:classically_determined}). 
\end{corollary}
\subsubsection{Proof of \cref{thm:bounded-collusion-resistant_conversion}}
In this section, we prove \cref{thm:bounded-collusion-resistant_conversion}.

We only give a full proof of collusion-resistant strong anti-piracy security, since collusion-resistant strong search anti-piracy security follows similarly (the remark in the unbounded setting (\Cref{rem:res_cr_strong_search_anti_piracy}) applies to the bounded setting as well). We also do not give a separate proof for classical-vendor schemes since it follows the same way.

 Assume, for the sake of contradiction, that there exist an inverse polynomial $\gamma$ and a QPT adversary $\qA$ that breaks collusion-resistant strong $\gamma$-anti-piracy security. 
Let $\qD_1=(\rho_1,U_1),...,\qD_{q+1}=(\rho_{q+1},U_{q+1})$ be (possibly entangled) $q+1$ quantum decryptors generated by $\qA$ and 
$(\msg_{i,0},\msg_{i,1})_{i\in[q+1]}$ be the challenge messages chosen by $\qA$.  
Define a distribution $D$ over $\cC\times \bit$ where $\cC$ is the ciphertext space of $\SDE$  and a family $\mathcal{P}_i=\{\mathcal{P}_{i,\ct,\coin}\}_{\ct\in\cC,\coin\in\bit}$ as in the test for a $\gamma$-good decryptor (\Cref{def:test_quantum_decryptor}) on the $i$-th quantum decryptor $\qD_{i}$. 
That is, they are defined as follows:
\begin{itemize}
\item 
$D_{i}$ is a distribution of $(\ct,\coin)$
where $\coin\gets\bit$ and $\ct\gets\Enc(\pk,\msg_{i,\coin})$.
More precisely, it chooses $\coin\gets \bit$, generates $\secp$-out-of-$\secp$ XOR secret sharing $(\msg_{i,\coin,1},...,\msg_{i,\coin,\secp})$ of $\msg_{i,\coin}$, generates $\one.\ct_{a,b}\gets \One.\Enc(\one.\pk_{a,b},\msg_{i,\coin,a})$ for $a\in[\secp]$ and $b\in[q^2]$, and sets  $\ct\seteq(\one.\ct_{a,b})_{a\in[\secp],b\in[q^2]}$.    
\item 
$\cP_{i,\ct,\coin}$ is a projective measurement corresponding to the following procedure: On a quantum state $\rho'$,  
run $\coin^\prime \gets \qD'_i(\ct)$, where $\qD'_i=(\rho',U_i)$.  
If $\coin^\prime = \coin$, output $1$, otherwise output $0$.
\end{itemize}

Let $\cP_{i,D_{i}}$ be a mixture of projective measurements corresponding to $\cP_i$ and $D_i$. 
Then, our assumption means that if we apply $\TI_{1/2+\gamma}(\cP_{i,D_{i}})$ on $\qD_i$ for all $i\in [q+1]$, then the probability that all of them return $1$ is non-negligible. 
Let $\epsilon'\seteq \gamma/4$ and $\delta'\seteq 2^{-\secp}$. 
By the first item of \Cref{lem:API}, if we apply $\ATI^{\epsilon',\delta'}_{\cP_i,D_i,1/2+\gamma-\epsilon'}$ on $\rho_i$ for $i\in[q+1]$, all of them return $1$ with a non-negligible probability.  
Suppose that it occurs, and let $\tilde{\rho}_1,...,\tilde{\rho}_{q+1}$ be the leftover states.

For $k\in[q]$, 
let $\qsk_k$ be the $k$-th secret key $\qsk_k$ given to the adversary and parse it as $\qsk_k=(b_{k,a},\one.\sk_{k,a})_{a\in[\secp]}$. 
\begin{lemma}\label{lem:prob_avoid_double_use}
 There exists $a\in [\secp]$ such that 
 $b_{1,a},...,b_{q,a}$ are pairwise distinct except for negligible probability. 
\end{lemma}
\begin{proof}
For each $a\in[\secp]$, the probability that  $b_{1,a},...,b_{q,a}$ are pairwise distinct is 
\[
1\cdot \left(1-\frac{1}{q^2}\right)\left(1-\frac{2}{q^2}\right)\cdots\left(1-\frac{q-1}{q^2}\right)\ge \left(1-\frac{1}{q}+\frac{1}{q^2}\right)^q\ge e^{-1}.
\]
Since the event is independent for each $a$, the probability that the above holds for some $a$ is at least 
\[
1- \left(1-e^{-1}\right)^\secp=1-\negl(\secp).
\]
\end{proof}
Let $a^*$ be the smallest index such that $b_{1,a^*},...,b_{q,a^*}$ are pairwise distinct.  If there is no such index, we define $a^*\seteq \bot$. By \Cref{lem:prob_avoid_double_use}, $a^*\neq \bot$ except for negligible probability. We assume that $a^*\neq \bot$ below.
Let $B\subseteq [q^2]$ be a subset that consists of all $b\in[q^2]$ such that $b_{k,a^*}=b$ for some $k\in[q]$.

For $i\in[q+1]$ and $j\in\{0,1,\ldots,q\}$, we define the following distributions $D_{i,j}$ over $\cC\times \bit$ where we highlight the difference from $D_i$ by underlines. 

\begin{description}
	\item[$D_{i,j}$:] It is a distribution of $(\ct,\coin)$ where it chooses $\coin\gets \bit$, generates $\secp$-out-of-$\secp$ XOR secret sharing $(\msg_{i,\coin,1},...,\msg_{i,\coin,\secp})$ of $\msg_{i,\coin}$, generates 
    $\one.\ct_{a,b}$ for $a\in[\secp]$ and $b\in[q^2]$ as follows: 
    \begin{itemize}
        \item For $(a,b)$ such that $a\ne a^*$, $\one.\ct_{a,b}\gets \One.\Enc(\one.\pk_{a,b},\msg_{i,\coin,a})$.
        \item For $(a,b)$ such that $a=a^*$ and $b\notin B$, $\one.\ct_{a,b}\gets \One.\Enc(\one.\pk_{a,b},\underline{0^\ell})$.
         \item For $(a,b)$ such that $a=a^*$ and $b\in B$, let 
         $k\in[q]$ be such that $b_{k,a^*}=b$. (Such a $k$ is unique since we assume $a^*\ne \bot$.) 
         \begin{itemize}
             \item If $k>j$, $\one.\ct_{a,b}\gets \One.\Enc(\one.\pk_{a,b},\msg_{i,\coin,a})$.
             \item If $k\le j$, $\one.\ct_{a,b}\gets \One.\Enc(\one.\pk_{a,b},\underline{0^\ell})$.
         \end{itemize}
    \end{itemize}
  and sets  $\ct\seteq(\one.\ct_{a,b})_{a\in[\secp],b\in[q^2]}$.    
\end{description}
Below, we use the notation $D\cind D'$ to mean that  distributions $D$ and $D'$ are computationally indistinguishable against the adversary in the security experiment of $q$-bounded collusion-resistant strong anti-piracy security. 
  
To apply \Cref{thm:collusion_resistant}, we show the following lemmas. 
\begin{lemma}\label{lem:bounded_cr_assumption_one}
    For $i\in[q]$, \[
    \Pr[\projimp(\cP_{i,D_{i,0}}) \tilde{\rho}_i\ge 
    1/2+\gamma/4]\ge 1-\negl(\secp)
    \]
\end{lemma}
\begin{proof}
The only difference between 
$D_{i,0}$ and $D_i$ is that for $(a,b)$ such that $a=a^*$ and $b\notin B$, $\one.\ct_{a,b}$ is generated as 
$\one.\ct_{a,b}\gets \One.\Enc(\one.\pk_{a,b},0^\ell)$ in the former but as $\one.\ct_{a,b}\gets \One.\Enc(\one.\pk_{a,b},\msg_{i,\coin,a})$ in the latter. 
For such $(a,b)$, the corresponding secret key $\one.\sk_{a,b}$ is not given to the adversary. Thus, by security of $\One$ (as a plain public key encryption scheme), we have $D_{i,0}\cind D_i$. 

Thus, by \Cref{lem:ATI_gap_to_dist}, 
\[\left|\Pr[\ATI^{\epsilon',\delta'}_{\cP_i,D_{i,0},1/2+\gamma-2\epsilon'} \tilde{\rho}_i=1]-\Pr[\ATI^{\epsilon',\delta'}_{\cP_i,D_i,1/2+\gamma-2\epsilon'} \tilde{\rho}_i=1]\right|\le \negl(\secp).\]

Note that $\tilde{\rho}_i$ is a residual state after $\ATI^{\epsilon',\delta'}_{\cP_i,D_i,1/2+\gamma-\epsilon'}$ outputs $1$. Thus, by the second item of \Cref{lem:API}, we have 

\[\Pr[\ATI^{\epsilon',\delta'}_{\cP_i,D_i,1/2+\gamma-2\epsilon'} \tilde{\rho}_i=1]\ge 1-\negl(\secp).\] 

Combining above, we have 
\[\Pr[\ATI^{\epsilon',\delta'}_{\cP_i,D_{i,0},1/2+\gamma-2\epsilon'} \tilde{\rho}_i=1]\ge 1-\negl(\secp).\] 
Then \Cref{lem:bounded_cr_assumption_one} follows from the above and  the first item of \Cref{lem:API} noting that $1/2+\gamma-2\epsilon'-\epsilon'=1/2+\gamma/4$. 
\end{proof}

\begin{lemma}\label{lem:bounded_cr_assumption_two}
For $i\in[q]$ and any  state $\sigma$,   $\Pr[\projimp(\cP_{i,D_{i,q}}) \sigma=1/2]= 1$. 
\end{lemma}
\begin{proof}
In $\ct$ generated by $D_{i,q}$, $\msg_{i,\coin,a^*}$ is not used at all for generating $\ct$. Thus, any algorithm given $\ct$ generated by $D_{i,q}$ can correctly guess $\coin$ with probability $1/2$. 
This implies \Cref{lem:bounded_cr_assumption_two}.
\end{proof}

Since $(1/2+\gamma/4)-1/2=\gamma/4=1/\poly(\secp)$, 
\Cref{lem:bounded_cr_assumption_one,lem:bounded_cr_assumption_two} imply that the assumptions of \Cref{thm:collusion_resistant} are satisfied. Thus, we have the following lemma.
\begin{lemma}\label{lem:apply_the_theorem_bounded}
There exist $\epsilon=1/\poly(\secp)$, $\delta=\negl(\secp)$, $\tau=1/\poly(\secp)$, and QPT algorithms $\Convert_1^{\epsilon,\delta}$ and $\Convert_2^{\epsilon,\delta}$ that satisfy the conclusion of \Cref{thm:collusion_resistant}. 
That is, consider the following experiment $\mathsf{Exp}_\tau^{\epsilon,\delta}$ where $\API_{i,j} \seteq \API_{\cP_i,D_{i,j}}^{\epsilon,\delta}$ for 
$i\in[q+1]$ and
$j\in \{0,1,\ldots,q\}$ and $\qreg{R}_i$ is the register for $\tilde{\rho}_i$ for $i\in[q+1]$. 
\paragraph{\underline{$\mathsf{Exp}_\tau^{\epsilon,\delta}$}}
\begin{enumerate}
\item Choose uniformly random distinct  $i_1,i_2\in [q+1]$ and $j\in [q]$.
\item Apply $\Convert_1^{\epsilon,\delta}(i_1,j,\cdot)$ and  $\Convert_2^{\epsilon,\delta}(i_2,j,\cdot)$ on $\qreg{R}_{i_1}$ and $\qreg{R}_{i_2}$,  which output quantum distinguishers $\qD'_1$ and $\qD'_2$, respectively.   
\item Apply the $\tau$-distinguishing test between $D_{i_1,j-1}$ and $D_{i_1,j}$ on $\qD'_1$ and the $\tau$-distinguishing test between $D_{i_2,j-1}$ and $D_{i_2,j}$ on $\qD'_2$ . If both tests accept, the experiment outputs $1$, otherwise outputs $0$. 
\end{enumerate}
Then $\Pr[\mathsf{Exp}_\tau^{\epsilon,\delta}=1]\ge 1/\poly(\secp)$. 
\end{lemma}

We construct an adversary $\qB$ that breaks single-key strong anti-piracy security of $\One$ that works as follows:  Let $\epsilon$, $\delta$, $\tau$, $\Convert_1^{\epsilon,\delta}$, and  $\Convert_2^{\epsilon,\delta}$ %
be as in \Cref{lem:apply_the_theorem_bounded}.  
\paragraph{\underline{$\qB$}}
\begin{enumerate}
\item After receiving $(1^\secp,\one.\pk,\one.\sk)$ from the challenger\footnote{In the case of classical-vendors, here $\mathcal{B}$ will instead simulate the receiver part of the adversary $\adve$, while simulating both sides (receiver and sender) of the protocols for non-challenge instances by sampling them himself (as done in the case of quantum vendor). Rest of the adversary is the same as the quantum vendor case.},
do the following.
\begin{enumerate}[label=\arabic*.]
\item Choose uniformly random distinct  $i_1,i_2\in [q+1]$ and $j\in [q]$. 
\item Choose   
$b_{k,a}\gets [q^2]$ for $k\in[q]$ and $a\in[\secp]$, and let $a^*$ be the smallest index such that for all $k_1\ne k_2 \in \{0,1,\ldots,q\}$, $b_{k_1,a^*}\ne b_{k_2,a^*}$. If there is no such index, $\qB$ aborts. 
\item 
Set $\one.\pk_{a^*,b_{j,a^*}}\seteq \one.\pk$, 
generate $(\one.\pk_{a,b},\one.\msk_{a,b})\lrun \One.\Setup(1^\secp)$ for $(a,b)\in ([\secp]\times [q^2])\setminus \{(a^*,b_{j,a^*})\}$, 
and set $\pk\seteq (\one.\pk_{a,b})_{a\in[\secp],b\in[q^2]}$. 
\item Set $\one.\qsk_{j,a^*}\seteq \one.\sk$, 
generate $\one.\qsk_{k,a} \lrun \One.\qKG(\one.\msk_{a,b_{k,a}})$ for 
$(k,a)\in([q]\times [\secp])\setminus \{(j,a^*)\}$, and set $\sk_k\seteq (b_{k,a},\one.\qsk_{k,a})_{a\in[\secp]}$ for $k\in[q]$. 
\item Send $(1^\secp,\pk,\qsk_1,...,\qsk_q)$ to $\cA$.
\end{enumerate}
\item When $\cA$ outputs $\setbk{(\msg_{i,0},\msg_{i,1})}_{i\in[q+1]}$ and $(q+1)$ (possibly entangled) quantum decryptors $\qD_1=(\rho[\qreg{R}_1],U_1) ,\ldots, \qD_{q+1}=(\rho[\qreg{R}_{q+1}],U_{q+1})$, do the following: 
\begin{enumerate}[label=\arabic*.]
\item Apply $\Convert_1^{\epsilon,\delta}(i_1,j,\cdot)$ and  $\Convert_2^{\epsilon,\delta}(i_2,j,\cdot)$ on $\qreg{R}_{i_1}$ and $\qreg{R}_{i_2}$,  which output quantum distinguishers $\qD'_1$ and $\qD'_2$, respectively.  
\item Apply the $(\epsilon'',\delta'')$-approximate  $(\tau-\epsilon'')$-distinguishing test between $D_{i_1,j-1}$ and $D_{i_1,j}$ on $\qD'_1$ and the 
$(\epsilon'',\delta'')$-approximate 
$(\tau-\epsilon'')$-distinguishing test between $D_{i_2,j-1}$ and $D_{i_2,j}$ on $\qD'_2$
 where $\epsilon''=\tau/4$ and $\delta''=2^{-\secp}$.\footnote{See \Cref{def:distinguishing_test} for the definition of approximate distinguishing tests.} If either test rejects, $\qB$ aborts. Otherwise, let  $\qD''_1$ 
 and $\qD''_2$ be the quantum distinguishers after the tests.  \label{item:approximate_distinguish_test_by_B-2}
\item Choose $\coin\gets \bit$  and define quantum decryptors $\One.\qD_1$ and  $\One.\qD_2$ as follows:
\begin{description}
    \item[$\One.\qD_1$:] 
    Upon receiving $\one.\ct$, 
    generate $\ct$ as in $D_{i_1,j-1}$ except that it sets $\one.\ct_{a^*,b_{j,a^*}}\seteq \one.\ct$, run $\qD''_1$ on $\ct$, and output whatever $\qD''_1$ outputs. 
    \item[$\One.\qD_2$:]   Upon receiving $\one.\ct$, 
    generate $\ct$ as in $D_{i_2,j-1}$ except that it sets $\one.\ct_{a^*,b_{j,a^*}}\seteq \one.\ct$, run $\qD''_2$ on $\ct$, and output whatever $\qD''_2$ outputs. 
\end{description}
\item Output challenge messages $(\msg_{i_1,\coin},0^\ell)$ and $(\msg_{i_2,\coin},0^\ell)$ and quantum decryptors $\One.\qD_1$ and  $\One.\qD_2$. 
\end{enumerate}
\end{enumerate}

By \Cref{lem:apply_the_theorem_bounded} and the first item of \Cref{lem:API}, 
 $\qB$ does not abort with an inverse polynomial probability. 
Below, we assume that $\qB$ does not abort. 

Let $\gamma'\seteq \tau/6$. 
Suppose that we apply the test for a $\gamma'$-good quantum decryptor with respect to $\one.\pk$ and $(\msg_{i_1,\coin},0^\ell)$ on $\One.\qD_1$
and the test for a $\gamma'$-good quantum decryptor with respect to $\one.\pk$ and $(\msg_{i_2,\coin},0^\ell)$ on $\One.\qD_2$. 
Let $\event{Acc}_1$ and $\event{Acc}_2$ be the events that the first and second tests accept, respectively.
\begin{lemma}\label{lem:bounded_cr_acc_one}
 It holds that $\Pr[\event{Acc}_1]\ge \gamma'$.   
\end{lemma}
\begin{proof}
 Since $\qD_1''$ is a quantum distinguisher after passing the  $(\epsilon'',\delta'')$-approximate  $(\tau-\epsilon'')$-distinguishing test between $D_{i_1,j-1}$ and $D_{i_1,j}$, by the first item of \Cref{lem:API}, if we apply the  $(\tau-2\epsilon'')$-distinguishing test between $D_{i_1,j-1}$ and $D_{i_1,j}$ on $\qD_1''$, it accepts with overwhelming probability. This means that $\qD_1''$ distinguishes $D_{i_1,j-1}$ and $D_{i_1,j}$ with advantage at least $\tau-2\epsilon''-\negl(\secp)\ge \tau/3$. This directly translates into that $\One.\qD_1$ distinguishes ciphertexts on messages $\msg_{i_1,\coin}$ and $0^\ell$ with advantage at least $\tau/3$. Thus, by a standard averaging argument, $\One.\qD_1$ passes the test for a $\gamma'$-good quantum decryptor with respect to $\one.\pk$ and $(\msg_{i_1,\coin},0^\ell)$ with probability at least $\gamma'$ noting that $\gamma'=\tau/6$.  
\end{proof}

\begin{lemma}\label{lem:bounded_cr_acc_two}
 It holds that $\Pr[\event{Acc}_2\mid \event{Acc}_1]\ge \gamma'$.   
\end{lemma}
\begin{proof}
    Let $\qD_2'''$ be the quantum distinguisher on $\qreg{R}_{i_2}$ conditioned on the event $\event{Acc}_1$. Since $\Pr[\event{Acc}_1]=1/\poly(\secp)$, due to a similar reason to that in the proof of \Cref{lem:bounded_cr_acc_one}, if we apply the  $(\tau-2\epsilon'')$-distinguishing test between $D_{i_2,j-1}$ and $D_{i_2,j}$ on $\qD_2'''$, it accepts with overwhelming probability. Then the rest of the proof can be done similarly to the proof of \Cref{lem:bounded_cr_acc_one}. 
\end{proof}
By \Cref{lem:bounded_cr_acc_one,lem:bounded_cr_acc_two}, we have $\Pr[\event{Acc}_1\land \event{Acc}_2]\ge {\gamma'}^2=1/\poly(\secp)$. This means that $\qB$ breaks single-key $\gamma'$-strong anti-piracy security of $\One$, which contradicts the assumption. This completes the proof of $q$-bounded collusion-resistant strong anti-piracy security. 

\begin{remark}[The case of SK-SDE]\label{rem:bounded_SK-SDE}
One might wonder whether the security proof still goes through for SK-SDE, since our reduction uses $\API$s corresponding to distributions that depend on the encryption key, which cannot be made public in the SK-SDE setting. However, the crucial observation is that the encryption key is needed only in the second stage, after the adversary has split its state into $q+1$ registers, and then it is used only to perform local operations on each register. Moreover, in our security definitions for SK-SDE, we reveal $\ek$ in the second stage, and therefore these local operations are allowed to depend on $\ek$. Hence, the same security proof can be adapted straightforwardly to the SK-SDE setting. 
\end{remark}
\subsection{Collusion-Resistant Identical-Challenge Search Security}\label{sec:identical_SDE}
In this section, we use \cref{thm:ind-to-iden} to show that collusion-resistant strong search anti-piracy security implies collusion-resistant search identical-challenge security. Previously, \cite{TCC:CakGoy24} had already shown that collusion-resistant strong search anti-piracy security implies plain collusion-resistant search anti-piracy security in the \emph{independent-challenge setting} (where the challenge ciphertexts given to each party are generated independently), and \cite{C:AnaKalLiu23} had shown strong anti-piracy security implies identical challenge security only in the case of only two adversaries.

We note that our result applies to any search copy-protection game due to generality of \cref{thm:ind-to-iden}. For conciseness, we only write out the details and definitions for the case of SDE. We also note that the results in this section work in both public key and secret key and both unbounded and bounded settings. Again for conciseness, we focus on the public key and unbounded setting here. 

We first give the definition of collusion-resistant identical-challenge search security.
\begin{definition}[Collusion-Resistant Identical-Challenge Search Anti-Piracy]\label{def:cr_iden-chal_search_anti_piracy}
Let $\SDE=(\Setup,\KG,\Enc,\Dec)$ be a public-key SDE scheme. We consider the collusion-resistant identical-challenge search-type anti-piracy game $\expc{\SDE,\cA}{iden}{chal}{search}(\secp,\gamma)$ between the challenger and an adversary $\cA$ below.
\begin{enumerate}
\item The challenger generates $(\pk,\msk)\lrun \Setup(1^\secp)$ and sends $(1^\secp,\pk)$ to $\cA$.
\item $\cA$ sends $1^q$ to the challenger, then the challenger generates $\sk_i \lrun \KG(\msk)$ for $i\in[q]$, and sends $(\sk_1,\ldots,\sk_q)$ to $\cA$.
\item $\cA$ outputs $(q+1)$ (possibly entangled) quantum decryptors $\qD_1=(\rho[\qreg{R}_1],\mat{U}_1) ,\ldots, \qD_{q+1}=(\rho[\qreg{R}_{q+1}],\mat{U}_{q+1})$, where $\rho$ is a quantum state over registers $\qreg{R}_1,\ldots, \qreg{R}_{q+1}$, and each $\mat{U}_i$ is a unitary quantum circuit.
\item The challenger samples $\msg \chosen \cM$ and generates $\ct \lrun \Enc(\pk,\msg)$. The challenger runs quantum decryptors $\qD_i$ on $\ct$ for all $i\in[q+1]$, and checks whether all $\qD_i$ output $\msg$. If so, the challenger outputs $1$; otherwise outputs $0$.
\end{enumerate}

We say that $\SDE$ is collusion-resistant identical-challenge search $\gamma$-anti-piracy secure if for any QPT adversary $\cA$, it satisfies that
\begin{align}
\advc{\SDE,\cA}{iden}{chal}{search}(\secp,\gamma)\seteq\Pr[ \expc{\SDE,\cA}{iden}{chal}{search}(\secp, \gamma)=1]\le \frac{1}{\abs{\cM}} +\gamma(\secp)+\negl(\secp).
\end{align}
We say that $\SDE$ is collusion-resistant identical-challenge search anti-piracy secure if it is collusion-resistant identical-challenge search $\gamma$-anti-piracy secure for any inverse polynomial $\gamma$.

We say that $\SDE$ is $q^*$-bounded collusion-resistant identical-challenge search anti-piracy secure if the above holds when $q$ is limited to be  at most $q^*$. 
\end{definition}
\begin{remark}
    A secret key variant of the above security notion is defined similarly except that $\ek$, which playes the role of $\pk$,  is not given to the adversary at the beginning and instead given after the adversary split its state into $q+1$ registers. 
\end{remark}

Now we state and prove our theorem.
\begin{theorem}\label{thm:ind-to-iden-sde}
If $\SDE$ is (resp. $q^*$-bounded) collusion-resistant strong-search anti-piracy secure and $1/\abs{\cM}$ is negligible, then it is also (resp. $q^*$-bounded)  collusion-resistant identical-challenge search anti-piracy secure.
\end{theorem}
\begin{proof}
By definition, we have
\[
\advc{\SDE,\cA}{iden}{chal}{search}(\secp,\gamma) = p_{\mathsf{iden}} 
\]
Then by \cref{thm:ind-to-iden}, we have $$p_{\mathsf{iden}} \le (q+2)(\gamma(\secp)+p_{\mathsf{strong,\gamma}})$$
If $\SDE$ is collusion-resistant strong-search $\gamma$-anti-piracy secure, then we get $p_{\mathsf{strong,\gamma}} \leq\negl(\lambda)$ for any inverse polynomial $\gamma$. Combining with above, we get $p_{\mathsf{iden}} \leq (q+2)(2\gamma(\secp))$ for any inverse polynomial $\gamma$, which means $\SDE$ satisfies collusion-resistant search identical-challenge anti-piracy security if $\frac{1}{\abs{\cM}}$ is negligible. This completes the proof.
\end{proof}

Applying \Cref{thm:ind-to-iden-sde} to \Cref{cor:cr_SDE_from_IO,cor:cr_SKSDE_from_LWE,cor:classical-vendor-sde} we obtain the following corollaries. 
\begin{corollary}\label{cor:identical_cr_SDE_from_IO}
If there exist a secure iO for polynomial size classical circuits and OWFs, then there exists an SDE scheme that satisfies collusion-resistant identical-challenge search anti-piracy security. Moreover, its key generation algorithm has classically determined outputs (\Cref{def:classically_determined}).
\end{corollary}

\begin{corollary}
If there exist a subexponentially secure iO for polynomial size classical circuits and LWE is subexponentially hard for quantum adversaries, then there exists a classical-vendor SDE scheme that satisfies collusion-resistant identical-challenge search anti-piracy security.
\end{corollary}

\begin{corollary}\label{cor:identical_cr_SKSDE_from_LWE}
Assuming the LWE assumption, 
for any polynomial $q^*$, 
there exists an SK-SDE scheme that satisfies $q^*$-bounded collusion-resistant identical-challenge search security. Moreover, its key generation algorithm has classically determined outputs (\Cref{def:classically_determined}). 
\end{corollary}

\subsection{Multi-Copy Secure Construction}\label{sec:multi-copy_SDE}
In this section, we apply our purification compiler (\cref{thm:compiler}) to convert our collusion-resistant SDE schemes into multi-copy secure schemes. 

First, we define multi-copy security of SDE.
\begin{definition}[Multi-Copy Anti-Piracy]
Let $\SDE = (\Setup, \KG, \Enc, \Dec)$ be a (secret key or public key) SDE scheme where $\KG$ has classically determined outputs (\Cref{def:classically_determined}). We define its multi-copy strong anti-piracy security, multi-copy strong search anti-piracy security, and multi-copy identical-challenge search anti-piracy security analogously to their collusion-resistant counterparts (\cref{def:cr_strong_anti_piracy,def:cr_strong_search_anti_piracy,def:cr_iden-chal_search_anti_piracy}), except that the adversary is given $q$ exact copies of the same secret key generated under the same randomness, rather than independently generated keys.
\end{definition}

Since the schemes in \Cref{cor:cr_SDE_from_IO,cor:identical_cr_SDE_from_IO} have classically determined outputs (\Cref{def:classically_determined}),  we can apply \cref{thm:compiler} to compile them into a multi-copy secure scheme while preserving security, and correctness also follows since our collusion-resistant schemes have keys that have the same output on the same input for each key.
As a result, we obtain the following theorem. 

\begin{theorem}\label{thm:sde}
    Assuming the existence of polynomially secure indistinguishability obfuscation and one-way functions, there exists an %
    SDE scheme that satisfies
    \begin{itemize}
        \item %
        multi-copy strong anti-piracy security
        \item %
        multi-copy strong search anti-piracy security
        \item %
        multi-copy identical-challenge search anti-piracy security.
    \end{itemize}
\end{theorem} 

Similarly, in the bounded setting, we can apply \cref{thm:compiler} to the schemes of \Cref{cor:cr_SKSDE_from_LWE,cor:identical_cr_SKSDE_from_LWE} to obtain the following theorem. 
\begin{theorem}\label{thm:sde-sk}
    Assuming the LWE assumption, for any polynomial $q^*$, there exists an %
    SK-SDE scheme that satisfies
    \begin{itemize}
        \item %
        $q^*$-bounded multi-copy strong anti-piracy security
        \item %
        $q^*$-bounded multi-copy strong search anti-piracy security
        \item %
        $q^*$-bounded multi-copy identical-challenge search anti-piracy security.
    \end{itemize}
\end{theorem} 

We omit writing the full descriptions of the schemes explicitly, since they are evident from \cref{thm:compiler} (also see \Cref{sec:money_PR_compiler} for a similar example). 

\newpage
\section{LOCC Leakage-Resilience}\label{sec:lr}
In this section, we construct multi-copy secure LOCC leakage-resilient and bounded quantum leakage-resilient PKE schemes.

In \Cref{sec:lrpke-def}, we recall the definitions of collusion-resistant LOCC leakage-resilient PKE schemes. In \Cref{sec:lrpke-compiler}, we obtain collusion-resistant schemes using our generic compiler (\cref{sec:fe-compiler}) and then multi-copy schemes using our purification compiler (\cref{sec:purification_compiler}).

\subsection{Definitions}\label{sec:lrpke-def}
We now recall LOCC leakage-resilience (\cite{CGLR23}) and bounded-quantum leakage resilience (\cite{qleak}).

\begin{definition}[Collusion-Resistant LOCC Leakage-Resilience Security]\label{def:coll-res-lr}
Let $\pke$ be a public-key encryption scheme with quantum protected keys (\Cref{def:qpke}).  Consider the following game between the challenger and an adversary $\adve$.

\paragraph{\underline{$\lrpkegame{\adve}(1^\lambda)$}}
\begin{enumerate}
\item The challenger samples $pk, msk \samp \pke.\mathsf{Setup}(1^\lambda)$  and submits $pk$ to $\adve$.
\item The adversary submits a number $q$. 
\item For $i \in [q]$,  the challenger samples $\regi{key,i} \samp \pke.\KG(msk)$.
\item The challenger sets $\regi{key} = (\regi{key,i})_{i \in [q]}$.
    \item The adversary obtains leakage on $\regi{key}$ adaptively over any number of rounds by sending the description of classical-quantum circuit and learning the outcome of the circuit applied to $\regi{key}$
    \item The adversary $\adve$ outputs two messages $m_0, m_1$.
    \item The challenger samples $b \samp \zo$ and $ct \samp \pke.\mathsf{Enc}(pk, m_b)$ and submits $ct$ to $\adve$.
    \item $\adve$ outputs a guess $b' \in \zo$.
    \item The challenger outputs $1$ if and only if $b' = b$.
\end{enumerate}
$\pke$ is said to satisfy  \emph{LOCC leakage-resilience} if for any QPT adversary $\adve$,
\begin{equation*}
    \Pr[\lrpkegame{\adve}(1^\lambda) = 1] \leq \frac{1}{2} + \negl(\lambda).
\end{equation*}
\end{definition}

\begin{definition}[Collusion-Resistant Bounded-Quantum-Leakage Resilience Security]\label{def:coll-res-lr-q}
Let $\pke$ be a public-key encryption scheme with quantum protected keys. Let $t(\lambda)$ denote a leakage bound. Consider the following game between the challenger and an adversary $\adve = ((\adve_1^{i})_{i \in [q(\lambda)]}, \adve_2)$.

\paragraph{\underline{$\mathsf{QuantumLeakageGame}{\adve}(1^\lambda)$}}
\begin{enumerate}
\item The challenger samples $pk, msk \samp \pke.\mathsf{Setup}(1^\lambda)$  and submits $pk$ to $\adve$.
\item The adversary submits a number $q$.
\item For $i \in [q]$, the challenger samples $\regi{key,i} \samp \pke.\KG(msk)$.
\item The challenger submits $\regi{key,i}$ to $\adve_1^i$ for $i \in [q]$.
\item The adversary $\adve_2$ communicates interactively with all $\adve_1^i$ (in any order, interleaved), exchanging arbitrary classical messages but exchanging at most $t(\lambda)$ qubits for each $\adve_1^i$.

    \item The adversary $\adve_2$ outputs two messages $m_0, m_1$.
    \item The challenger samples $b \samp \zo$ and $ct \samp \pke.\mathsf{Enc}(pk, m_b)$ and submits $ct$ to $\adve$.
    \item $\adve_2$ outputs a guess $b' \in \zo$.
    \item The challenger outputs $1$ if and only if $b' = b$.
\end{enumerate}
$\pke$ is said to satisfy  \emph{$t(\lambda)$-bounded  leakage-resilience} if for any QPT adversary $\adve$,
\begin{equation*}
    \Pr[\mathsf{QuantumLeakageGame}{\adve}(1^\lambda) = 1] \leq \frac{1}{2} + \negl(\lambda).
\end{equation*}
\end{definition}

We recall the following results about single-key schemes. 
\begin{theorem}[\cite{qleak}]\label{thm:singlekey-lr}
Assuming polynomially secure iO for classical circuits and one-way functions, there exists a $\pke$ scheme that satisfies single-key $(c\cdot s(\lambda))$-bounded-quantum-leakage resilience with keys of size $s(\lambda)$ qubits and $0 < c < 1$ is a constant.
\end{theorem}

\subsection{Collusion-Resistant and Multi-Copy Schemes}\label{sec:lrpke-compiler}
In this section, we show that the compiler from \cref{sec:fe-compiler} yields a collusion-resistant leakage-resilient scheme when the single-key scheme $\One$ is secure. We then instantiate it with single-key schemes to obtain collusion-resistant schemes in the plain model. Finally, we use our purification compiler to obtain multi-copy secure schemes.

\begin{theorem}
When $\one$ is a single-key secure LOCC leakage-resilient scheme, then the compiled scheme from \cref{sec:fe-compiler} satisfies collusion-resistant LOCC leakage-resilience.

When $\one$ is a single-key secure bounded quantum leakage-resilient scheme, then the compiled scheme from \cref{sec:fe-compiler} satisfies collusion-resistant bounded quantum leakage-resilience with the same bound.
\end{theorem}
\begin{proof}
This is a simple corollary of \cref{thm:fe-coll_res}. Given a successful adversary $\adve$ for the collusion-resistant game, we can obtain an adversary $\adve_{one}$ for the security game of $\one$ by simply simulating this adversary $\adve$ with our own FE instance and $\one$ instances for indices $j\neq j^*$. At the end, when $\adve_{one}$ receives the challenge ciphertext $ct$, it simulates $\mathcal{D}_{j^*}$ using $ct$, and outputs the guess of $\adve_2'$ (as in \cref{thm:fe-coll_res}) on this simulated ciphertext. Note that \cref{thm:fe-coll_res} only modifies the last part of the adversary, thus, $\adve_{one}$ will be a valid adversary for the security game of $\one$.
\end{proof}

Instantiating the single-key scheme $\one$ with the scheme from \cref{thm:singlekey-lr}, and invoking our purification compiler (\cref{sec:purification_compiler}), we get the following corollary. 
\begin{corollary}\label{cor:leakpure}
Assuming polynomially secure iO for classical circuits and one-way functions, there exists a $\pke$ scheme that satisfies multi-copy LOCC leakage-resilience and $(c\cdot s(\lambda))$-bounded-quantum-leakage resilience with keys of size $s(\lambda)$ qubits and $0 < c < 1$ is a constant.
\end{corollary}
The security follows by \cref{thm:compiler} and correctness follows since our collusion-resistant schemes have keys that have the same output on the same input for each key.
\newpage
\section{Secure Key Leasing}\label{sec:skl}
In this section, we construct collusion-resistant classical-vendor and multi-copy secure secure leasing schemes for public-key encryption.

In \Cref{sec:skl-def}, we recall the definitions of secure leasing schemes. In \Cref{sec:skl-compiler}, we obtain collusion-resistant schemes using our generic compiler (\cref{sec:fe-compiler}) and then multi-copy schemes using our purification compiler (\cref{sec:purification_compiler}).

\subsection{Definitions}\label{sec:skl-def}
We now recall collusion-resistant classical-vendor secure leasing schemes.

\begin{definition}[Collusion-Resistant Classical-Vendor Secure Leasing for PKE]
   Let $\pke$ be a public-key encryption scheme (\cref{def:qpke}) with classical-vendor quantum protected keys.   Consider the following game between the challenger and an adversary $\adve$.

\paragraph{\underline{$\mathsf{SecureLeasingGame}{\adve}(1^\lambda)$}}
\begin{enumerate}
\item The challenger samples $pk, msk \samp \pke.\mathsf{Setup}(1^\lambda)$  and submits $pk$ to $\adve$.
\item For any number of rounds, $\adve$ can interact with the challenger asking for a key, in which case the challenger executes the classical-vending algorithm $\mathsf{Sender}$ of $\pke$. During any time, the adversary can also send queries to the deletion-certificate verification oracle.
\item Let $q(\lambda)$ be the number of times the classical-vending protocol ended without terminating with $\bot$. That is, this is the number of keys obtained by the adversary. Then, the adversary sends $q(\lambda)$ deletion-verification certificates, which are then verified by the challenger. If some certificates do not pass, the challenger outputs $0$ and terminates. 

    \item The challenger samples $b \samp \zo$ and $ct \samp \pke.\mathsf{Enc}(pk, m_b)$ and submits $ct$ to $\adve$.
    \item $\adve$ outputs a guess $b' \in \zo$.
    \item The challenger outputs $1$ if and only if $b' = b$.
\end{enumerate}
$\pke$ is said to satisfy collusion-resistant secure leasing security if for any QPT adversary $\adve$,
\begin{equation*}
    \Pr[\mathsf{SecureLeasingGame}{\adve}(1^\lambda) = 1] \leq \frac{1}{2} + \negl(\lambda).
\end{equation*}
\end{definition}

We recall the following results about single-key secure key-leasing.
\begin{theorem}[\cite{10.1007/978-981-96-0947-5_7}]\label{thm:singlekey-cvendor}
Assuming LWE, there exists a single-key classical-lessor secure-leasing scheme  for PKE whose leasing/key-sending protocol is 2-rounds.
\end{theorem}

\begin{theorem}[\cite{cryptoeprint:2024/1564}]\label{thm:singlekey-cdel}
Assuming a classical PKE scheme, there exists a single-key secure-leasing scheme for PKE with classical deletion certificates. 
\end{theorem}

\subsection{Collusion-Resistant and Multi-Copy Schemes}\label{sec:skl-compiler}
In this section, we show that the compiler from \cref{sec:fe-compiler} yields a collusion-resistant secure leasing scheme when the single-key scheme $\One$ is secure. We then instantiate it with single-key schemes to obtain collusion-resistant schemes in the plain model. Finally, we use our purification compiler to obtain multi-copy secure schemes.

\begin{theorem}
If $\one$ is a scheme that satisfies single-key secure leasing security, then the compiled scheme from \cref{sec:fe-compiler} (with the additions of the deletion and verification algorithms below) satisfies collusion-resistant secure-leasing security. Moreover, if $\one$ is classical-deletion or classical-vendor scheme or public-certificate-verification scheme, then so is the compiled scheme.

For the deletion algorithm, the compiled scheme simply runs the deletion of $\one$, along with sending the classical functional key. For the certificate verification algorithm, the scheme simply runs the certificate verification of the associated instance\footnote{In case of secret-verification-key schemes, we simply add a signed encryption of the verification key to the leased key, alongside the functional key} of $\one$.
\end{theorem}
\begin{proof}
This is a simple corollary of \cref{thm:fe-coll_res}. Given a successful adversary $\adve$ for the collusion-resistant game, we can obtain an adversary $\adve_{one}$ for the security game of $\one$ by simply simulating this adversary $\adve$ with our own FE instance and $\one$ instances for indices $j\neq j^*$. At the end, when $\adve_{one}$ receives the challenge ciphertext $ct$, it simulates $\mathcal{D}_{j^*}$ using $ct$, and outputs the guess of $\adve_2'$ (as in \cref{thm:fe-coll_res}) on this simulated ciphertext. Note that \cref{thm:fe-coll_res} only modifies the last part of the adversary, thus, $\adve_{one}$ will be a valid adversary for the security game of $\one$.

For a similar example, see the adversary $\mathcal{B}$ in the proof of \cref{thm:collusion-resistant_conversion}.
\end{proof}

Instantiating the single-key scheme $\one$ with the scheme from \cref{thm:singlekey-cvendor}, we get the following corollary.
\begin{corollary}\label{cor:skl-cl}
    Assuming LWE and existence of FE, there exists a collusion-resistant classical-lessor secure-leasing scheme for PKE whose leasing/key-sending protocol is 2-rounds.
\end{corollary}

Instantiating the single-key scheme $\one$ with the scheme from \cref{thm:singlekey-cdel} and invoking our purification compiler (\cref{sec:purification_compiler}), we get the following corollary. 
\begin{corollary}\label{cor:skl-cdel}
   Assuming a classical FE scheme, there exists a multi-copy secure-leasing scheme for PKE with classical deletion certificates.
\end{corollary}
The security follows by \cref{thm:compiler} and correctness follows since our collusion-resistant schemes have keys that have the same output on the same input for each key.

Finally, since bounded secure FE can be achieved from any PKE, we get the following.
\begin{corollary}\label{cor:skl-bdd}
   Assuming LWE, there exists a collusion-resistant classical-lessor secure-leasing scheme for PKE whose leasing/key-sending protocol is 2-rounds and a multi-copy secure scheme with classical deletion certificates.
\end{corollary}

\newpage
\part{Unclonable Encryption}
\section{Unclonable Encryption}\label{Sec:UE} 
In this section, we construct unclonable encryption (UE) schemes that satisfy multi-copy security.  

In \cref{sec:uedef}. we recall the definitions of unclonable encryption in various settings. Then, in \cref{sec:SDE_to_UE}, we show how to obtain search-secure multi-challenge UE from any collusion-resistant SDE scheme with identical-challenge security (which is constructed in \cref{sec:CR-transformation}). In \cref{sec:ueqrom}, we show how to construct indistinguishable-secure multi-challenge UE in QROM. Finally, in \cref{sec:multi_copy_UE}, by invoking \Cref{thm:compiler} we obtain multi-copy secure schemes from our multi-challenge secure schemes.

\subsection{Definitions}\label{sec:uedef}
In this section, we recall the definitions of unclonable encryption in various settings.

First we recall the syntax for unclonable encryption.
\begin{definition}[UE (Syntax)]\label{def:unclonable_ske}
A UE scheme with the message space $\cM$ is a tuple of QPT algorithms $(\KG,\Enc,\Dec)$.
\begin{description}
\item[$\KG(1^\secp)\ra(\ek,\dk)$:] The key generation algorithm takes as input the security parameter $1^\secp$ and outputs an encryption key $\ek$ and a decryption key $\dk$.
    \item[$\Enc(\ek,\msg) \ra \ct$:] The encryption algorithm takes as input $\ek$ and a plaintext $\msg\in\cM$ and outputs a ciphertext $\ct$.
    \item[$\Dec(\dk,\ct) \ra \msg^\prime$:] The decryption algorithm takes as input $\dk$ and $\ct$ and outputs a plaintext $\msg^\prime \in \cM$ or $\bot$.

\item[Decryption correctness:] For any $\msg\in\cM$, it holds that
\begin{align}
\Pr\left[
\Dec(\dk,\ct)= \msg
\ \middle |
\begin{array}{ll}
(\ek,\dk)\lrun \KG(1^\secp)\\
\ct \lrun \Enc(\ek,\msg)
\end{array}
\right] 
\geq 1-\negl(\secp).
\end{align}
\end{description}
\end{definition}

\begin{remark}
The UE syntax above is slightly different from the standard one since it uses two different secret keys $\ek$ and $\dk$. Nevertheless, in \Cref{sec:UE_ek=dk}, we show how to convert any such scheme to a scheme with identical encryption-decryption key through a simple one-time pad technique.

\end{remark}

We now recall the search and indistinguishability security notions in the setting of multi-challenge security where the adversary receives multiple challenge ciphertexts.
\begin{definition}[Multi-Challenge Search Security]\label{def:UE_OT_MC_OW}
Let $\UE=\UE.(\KG,\Enc,\Dec)$ be a UE scheme  with the message space $\cM$. 
We consider the following security experiment $\expc{\UE,\cA}{multi}{chal}{search}(\secp)$, where $\cA$ is a QPT adversary.

\begin{enumerate}
\item The challenger generates $(\ek,\dk) \lrun \UE.\KG(1^\secp)$ and sends $1^\secp$ to $\cA$.
\item $\cA$ sends $1^q$ to the challenger.
\item The challenger samples %
$\msg \chosen \cM$,  
generates $\ct_i \lrun \UE.\Enc(\ek,\msg)$ for all $i\in[q]$, and sends $(\ct_i)_{i\in[q]}$ to $\cA$.
\item $\cA$ creates a $(q+1)$-partite state $\rho$ over registers $\qreg{R}_1,\ldots,\qreg{R}_{q+1}$ and $(q+1)$ unitaries $(U_1,\ldots,U_{q+1})$, and lets $\cA_i \seteq (\rho[\qreg{R}_i],U_i)$ for each $i\in[q+1]$.
\item The challenger sends $\dk$ to $\cA_i$ for each $i\in[q+1]$.
Each $\cA_i$ outputs $\msg_i^\prime$.
If $\msg_i^\prime = \msg$ for all $i\in [q+1]$, the challenger outputs $1$, otherwise outputs $0$.

\end{enumerate}
We say that $\UE$ is multi-challenge search secure if for any QPT $\cA$, it holds that
\begin{align}
\advc{\UE,\cA}{multi}{chal}{search}(\secp)\seteq \Pr[ \expc{\UE,\cA}{multi}{chal}{search}(\secp)=1]\le \frac{1}{\abs{\cM}} + \negl(\secp).
\end{align} 
We say that $\UE$ is $q^*$-bounded multi-challenge search secure if the above holds when $q$ is limited to be at most $q^*$. 
\end{definition}

\begin{definition}[Multi-Challenge Indistinguishability  Security]\label{def:UE_OT_MC_IND}
Let $\UE=\UE.(\KG,\Enc,\Dec)$ be a UE scheme  with the message space $\cM$. 
We consider the following security experiment $\expc{\UE,\cA}{multi}{chal}{ind}(\secp)$, where $\cA$ is a QPT adversary.

\begin{enumerate}
\item The challenger generates $(\ek,\dk) \lrun \UE.\KG(1^\secp)$ and sends $1^\secp$ to $\cA$.
\item $\cA$ sends $1^q$ and $(\msg_{i,0},\msg_{i,1})_{i\in[q]}\in \cM^{2\times q}$ to the challenger.
\item The challenger chooses $\coin\gets \bit$,  
generates $\ct_i \lrun \UE.\Enc(\ek,\msg_{i,\coin})$ for all $i\in[q]$, and sends $(\ct_i)_{i\in[q]}$ to $\cA$.
\item $\cA$ creates a $(q+1)$-partite state $\rho$ over registers $\qreg{R}_1,\ldots,\qreg{R}_{q+1}$ and $(q+1)$ unitaries $(U_1,\ldots,U_{q+1})$, and lets $\cA_i \seteq (\rho[\qreg{R}_i],U_i)$ for each $i\in[q+1]$.
\item The challenger sends $\dk$ to $\cA_i$ for each $i\in[q+1]$.
Each $\cA_i$ outputs $\coin_i^\prime$.
If $\coin_i^\prime = \coin$ for all $i\in [q+1]$, the challenger outputs $1$, otherwise outputs $0$.

\end{enumerate}
We say that $\UE$ is multi-challenge indistinguishability secure if for any QPT $\cA$, it holds that
\begin{align}
\advc{\UE,\cA}{multi}{chal}{ind}(\secp)\seteq \Pr[ \expc{\UE,\cA}{multi}{chal}{ind}(\secp)=1]\le \frac{1}{2} + \negl(\secp).
\end{align}

We say that $\UE$ is $q^*$-bounded multi-challenge indistinguishability secure if the above holds when $q$ is limited to be at most $q^*$. 
\end{definition}

Multi-copy security is defined as in the other cases (i.e. the adversary receives the same state multiple times) - see \cref{sec:pure-ue-def}.

\subsection{From Collusion-Resistant SDE to Multi-Challenge Search Secure UE}\label{sec:SDE_to_UE}
In this section, we give a compiler that constructs $X \in \{\text{bounded}, \text{unbounded}\}$ multi-challenge search secure unclonable encryption from any secret-key SDE scheme with $X$ collusion-resistant identical-challenge anti-piracy security.

Now we present our compiler, which simply switches the roles of ciphertexts and keys. Let $\SDE =\SDE.(\Setup,\qKG,\Enc,\qDec)$ be an SK-SDE scheme with message space $\zo^\ell$. Then we construct a UE scheme  $\UE=\UE.(\KG,\Enc,\Dec)$ with message space $\zo^\ell$ as follows. 
\begin{description}
\item[$\UE.\KG(1^\secp)$:] $ $
\begin{enumerate}
\item Generate $(\sde.\ek,\sde.\msk) \lrun \SDE.\Setup(1^\secp)$.
\item Choose $s \chosen \zo^\ell$.
\item Generate $\sde.\ct \lrun \SDE.\Enc(\sde.\ek,s)$.
\item Output $\ue.\ek \seteq (\sde.\msk,s)$ and $\ue.\dk\seteq \sde.\ct$.
\end{enumerate}
\item[$\UE.\Enc(\ue.\ek,\msg)$:] $ $
\begin{enumerate}
    \item Parse $\ue.\ek = (\sde.\msk,s)$.
\item Generate $\sde.\sk\lrun \SDE.\KG(\sde.\msk)$.
\item Output $\ue.\ct \seteq (\sde.\sk,\msg \oplus s)$
\end{enumerate}
\item[$\UE.\Dec(\ue.\dk,\ue.\ct)$:] $ $
\begin{enumerate}
\item Parse $\ue.\dk = \sde.\ct$ and $\ue.\ct=(\sde.\sk,\mu)$.
\item Compute $s^\prime \lrun \SDE.\Dec(\sde.\sk,\sde.\ct)$.
\item Output $\msg^\prime \seteq \mu \oplus s^\prime$.
\end{enumerate}
\end{description}

Correctness is straightforward. 
Security is stated in the following theorem. 
\begin{theorem}\label{thm:from_SDE_to_UE}
If $\SDE$ is (resp. $q^*$-bounded) collusion-resistant identical-challenge search anti-piracy secure, then $\UE$ is  (resp. $q^*$-bounded) multi-challenge search secure UE.
\end{theorem}

We prove this theorem in \cref{sec:ue-search-proof}.

By applying the above compiler to the schemes in \Cref{cor:identical_cr_SDE_from_IO,cor:identical_cr_SKSDE_from_LWE}, we obtain the following corollaries. 

\begin{corollary}\label{cor:UE_from_IO}
If there exist a secure iO for polynomial size classical circuits and OWFs, then there exists a UE scheme that satisfies multi-challenge search security. Moreover, its key generation algorithm has classically determined outputs (\Cref{def:classically_determined}). 
\end{corollary}
\begin{corollary}\label{cor:bounded_UE_from_LWE}
Assuming the LWE assumption, 
for any polynomial $q^*$, 
there exists a UE scheme that satisfies $q^*$-bounded  multi-challenge search security. Moreover, its key generation algorithm has classically determined outputs (\Cref{def:classically_determined}). 
\end{corollary}

\subsubsection{Proof of \cref{thm:from_SDE_to_UE}}\label{sec:ue-search-proof}
We only explicitly prove the theorem in the unbounded case, and the proof for the bounded case is almost identical. 

We construct an adversary $\cA_\SDE$ for SDE by using an adversary $\cA_\UE$ for UE.
$\cA_\SDE$ proceeds as follows.
\begin{enumerate}
\item $\cA_\SDE$ is given $1^\secp$, forwards $1^\secp$ to $\cA_{\UE}$, and receives $1^q$.
\item $\cA_\SDE$ sends $1^q$ to its challenger, receieves $(\sde.\sk_1,\ldots,\sde.\sk_q)$.
\item $\cA_\SDE$ samples a uniformly random message $t \chosen \zo^\ell$, set $\ue.\ct_i \seteq (\sde.\sk_i,t)$ for all $i\in[q]$, and sends $(\ue.\ct_1,\ldots,\ue.\ct_q)$ to $\cA_{\UE}$. 
\item When $\cA_{\UE}$ creates $(q+1)$-partite state $\delta$ over registers $\qreg{R}_1,\ldots,\qreg{R}_{q+1}$ and $(q+1)$ unitaries $(U_1,\ldots,U_{q+1})$, $\cA_\SDE$ sets $\cA_{\UE,i} \seteq (\rho[\qreg{R}_i],U_i)$ for each $i\in[q+1]$ and constructs the following quantum decryptors $\qD_i$:
\begin{itemize}
\item Receive an input $\ue.\dk$, interpret it as an SDE ciphertext $\sde.\ct \seteq \ue.\dk$, and run $\msg_i^\prime \lrun \cA_{\UE,i}(\sde.\ct)$.
\item Output $\msg_i^\prime \oplus t$.
\end{itemize}
\item $\cA_\SDE$ sends $(q+1)$ quantum decryptors $\qD_1,\ldots,\qD_{q+1}$ to its challenger.
\end{enumerate}
The challenger of the collusion-resistant identical-challenge search anti-piracy security game samples $s \chosen \zo^\ell$, generates $\sde.\ct \lrun \SDE.\Enc(\sde.\ek,s)$, and runs $\qD_i (\sde.\ct)$ for all $i\in[q+1]$. In $\UE$, $\sde.\ct$ works as the decryption key $\ue.\dk$ of $\UE$. In addition, $\cA_\SDE$ perfectly simulates $\ue.\ct_i$ since $t$ is uniformly random.  Hence, $\cA_{\UE,i}(\sde.\ct)$ successfully decrypts $\ue.\ct$ and is expected to output $t \oplus s$ for all $i\in [q+1]$ with non-negligible probability. That is, $\qD_i$ outputs $s$ for all $i\in[q+1]$ with non-negligible probability. This completes the proof.

\begin{remark}
By a similar proof, we can also show that the above UE scheme satisfies a variant of multi-challenge search security where each ciphertext encrypts independently random message rather than the same message.  
\end{remark}

\subsection{Multi-Challenge Indistinguishability Secure UE in QROM}\label{sec:ueqrom}
In this section, we construct multi-challenge indistinguishability secure UE scheme in the quantum random oracle model (QROM) unconditionally.

For simplicity of the proof, we use two different oracles, but through standard tricks this can be converted into a single random oracle.
Let 
$G:\bit^\secp\times \bit^\secp\rightarrow \bit^\secp$ and
$H:\bit^\secp\rightarrow \bit^\ell$ be random oracles. 
We construct a UE scheme  $\UE=\UE.(\KG,\Enc,\Dec)$ with message space $\zo^\ell$ as follows. 
\begin{description}
\item[$\UE.\KG(1^\secp)$:] $ $
\begin{enumerate}
\item Choose $s \chosen \zo^\secp$.
\item Output $\ek \seteq s$ and $\dk\seteq s$.
\end{enumerate}
\item[$\UE.\Enc(\ek,\msg)$:] $ $
\begin{enumerate}
    \item Parse $\ek = s$.
    \item Choose 
    $r\gets \bit^\secp$, 
    $x\gets \zo^\secp$, and $\theta\gets \zo^\secp$. 
\item Output $\ct \seteq (r,\ket{x_\theta},G(s,r)\oplus \theta,H(x)\oplus \msg)$.
\end{enumerate}
\item[$\UE.\Dec(\dk,\ct)$:] $ $
\begin{enumerate}
\item Parse $\dk = s$ and $\ct = (r,\ket{\psi},u,v)$.
\item Compute $\theta^\prime\seteq G(r,s)\oplus u$.
\item Compute $H^{\theta}\ket{\psi}$ and measure it in computational basis to obtain $x^\prime$.
\item Compute and output $\msg^\prime \seteq H(x^\prime)\oplus w$.
\end{enumerate}
\end{description}
Correctness is clear. 
Security is stated in the following theorem. 
\begin{theorem}\label{thm:UE_IND_QROM}
$\UE$ is a multi-challenge  indistinguishability secure UE scheme in the QROM. 
\end{theorem}
We prove this theorem in \cref{sec:ue-ind-proof}.
\subsubsection{Proof of \cref{thm:UE_IND_QROM}}\label{sec:ue-ind-proof}
For a contradiction, assume that  there is a QPT adversary $\qA$ that breaks multi-challenge indistinguishability security of $\UE$. Let $q(\lambda)$ denote the number of ciphertexts it obtains, and let $Q$ be the number of total queries to the random oracles by $\qA$. 

Toward contradiction, 
We consider the following sequence of hybrids. 
\paragraph{$\underline{\hyb_0}$}: The original game $\expc{\UE,\cA}{multi}{chal}{ind}(\secp)$. That is, it works as follows. 
\begin{enumerate}
\item The challenger chooses uniformly random functions $G:\bit^\secp\times \bit^\secp\rightarrow \bit^\secp$ and
$H:\bit^\secp\rightarrow \bit^\ell$.
$\qA$ has quantum oracle access to $G$ and $H$ throughout the game. 
\item Upon receiving $1^\secp$ from the challenger, $\qA$ sends $(\msg_{k,0},\msg_{k,1})_{k\in[q]}$ to the challenger. 
\item The challenger chooses 
$s\gets \bit^\secp$, 
$\coin\gets \bit$, 
$r_k\gets \bit^\secp$, 
$x_k\gets \zo^\secp$ and $\theta_k\gets \zo^\secp$ for $k\in[q]$, sets 
$\ct_k \seteq (r_k,\ket{x_{k,\theta_k}},G(s,r_k)\oplus \theta_k,H(x_k)\oplus \msg_{k,\coin})$ for $k\in[q]$, and sends $(\ct_k)_{k\in[q]}$ to $\cA$.\footnote{$\ket{x_{k,\theta_k}}$ means the BB84 state corresponding to $x_k$ and $\theta_k$.
} 
\item  $\cA$ creates a $(q+1)$-partite state $\rho$ over registers $\qreg{R}_1,\ldots,\qreg{R}_{q+1}$ and $(q+1)$ unitaries $(U_1,\ldots,U_{q+1})$, and sets $\cA_i \seteq (\qreg{R}_i,U_i)$ for each $i\in[q+1]$.
\item The challenger sends $\dk=s$ to $\cA_i$ for each $i\in[q+1]$.
Each $\cA_i$ outputs $\coin_i^\prime$.
If $\coin_i^\prime = \coin$ for all $i\in [q+1]$, the challenger outputs $1$, otherwise outputs $0$.
\end{enumerate}

\paragraph{$\underline{\hyb_1}$}: This hybrid is identical to $\hyb_0$ except that if there are $k_1\ne k_2$ such that $r_{k_1}=r_{k_2}$ or $x_{k_1}=x_{k_2}$, then the hybrid immediately halts by outputting $0$. 

The probability that the hybrid aborts is at most $2q^2/2^{-\secp}=\negl(\secp)$. Thus, we have 
\[
|\Pr[\hyb_0=1]-\Pr[\hyb_1=1]|=\negl(\secp).
\]

\paragraph{$\underline{\hyb_2}$}: This hybrid is identical to $\hyb_1$ except that the challenger chooses $s$ at the beginning and the oracle $G$ for $\qA$ is replaced with $G_\bot^s$ defined below:  
\begin{align}
G_\bot^s(s^\prime,r)=
\begin{cases}
\bot~&\text{if}~s^\prime=s\\
G(s^\prime,r)~&\text{else}
    \end{cases}. 
\end{align}
Note that the oracles for $\qA_1,...,\qA_{q+1}$ are unchanged.

In $\hyb_2$, the only information about $s$ given to $\qA$ is through the oracle $G_\bot^s$, which returns $\bot$ if and only if the query has prefix $s$. 
Thus, by the well-known bound for unstructured search~\cite{BBBV97}, the probability that  measurement of its random query happens to have prefix $s$ is $O(Q^2\cdot 2^{-\secp})$.
Thus, \Cref{lem:O2H} implies 
\[
|\Pr[\hyb_2=1]-\Pr[\hyb_3=1]|\le 2Q\cdot O(Q\cdot 2^{-\secp/2})=\negl(\secp).
\]

\paragraph{$\underline{\hyb_3}$}: This hybrid is identical to $\hyb_1$ except that the challenger chooses 
$r_k\gets \bit^\secp$, 
 $\theta_k\gets \zo^\secp$, and 
$u_k\gets \bit^\secp$ 
for $k\in[q]$ at the beginning, 
 sets $\ct_k$ as 
$\ct_k \seteq (r_k,\ket{x_{k,\theta_k}},u_k,H(x_k)\oplus \msg_{k,\coin})$ for $k\in[q]$, 
and the oracle $G$ for $\qA_1,\ldots,\qA_{q+1}$ is replaced with $G^{s,\bar{r},\bar{u},\bar{\theta}}$  defined below:\footnote{We use the following shorthand notation: $\bar{r}=(r_1,\ldots,r_q)$, $\bar{u}=(u_1,\ldots,u_q)$, $\bar{\theta}=(\theta_1,\ldots,\theta_q)$.}  
\begin{align}
G^{s,\bar{r},\bar{u},\bar{\theta}}(s^\prime,r)=
\begin{cases}
u_k\oplus \theta_k~&\text{if}~s^\prime=s\text{~and~}r=r_k \text{~for~some~}k\in[q]\\
G(s^\prime,r)~&\text{else}
    \end{cases}. 
\end{align}
Note that $G^{s,\bar{r},\bar{u},\bar{\theta}}$   is well-defined when the hybrid does not abort since $r_1,...,r_q$ are pairwise distinct.  

This modification does not change the  adversary's view. Thus, we have 
\[
\Pr[\hyb_1=1]=\Pr[\hyb_2=1].
\]

\paragraph{$\underline{\hyb_4}$}: This hybrid is identical to $\hyb_3$ except that the oracle $H$ for $\qA$ is replaced with $H_\bot^{\bar{x}}$ defined below:  
\begin{align}
H_\bot^{\bar{x}}(x^\prime)=
\begin{cases}
\bot~&\text{if}~x^\prime=x_k\text{~for~some~}k\in[q]\\
H(x^\prime)~&\text{else}
    \end{cases}. 
\end{align}
Note that the oracles for $\qA_1,...,\qA_{q+1}$ are unchanged.

Similarly to the hop from $\hyb_1$ to $\hyb_2$, 
by the well-known bound for unstructured search~\cite{BBBV97}, 
the probability that  measurement of its random query happens to be $x_k$ for some $k\in[q]$ is $O(k\cdot Q^2\cdot 2^{-\secp})$. Thus, \Cref{lem:O2H} implies 
\[
|\Pr[\hyb_3=1]-\Pr[\hyb_4=1]|\le 2Q\cdot O(k^{1/2}\cdot Q\cdot 2^{-\secp/2})=\negl(\secp).
\]

\paragraph{$\underline{\hyb_5}$}: This hybrid is identical to $\hyb_4$ except that the challenger chooses 
$\coin\gets \bit$, 
$x_k\gets \zo^\secp$, 
and $v_k\gets \bit^\secp$ 
for $k\in[q]$ at the beginning, 
 sets $\ct_k$ as 
$\ct_k \seteq (r_k,\ket{x_{k,\theta_k}},u_k,v_k)$ for $k\in[q]$, 
and the oracle $H$ for $\qA_1,...,\qA_{q+1}$ is replaced with  $H^{\bar{x},\bar{v},\bar{\msg}}_\coin$ defined below:\footnote{We use the following shorthand notation: $\bar{x}=(x_1,\ldots,x_q)$, $\bar{v}=(v_1,\ldots,v_q)$, and $\bar{\msg}=(\msg_{1,0},\msg_{1,1},\ldots,\msg_{q,0},\msg_{q,1})$.}  
\begin{align}
H^{\bar{x},\bar{v},\bar{\msg}}_\coin(x^\prime)=
\begin{cases}
v_k\oplus \msg_{k,\coin}~&\text{if}~x^\prime=x_k\text{~for~some~}k\in[q]\\
H(x^\prime)~&\text{else}
    \end{cases}. 
\end{align}
Note that  $H^{\bar{x},\bar{v},\bar{\msg}}_\coin$ is well-defined when the hybrid does not abort since $x_1,...,x_q$ are pairwise distinct.  

This modification does not change the  adversary's view. Thus, we have 
\[
\Pr[\hyb_4=1]=\Pr[\hyb_5=1].
\]

Below, we prove $\Pr[\hyb_5=1]\le 1/2+\negl(\secp)$. 

Using $\qA$, we construct an algorithm $(\widetilde{\cA},\widetilde{\cB},\widetilde{\cC})$ against the augmented search security of BB84 states. 

\begin{description}
    \item[$\widetilde{\cA}^{O_x}$:] Upon receiving $\ket{x_\theta}$, uniformly choose functions
    $G:\bit^\secp\times \bit^\secp\rightarrow \bit^\secp$ and
$H:\bit^\secp\rightarrow \bit^\ell$,\footnote{The description sizes of $G$ and $H$ are exponential; however, this poses no problem, since the augmented search security of BB84 states is an information-theoretic property.}  $i_1,i_2\gets [q+1]$ such that $i_1\ne i_2$, $j\in [q]$, 
    $s\gets \bit^\secp$,  $r_k\gets \bit^\secp$, $u_k\gets \bit^\secp$, and $v_k\gets \bit^\secp$ for $k\in[q]$. For $k\in[q]\setminus\{j\}$,  
    set 
$\ct_k \seteq (r_k,\ket{x_{k,\theta_k}},u_k,v_k)$ where
$x_k\gets \zo^\secp$ and $\theta_k\gets \zo^\secp$, and $\ct_j\seteq (r_j,\ket{x_{\theta}},u_j,v_j)$, which implicitly defines $x_k\seteq x$ and $\theta_k\seteq \theta$.   
Send $1^\secp$ to $\qA$ and simulate the oracles $G_\bot^s$ and $H_\bot^{\bar{x}}$ for $\qA$. Note that 
$\widetilde{\cA}$ can simulate 
$H_\bot^{\bar{x}}$ using $O_x$ even though it does not know $x_j=x$. 
Receive $(\msg_{k,0},\msg_{k,1})_{k\in[q]}$ from $\qA$. Send $(\ct_k)_{k\in[q]}$ to $\cA$. 
Receive a $(q+1)$-partite state $\rho$ over registers $\qreg{R}_1,\ldots,\qreg{R}_{q+1}$ and $(q+1)$ unitaries $(U_1,\ldots,U_{q+1})$ from $\qA$. Let $\cA_i\seteq (\qreg{R}_i,U_i)$ for each $i\in[q+1]$. Send $(G,H,s,(r_k)_{k\in[q]},(u_k)_{k\in[q]},(v_k)_{k\in[q]},(x_k)_{k\in[q]\setminus\{j\}},(\theta_k)_{k\in[q]\setminus\{j\}},\allowbreak (\msg_{k,0},\msg_{k,1})_{k\in[q]}, i_1,j,\qreg{R}_{i_1})$ to $\widetilde{B}$ and $(G,H,s,(r_k)_{k\in[q]},(u_k)_{k\in[q]},(v_k)_{k\in[q]},(x_k)_{k\in[q]\setminus\{j\}},(\theta_k)_{k\in[q]\setminus\{j\}},\allowbreak (\msg_{k,0},\msg_{k,1})_{k\in[q]},i_2,j,\qreg{R}_{i_2})$ to $\widetilde{C}$. 
\item[$\widetilde{\cB}^{O_x}$:]
Receive $\theta=\theta_k$ from the challenger. 
Let $\mathsf{Ext}$ be the algorithm in \Cref{lem:cr_O2H} where $H_0\seteq H^{\bar{x},\bar{v},\bar{\msg}}_0$, $H_1\seteq H^{\bar{x},\bar{v},\bar{\msg}}_1$, and 
$S_j\seteq \{x_j\}$ for all $j\in[q]$. 
Run $x_\qB \gets \mathsf{Ext}^{H_0,H_1,O_{S}}(1^\secp,i_1,j,\qA_{i_1}[G^{s,\bar{r},\bar{u},\bar{\theta}}],\rho[i_1])$ where 
$O_S$ is as in  \Cref{lem:cr_O2H}  and 
$\qA_{i_1}[G^{s,\bar{r},\bar{u},\bar{\theta}}]$ is an algorithm that works similarly to $\qA_{i_1}$ except that its queries to the first oracle are returned by $G^{s,\bar{r},\bar{u},\bar{\theta}}$ which is embedded in its description. 
Note that $\widetilde{\cB}$ knows the full description of $G^{s,\bar{r},\bar{u},\bar{\theta}}$  since it receives $\theta_k=\theta$ from the challenger. 
Note also that $\widetilde{\cB}$ can simulate the oracle $O_{S}$ using $O_x$ even though it does not know $x_j=x$. 
 Output $x_\qB$.
 \item[$\widetilde{\cC}^{O_x}$:]
 Work similarly to $\widetilde{\cB}$ except that $i_1$ is replaced with $i_2$. Let $x_\qC$ be the output. 
\end{description}

Suppose that we have  $\Pr[\hyb_5=1]= 1/2+\epsilon$ for some $\epsilon=1/\poly(\secp)$ for infinitely many $\secp$. 
By \Cref{lem:cr_O2H}, noting that $S_j=\{x\}$,  
we have 
$\Pr[x_\qB=x_\qC=x]\ge \poly(\secp)$ for infinitely many $\secp$. 
This contradicts \Cref{lem:BB84_aug}. Thus, we have $\epsilon=\negl(\secp)$.

Combining the above, we have $\Pr[\hyb_0]\le 1/2+\negl(\secp)$. This completes the proof of \Cref{thm:UE_IND_QROM}.

\subsection{Multi-Copy Secure Unclonable Encryption}\label{sec:pure-ue-def}\label{sec:multi_copy_UE}
We define multi-copy security, where the adversary receives the same ciphertext state multiple times instead of receiving independent encryptions. 
\begin{definition}[Multi-Copy Search Security]\label{def:UE_OT_MCopy_OW}
Let $\UE=\UE.(\KG,\Enc,\Dec)$ be a UE scheme such that $\UE.\Enc$ has classically determined outputs (\Cref{def:classically_determined}).  
We define multi-copy search security of $\UE$ analogously to its multi-challenge counterpart (\Cref{def:UE_OT_MC_OW}) except that the adversary is given $q$ exact copies of the same ciphertext generated under the same randomness, rather than independently generated ciphertexts. 
We define multi-copy indistinguishability security of $\UE$ analogously to its multi-challenge counterpart (\Cref{def:UE_OT_MC_IND}) except that the adversary is 
forced to submit messages $(\msg_{i,0},\msg_{i,1})_{i\in[q]}$ such that $\msg_{1,0}=\msg_{2,0}=\ldots=\msg_{q,0}$ and 
$\msg_{1,1}=\msg_{2,1}=\ldots=\msg_{q,1}$, and given $q$ exact copies of the same ciphertext generated under the same randomness, rather than independently generated ciphertexts.\footnote{The requirement of $\msg_{1,0}=\msg_{2,0}=\ldots=\msg_{q,0}$ and 
$\msg_{1,1}=\msg_{2,1}=\ldots=\msg_{q,1}$ is needed to make sure that all the $q$ ciphertexts given to the adversary encrypts the same message so that the multi-copy security is well-defined.} 
\end{definition}

In this section, we use our purification compiler (\cref{thm:compiler}) to obtain multi-copy secure schemes from our multi-challenge secure schemes.

Observe that in the schemes in \Cref{cor:UE_from_IO,cor:bounded_UE_from_LWE,thm:UE_IND_QROM}, the encryption algorithms have classically determined outputs (\Cref{def:classically_determined}). 
Therefore, we can apply the compiler of \Cref{thm:compiler} to those schemes obtain multi-copy secure schemes.  Thus, we obtain the following theorems. 
\begin{theorem}\label{thm:ue}
    Assuming the existence of polynomially secure indistinguishability obfuscation and one-way functions, there exists an UE scheme with multi-copy search security.
\end{theorem}

\begin{theorem}\label{thm:bounded_ue}
    Assuming the LWE assumption, 
    for any polynomial $q^*$, 
    there exists an UE scheme with $q^*$-bounded multi-copy search security.
\end{theorem}

Since pseudorandom functions and pseudorandom states exist unconditionally in the QROM, we also get the following corollary of \cref{thm:UE_IND_QROM}.

\begin{theorem}\label{thm:ue_ind}
    There exists an UE scheme with multi-copy indistinguishability security in the QROM unconditionally.
\end{theorem}

We omit writing the full descriptions of our our schemes explicitly, since they are evident from \cref{thm:compiler} (also see \Cref{sec:money_PR_compiler} for a similar example).

\newpage
\addrefs

\addcontentsline{toc}{part}{Appendix}
\appendix

\newpage
\section{Omitted Proofs in \Cref{sec:o2h}}

\subsection{Proof of \Cref{lem:success_prob}}\label{app:success_prob}
For the proof of \Cref{lem:success_prob}, we rely on the following lemma. 
\begin{lemma}[{\cite[Corollary 2.4]{C:AKLLZ22}}]\label{lem:cross_term}
    Let $\mat{P}_0$ and $\mat{P}_1$ be projections and $\ket{\psi_0}$ and $\ket{\psi_1}$ be eigenvectors of $(\mat{P}_0+\mat{P}_1)/2$ with eigenvalues $\lambda_0,\lambda_1$, respectively. If $\lambda_0\ne\lambda_1$ and $\lambda_0+\lambda_1\ne 1$, then 
    \[
\bra{\psi_0}\mat{P}_0\ket{\psi_1}=\bra{\psi_0}\mat{P}_1\ket{\psi_1}=0.
    \]
\end{lemma}
Let $\rho$ be as in \cref{lem:cr_O2H} and $\cP_{i,0,D}$ be as in its proof. 
Suppose that we apply $\projimp(\cP_{i,0,D})$  to $\rho[i]$ and let $p_{i}$ be the outcome for all $i\in[q+1]$. For $\bar{p}=(p_{1},\ldots,p_{q+1})$, 
Let $P_{\bar{p}}$ be the probability that the outcomes are equal to $\bar{p}$. 
In the following, whenever we take a sum over $\bar{p}$, it should be understood as a sum over all $\bar{p}$ such that $P_{\bar{p}}\ne 0$. 
We prove the following lemma.
\begin{lemma}\label{lem:good_component}
It holds that
\begin{align*}
\sum_{
    \substack{
    \bar{p}:\\
    \forall i \in [q+1]~ |p_{i}-1/2|>\epsilon/2
    }
    } P_{\bar{p}}\ge \frac{\epsilon}{2}.
\end{align*} 
\end{lemma} \Cref{lem:success_prob} immediately follows from \Cref{lem:good_component} and the first item of \Cref{lem:API} noting that 
$\epsilon/2-\epsilon'\ge 3\epsilon/8$ and 
$\epsilon/2-(q+1)\delta\ge \epsilon/4$.  Thus, we are left to prove \Cref{lem:good_component}. 
\begin{proof}[Proof of \Cref{lem:good_component}]
We can assume without loss of generality  that $H_0,H_1$ are fixed functions and $\rho$ is a fixed pure state rather than being randomly chosen from some distribution by convexity. 
Thus, we write $\rho=\ket{\psi}\bra{\psi}$ using a pure state $\ket{\psi}$. 

For $i\in[q+1]$ and $\coin\in\bit$, let $\mat{P}_{i,\coin}$ be the projection corresponding to  post-selecting the outcome $1$ for $\cP_{i,0,\coin}$, and let $\mat{P}_{i,D}\seteq (\mat{P}_{i,0}+\mat{P}_{i,1})/2$. 
For $i\in[q+1]$, let $\{\ket{\beta_{i,j}}\}_j$
be the orthonormal eigenbasis of $\mat{P}_{i,D}$ where $\ket{\beta_{i,j}}$ has an eigenvalue $p_{i,j}$. 
Then we can express $\ket{\psi}$ as follows where we use shorthand $\bar{j}=(j_1,\ldots,j_{q+1})$:
\begin{align}
   \ket{\psi}=\sum_{\bar{j}}c_{\bar{j}}\bigotimes_{i\in[q+1]}\ket{\beta_{i,j_i}}. 
\end{align}
For $i\in[q+1]$, 
define projective measurements $\mathcal{E}_i= \{\Pi_{i,p_i}\}_{p_i}$ 
where $\Pi_{i,p_i}\seteq\sum_{j:p_{i,j}=p_i}\ket{\beta_{i,j}}\bra{\beta_{i,j}}$. 
Then it is easy to see that $\mathcal{E}_i$ is the projective implementation of $\cP_{i,
D}$. 
Thus, we have 
\begin{align}\label{eq:sum_prob}
\sum_{
\bar{j}:p_{i,j_i}=p_i~\forall i\in[q+1]
}
|c_{\bar{j}}|^2=P_{\bar{p}}.
\end{align}
Decompose $\ket{\psi}=\sum_{k\in\{0,1,.\ldots,q+1\}}\ket{\psi_k}$ where $\ket{\psi_k}$ is defined as follows: 
\begin{align}
\ket{\psi_0}\seteq \sum_{
\substack{
\bar{j}:\\
|p_{1,j_1}-1/2|\le \epsilon/2
}
}c_{\bar{j}}\bigotimes_{i\in[q+1]}\ket{\beta_{i,j_i}},
\end{align}
for $k\in[q]$, 
\begin{align}
\ket{\psi_k}\seteq \sum_{
\substack{
\bar{j}:\\
\forall i\le k~ |p_{i,j_i}-  1/2|>\epsilon/2\\
|p_{k+1,j_{k+1}}-1/2|\le \epsilon/2
}
}c_{\bar{j}}\bigotimes_{i\in[q+1]}\ket{\beta_{i,j_i}},
\end{align}
and 
\begin{align}
\ket{\psi_{q+1}}\seteq \sum_{
\substack{
\bar{j}:\\
\forall i\in[q+1]~ |p_{i,j_i}-  1/2|>\epsilon/2
}
}c_{\bar{j}}\bigotimes_{i\in[q+1]}\ket{\beta_{i,j_i}}.
\end{align}
By \Cref{eq:sum_prob}, we have
\begin{align}
\sum_{
    \substack{
    \bar{p}:\\
   \forall i\in [q+1]~ |p_{i}- 1/2|>\epsilon/2
    }
    } P_{\bar{p}}=\|\ket{\psi_{q+1}}\|^2.
\end{align}

Let 
$\mat{P}\seteq \frac{\bigotimes_{i\in[q+1]}\mat{P}_{i,0}+\bigotimes_{i\in[q+1]}\mat{P}_{i,1}}{2}$. 
We show the following claim:
\begin{claim}\label{cla:cross_term_vanish}
It holds that
\begin{align}
\bra{\psi}\mat{P}\ket{\psi}
   =
\sum_{k\in\{0,1,\ldots,q+1\}}\bra{\psi_k}\mat{P}\ket{\psi_k}.
\end{align}
\end{claim}
\begin{proof}[Proof of \Cref{cla:cross_term_vanish}]
For any $i\in[q+1]$ and 
any $j,j^\prime$ such that $|p_{i,j}- 1/2|\le \epsilon/2$ and 
$|p_{i,j^\prime}-1/2|>\epsilon/2$,  
we have $p_{i,j}\ne p_{i,j^\prime}$ and $p_{i,j}+p_{i,j^\prime}\ne 1$. 
Thus, by \Cref{lem:cross_term},  we have 
$\bra{\beta_{i,j}}
\mat{P}_{i,0}
\ket{\beta_{i,j^\prime}}=\bra{\beta_{i,j}}
\mat{P}_{i,1}
\ket{\beta_{i,j^\prime}}=0$. 
For $k,k'\in\{0,1,\ldots,q+1\}$ such that $k'>k$, the state in the $(k+1)$-th register of $\ket{\psi_k}$ is spanned by $\{\ket{\beta_{k+1,j}}:|p_{k+1,j}- 1/2|\le \epsilon/2\}$ and the state in the $(k+1)$-th register of $\ket{\psi_{k'}}$ is spanned by $\{\ket{\beta_{k+1,j}}:|p_{k+1,j}-1/2|>\epsilon/2\}$. 
By the above, for any $k\ne k'\in\{0,1,\ldots,q+1\}$, we have
\begin{align}
\bra{\psi_k}\bigotimes_{i\in[q+1]}\mat{P}_{i,0}\ket{\psi_{k'}}=0
\end{align}
and 
\begin{align}
\bra{\psi_k}\bigotimes_{i\in[q+1]}\mat{P}_{i,1}\ket{\psi_{k'}}=0.
\end{align}
Thus, we have 
\begin{align}
\bra{\psi_k}\mat{P}\ket{\psi_{k'}}=0.
\end{align}
This immediately implies \Cref{cla:cross_term_vanish}.
\end{proof}

\begin{claim}\label{cla:mid_term_small}
For $k\in\{0,1,\ldots,q\}$, 
it holds that
\begin{align}
\bra{\psi_k}\mat{P}\ket{\psi_k}\le \left(\frac{1}{2}+\frac{\epsilon}{2}\right)\|\psi_k\|^2.  
\end{align}
\end{claim}
\begin{proof}[Proof of \Cref{cla:mid_term_small}]
Let 
\begin{align}
   \mat{P}^{(k)}_0\seteq \overbrace{\mat{I} \cdots \otimes \mat{I}}^{k}\otimes \mat{P}_{k+1,0}\otimes \overbrace{\mat{I} \cdots \otimes \mat{I}}^{q-k}
\end{align}
and 
\begin{align}
   \mat{P}^{(k)}_1\seteq \overbrace{\mat{I} \cdots \otimes \mat{I}}^{k}\otimes \mat{P}_{k+1,1}\otimes \overbrace{\mat{I} \cdots \otimes \mat{I}}^{q-k}.
\end{align}
Then we have 
\begin{align}
\bra{\psi_k}\mat{P}\ket{\psi_k}\le \bra{\psi_k}\left(\frac{\mat{P}_{0}^{(k)}+\mat{P}_{1}^{(k)}}{2}\right)\ket{\psi_k}\le \left(\frac{1}{2}+\frac{\epsilon}{2}\right)\|\ket{\psi_k}\|^2 
\end{align}
where the second inequality follows from the fact that 
the $(k+1)$-th register of $\ket{\psi_k}$ is spanned by $\{\ket{\beta_{k+1,j}}:p_{k+1,j}\le 1/2+\epsilon/2\}$.
\end{proof}

By the assumption of \Cref{lem:cr_O2H}, 
we have
\begin{align}
   \frac{1}{2}+\epsilon
    &\le \bra{\psi}\mat{P}
   \ket{\psi}\\
&
=\sum_{k\in\{0,1,\ldots,q+1\}}\bra{\psi_k}\mat{P}\ket{\psi_k}\\ 
&\le \frac{1}{2}+\frac{\epsilon}{2}+\|\ket{\psi_{q+1}}\|^2.
\end{align}
where the second line follows from \cref{cla:cross_term_vanish} 
and the third line follows from \Cref{cla:mid_term_small}. 
From this, we obtain 
$\|\ket{\psi_{q+1}}\|^2\ge \epsilon/2$. This completes the proof of 
\Cref{lem:good_component}. 
\end{proof} 

\section{Making Encryption and Decrytpion Keys Identical in UE}\label{sec:UE_ek=dk}
In the syntax of UE given in \Cref{def:unclonable_ske}, the encryption and decryption keys are allowed to differ, which deviates from the standard definition. Here, we show that this discrepancy can be generically resolved while preserving multi-copy security.
\begin{theorem}\label{thm:ek=dk}
 If a UE scheme with classical key generation algorithm that satisfies multi-copy search security exists, then there also exists a UE scheme that satisfies multi-copy search security and has identical encryption and decryption keys.
\end{theorem}
\begin{proof}
    Let $\UE=\UE.(\KG,\Enc,\Dec)$ be a multi-copy secure UE scheme.
    Let $\ell$ be the length of the decryption key of $\UE$.  
    Then we construct another UE scheme $\UE'=\UE'.(\KG,\Enc,\Dec)$ as follows. 
    \begin{description}
\item[$\UE'.\KG(1^\secp)$:] $ $
\begin{itemize}
\item Choose $s\lrun \zo^\ell$. 
\item Output $\ek'=s$ and $\dk'=s$.  
\end{itemize}
\item[$\UE'.\Enc(\ek'=s,\msg)$:] $ $
\begin{itemize}
\item Generate $(\ek,\dk)\lrun \UE.\KG(1^\secp)$. 
\item Generate $\ct\lrun\UE.\Enc(\ek,\msg)$.
\item Output $\ct'=(\ct,\dk\oplus s)$. 
\end{itemize}
\item[$\UE'.\Dec(\dk'=s,\ct')$:] $ $
\begin{itemize}
\item Parse $\ct'=(\ct,t)$.
\item Compute $\dk \lrun s\oplus t$. 
\item Output $\msg\gets \UE.\Dec(\dk,\ct)$.
\end{itemize}
\end{description}
It is straightforward to reduce the correctness and multi-copy search security of $\UE'$ to those of $\UE$.  
\end{proof}

It is easy to see that the multi-copy search secure UE scheme constructed in \Cref{sec:multi_copy_UE} has a classical key generation algorithm. 
Thus, we obtain the following corollary. 
\begin{corollary}\label{cor:ue}
    Assuming the existence of polynomially secure indistinguishability obfuscation and one-way functions, there exists an unclonable encryption scheme with multi-copy search security that has identical encryption and decryption keys. 
\end{corollary}
\begin{corollary}\label{cor:ue_bounded}
    Assuming the LWE assumption, for any polynomial $q^*$,  there exists an unclonable encryption scheme with $q^*$-bounded multi-copy search security that has identical encryption and decryption keys. 
\end{corollary}

\newcommand{\etalchar}[1]{$^{#1}$}
\begin{thebibliography}{{\c{C}}GLZR25}

\bibitem[Aar09]{Aar09}
Scott Aaronson.
\newblock Quantum copy-protection and quantum money.
\newblock In {\em 2009 24th Annual IEEE Conference on Computational
  Complexity}, pages 229--242, 2009.

\bibitem[Aar16]{aarlemma}
Scott Aaronson.
\newblock The complexity of quantum states and transformations: From quantum
  money to black holes, 2016.

\bibitem[ABH{\etalchar{+}}25]{cryptoeprint:2025/1880}
Prabhanjan Ananth, Amit Behera, Zikuan Huang, Fuyuki Kitagawa, and Takashi
  Yamakawa.
\newblock Copy-protection from unclonable puncturable obfuscation, revisited.
\newblock Cryptology {ePrint} Archive, Paper 2025/1880, 2025.

\bibitem[ABSV15]{C:ABSV15}
Prabhanjan Ananth, Zvika Brakerski, Gil Segev, and Vinod Vaikuntanathan.
\newblock From selective to adaptive security in functional encryption.
\newblock In Rosario Gennaro and Matthew J.~B. Robshaw, editors, {\em Advances
  in Cryptology -- {CRYPTO}~2015, Part~II}, volume 9216 of {\em Lecture Notes
  in Computer Science}, pages 657--677, Santa Barbara, CA, USA, August~16--20,
  2015. Springer Berlin Heidelberg, Germany.

\bibitem[AC13]{TOC:AarChr13}
Scott Aaronson and Paul~F. Christiano.
\newblock Quantum money from hidden subspaces.
\newblock {\em Theory Comput.}, 9:349--401, 2013.

\bibitem[AFG{\etalchar{+}}12]{Aaronson2012QuantumMoneyACM}
Scott Aaronson, Edward Farhi, David Gosset, Avinatan Hassidim, Jonathan Kelner,
  and Andrew Lutomirski.
\newblock Quantum money.
\newblock {\em Communications of the ACM}, 55(8):84--92, August 2012.

\bibitem[AG25]{AG26concurrent}
Prabhanjan Ananth and Eli Goldin.
\newblock Less is more: On copy complexity in quantum cryptography.
\newblock Cryptology {ePrint} Archive, Paper 2025/1855, 2025.

\bibitem[AHH24]{TCC:AnaHuHua24}
Prabhanjan Ananth, Zihan Hu, and Zikuan Huang.
\newblock Quantum key-revocable dual-regev encryption, revisited.
\newblock In Elette Boyle and Mohammad Mahmoody, editors, {\em TCC~2024: 22nd
  Theory of Cryptography Conference, Part~III}, volume 15366 of {\em Lecture
  Notes in Computer Science}, pages 257--288, Milan, Italy, December~2--6,
  2024. Springer, Cham, Switzerland.

\bibitem[AHU19]{C:AmbHamUnr19}
Andris Ambainis, Mike Hamburg, and Dominique Unruh.
\newblock Quantum security proofs using semi-classical oracles.
\newblock In Alexandra Boldyreva and Daniele Micciancio, editors, {\em Advances
  in Cryptology -- {CRYPTO}~2019, Part~II}, volume 11693 of {\em Lecture Notes
  in Computer Science}, pages 269--295, Santa Barbara, CA, USA, August~18--22,
  2019. Springer, Cham, Switzerland.

\bibitem[AKL{\etalchar{+}}22]{C:AKLLZ22}
Prabhanjan Ananth, Fatih Kaleoglu, Xingjian Li, Qipeng Liu, and Mark Zhandry.
\newblock On the feasibility of unclonable encryption, and more.
\newblock In Yevgeniy Dodis and Thomas Shrimpton, editors, {\em Advances in
  Cryptology -- {CRYPTO}~2022, Part~II}, volume 13508 of {\em Lecture Notes in
  Computer Science}, pages 212--241, Santa Barbara, CA, USA, August~15--18,
  2022. Springer, Cham, Switzerland.

\bibitem[AKL23]{C:AnaKalLiu23}
Prabhanjan Ananth, Fatih Kaleoglu, and Qipeng Liu.
\newblock Cloning games: {A} general framework for unclonable primitives.
\newblock In Helena Handschuh and Anna Lysyanskaya, editors, {\em Advances in
  Cryptology -- {CRYPTO}~2023, Part~V}, volume 14085 of {\em Lecture Notes in
  Computer Science}, pages 66--98, Santa Barbara, CA, USA, August~20--24, 2023.
  Springer, Cham, Switzerland.

\bibitem[AKY25]{ananth2025simultaneous}
Prabhanjan Ananth, Fatih Kaleoglu, and Henry Yuen.
\newblock Simultaneous haar indistinguishability with applications to
  unclonable cryptography.
\newblock In {\em 16th Innovations in Theoretical Computer Science Conference
  (ITCS 2025)}, pages 7--1. Schloss Dagstuhl--Leibniz-Zentrum f{\"u}r
  Informatik, 2025.

\bibitem[ALL{\etalchar{+}}21]{C:ALLZZ21}
Scott Aaronson, Jiahui Liu, Qipeng Liu, Mark Zhandry, and Ruizhe Zhang.
\newblock New approaches for quantum copy-protection.
\newblock In Tal Malkin and Chris Peikert, editors, {\em Advances in Cryptology
  -- {CRYPTO}~2021, Part~I}, volume 12825 of {\em Lecture Notes in Computer
  Science}, pages 526--555, Virtual Event, August~16--20, 2021. Springer, Cham,
  Switzerland.

\bibitem[ALP21]{AP21}
Prabhanjan Ananth and Rolando~L. La~Placa.
\newblock Secure software leasing.
\newblock In Anne Canteaut and Fran{\c{c}}ois-Xavier Standaert, editors, {\em
  Advances in Cryptology -- EUROCRYPT 2021}, pages 501--530, Cham, 2021.
  Springer International Publishing.

\bibitem[AMP25]{ITC:AnaMutPor25}
Prabhanjan Ananth, Saachi Mutreja, and Alexander Poremba.
\newblock Revocable encryption, programs, and more: The case of multi-copy
  security.
\newblock In Niv Gilboa, editor, {\em ITC 2025: 6th Conference on
  Information-Theoretic Cryptography}, volume 343 of {\em Leibniz International
  Proceedings in Informatics (LIPIcs)}, pages 9:1--9:23, Santa Barbara, CA,
  USA, August~16--17, 2025. Schloss Dagstuhl - Leibniz-Zentrum fuer Informatik.

\bibitem[BB20]{BB84}
Charles~H Bennett and Gilles Brassard.
\newblock Quantum cryptography: Public key distribution and coin tossing.
\newblock {\em arXiv preprint arXiv:2003.06557}, 2020.

\bibitem[BBBV97]{BBBV97}
Charles~H Bennett, Ethan Bernstein, Gilles Brassard, and Umesh Vazirani.
\newblock Strengths and weaknesses of quantum computing.
\newblock {\em SIAM journal on Computing}, 26(5):1510--1523, 1997.

\bibitem[BGI14]{PKC:BoyGolIva14}
Elette Boyle, Shafi Goldwasser, and Ioana Ivan.
\newblock Functional signatures and pseudorandom functions.
\newblock In Hugo Krawczyk, editor, {\em PKC~2014: 17th International
  Conference on Theory and Practice of Public Key Cryptography}, volume 8383 of
  {\em Lecture Notes in Computer Science}, pages 501--519, Buenos Aires,
  Argentina, March~26--28, 2014. Springer Berlin Heidelberg, Germany.

\bibitem[BKS{\etalchar{+}}18]{berry2018improved}
Dominic~W Berry, M{\'a}ria Kieferov{\'a}, Artur Scherer, Yuval~R Sanders,
  Guang~Hao Low, Nathan Wiebe, Craig Gidney, and Ryan Babbush.
\newblock Improved techniques for preparing eigenstates of fermionic
  hamiltonians.
\newblock {\em npj Quantum Information}, 4(1):22, 2018.

\bibitem[BL20]{TQC:BroLor20}
Anne Broadbent and S{\'{e}}bastien Lord.
\newblock Uncloneable quantum encryption via oracles.
\newblock In Steven~T. Flammia, editor, {\em 15th Conference on the Theory of
  Quantum Computation, Communication and Cryptography, {TQC} 2020, June 9-12,
  2020, Riga, Latvia}, volume 158 of {\em LIPIcs}, pages 4:1--4:22. Schloss
  Dagstuhl - Leibniz-Zentrum f{\"{u}}r Informatik, 2020.

\bibitem[BS21]{bs21}
Amit Behera and Or~Sattath.
\newblock Almost public quantum coins.
\newblock {\em QIP}, 2021.

\bibitem[BW13]{AC:BonWat13}
Dan Boneh and Brent Waters.
\newblock Constrained pseudorandom functions and their applications.
\newblock In Kazue Sako and Palash Sarkar, editors, {\em Advances in Cryptology
  -- {ASIACRYPT}~2013, Part~II}, volume 8270 of {\em Lecture Notes in Computer
  Science}, pages 280--300, Bengalore, India, December~1--5, 2013. Springer
  Berlin Heidelberg, Germany.

\bibitem[BZ13a]{boneh2013secure}
Dan Boneh and Mark Zhandry.
\newblock Secure signatures and chosen ciphertext security in a quantum
  computing world.
\newblock In {\em Advances in Cryptology--CRYPTO 2013: 33rd Annual Cryptology
  Conference, Santa Barbara, CA, USA, August 18-22, 2013. Proceedings, Part
  II}, pages 361--379. Springer, 2013.

\bibitem[BZ13b]{C:BonZha13}
Dan Boneh and Mark Zhandry.
\newblock Secure signatures and chosen ciphertext security in a quantum
  computing world.
\newblock In Ran Canetti and Juan~A. Garay, editors, {\em Advances in
  Cryptology -- {CRYPTO}~2013, Part~II}, volume 8043 of {\em Lecture Notes in
  Computer Science}, pages 361--379, Santa Barbara, CA, USA, August~18--22,
  2013. Springer Berlin Heidelberg, Germany.

\bibitem[{\c C}G24]{TCC:CakGoy24}
Alper {\c C}akan and Vipul Goyal.
\newblock Unclonable cryptography with unbounded collusions and impossibility
  of hyperefficient shadow tomography.
\newblock In Elette Boyle and Mohammad Mahmoody, editors, {\em TCC~2024: 22nd
  Theory of Cryptography Conference, Part~III}, volume 15366 of {\em Lecture
  Notes in Computer Science}, pages 225--256, Milan, Italy, December~2--6,
  2024. Springer, Cham, Switzerland.

\bibitem[CG25]{cryptoeprint:2025/1197}
Alper Cakan and Vipul Goyal.
\newblock How to copy-protect all puncturable functionalities without
  conjectures: A unified solution to quantum protection.
\newblock Cryptology {ePrint} Archive, Paper 2025/1197, 2025.

\bibitem[CGJL25]{EC:CGJL25}
Orestis Chardouvelis, Vipul Goyal, Aayush Jain, and Jiahui Liu.
\newblock Quantum key leasing for {PKE} and {FHE} with a classical lessor.
\newblock In Serge Fehr and Pierre-Alain Fouque, editors, {\em Advances in
  Cryptology -- {EUROCRYPT}~2025, Part~III}, volume 15603 of {\em Lecture Notes
  in Computer Science}, pages 248--277, Madrid, Spain, May~4--8, 2025.
  Springer, Cham, Switzerland.

\bibitem[{\c C}GLZR24]{TCC:CGLR24}
Alper {\c C}akan, Vipul Goyal, Chen-Da Liu-Zhang, and Jo{\~a}o Ribeiro.
\newblock Unbounded leakage-resilience and intrusion-detection in a quantum
  world.
\newblock In Elette Boyle and Mohammad Mahmoody, editors, {\em TCC~2024: 22nd
  Theory of Cryptography Conference, Part~II}, volume 15365 of {\em Lecture
  Notes in Computer Science}, pages 159--191, Milan, Italy, December~2--6,
  2024. Springer, Cham, Switzerland.

\bibitem[{\c{C}}GLZR25]{CGLR23}
Alper {\c{C}}akan, Vipul Goyal, Chen-Da Liu-Zhang, and Jo{\~a}o Ribeiro.
\newblock Unbounded leakage-resilience and intrusion-detection in a quantum
  world.
\newblock In Elette Boyle and Mohammad Mahmoody, editors, {\em Theory of
  Cryptography}, pages 159--191, Cham, 2025. Springer Nature Switzerland.

\bibitem[{\c{C}}GS25]{ccakan2025public}
Alper {\c{C}}akan, Vipul Goyal, and Omri Shmueli.
\newblock Public-key quantum fire and key-fire from classical oracles.
\newblock {\em arXiv preprint arXiv:2504.16407}, 2025.

\bibitem[CHV23]{chevalier2023semi}
C{\'e}line Chevalier, Paul Hermouet, and Quoc-Huy Vu.
\newblock Semi-quantum copy-protection and more.
\newblock In {\em Theory of Cryptography Conference}, pages 155--182. Springer,
  2023.

\bibitem[CKN{\etalchar{+}}26]{qleak}
Alper Cakan, Fuyuki Kitagawa, Ryo Nishimaki, , Manasi Shingane, and Takashi
  Yamakawa.
\newblock Pke schemes secure against quantum leakage, 2026.

\bibitem[CLLZ21]{C:CLLZ21}
Andrea Coladangelo, Jiahui Liu, Qipeng Liu, and Mark Zhandry.
\newblock Hidden cosets and applications to unclonable cryptography.
\newblock In Tal Malkin and Chris Peikert, editors, {\em Advances in Cryptology
  -- {CRYPTO}~2021, Part~I}, volume 12825 of {\em Lecture Notes in Computer
  Science}, pages 556--584, Virtual Event, August~16--20, 2021. Springer, Cham,
  Switzerland.

\bibitem[CLPY25]{C:CLPY25}
Rohit Chatterjee, Xiao Liang, Omkant Pandey, and Takashi Yamakawa.
\newblock The round complexity of black-box post-quantum secure computation.
\newblock In Yael~Tauman Kalai and Seny~F. Kamara, editors, {\em Advances in
  Cryptology -- {CRYPTO}~2025, Part~IV}, volume 16003 of {\em Lecture Notes in
  Computer Science}, pages 3--35, Santa Barbara, CA, USA, August~17--21, 2025.
  Springer, Cham, Switzerland.

\bibitem[GGM86]{JACM:GolGolMic86}
Oded Goldreich, Shafi Goldwasser, and Silvio Micali.
\newblock How to construct random functions.
\newblock {\em Journal of the {ACM}}, 33(4):792--807, 1986.

\bibitem[Gol01]{Goldreich_2001}
Oded Goldreich.
\newblock {\em Foundations of Cryptography}.
\newblock Cambridge University Press, 2001.

\bibitem[Gol04]{Goldreich_2004}
Oded Goldreich.
\newblock {\em Foundations of Cryptography}.
\newblock Cambridge University Press, 2004.

\bibitem[GVW12]{gorbunov2012functional}
Sergey Gorbunov, Vinod Vaikuntanathan, and Hoeteck Wee.
\newblock Functional encryption with bounded collusions via multi-party
  computation.
\newblock In {\em Advances in Cryptology--CRYPTO 2012: 32nd Annual Cryptology
  Conference, Santa Barbara, CA, USA, August 19-23, 2012. Proceedings}, pages
  162--179. Springer, 2012.

\bibitem[Har13]{harrow2013church}
Aram~W Harrow.
\newblock The church of the symmetric subspace.
\newblock {\em arXiv preprint arXiv:1308.6595}, 2013.

\bibitem[HMNY21]{hiroka2021quantum}
Taiga Hiroka, Tomoyuki Morimae, Ryo Nishimaki, and Takashi Yamakawa.
\newblock Quantum encryption with certified deletion, revisited: Public key,
  attribute-based, and classical communication.
\newblock In {\em Advances in Cryptology--ASIACRYPT 2021: 27th International
  Conference on the Theory and Application of Cryptology and Information
  Security, Singapore, December 6--10, 2021, Proceedings, Part I 27}, pages
  606--636. Springer, 2021.

\bibitem[ISV{\etalchar{+}}17]{PKC:ISVWY17}
Gene Itkis, Emily Shen, Mayank Varia, David Wilson, and Arkady Yerukhimovich.
\newblock Bounded-collusion attribute-based encryption from minimal
  assumptions.
\newblock In Serge Fehr, editor, {\em PKC~2017: 20th International Conference
  on Theory and Practice of Public Key Cryptography, Part~II}, volume 10175 of
  {\em Lecture Notes in Computer Science}, pages 67--87, Amsterdam, The
  Netherlands, March~28--31, 2017. Springer Berlin Heidelberg, Germany.

\bibitem[JLS18]{C:JiLiuSon18}
Zhengfeng Ji, Yi-Kai Liu, and Fang Song.
\newblock Pseudorandom quantum states.
\newblock In Hovav Shacham and Alexandra Boldyreva, editors, {\em Advances in
  Cryptology -- {CRYPTO}~2018, Part~III}, volume 10993 of {\em Lecture Notes in
  Computer Science}, pages 126--152, Santa Barbara, CA, USA, August~19--23,
  2018. Springer, Cham, Switzerland.

\bibitem[KMY24]{cryptoeprint:2024/1564}
Fuyuki Kitagawa, Tomoyuki Morimae, and Takashi Yamakawa.
\newblock A simple framework for secure key leasing.
\newblock Cryptology {ePrint} Archive, Paper 2024/1564, 2024.

\bibitem[KN23]{TCC:KitNis23}
Fuyuki Kitagawa and Ryo Nishimaki.
\newblock One-out-of-many unclonable cryptography: Definitions, constructions,
  and more.
\newblock In Guy~N. Rothblum and Hoeteck Wee, editors, {\em TCC~2023: 21st
  Theory of Cryptography Conference, Part~IV}, volume 14372 of {\em Lecture
  Notes in Computer Science}, pages 246--275, Taipei, Taiwan,
  November~29~--~December~2, 2023. Springer, Cham, Switzerland.

\bibitem[KNP25a]{cryptoeprint:2025/1842}
Fuyuki Kitagawa, Ryo Nishimaki, and Nikhil Pappu.
\newblock Collusion-resistant quantum secure key leasing beyond decryption.
\newblock Cryptology {ePrint} Archive, Paper 2025/1842, 2025.

\bibitem[KNP25b]{kitagawa2025pke}
Fuyuki Kitagawa, Ryo Nishimaki, and Nikhil Pappu.
\newblock Pke and abe with collusion-resistant secure key leasing.
\newblock In {\em Annual International Cryptology Conference}, pages 35--68.
  Springer, 2025.

\bibitem[KPTZ13]{CCS:KPTZ13}
Aggelos Kiayias, Stavros Papadopoulos, Nikos Triandopoulos, and Thomas
  Zacharias.
\newblock Delegatable pseudorandom functions and applications.
\newblock In Ahmad-Reza Sadeghi, Virgil~D. Gligor, and Moti Yung, editors, {\em
  ACM CCS 2013: 20th Conference on Computer and Communications Security}, pages
  669--684, Berlin, Germany, November~4--8, 2013. {ACM} Press.

\bibitem[KY25]{EPRINT:KitYam25}
Fuyuki Kitagawa and Takashi Yamakawa.
\newblock Foundations of single-decryptor encryption.
\newblock {\em {IACR} Cryptol. ePrint Arch.}, page 1219, 2025.

\bibitem[LLQZ22]{TCC:LLQZ22}
Jiahui Liu, Qipeng Liu, Luowen Qian, and Mark Zhandry.
\newblock Collusion resistant copy-protection for watermarkable
  functionalities.
\newblock In Eike Kiltz and Vinod Vaikuntanathan, editors, {\em TCC~2022: 20th
  Theory of Cryptography Conference, Part~I}, volume 13747 of {\em Lecture
  Notes in Computer Science}, pages 294--323, Chicago, IL, USA, November~7--10,
  2022. Springer, Cham, Switzerland.

\bibitem[NC10]{NC10}
Michael~A. Nielsen and Isaac~L. Chuang.
\newblock {\em Quantum Computation and Quantum Information: 10th Anniversary
  Edition}.
\newblock Cambridge University Press, 2010.

\bibitem[PRV24]{PRV24}
Alexander Poremba, Seyoon Ragavan, and Vinod Vaikuntanathan.
\newblock Cloning games, black holes and cryptography.
\newblock Cryptology {ePrint} Archive, Paper 2024/1826, 2024.

\bibitem[PWYZ25]{10.1007/978-981-96-0947-5_7}
Duong~Hieu Phan, Weiqiang Wen, Xingyu Yan, and Jinwei Zheng.
\newblock Adaptive hardcore bit and quantum key leasing over classical channel
  from lwe with polynomial modulus.
\newblock In Kai-Min Chung and Yu~Sasaki, editors, {\em Advances in Cryptology
  -- ASIACRYPT 2024}, pages 185--214, Singapore, 2025. Springer Nature
  Singapore.

\bibitem[SS10]{sahai2010worry}
Amit Sahai and Hakan Seyalioglu.
\newblock Worry-free encryption: functional encryption with public keys.
\newblock In {\em Proceedings of the 17th ACM conference on Computer and
  communications security}, pages 463--472, 2010.

\bibitem[Unr15]{JACM:Unruh15}
Dominique Unruh.
\newblock Revocable quantum timed-release encryption.
\newblock {\em J. {ACM}}, 62(6):49:1--49:76, 2015.

\bibitem[Wat15]{C:Waters15}
Brent Waters.
\newblock A punctured programming approach to adaptively secure functional
  encryption.
\newblock In Rosario Gennaro and Matthew J.~B. Robshaw, editors, {\em Advances
  in Cryptology -- {CRYPTO}~2015, Part~II}, volume 9216 of {\em Lecture Notes
  in Computer Science}, pages 678--697, Santa Barbara, CA, USA, August~16--20,
  2015. Springer Berlin Heidelberg, Germany.

\bibitem[Wie83]{Wie83}
Stephen Wiesner.
\newblock Conjugate coding.
\newblock {\em SIGACT News}, 15(1):78–88, jan 1983.

\bibitem[Zha12a]{Z12}
Mark Zhandry.
\newblock How to construct quantum random functions.
\newblock In {\em 2012 IEEE 53rd Annual Symposium on Foundations of Computer
  Science}, pages 679--687, 2012.

\bibitem[Zha12b]{Z12b}
Mark Zhandry.
\newblock Secure identity-based encryption in the quantum random oracle model.
\newblock In Reihaneh Safavi-Naini and Ran Canetti, editors, {\em Advances in
  Cryptology -- CRYPTO 2012}, pages 758--775, Berlin, Heidelberg, 2012.
  Springer Berlin Heidelberg.

\bibitem[Zha19a]{zhandry2019record}
Mark Zhandry.
\newblock How to record quantum queries, and applications to quantum
  indifferentiability.
\newblock In {\em Advances in Cryptology--CRYPTO 2019: 39th Annual
  International Cryptology Conference, Santa Barbara, CA, USA, August 18--22,
  2019, Proceedings, Part II 39}, pages 239--268. Springer, 2019.

\bibitem[Zha19b]{EC:Zhandry19b}
Mark Zhandry.
\newblock Quantum lightning never strikes the same state twice.
\newblock In Yuval Ishai and Vincent Rijmen, editors, {\em Advances in
  Cryptology -- {EUROCRYPT}~2019, Part~III}, volume 11478 of {\em Lecture Notes
  in Computer Science}, pages 408--438, Darmstadt, Germany, May~19--23, 2019.
  Springer, Cham, Switzerland.

\bibitem[Zha20]{TCC:Zhandry20}
Mark Zhandry.
\newblock {Schr{\"o}dinger}'s pirate: How to trace a quantum decoder.
\newblock In Rafael Pass and Krzysztof Pietrzak, editors, {\em TCC~2020: 18th
  Theory of Cryptography Conference, Part~III}, volume 12552 of {\em Lecture
  Notes in Computer Science}, pages 61--91, Durham, NC, USA, November~16--19,
  2020. Springer, Cham, Switzerland.

\end{thebibliography}
\end{document}